\IfFileExists{./.config.tex}{\input{.config}}{}
\providecommand{\style}{1}
\providecommand{\isanonymous}{0}
\providecommand{\isfullversion}{1}
\providecommand{\iscomments}{1}

\ifnum\style=2
  \documentclass[conference,compsoc]{IEEEtran}
\else
  \documentclass[a4paper]{llncs}
\fi

\ifnum\style=2
  \usepackage{amsmath}
  \usepackage{amsthm}
  \makeatletter
  \renewenvironment{proof}[1][\proofname]{\par
    \normalfont \topsep6\p@\@plus6\p@\relax
    \trivlist
    \item[\hskip\labelsep
          \itshape
      #1\@addpunct{.}]\ignorespaces
  }{\endtrivlist}
  \makeatother
\fi

\makeatletter
\ifnum\style=2\else
  \renewcommand*{\l@author}[2]{}
  \renewcommand*{\l@title}[2]{}
\fi
\makeatother

\PassOptionsToPackage{dvipsnames}{xcolor}
\usepackage[l2tabu,orthodox]{nag}
\usepackage{microtype}
\ifnum\style=2  
  \usepackage[bookmarksdepth=2]{hyperref}
\else
  \usepackage[pagebackref=true,bookmarksdepth=2]{hyperref}
\fi
 \usepackage{booktabs}  
\usepackage{comment}
\usepackage[inline]{enumitem}
\setlist[itemize,1]{leftmargin=*, labelindent=\labelsep}
\usepackage{nicematrix}
\usepackage[normalem]{ulem}
\usepackage{stackrel}
\usepackage{subcaption}
\usepackage{tabularx}
\usepackage{thm-restate}
\usepackage{nicematrix}
\usepackage[pass]{geometry}
\usepackage{orcidlink}
\usepackage{breakcites}
\usepackage{dashbox}
\usepackage{pifont}
\usepackage[algo2e,linesnumbered,ruled,vlined]{algorithm2e}
\ifnum\style=2 \SetAlFnt{\small}\SetAlCapFnt{\small}\SetAlCapNameFnt{\small}\fi
\usepackage{dashbox}
\usepackage{multirow}
\usepackage{adjustbox}  
\usepackage{placeins}   
\usepackage{xfrac}
\usepackage{pgfplots}
\pgfplotsset{compat=1.18}
\ifnum\style=2 \usepackage{flushend}\fi  

\usepackage{lineno}

\ifnum\iscomments=1
  \usepackage[draft]{comments-bugs-todos}
\else
  \usepackage{comments-bugs-todos}
\fi
\NewDocumentCommand\iwoo{s m o}{\IfBooleanTF{#1}{\cmnt*{iwoo}{#2}}{\cmnt{iwoo}{#2}[#3]}}
\NewDocumentCommand\rlai{s m o}{\IfBooleanTF{#1}{\cmnt*{rlai}{#2}}{\cmnt{rlai}{#2}[#3]}}
\NewDocumentCommand\osdnk{s m o}{\IfBooleanTF{#1}{\cmnt*{osdnk}{#2}}{\cmnt{osdnk}{#2}[#3]}}
\AtBeginDocument{%
  \setlength{\marginparsep}{0.1cm}%
  \setlength{\marginparwidth}{\dimexpr
    \paperwidth-\textwidth-1in-\hoffset-\oddsidemargin-\marginparsep-2mm\relax}%
}
\NewDocumentCommand\vcini{s m o}{\IfBooleanTF{#1}{\cmnt*{vcini}{#2}}{\cmnt{vcini}{#2}[#3]}}
\NewDocumentCommand\antg{s m o}{\IfBooleanTF{#1}{\cmnt*{Antonio}{#2}}{\cmnt{Antonio}{#2}[#3]}}
\NewDocumentCommand\claude{s m o}{\IfBooleanTF{#1}{\cmnt*{claude}{#2}}{\cmnt{claude}{#2}[#3]}}
\NewDocumentCommand\contrclaude{s m o}{\IfBooleanTF{#1}{\cmnt*{claude-ans}{#2}}{\cmnt{claude-ans}{#2}[#3]}}
\definecolor{claudecolor}{HTML}{DE7356}
\NewDocumentCommand\claudetext{+m}{\ifnum\iscomments=1\relax{\color{claudecolor}#1}\else#1\fi}
\colorlet{claudedelcolor}{black!35}
\makeatletter
\def\claudedel@gap{\ifx\@let@token\claudetext\hspace{0.3em}\fi}
\def\claudedel@accepted{\ifx\@let@token\@sptoken\ifhmode\unskip\fi\fi}
\long\def\claudedel@on#1{{\color{claudedelcolor}#1}\futurelet\@let@token\claudedel@gap}
\long\def\claudedel@off#1{\futurelet\@let@token\claudedel@accepted}
\NewDocumentCommand\claudedel{+m}{%
  \ifnum\iscomments=1\relax\expandafter\claudedel@on\else\expandafter\claudedel@off\fi{#1}}
\makeatother

\usepackage{ifthen}
\newcommand{\anonymous}[2]{%
\ifthenelse{\equal{\isanonymous}{1}}%
{#1}%
{#2}%
}%

\newcommand{\fullversion}[2]{%
\ifthenelse{\equal{\isfullversion}{1}}%
{#1}%
{#2}%
}%

\newcommand{\ifstyle}[3]{%
\ifthenelse{\equal{\style}{#1}}%
{#2}%
{#3}%
}%

\usepackage{amsmath,amsfonts,amssymb}
\usepackage{xspace}
\usepackage[lambda,landau,operators,probability,sets,adversary,logic,ff,primitives,complexity,oracles,asymptotics,keys]{cryptocode}

\ifnum\style=2
\usepackage{newtxtext,newtxmath} 
\else 
\makeatletter
\newcommand{\hathat}[1]{%
\begingroup%
  \let\macc@kerna\z@%
  \let\macc@kernb\z@%
  \let\macc@nucleus\@empty%
  \hat{\mathchoice%
    {\raisebox{.2ex}{\vphantom{\ensuremath{\displaystyle #1}}}}%
    {\raisebox{.2ex}{\vphantom{\ensuremath{\textstyle #1}}}}%
    {\raisebox{.16ex}{\vphantom{\ensuremath{\scriptstyle #1}}}}%
    {\raisebox{.14ex}{\vphantom{\ensuremath{\scriptscriptstyle #1}}}}%
    \smash{\hat{#1}}}%
\endgroup%
}
\makeatother
\fi

\newcommand{\ourwork}{\ensuremath{\mathtt{Tulva}}\xspace}                  
\newcommand{\ourworktfhe}{\ensuremath{\mathtt{tulva\mhyphen tfhe}}\xspace}  
\newcommand{\thefullversion}{the full version \cite{FULLVERSION}}

\mathchardef\mhyphen="2D
\renewcommand{\vec}[1]{\ensuremath{\mathbf{#1}}\xspace}
\newcommand{\mat}[1]{\ensuremath{\mathbf{#1}}\xspace}
\newcommand{\ring}{\mathcal{R}}
\newcommand{\field}{\mathcal{K}}

\newcommand{\conductor}{\mathfrak{f}}
\newcommand{\expansion}{\mathfrak{c}}
\newcommand{\tensorexpansion}[1]{{\expansion^{\otimes}_{#1}}} 
\newcommand{\rad}{\mathsf{rad}}

\newcommand{\transpose}{\mathtt{T}}
\newcommand{\tr}{\mathtt{T}}
\newcommand{\tuple}[1]{\ensuremath{\left(#1\right)}}
\newcommand{\diag}[1]{\ensuremath{\mathsf{diag}\left(#1\right)}}
\newcommand{\Span}{\mathsf{span}}

\newcommand{\round}[1]{\left\lfloor #1 \right\rceil}
\newcommand{\fracpart}[1]{\left\{ #1 \right\}}

\newcommand{\fracpartfrac}[2]{\fracpart{\frac{#1}{#2}}}
\newcommand{\diff}[1]{{\color{red}#1}}
\newcommand{\DD}{{D}}
\newcommand{\slack}{\xi}

\newcommand{\smax}{s_{\max}}

\newcommand{\tildevec}[1]{\tilde{\vec{#1}}}

\newcommand{\hatvec}[1]{\hat{\vec{#1}}}

\newcommand{\barvec}[1]{\bar{\vec{#1}}}
\newcommand{\barmat}[1]{\bar{\mat{#1}}}

\newcommand{\ddvA}{\ddv_{\mathtt{A}}}

\newcommand{\ddvt}{\ddv_{\mathtt{t}}}
\newcommand{\hdv}{\mathcal{H}}

\newcommand{\cL}{\mathcal{L}}
\newcommand{\cM}{\mathcal{M}}

\newcommand{\FHE}{\mathsf{FHE}}
\newcommand{\tFHE}{\mathsf{thFHE}}

\newcommand{\chisub}{\chi}

\newcommand{\Adp}{\mathsf{Adp}}

\newcommand{\Hyb}{\mathsf{Hyb}}
\newcommand{\hyb}{\mathsf{Hyb}}
\newcommand{\ParDec}{\mathsf{ParDec}}
\newcommand{\Chal}{\mathsf{Chal}}

\newcommand{\Setup}{{\mathsf{Setup}}}
\newcommand{\setup}{\ensuremath{\mathsf{Setup}}\xspace}

\newcommand{\Dec}{{\mathsf{Combine}}}
\newcommand{\Eval}{{\mathsf{Eval}}}

\newcommand{\Share}{{\mathsf{Share}}}

\newcommand{\Leak}{{\mathsf{Leak}}}
\newcommand{\leak}{{\mathsf{leak}}}
\newcommand{\leakdim}{L} 

\newcommand{\PreDecrypt}{\mathsf{Preproc}}

\newcommand{\ctxt}{\mathsf{ctxt}}

\newcommand{\pd}{\mathsf{pd}}

\newcommand{\params}{\ensuremath{\mathsf{par}}\xspace}
\newcommand{\pubparam}{\ensuremath{\mathsf{pp}}\xspace}

\newcommand{\id}{\mathsf{id}} 

\renewcommand{\state}{\mathsf{st}}

\newcommand{\rtag}{\vect}
\newcommand{\pretag}{\tildevec{t}}
\newcommand{\prectxt}{\tildevec{c}}

\newcommand{\LShare}{L_{\mathsf{Share}}}
\newcommand{\Lquery}{L_{\mathsf{Query}}}
\newcommand{\Lshare}{L_{\mathsf{Share}}}
\newcommand{\Lindex}{L_{\mathsf{Party}}}
\newcommand{\Lleak}{L_{\mathsf{leak}}}
\newcommand{\Ltag}{L_{\mathsf{t}}}
\newcommand{\sdvinit}{\sdv_{\mathsf{Init}}}

\newcommand{\Gen}{\mathsf{Gen}}

\newcommand{\veca}{\mathbf{a}}
\newcommand{\vecb}{\mathbf{b}}
\newcommand{\vecc}{\mathbf{c}}
\newcommand{\vecd}{\mathbf{d}}
\newcommand{\vece}{\mathbf{e}}
\newcommand{\vecf}{\mathbf{f}}
\newcommand{\vecg}{\mathbf{g}}

\newcommand{\vecl}{\mathbf{l}}
\newcommand{\vecm}{\mathbf{m}}

\newcommand{\vecr}{\mathbf{r}}
\newcommand{\vecs}{\mathbf{s}}
\newcommand{\vect}{\mathbf{t}}
\newcommand{\vecu}{\mathbf{u}}
\newcommand{\vecv}{\mathbf{v}}
\newcommand{\vecw}{\mathbf{w}}
\newcommand{\vecx}{\mathbf{x}}
\newcommand{\vecy}{\mathbf{y}}
\newcommand{\vecz}{\mathbf{z}}

\newcommand{\veczero}{\mathbf{0}}

\newcommand{\matA}{\mathbf{A}}
\newcommand{\matB}{\mathbf{B}}
\newcommand{\matC}{\mathbf{C}}
\newcommand{\matD}{\mathbf{D}}
\newcommand{\matE}{\mathbf{E}}
\newcommand{\matF}{\mathbf{F}}
\newcommand{\matG}{\mathbf{G}}

\newcommand{\matI}{\mathbf{I}}

\newcommand{\matM}{\mathbf{M}}

\newcommand{\matP}{\mathbf{P}}

\newcommand{\matR}{\mathbf{R}}
\newcommand{\matS}{\mathbf{S}}
\newcommand{\matT}{\mathbf{T}}
\newcommand{\matU}{\mathbf{U}}
\newcommand{\matV}{\mathbf{V}}
\newcommand{\matW}{\mathbf{W}}
\newcommand{\matX}{\mathbf{X}}
\newcommand{\matY}{\mathbf{Y}}
\newcommand{\matZ}{\mathbf{Z}}
\newcommand{\matzero}{\mathbf{0}}

\DeclareRobustCommand{\SIS}[1]%
{\ifthenelse{\equal{#1}{}}%
  {\ensuremath{\mathsf{SIS}}}%
  {\ensuremath{\mathsf{SIS}_{\/#1}}}%
  \xspace%
}
\DeclareRobustCommand{\LWE}[1]%
{\ifthenelse{\equal{#1}{}}%
  {\ensuremath{\mathsf{LWE}}}%
  {\ensuremath{\mathsf{LWE}_{\/#1}}}%
  \xspace%
}

\DeclareRobustCommand{\elLWE}[1]%
{\ifthenelse{\equal{#1}{}}%
  {\ensuremath{\mathsf{LLWE}}}%
  {\ensuremath{\mathsf{LLWE}_{\/#1}}}%
  \xspace%
}
\DeclareRobustCommand{\LLWE}[1]%
{\ifthenelse{\equal{#1}{}}%
  {\ensuremath{\mathsf{LLWE}}}%
  {\ensuremath{\mathsf{LLWE}_{\/#1}}}%
  \xspace%
}

\DeclareRobustCommand{\thLWE}[1]%
{\ifthenelse{\equal{#1}{}}%
  {\ensuremath{\mathsf{thLWE}}}%
  {\ensuremath{\mathsf{thLWE}_{\/#1}}}%
  \xspace%
}

\newcommand{\leco}{\vecu} 

\newcommand{\INDCPA}{{\mathsf{IND}\text{-}\mathsf{CPA}}} 
\newcommand{\INDCPAwL}{{\mathsf{IND}\text{-}\mathsf{CPA}\text{-}\mathsf{wL}}} 
\newcommand{\INDCPAwLd}{{\mathsf{IND}\text{-}\mathsf{CPA}\text{\dbox{\text{-}$\mathsf{wL}$}}}} 
\newcommand{\Adv}{{\mathsf{Adv}}}

\usepackage{booktabs}
\usepackage{tabularx}
\usepackage{array}

\newcolumntype{Y}{>{\raggedright\arraybackslash}X}

\colorlet{hlcolor}{yellow!100!white!50}
\newcommand{\highlight}[2][hlcolor]{%
{\setlength{\fboxsep}{0pt}\colorbox{#1}{\ifmmode$\displaystyle#2$\else#2\fi}}%
}
        
\fullversion{\usepackage{fullpage}}{}

\hypersetup{
    colorlinks,
    linkcolor={red!50!black},
    citecolor={blue!50!black},
    urlcolor={blue!80!black}
}
\allowdisplaybreaks

\ifstyle{2}{
  \theoremstyle{plain}
  \newtheorem{theorem}{Theorem}
  \newtheorem{lemma}{Lemma}
  \newtheorem{corollary}{Corollary}
  \newtheorem{proposition}{Proposition}

  \theoremstyle{definition}
  \newtheorem{definition}{Definition}
  
  \theoremstyle{remark}
  \newtheorem{remark}{Remark}
}{
  \spnewtheorem{Heuristic}{Heuristic}{\bfseries}{\itshape} 
}

\usepackage[capitalize]{cleveref}
\ifstyle{2}{
  \crefname{lemma}{Lem.}{Lems.}%
  \crefname{corollary}{Cor.}{Cors.}%
  \crefname{theorem}{Thm.}{Thms.}%
  \crefname{proposition}{Prop.}{Props.}%
  \crefname{definition}{Def.}{Defs.}%
  \crefname{claim}{Claim}{Claims}%
  \crefname{remark}{Rem.}{Rems.}%
  \crefname{section}{\S}{\S\S}%
  \crefname{subsection}{\S}{\S\S}%
  \crefname{figure}{Fig.}{Figs.}%
  \crefname{table}{Tab.}{Tabs.}%
  \crefname{Heuristic}{Heuristic}{Heuristics}

  \crefformat{section}{\S#2#1#3}%
  \Crefformat{section}{\S#2#1#3}%
  \AtBeginDocument{%
    \crefformat{subsection}{\S#2#1#3}%
    \Crefformat{subsection}{\S#2#1#3}%
  }
}{
  
}

\AddToHook{cmd/appendix/before}{\crefalias{section}{appendix}}
\AddToHook{cmd/appendix/before}{\crefalias{subsection}{appendix}}
\AddToHook{cmd/appendix/before}{} 

\DeclareRobustCommand{\abbrevcrefs}{}

\usepackage[final]{listings}
\lstdefinelanguage{Sage}[]{Python}{morekeywords={True,False,sage,cdef,cpdef,ctypedef,self},sensitive=true}

\usepackage[color={.9 .2 .2},final]{attachfile2}

\DeclareRobustCommand{\cshref}[1]{{\abbrevcrefs\cref{#1}}}

\fullversion{
\title{\ourwork: Practical Asynchronous Threshold FHE \\ \textnormal{(Full version)}}
}{
\title{\ourwork: Practical Asynchronous Threshold FHE}
}

\newcommand{\authorCini}{Valerio Cini}
\newcommand{\orcidCini}{0009-0003-6876-0954}
\newcommand{\authorGuimaraes}{Antonio Guimar{\~a}es}
\newcommand{\orcidGuimaraes}{0000-0001-5110-6639}
\newcommand{\authorLai}{Russell W. F. Lai}
\newcommand{\orcidLai}{0000-0001-9126-1887}
\newcommand{\authorOsadnik}{Micha{\fontencoding{T1}\selectfont\l}\nobreakspace Osadnik}
\newcommand{\orcidOsadnik}{0000-0003-4000-3241}
\newcommand{\authorWoo}{Ivy K. Y. Woo}
\newcommand{\orcidWoo}{0000-0001-8905-1207}
\newcommand{\affBocconi}{Bocconi University, Milan, Italy}
\newcommand{\affIMDEA}{IMDEA Software Institute, Madrid, Spain}
\newcommand{\affAalto}{Aalto University, Espoo, Finland}

\anonymous{
  \ifstyle{2}{}{
  \author{}
  \institute{}
  }
}{
\ifstyle{2}{
  \author{
    \IEEEauthorblockN{%
      \authorCini\orcidlink{\orcidCini},
      \authorGuimaraes\orcidlink{\orcidGuimaraes},
      \authorLai\orcidlink{\orcidLai},
      \authorOsadnik\orcidlink{\orcidOsadnik},
      \authorWoo\orcidlink{\orcidWoo}%
    }
    \IEEEauthorblockA{%
      \affBocconi \\
      \affIMDEA \\
      \affAalto%
    }
  }
}{
  \author{
      \authorCini\inst{1}%
      \orcidlink{\orcidCini}
      \and
      \authorGuimaraes\inst{2}%
      \orcidlink{\orcidGuimaraes}
      \and
      \authorLai\inst{3}%
      \orcidlink{\orcidLai}
      \and
      \authorOsadnik\inst{3}%
      \orcidlink{\orcidOsadnik}
      \and
      \authorWoo\inst{3}%
      \orcidlink{\orcidWoo}
  }
  \institute{
      \affBocconi
      \and
      \affIMDEA
      \and
      \affAalto
  }
}
}

\begin{document}

\maketitle
\ifnum\style=2\else\setcounter{footnote}{0}\fi

\begin{abstract}
  Fully homomorphic encryption (FHE) allows arbitrary computation over encrypted data, but its secret key presents a single point of failure which limits its applications. Threshold FHE patches this gap by distributing the secret key to $K$ parties, requiring a threshold $t \leq K$ of them to cooperate for decryption. We focus on the asynchronous setting where a ciphertext can be decrypted by collecting decryption shares from any $t$ parties who do not need to coordinate or even be online simultaneously. Asynchronous schemes are sought in applications such as privacy-preserving Ethereum Virtual Machines (EVM), processing of long-lived private data with quorum-based decryption and the thresholdisation of any cryptographic primitive in general.

We present \ourwork, a practical threshold FHE scheme which equips a base FHE scheme with an asynchronous threshold decryption mechanism in a modular way. Our scheme simultaneously (1) is compatible with virtually any base FHE scheme and leaves its ciphertexts unchanged apart from an appended random tag, (2) supports a wide range of threshold profiles $(t,K)$ and ring degrees $\geq 256$ with 6--25 bits partial decryption per plaintext bit, and (3) achieves security in the presence of partial decryption queries on challenge ciphertexts based on the standard LWE assumption.

We provide an open-source implementation, benchmark its threshold decryption performance and evaluate the end-to-end costs for typical privacy-preserving EVM use cases. For 16-out-of-512 threshold and a publicly preprocessed ciphertext, computing secret-dependent partial decryption and public share combination requires only 3.7ms and 0.5ms, respectively, on commodity hardware. Compared with the state-of-the-art synchronous scheme [Zyskind et al.,  CCS'25], we replace their ciphertext-independent yet interactive preparation of secret-correlated randomness with a public, ciphertext-dependent 0.5s preprocessing step, eliminate interaction in the online phase while maintaining similar decryption time on the same machine, and support 32$\times$ more total users.

\end{abstract}

\fullversion{%
  \thispagestyle{plain}
  \vfill
  \begingroup
  \footnotesize
  \noindent\rule{0.35\linewidth}{0.4pt}\par\smallskip
  \noindent
  This is the full version of a paper to appear in the proceedings of the 48th IEEE Symposium on Security and Privacy (S\&P 2027) with the same title.
  This version includes additional proofs that are omitted in the proceedings version.
  \par\endgroup
}{}

\fullversion{%
  \clearpage
  \setcounter{tocdepth}{2}
  \tableofcontents
  \thispagestyle{plain}
  \clearpage
}{}


\section{Introduction}

Fully homomorphic encryption (FHE) has rapidly emerged in the past decade as a practically viable tool for computing over private data. 
However, in many real-world applications, storing the secret decryption key on a single device is undesirable.
Threshold FHE distributes the single point of trust by splitting the secret key among $K$ parties, 
such that any subset of at least $t$ parties can jointly decrypt a ciphertext, 
while any coalition of fewer than $t$ parties learns nothing about the underlying message.\footnote{For some schemes, the recovery threshold could be much larger than the corruption threshold.}
The practical significance of threshold FHE is recognised: it constitutes a standalone category in NIST's call for multiparty threshold schemes \cite{NIST:ThresholdCrypto}.

\subsection{Threshold FHE} 

Threshold FHE can be roughly classified into two classes: asynchronous and synchronous/interactive.
For an asynchronous scheme (e.g.~\cite{C:BGGJKR18,AC:BouSch23,EPRINT:BFMMST25}), any of the $K$ shareholders can locally compute a partial decryption without coordinating with other parties. 
Anyone can recover the encrypted message by collecting $t$ partial decryptions in any order and possibly across an extended time period.
In contrast, the decrypting parties of a synchronous scheme (e.g.~\cite{EPRINT:BCDFLO25,CCS:ZZLP25}) must be aware of the set of $t$ participating parties in advance to produce valid partial decryptions. 
Some synchronous schemes are further interactive (e.g.~\cite{CCS:ZZLP25}), i.e.~the partial decryptors must engage in at least two rounds of communication to decrypt the ciphertext, meaning that all decrypting parties must be simultaneously online during the decryption process.\footnote{Some schemes have a synchronous/interactive offline phase and an asynchronous online phase.}

The above differences dictate where a threshold FHE can be deployed, 
and the non-interactivity of asynchronous schemes makes them particularly suited for decentralised, dynamic environments. 
A prominent application is FHEVM \cite{ZamaFhevm}, a confidential smart-contract platform where encrypted account balances and contract states reside on a public ledger. A validator set homomorphically executes smart-contract logic over these ciphertexts, while a separate, rotating committee of decrypting parties is invoked to decrypt only specific outputs (e.g.~a transaction-success bit or an auction's clearing price).
The rise of production FHEVM stacks from companies such as Zama \cite{ZamaFhevm} and Fhenix \cite{Fhenix} highlights the growing demand for this architecture. 
Because contracts trigger decryption requests unpredictably and the decryption committee rotates across epochs, the underlying decryption pattern is inherently asynchronous. Synchronous schemes, which require the decrypting parties to be agreed in advance and/or simultaneously online, are structurally incompatible. 
Beyond FHEVMs, asynchronous threshold FHE further enables generically thresholdising cryptographic primitives \cite{C:BGGJKR18} without imposing constraints on set-agreement or communication that are inherent to synchronous alternatives.

Despite strong market demand, existing threshold FHE schemes struggle to meet all practical constraints simultaneously. 
Constructing schemes that simultaneously maximise the range of threshold profiles while minimising partial decryption size, time and communication rounds is challenging, even without asynchronous decryption. 
This forces a trade-off between asynchrony and efficiency. 
Among the best benchmarked schemes, Noah's Ark~\cite{WAHC:DDE23} is asynchronous only when $t = O(1)$: It requires interactive offline $K$-party preprocessing for $t \geq \omega(1)$, such that each decryption process consumes one preprocessing run; Zyskind et al.~\cite{CCS:ZZLP25} is restricted to the $K$-out-of-$K$ setting and is highly interactive.

Beyond efficiency trade-offs, there is also a subtlety in how threshold FHE security is typically modeled. While the exact security definition of threshold FHE differs greatly from work to work (cf.~\cref{sec:related}), none of the definitions that we are aware of models attack scenarios where the adversary can learn partial decryptions of ``challenge ciphertexts'', whose underlying plaintexts are (supposedly) unknown to the adversary.%
\footnote{This restriction is sometimes explicit (e.g.~\cite{AC:BouSch23}), but is more often implicitly implied by restricting the adversary to corrupt exactly $t-1$ users (e.g.~\cite{C:BGGJKR18,AC:OkaTak25,EPRINT:BFMMST25}), so that any additional partial decryption leads to trivial decryption. In the context of proof, assuming the corruption of exactly $t-1$ users seems to make a proof much easier,  
since any partial decryption would imply recovery of the plaintext, which can then be trivially simulated.}
This is a strict~\cite{PKC:BrzKloWoo26} limitation that is especially significant when the threshold FHE is used to thresholdise another cryptographic primitive. 
Taking threshold signatures as an example, the ability to learn partial decryptions of challenge ciphertexts corresponds to that of learning transcripts of unfinished threshold signing protocols, the latter of which is captured in the strongest \mbox{TS-UF-4} security definition of threshold signatures defined in~\cite{C:BCKMTZ22}.

\subsection{Our Contributions}

We present \ourwork, a practical compiler that turns virtually any\footnote{We require the FHE to satisfy a ``linear decryption'' property (\cref{def:linear-FHE-fine-grained}) which holds for virtually all practical FHE families except CKKS.} FHE into an asynchronous threshold FHE. The instantiation with TFHE, \ourworktfhe, efficiently supports a wide range of threshold profiles\footnote{There is no hard limit and the costs scale gracefully.}, e.g.~any $t \leq K \leq 128$ or any $t$ with $\min\{t, K-t\} \leq 16$ for up to $K = 512$, with $\secpar=128$ bits of estimated security. 
Our headline result is a 16-out-of-512 scheme whose threshold decryption process can be decomposed into a 0.5s public ciphertext preprocessing step, a 3.7ms partial decryption step, and a 0.5ms public share combination step, where the median timing measurements are obtained from an AWS \texttt{c7a.32xlarge} instance with AMD EPYC CPU.

\medskip\noindent\textbf{Modular construction.} 
\ourwork consists of a decryption process which first performs a (public) key-switch operation on the FHE ciphertext, turning it into a threshold public-key encryption (PKE) ciphertext of a compatible threshold PKE scheme, and then performs the threshold decryption procedure of the threshold PKE. 
Such a modular construction is essentially agnostic to the underlying FHE scheme, hence enabling flexible instantiations depending on the use cases.

We highlight two advantages of key-switching to a dedicated threshold PKE scheme:
\begin{enumerate*}
  \item It enables cleaner decryption noise management and parameter selection: The FHE layer only needs to expose a small set of parameters (e.g.~plaintext and ciphertext moduli, decryption noise budget) to the threshold layer, making the parameter selection process of the two layers mostly independent;
  \item It allows us to inject additional (pseudo)randomness to ciphertexts to be decrypted. This is crucial for proving that \ourwork remains secure even when the adversary learns partial decryptions of challenge ciphertexts, under the standard (module) learning with errors (LWE) assumption. 
\end{enumerate*}

\medskip\noindent\textbf{Security given partial decryptions of challenge ciphertexts.} 
Proving this strong security property requires constructing a reduction which simulates partial decryptions that are jointly correlated with all key shares instead of only those of corrupt parties. 
While proof techniques for this scenario have recently been developed in the simpler context of threshold PKE~\cite{AC:CinLaiWoo25}, the threshold FHE setting introduces more complications since many correlated FHE ciphertexts can be derived from a single one.
The injection of additional randomness to ciphertexts allows us to extend the technique of~\cite{AC:CinLaiWoo25} to the threshold FHE setting.  

\medskip\noindent\textbf{Implementation, benchmark and use cases.}
Our TFHE-based instantiation \ourworktfhe consumes ciphertexts produced by the \texttt{tfhe-rs}~\cite{Zama} library directly: a thin compatibility layer repacks and key-switches them into the form expected by our threshold layer, whose performance-critical routines are written in C with \texttt{AVX-512} optimisations and orchestrated from Python.
We benchmark across ring degrees $\varphi \in \{256, 2048\}$ and threshold profiles up to $K = 512$ parties; the small ring minimises single-ciphertext latency, while the large ring packs $\varphi = 2048$ messages per batch for throughput.
We further provide end-to-end benchmarks for four confidential smart contract use cases in an FHEVM (\cref{sec:fhevm}): a sealed-bid auction, a confidential ERC-20 transfer, a dark-pool DEX, and a median price oracle.


\subsection{Related Work}\label{sec:related}

There is a substantial literature on threshold FHE, including~\cite{EC:AJLTVW12,C:BGGJKR18,AC:BouSch23,AC:PasSte24,AC:OkaTak25,EPRINT:BFMMST25}. In the following, we focus on comparing with~\cite{WAHC:DDE23,C:KLLPS26,CCS:ZZLP25} which are the most practically relevant, of which \cite{CCS:ZZLP25} is the only work that makes its implementation publicly available.

Pilvi~\cite{AC:CinLaiWoo25} is a threshold PKE scheme that uses subtractive-set elements as the evaluation points of Shamir sharing, and it is secure even when the adversary obtains partial decryptions of challenge ciphertexts.
\ourwork takes from Pilvi the subtractive-set evaluation points, the security model, and its proof techniques.
Applying Pilvi directly to threshold FHE would mean secret-sharing the FHE secret key itself in the same way, as \cite{C:KLLPS26} does, which requires non-black-box changes to the FHE and its parameters.
Instead, \ourwork keeps the FHE and threshold layers separate: a public key switch moves an FHE ciphertext to a dedicated decryption key, and only that key is secret-shared, so the FHE is used as a black box and its ciphertexts and evaluation are unchanged apart from an appended random tag.
On the proof side, these random tags are what let us adapt Pilvi's proof techniques to homomorphically evaluated ciphertexts (see \cref{sec:overview-analysis}).

In terms of construction, both \ourwork and Noah's Ark \cite{WAHC:DDE23} secret-share a dedicated decryption key\footnote{A subtle difference is that in \ourwork a ciphertext is key-switched to one under the decryption key while in Noah's Ark \cite{WAHC:DDE23} it is bootstrapped.}, while the schemes of \cite{CCS:ZZLP25,C:KLLPS26} perform threshold decryption ``in place''.
All three schemes, \ourwork and \cite{WAHC:DDE23,C:KLLPS26}, use noise flooding to mask leakages of the FHE decryption noise and the secret key share.
The scheme of \cite{CCS:ZZLP25} avoids noise flooding using a rounding protocol based on multiparty computation (MPC), at the cost of having online interactions.

\ourwork and \cite{EPRINT:BFMMST25,C:KLLPS26} consider the asynchronous model.
Noah's Ark \cite{WAHC:DDE23} is only asynchronous for $t = O(1)$. For $t \geq \omega(1)$, each threshold decryption consumes a ciphertext-independent $K$-party interactive offline preprocessing, although the online phase is asynchronous.
The scheme of \cite{CCS:ZZLP25} requires interactions among all $K$ parties in both offline and online phases.

There are mainly two formal models of threshold FHE which define the notion either (1) as a standalone primitive (e.g.~this work and \cite{C:KLLPS26}) or (2) as MPC protocols for key generation and threshold decryption associated to a base FHE scheme (e.g.~\cite{WAHC:DDE23,CCS:ZZLP25}). 
Typically, security in the former model is described via security experiments for individual properties, and in the latter via ideal functionalities.
All compared works use distinct security definitions, and we are not aware of any result that connects game-/simulation-/UC-based security definitions for threshold FHE, making them difficult to compare. 
This work inherits the strong security model for threshold PKE considered in \cite{AC:CinLaiWoo25,PKC:BrzKloWoo26} where the adversary can non-maximally corrupt $|\cdv| < t$ parties and make less than $t - |\cdv|$ partial decryption queries on each challenge ciphertext. 
The model in \cite{C:KLLPS26} is weaker \cite{PKC:BrzKloWoo26} in the sense that it forces $|\cdv| = t-1$. 
Understanding the correspondence of this property in UC-based models~\cite{WAHC:DDE23,CCS:ZZLP25} is beyond the scope of this work.


\section{Technical Overview}

We overview our construction and the security proof. 
Throughout, we work over a cyclotomic field $\field = \QQ(\zeta)$ and its ring of integers $\ring = \ZZ[\zeta]$ with norms measured in the coefficient embedding. 
The fractional part of $x \in \field$ wrt.~the coefficient embedding is written as $\fracpart{x} \coloneqq x - \round{x}$. 

\subsection{Requirements on Underlying FHE}\label{sec:overview-FHE-requirements}

Our transformation assumes an FHE scheme $\FHE$ to satisfy a structural property and a security property.

\medskip\noindent\textbf{Structural property.} We assume that $\FHE$ has ``linear decryption''\footnote{The formal definition involves a ``slack'' factor $\xi$. We assume $\xi=1$ in this overview for simplicity.} (\cref{def:linear-FHE-fine-grained}), i.e.~the following (typical) properties: 
\begin{itemize}
    \item The decryption key is a vector $\sk_\FHE = \hat{\vec{s}} \in \ring_q^{\hat{n}}$. 
    \item There exists a deterministic transformation $\PreDecrypt$ that turns a ciphertext into a tuple 
    \[
        (\hat{\vec{c}}_0, \hat{c}_1) \gets \PreDecrypt(\ctxt_\FHE) \in \ring_q^{\hat{n}} \times \ring_q. 
    \]
    In particular, we do \emph{not} require the FHE scheme to operate over the same ring $\ring_q$ (e.g.~it can work over entirely different rings or moduli), as long as the above $\PreDecrypt$ procedure exists.
    \item If a ciphertext $\ctxt_\FHE$ is intended to encrypt a message $\mu \in \ring_p$ (potentially after some homomorphic evaluations), then it holds that 
    \begin{align}\label{eq:overview-linear-dec}
        \hat{c}_1 - \hat{\vec{s}}^\transpose \cdot \hat{\vec{c}}_0 = \round{\frac{q}{p}} \cdot \mu + \hat{e} \bmod q 
        \quad \text{with} \quad 
        \norm{\hat{e}}_\infty \leq B,
    \end{align}
    for some noise budget $B$ depending on the FHE scheme. We call $\hat{e}$ the FHE decryption noise.
    This implies $\round{ \frac{p}{q} \cdot (\hat{c}_1 - \hat{\vec{s}}^\transpose \cdot \hat{\vec{c}}_0) } = \mu \bmod p$ when $q/2p > B + \fracpart{\frac{q}{p}}$. 
\end{itemize}
Virtually all lattice-based constructions of FHE have linear decryption. For most of them, $\PreDecrypt$ can be implemented with well-known and relatively inexpensive procedures, such as Ring Switching~\cite{SCN:GHPS12} and Packing Key Switching~\cite{AC:CGGI16,ICALP:MicSor18,ACNS:CDKS21}. In fact, some schemes (e.g.~BFV~\cite{C:Brakerski12,EPRINT:FanVer12}) already often work with ciphertexts satisfying our requirements. An exception is CKKS~\cite{AC:CKKS17}, which, being an approximate scheme, would not immediately satisfy our correct decryption condition, requiring more complex procedures. 
We discuss the potential instantiations of $\PreDecrypt$ and their compatibility with different FHE schemes in \cref{sec:implementation}.
For our proof-of-concept implementation, we consider the TFHE scheme~\cite{AC:CGGI16}, as provided by the \texttt{tfhe-rs} library~\cite{Zama}, and implement $\PreDecrypt$ using bootstrappings (to increase moduli) and packing key switching.

\medskip\noindent\textbf{Security property.}
We assume that $\FHE$ has IND-CPA-with-leakage security
(\cref{def:FHE-with-leakage-security}). Informally, this means that $\FHE$ remains IND-CPA secure even when the adversary has access to a leakage oracle that, on input an honestly generated ciphertext of its choice, returns the corresponding FHE decryption noise $\hat{e}$ masked
by additional noise. We show (\cref{thm:leakage-security}) that any IND-CPA-secure FHE scheme with linear decryption satisfies this notion,
provided the masking noise is super-polynomially large. We leave the case of polynomially large masking noise to future work.

Apart from these structural and security requirements, our
transformation only requires the modulus-to-noise ratio $q/B$ of the underlying FHE scheme to be sufficiently large (super-polynomial in our current analysis).
The transformation is otherwise oblivious to the homomorphic evaluation procedures of the underlying FHE scheme. Hence, it preserves any optimisations implemented in the underlying FHE scheme.

\subsection{From FHE to Threshold FHE}

Next, we describe our generic compiler, \ourwork, that turns any FHE scheme satisfying the requirements in \cref{sec:overview-FHE-requirements} into a threshold FHE scheme.
Conceptually, our transformation is to pair the FHE with an external linearly-homomorphic threshold PKE (tPKE) scheme which takes care of threshold decryption. 
The basic idea is to encrypt the FHE secret key under the tPKE, and homomorphically evaluate (the linear part) of the FHE's decryption function on an FHE ciphertext, i.e.~perform a key-switching operation to turn an FHE ciphertext into a tPKE ciphertext. The partial decryption procedure of the tPKE can then be applied.
Looking ahead, we will instantiate the tPKE with the classic threshold Regev-PKE scheme. 
The construction involves two moduli $p \leq q$ for plaintexts and  threshold layer ciphertexts, respectively.
It is agnostic to the FHE layer modulus, which is left implicit. 

Our modular framework offers two major advantages: 
(1) The parameters for the FHE and tPKE layers are almost decoupled: only the noise budget, the plaintext and ciphertext moduli and some dimensions of the former are exposed to the latter, allowing a simpler parameter selection process. 
(2) Our security proof can be decomposed into distinct parts concerning solely the FHE or the tPKE components, which allows us to borrow techniques for tPKE to argue strong threshold FHE security (see \cref{sec:overview-analysis}).

\medskip\noindent\textbf{Key Generation.}
A trusted dealer first generates the FHE key pair $(\pk_\FHE, \sk_\FHE)$, where $\sk_\FHE = \hat{\vec{s}}$.
It then generates a key-switching key $(\mat{A}, \vec{b})$ which allows anyone to turn an FHE ciphertext into a tPKE ciphertext, as follows:
Sample the matrices $\mat{A} \sample \ring_q^{n \times m}$ and $\mat{R} \sample \ring_q^{t \times n}$ uniformly at random.
Compute $\vec{r}^\transpose = \vec{v}_0^\transpose \cdot \mat{R} \bmod q$, where $\vec{v}_0 = (1,0,\ldots,0)^\tr$ denotes the $t$-dimensional first unit vector, i.e.~$\vec{r}^\transpose$ is the first row of $\mat{R}$.
Sample a short error vector $\vec{e} \sample \ring^m$.
Compute 
\begin{align}\label{eq:overview-vecb}
    \vec{b} \coloneqq \vec{r}^\transpose \cdot \mat{A} + \vec{e}^\transpose + (\hat{\vec{s}}^\transpose, \vec{0}_h^\transpose) \cdot \matG \bmod q,
\end{align}
where $\mat{G}$ is an appropriate gadget matrix \cite{EC:MicPei12}.
The public key is $\pk = (\pk_\FHE, \mat{A}, \vec{b})$.
The secret key $\sk_k$ for the $k$-th party is 
the $k$-th Shamir's secret-share of the vector $\vecr$,
which can be expressed as the $k$-th row of $\mat{V} \cdot \mat{R} \bmod q$, where $\mat{V} \in \ring^{K \times t}$ is the Vandermonde matrix with the evaluation points as entries.%
\footnote{Our security proof will not directly rely on the security of secret-sharing, but via analysing the structure of the Vandermonde matrix.}
The norm growth issue can be handled by carefully chosen evaluation points, similar to other recent works \cite{AC:CinLaiWoo25,C:KLLPS26} (see \cref{sec:overview-analysis} ``optimisations'').
\begin{remark}[Distributed key generation]
    Removing the trusted-dealer assumption requires generating the underlying FHE key material, our key-switching key, and the secret shares in a distributed manner.
    Generic secure MPC can in principle implement $\kgen$, but an efficient distributed key generation protocol must address both distributed key generation for the underlying FHE scheme and joint sampling of the Gaussian error in the key-switching key. We leave its design and evaluation to future work.
\end{remark}

\medskip\noindent\textbf{Encryption and Homomorphic Evaluation.}
The encryption and homomorphic evaluation algorithms are almost exactly those of $\FHE$, except that each ciphertext $\ctxt = (\ctxt_\FHE, \rtag)$ is accompanied by an independently uniformly random tag $\rtag \sample \ring_q^h$, which is generated during encryption and refreshed every time a homomorphic operation is performed. 
This seemingly useless tag injects additional randomness to a ciphertext and is crucial for proving the strong security property in \cref{sec:overview-security-model}.
In practice, $\rtag$ could e.g.~be derived by hashing $\ctxt_\FHE$ with a hash function (fixed in the public parameters) that can reasonably be modelled as a random oracle (cf.~\cref{rem:alt-security-def,rem:pseudorandom-tag}).
We note again that fresh FHE ciphertexts $\ctxt_\FHE$ and those undergoing homomorphic evaluations may be over rings different from $\ring_q$, 
as long as they can be transformed to take the form $(\hat{\vec{c}}_0, \hat{c}_1) \in \ring_q^{\hat{n}} \times \ring_q$ 
by the procedure $\PreDecrypt$.

\medskip\noindent\textbf{Partial Decryption.}
For a ciphertext $\ctxt = (\ctxt_\FHE, \rtag)$, 
any party $k$ who wants to contribute its partial decryption first runs the deterministic procedure $\PreDecrypt$ to obtain $(\hat{\vec{c}}_0, \hat{c}_1) \gets \PreDecrypt(\ctxt_\FHE)$, 
then performs (part of) a key-switching operation to obtain%
\footnote{\label{fot:keygen}In practice, the computation of $\vecc_0$ can alternatively be done by a designated party, e.g.~whoever is performing the last homomorphic evaluation step on the ciphertext.}
\begin{align}\label{eq:overview-c_0}
    \vec{c}_0 \coloneqq \mat{A} \cdot \mat{G}^{-1} \begin{pmatrix} \hat{\vec{c}}_0 \\ \rtag \end{pmatrix} \bmod q,
\end{align}
where $\matG^{-1}(\cdot)$ is the ``gadget-decomposition'' operation \cite{EC:MicPei12}. 
Each party $k$ then releases the partial decryption 
\[
    \pd_k = \vec{s}_k^\tr \cdot \vec{c}_0 + e_k \bmod q
\] 
where $e_k \sample \ring$ is a fresh masking noise.

\begin{remark}
    In applications where ciphertexts are processed by a dedicated entity, e.g.~a server, before they are partial-decrypted, 
    $\PreDecrypt$ can be run only once by this entity, instead of multiple times by individual partial decryptors. 
\end{remark}

\medskip\noindent\textbf{Combining.}
After collecting $t$ partial decryptions $(\pd_k)_{k \in T}$ from a set $T \subseteq_t [K]$ of parties, the ciphertext can be fully decrypted by a public share combination procedure. 
First, the combiner completes the key-switching operation, i.e.~computes $c_1 \coloneqq \vec{b}^\transpose \cdot \mat{G}^{-1} \begin{pmatrix} \hat{\vec{c}}_0 \\ \rtag \end{pmatrix} \bmod q$.
Next, it combines the set of partial decryptions $(\pd_k)_{k \in T}$ into
\[
    \pd \coloneqq \vec{v}_0^\transpose \cdot \mat{V}_T^{-1} \cdot (\pd_k)_{k \in T} \bmod q,
\]
where $\matV_T$ is the submatrix of $\matV$ corresponding to the set $T$ of parties.
Equivalently, $\vecv_0^\tr \cdot \matV_T^{-1}$ are the recovery coefficients of the secret-sharing.
Then, it removes the encryption in the threshold layer by computing
 \ifstyle{2}{\begin{equation*}
\hat{\mu} \coloneqq \round{ \hat{c}_1\! -\! (c_1\! - \!\pd) } \bmod q
    \quad\text{then} \quad 
    \mu \coloneqq \round{ \frac{p}{q}\! \cdot\! \hat{\mu} } \bmod p.
\end{equation*}}{
\begin{align*}
    \hat{\mu} \coloneqq \round{ \hat{c}_1 - (c_1 - \pd) } \bmod q
    && \text{then} && 
    \mu \coloneqq \round{ \frac{p}{q} \cdot \hat{\mu} } \bmod p.
\end{align*}}

\subsection{Analysis}\label{sec:overview-analysis}

\ifstyle{2}{\subsubsection{Main parameter constraint}}{\subsubsection{Main parameter constraint.}}
In our transformation, each partial decryptor $k$ adds a moderate amount of noise $e_k$ to the partial decryption $\pd_k$. 
On the one hand, for security, this noise $e_k$ must be large enough to mask out the FHE decryption noise $\hat{e}$.
This ultimately allows us to argue the security of the resulting threshold FHE scheme based on the IND-CPA-with-Leakage security of the underlying FHE scheme.
On the other hand, for correctness, the modulus-to-noise ratio $q/B$ of the underlying FHE scheme must be large enough to accommodate the extra masking noise $e_k$ in addition to the existing FHE decryption noise $\hat{e}$.

\ifstyle{2}{\subsubsection{Security model}}{\subsubsection{Security model.}}\label{sec:overview-security-model}
Following recent works which studied and refined tPKE security \cite{AC:CinLaiWoo25,PKC:BrzKloWoo26}, we define an analogously refined IND-CPA security model for threshold FHE (\cref{def:tFHE-security}). For a threshold $t$, our model supports corrupting a subset of arbitrary $\leq t-1$ many users, and allows an adversary $\adv$ to access limited partial decryptions even for ciphertexts that it is unauthorised to decrypt, which captures security under thresholdisation.  
In more detail, our IND-CPA model is as follows:
\begin{itemize}
    \item The adversary $\adv$ statically corrupts a subset $\cdv$ of size $\leq t-1$ (instead of exactly $t-1$) users.
    \item $\adv$ is given the public key and the secret keys of all corrupt users.
    \item $\adv$ has access to a challenge oracle $\Chal\oracle$, an encryption oracle $\enc\oracle$, an evaluation oracle $\eval\oracle$, and a partial decryption oracle $\ParDec\oracle$, which do the following:
    \begin{itemize}
        \item $\Chal\oracle$ encrypts one of the two challenge messages $\mu_0,\mu_1$ chosen by $\adv$.  
        \item $\enc\oracle$ encrypts any message $\mu$ chosen by $\adv$.
        \item $\Eval\oracle$ allows $\adv$ to instruct any homomorphic evaluation on any set of ciphertexts that have been honestly generated (and possibly evaluated) by the experiment.
        \item $\ParDec\oracle$ adaptively answers partial decryption queries of any user on any honestly generated (and possibly evaluated) ciphertext, with a condition: for each ciphertext whose underlying plaintext is different in games $0$ and $1$,
        the number of partial decryption queries is less than $t - |\cdv|$. 
    \end{itemize}
\end{itemize}
Relative to prior works \cite{C:BGGJKR18,AC:BouSch23,AC:OkaTak25,EPRINT:BFMMST25,C:KLLPS26}, the main distinction lies in $\ParDec\oracle$, where we model security even when giving $<t-|\cdv|$ (instead of $0$) partial decryptions for each ``challenge ciphertext''%
\footnote{In the simpler case of threshold PKE, one may easily classify ciphertexts from $\Chal\oracle$ as ``challenges'', which an adversary is unauthorised to decrypt, and those from $\enc\oracle$ as ``non-challenges'', see e.g.~\cite{PKC:DLNPY21,AC:CinLaiWoo25,PKC:BrzKloWoo26}. For threshold FHE, because of the freedom of homomorphic evaluation, alternative classification is needed. In our model, any (possibly evaluated) ciphertext is considered ``challenge'' if the underlying plaintext is different in games $0$ and $1$, else it is ``non-challenge''. This aspect is consistent with the model of \cite{AC:BouSch23}.} to $\adv$.

\ifstyle{2}{\subsubsection{Security analysis}}{\subsubsection{Security analysis.}}
We show (\cref{thm:tfhe-security}) that under appropriate parameter constraints, if the underlying FHE has linear decryption and is IND-CPA-with-Leakage secure, and if the standard LWE assumption holds, then our threshold FHE transformation satisfies the above IND-CPA security. Since IND-CPA-with-Leakage is implied by plain IND-CPA up to appropriate masking noise (\cref{thm:leakage-security}), overall, our threshold FHE can be based on any IND-CPA secure FHE with linear decryption and the LWE assumption\fullversion{ (\cref{cor:tFHE-from-FHE})}{}. 

Our security proof goes via a hybrid argument consisting of two main parts, corresponding to the ``threshold layer'' and the ``FHE layer'' of the construction. 

\medskip\noindent\textbf{Threshold Layer.}
The first part \fullversion{($\hyb_{b,0} \approx_c \hyb_{b,1}$ via \cref{lem:thLWE-hybrid-step} and $\hyb_{b,1} \equiv \hyb_{b,2}$ via \cref{lem:H1-H2})}{($\hyb_{b,0} \approx_c \hyb_{b,1}$ and $\hyb_{b,1} \equiv \hyb_{b,2}$)} unwraps the threshold layer.
We hop from the real experiment, where the challenge oracle $\Chal\oracle$ encrypts message $\mu_b$, to a hybrid experiment with the following main changes:
\begin{itemize}
    \item The vector $\vecb$ in the public key which encodes the FHE secret key $\hatvec{s}$ (ref. \cref{eq:overview-vecb}) is swapped to uniformly random. In exchange, the correlation between $\vecb$ and all partial decryptions $\pd_k$'s (which arises naturally from decryption correctness) is embedded in the $\pd_k$'s. 
    \item For ciphertexts whose underlying plaintexts differ in games $0$ and $1$ (i.e.~the ``challenges''), their partial decryptions are swapped to uniformly random. 
\end{itemize} 
The proof ideas for this part borrow those in \cite{AC:CinLaiWoo25} established for tPKE: 
The first change is by carefully manipulating the underlying linear correlations and invoking LWE along the way; 
the second change relies on the algebraic structure of the Vandermonde matrix, together with the leverage that only $\leq t-|\cdv|$ partial decryptions of each ``challenge'' are handed to the adversary. 

Relative to \cite{AC:CinLaiWoo25}, however, we need to handle a new complication: 
During partial decryption, the component $(\mat{A}, \vec{b})$ is multiplied with $\vecx \coloneqq \mat{G}^{-1} \begin{pmatrix} \hat{\vec{c}}_0 \\ \rtag \end{pmatrix} \bmod q$ (\cref{eq:overview-c_0}), where the top part $\mat{G}^{-1}(\hat{\vec{c}}_0)$ is (the encoding of) an arbitrary honestly generated (and possibly evaluated) FHE ciphertext. Since we do not assume any specific FHE scheme, the experiment has no guarantee on the distribution of $\hatvec{c}_0$. Yet, it needs to adaptively answer partial decryption queries on these -- possibly adversarially influenced -- ciphertexts.
This is in contrast to the tPKE setting of \cite{AC:CinLaiWoo25} where $\vecx$ is simply the encryption randomness which the security experiment has full control over.

We tackle the above issue as follows. First, we utilise the random tag $\vect$ appended to the threshold FHE ciphertext, which serves as a rerandomisation (in the form of $\matG^{-1}(\vect)$), allowing us to invoke standard LWE w.r.t. $\matG^{-1}(\vect)$ when analysing the modified partial decryptions.
Second, to handle the adaptivity issue, we rely on standard LWE and a super-polynomial masking error, with which we argue the distribution of partial decryptions in an information-theoretic way.
The formal security proof, which we provide in \fullversion{the appendices}{\thefullversion}, involves reestablishing the counterparts of all lemmas in \cite{AC:CinLaiWoo25} with the above modifications incorporated.

\medskip\noindent\textbf{FHE Layer.}
The second part handles partial decryptions for the ciphertexts which an adversary is authorised to fully decrypt and recover the plaintexts. 
For this, we further rewrite these partial decryptions, this time relying on the linear decryption (\cref{eq:overview-linear-dec}) of the underlying FHE \fullversion{($\hyb_{b,2} \equiv \hyb_{b,3}$ via \cref{lem:H2-H3})}{($\hyb_{b,2} \equiv \hyb_{b,3}$)}, such that they can be simulated given the leakages from the IND-CPA-with-Leakage experiment of the FHE \fullversion{($\hyb_{0,3} \approx_c \hyb_{1,3}$ via \cref{lem:H3-H4})}{($\hyb_{0,3} \approx_c \hyb_{1,3}$)}.

\medskip\noindent\textbf{Optimisations.}
As in \cite{EPRINT:BFMMST25,AC:CinLaiWoo25,C:KLLPS26}, to support $t$-out-of-$K$ decryption, our transformation utilises $\xi$-subtractive sets \cite{C:AlbLai21,AC:KLNO24} of cardinality $K$, where $\xi = \xi(t): \NN \to \ring$ is a slack function. 
In brief, if $\Xi \subset \ring$ is a $\xi$-subtractive set, then for any $t$-subset $T \subseteq \Xi$ which defines a Vandermonde matrix $\mat{V}_T$, it holds that $\xi(t) \cdot \mat{V}_T^{-1}$ is integral, i.e.~has entries in $\ring$.  
This allows choosing the modulus $q$ such that $\xi(t) \cdot \mat{V}_T^{-1}$ has low norm relative to $q$.
Since the norm of $\xi(t) \cdot \mat{V}_T^{-1}$ is empirically only moderate-sized and especially small with small thresholds, 
this means the modulus $q$ can also be kept practically small, leading to the concrete efficiency of our threshold FHE.
To further compress partial decryption size, we also adopt the standard bit-dropping technique, where the errornous low-order bits are dropped.

Finally, in our formal security proof, we use an intermediate computational assumption called Leaky LWE to modularise proof steps and ease presentation. In more detail, \cite{CiC:LaiSwaWoo25} and \cite{EPRINT:HKMNS26} gave reductions from standard LWE to the semi-adaptive and fully-adaptive versions of Leaky LWE, respectively, both in the polynomial modulus-to-noise ratio regime. The security proof of \ourwork \fullversion{(\cref{app:core2,app:core1-alt}) }{}assumes fully-adaptive Leaky LWE (implied by LWE) and uses proof techniques adapted from the semi-adaptive reduction of \cite{CiC:LaiSwaWoo25} in a non-black-box way.


\section{Preliminaries}\label{sec:prelim}

Let $\secpar \in \NN$ and $\delta \in \NN$ denote the computational and statistical security parameters respectively. 
For any $n \in \NN$, we write $\poly[n]$ and $\negl[n]$ for the sets of all polynomials and negligible functions in $n$, respectively. 
For $k,n \in \NN$, $k \leq n$, we write $[n]$ for $\set{1,\ldots,n}$ and $[k,n]$ for $\set{k,\ldots,n}$.
If $S$ is a set, we write $x \sample S$ for sampling a uniformly random element from $S$.
If $\ddv$ is a distribution over $S$, denoted as $\ddv \sim S$, we write $x \sample \ddv$ for sampling a random element from $S$ according to the distribution $\ddv$. 
We use bold capital and lower-case letters, e.g.~$\mat{A}$ and $\vec{b}$, to denote matrices and vectors, respectively. 

Let $\mathcal{K} = \QQ(\zeta)$ be a cyclotomic field of conductor $\conductor$ and degree $\varphi = \varphi(\conductor)$ and $\ring = \ZZ[\zeta]$ its ring of integers. 
For $q \in \NN$, define the quotient ring $\ring_q \coloneqq \ring / q\ring$. 
We denote by $\ring^\times$ and $\ring^\times_q$ the sets of units in $\ring$ and $\ring_q$ respectively. 

For any $x \in \RR$, we write $\round{x}$ for the nearest integer to $x$ with an arbitrary fixed tie-breaking rule. 
The fractional part of $x$ is defined as $\fracpart{x} \coloneqq x - \round{x}$. 
This notation is extended naturally to vectors of real numbers.
For $x \in \ZZ$ and $z \in \NN$, we write $\round{x}_{2^z} \coloneqq 2^z \cdot \round{x/2^z}$ for $x$ rounded to the nearest multiple of $2^z$.
For $x \in \ZZ_q$, $\round{x}_{2^z}$ is computed on the representative of $x$ in $[-q/2, q/2)$, i.e.~the one with the smallest absolute value, and the result is read again in $\ZZ_q$.
The notation extends coefficient-wise to ring elements and vectors.
We abuse the notation of modular arithmetic so that it expresses two meanings depending on the context: e.g.~
$a = b \bmod q$ means $a$ and $b$ are congruent modulo $q$ but $a \bmod q = b$ means $b$ is the smallest (in absolute value) of all integers congruent to $a$ modulo $q$.
Relatedly, $\frac{p}{q} \cdot (x \bmod q)$ means multiplying the smallest representative of $x \bmod q$ by $p/q$.
Note that $\fracpartfrac{q}{p} = \frac{q \bmod p}{p}$.

\medskip
\noindent\textbf{Embeddings and norms.} 
We write $\psi: \field \to \QQ^\varphi$ and $\sigma = (\sigma_i)_{i \in \ZZ_\conductor^\times}: \field \to \CC^\varphi$ for the coefficient (wrt.~the standard power basis) and canonical embeddings of $\field$, respectively.
For $x \in \field$, its $\ell_p$-norm is measured via its coefficient embedding, i.e.~$\norm{x}_p \coloneqq \norm{\psi(x)}_p$, and $p = \infty$ when omitted. 
If $\field$ is a power-of-two field, then it is known that $\norm{\psi(x)} \leq \norm{\sigma(x)}$ and $\norm{\sigma(x)}_2 = \sqrt{\varphi} \cdot \norm{\psi(x)}_2$.
The field expansion factor of $\field$ w.r.t. the coefficient $\ell_p$-norm is defined as 
$\expansion_p \coloneqq \max_{a,b\in \field} \norm{a\cdot b}_p/(\norm{a}_p\cdot\norm{b}_p)$.

For a vector $\vec{x} \in \field^m$, the embedding is extended component-wise, e.g.~$\psi(\vec{x}) = (\psi(x_i))_{i \in [m]}$, and 
$\norm{\vecx}_p \coloneqq \norm{\psi(\vecx)}_p$, and $p = \infty$ when omitted.
For a matrix $\matM \in \field^{n \times m}$, we write $\Psi(\matM) \in \QQ^{\varphi n \times \varphi m}$ for the matrix of the $\QQ$-linear map $\vecx \mapsto \matM \vecx$ from $\field^m$ to $\field^n$ in the coefficient embedding, i.e.~$\psi(\matM \vecx) = \Psi(\matM) \cdot \psi(\vecx)$ for all $\vecx \in \field^m$.

For a matrix $\mat{M} = (\vecm_1,\ldots,\vecm_m) \in \field^{n \times m}$, its $\ell_p$-induced-norm is
$\norm{\mat{M}}_{p,p} \coloneqq \sup_{\vec{x} \in \field^m} \norm{\mat{M} \cdot \vec{x}}_p/\norm{\vec{x}}_p$.
Consequently, $\norm{\mat{M}\vec{x}}_{p} \leq \norm{\mat{M}}_{p,p} \cdot \norm{\vec{x}}_p$ and $\norm{\mat{M}\mat{N}}_{p,p} \leq \norm{\mat{M}}_{p,p} \cdot \norm{\mat{N}}_{p,p}$.
Similarly, by default we write $\norm{\matM}$ for $\norm{\matM}_{\infty,\infty} = \norm{\Psi(\matM)}_{\infty,\infty}$.

For a real or complex matrix $\matB$, we write $\smax(\matB)$ for its largest singular value, i.e.~the square root of the largest eigenvalue of $\matB^* \matB$, where $\matB^*$ is the conjugate transpose of $\matB$, and we define its $\ell_2$-induced norm as $\norm{\matB}_{2,2} \coloneqq \smax(\matB)$.
By \cite[Lemma~2]{CiC:LaiSwaWoo25}, for any $\matM \in \field^{n \times m}$, it holds that $\norm{\matM}_{2,2} = \norm{\Psi(\matM)}_{2,2} = \smax(\Psi(\matM))$.
We note that matrices over $\field$ are not identified as complex matrices: a complex matrix arises from $\matM$ only as $\sigma_i(\matM)$ for some $i \in \ZZ_\conductor^\times$, obtained by applying the $i$-th component $\sigma_i$ of the canonical embedding $\sigma$ entrywise.
In particular, $\norm{\matM}_{2,2}$ is in general not $\smax(\sigma_i(\matM))$ for a fixed $i$\fullversion{ (cf.~\cref{lem:induced-norm-embeddings})}{}.

\medskip
\noindent\textbf{Subtractive sets.}
We rely on the notion of subtractive sets, 
proposed by \cite{C:AlbLai21}, further studied by e.g.~\cite{AC:KLNO24}, and used by e.g.~\cite{AC:CinLaiWoo25} for constructing threshold PKE.

\begin{definition}[Subtractive Sets]\label{def:subtractive-sets}
    For a set $\Xi = \{ \mu_1, \ldots, \mu_K \} \subseteq \ring$, any $t \leq K$, and any subset $T \subseteq \set{0} \cup \Xi$,
    let $\mat{V}_T$ (with implicit dependency on $t$) denote the (row-style) $T \times t$ Vandermonde matrix, 
    i.e.~its rows are indexed by $T$ and row $\mu \in T$ is given by $(1,\mu,\ldots,\mu^{t-1})$.

    Let $\slack = \slack(t): [K] \to \ring$. 
    The set $\Xi$ is said to be $\slack$-subtractive \cite{C:AlbLai21} if for any $t \in [K]$ and any size-$t$ subset $T \subseteq \Xi$, 
    it holds that $\slack \cdot \mat{V}_T^{-1} \in \ring^{t \times t}$.

    Consider $\Xi$ with $0 \notin \Xi$. Denote by $\vec{v}_0^\transpose \coloneqq \mat{V}_{\set{0}} = (1, 0,\ldots, 0)$ the first unit vector. We measure the quality of $\Xi$ by the following ``quality measurements'':
    \fullversion{
        \begin{align}\label{eq:subtractive-set-qualities}
        \begin{aligned}
            \rho(t) \coloneqq~& \max_{T \subseteq_t \Xi} \norm{\slack \cdot \vec{v}_0^\transpose \cdot \mat{V}_T^{-1}}, &
            \gamma(t) \coloneqq~& \max_{T^* \subseteq_{t-1} \Xi} \norm{\slack \cdot \mat{V}_{\Xi} \cdot \mat{V}_{\{0\} \cup T^*}^{-1}}_{2,2}, \\
            \theta(t) \coloneqq~& \max_{T^* \subseteq_{t-1} \Xi} \norm{\slack \cdot \mat{V}_{\Xi\setminus T^*} \cdot \mat{V}_{\{0\} \cup T^*}^{-1} \cdot \vecv_0},
        \end{aligned}
    \end{align}
    }{
    \begin{align}\label{eq:subtractive-set-qualities}
        \begin{split}
            \rho(t) \coloneqq~& \max_{T \subseteq_t \Xi} \norm{\slack \cdot \vec{v}_0^\transpose \cdot \mat{V}_T^{-1}},\\
            \gamma(t) \coloneqq~& \max_{T^* \subseteq_{t-1} \Xi} \norm{\slack \cdot \mat{V}_{\Xi} \cdot \mat{V}_{\{0\} \cup T^*}^{-1}}_{2,2}, \\
            \theta(t) \coloneqq~& \max_{T^* \subseteq_{t-1} \Xi} \norm{\slack \cdot \mat{V}_{\Xi\setminus T^*} \cdot \mat{V}_{\{0\} \cup T^*}^{-1} \cdot \vecv_0},
        \end{split}
    \end{align}}
    where the top row of $\mat{V}_{\{0\} \cup T^*}$ corresponds to the index $0$.
    Here $\rho$ and $\theta$ use the default coefficient $\ell_\infty$-norm, whereas $\gamma$ uses the $\ell_2$-induced operator norm $\norm{\cdot}_{2,2}$.
\end{definition}%

The Vandermonde matrix $\mat{V}_\Xi$ defined above 
is equivalent to the share-generating matrix of the $t$-out-of-$K$ Shamir's secret sharing with evaluation points in $\Xi$. The subtractive property ensures that
multiplying $\mat{V}_T^{-1}$ by $\xi$ clears its denominators
for every size-$t$ subset $T\subseteq \Xi$.

The quality measurements provide norm bounds for various scaled interpolation maps and associated coefficient vectors.
These bounds control noise growth and the parameters required by our later analysis.
The quantity $\rho(t)$ bounds the norm of the $\xi$-scaled interpolation coefficients used to reconstruct the value at zero from any $t$ shares.
It therefore controls the noise growth when $\Dec$ combines partial decryptions, and hence correctness (\cref{thm:correctness}).
The quantity $\gamma(t)$ bounds the growth in size when interpolating from the value at zero and the evaluations at $t-1$ points in $\Xi$ to the evaluations at all points in $\Xi$, scaled by $\xi$.
It appears in the conditions of the security theorem (\cref{thm:tfhe-security}).
Finally, fix the evaluations at any $t-1$ points in $\Xi$.
Each remaining evaluation, scaled by $\xi$, can be written as an interpolation coefficient times the value at zero, plus a term determined by the fixed evaluations.
The quantity $\theta(t)$ bounds the sizes of these interpolation coefficients.
It determines how wide the flooding noise $\chi$ must be to hide the FHE decryption noise (cf. \fullversion{\cref{cor:tFHE-from-FHE}}{\thefullversion}).
We always have $\theta(t) \leq \gamma(t)$.
Upper bounds on these quantities for power-of-two fields are given in \fullversion{\cref{lem:xi-subtactive-trig-spaced} in \cref{sec:add-prelim}}{\thefullversion}.

\medskip
\noindent\textbf{Discrete Gaussians.}
We denote by $\DD_{\ring,\chi}^m$ the (centred) discrete Gaussian distribution over $\ring^m$ with parameter $\chi$, i.e.~the
distribution over $\ring^m$ where for all $\vec{x} \in \ring^m$, the probability mass function $\DD_{\ring,\chi}^m(\vec{x}) \propto \rho_\chi(\vec{x})$ where $\rho_\chi(\vec{x}) \coloneqq e^{-\pi\norm{\vec{x}}_2^2/\chi^2}$ is the Gaussian function of width $\chi$, extended to a set $S$ by $\rho_\chi(S) \coloneqq \sum_{\vec{x} \in S} \rho_\chi(\vec{x})$.
More generally, for a positive definite $\Sigma \in \RR^{\varphi m \times \varphi m}$ we denote by $\DD_{\ring,\sqrt{\Sigma}}^m$ the discrete Gaussian over $\ring^m$ with covariance $\Sigma$, i.e.~the distribution with probability mass function proportional to $e^{-\pi \cdot \psi(\vec{x})^\transpose \Sigma^{-1} \psi(\vec{x})}$, which for $\Sigma = \chi^2 \cdot \mat{I}_{\varphi m}$ is $\DD_{\ring,\chi}^m$.
We recall the smoothing parameter~\cite{FOCS:MicReg04}: for a full-rank lattice $\Lambda$ and $\epsilon > 0$, $\eta_\epsilon(\Lambda)$ is the smallest $\chi > 0$ such that $\rho_{1/\chi}(\Lambda^* \setminus \set{\vec{0}}) \leq \epsilon$, where $\Lambda^*$ is the dual lattice. Throughout, $\eta_\epsilon(\ring^m)$ denotes the smoothing parameter of $\ring^m$ identified with $\ZZ^{\varphi m}$ via the coefficient embedding.

\medskip
\noindent\textbf{Lattice assumption.}
We recall the LWE assumption with the matrix $\mat{A}$ sampled from a
specified, not necessarily uniform distribution $\ddvA$.
This formulation accommodates both instances used in our security proof:
uniform matrices and matrices whose columns are gadget decompositions
of uniformly sampled vectors.
Their hardness follows from the usual LWE assumption with uniform
matrices, using \cite[Theorem 4.3]{C:BLMR13} for the latter.

\begin{definition}[$\LWE{}$]\label{def:LWE}
    Let 
    \(\params = (\ring,n,m,t,q,\chi,\ddvA)\) 
    be parametrised by $\secpar$.
    The (decision) $\LWE{\params}$ assumption states that 
    for any PPT adversary $\adv$,
    it holds that $\Adv_{\params,\adv}^{\LWE{}}(\secpar) \coloneqq$
    \begin{align*}
        \abs{
        \ifstyle{2}{%
        \begin{aligned}
            &\condprob{
                \adv(\mat{A}, \mat{B}) = 1
            }{
                \begin{aligned}
                     &\mat{A} \gets \ddvA\\ 
                     &\mat{R} \sample \ring_q^{t \times n} ;~ \matE \sample \DD_{\ring,\chi}^{t \times m} \\
                     &\mat{B} \coloneqq \matR \cdot \mat{A} + \matE \bmod q 
                \end{aligned}
            } \\
            -~& 
            \condprob{
                \adv(\mat{A}, \mat{B}) = 1
            }{
                \begin{aligned}
                     &\mat{A} \gets \ddvA \\
                     &\mat{B} \sample \ring_q^{t \times m} 
                \end{aligned}
            }
        \end{aligned}
        }{
            \condprob{
                \adv(\mat{A}, \mat{B}) = 1
            }{
                \begin{aligned}
                     &\mat{A} \gets \ddvA\\ 
                     &\mat{R} \sample \ring_q^{t \times n} ;~ \matE \sample \DD_{\ring,\chi}^{t \times m} \\
                     &\mat{B} \coloneqq \matR \cdot \mat{A} + \matE \bmod q 
                \end{aligned}
            }
            - 
            \condprob{
                \adv(\mat{A}, \mat{B}) = 1
            }{
                \begin{aligned}
                     &\mat{A} \gets \ddvA \\
                     &\mat{B} \sample \ring_q^{t \times m} 
                \end{aligned}
            }
        }
        }
        \leq \negl.
    \end{align*}
    The parameter $t$ is suppressed if $t=1$. 
    The parameter $\ddvA$ is suppressed if it is the uniform distribution over $\ring_q^{n \times m}$. 
    We write $m = \poly$ for the version where $\adv$ is given unbounded polynomially many LWE samples, in which case $\ddvA$ is the distribution per sample which is over $\ring_q^n$. 
\end{definition}

\medskip
\noindent\textbf{Gadget Matrix.}
For any $\beta,q \in \NN$ with $\beta \leq q$, we write $\vecg_{\beta,q}^\tr \coloneqq (\beta^0,\beta^1,\ldots,\beta^{\ceil{\log_\beta q}-1})$ for the gadget vector
and $\matG_{\beta,q} \coloneqq \matI_n \otimes \vecg_{\beta,q}^\tr$ for the gadget matrix respectively \cite{EC:MicPei12}. 
We denote by $\matG_{\beta,q}^{-1}(\cdot)$ the corresponding balanced $\beta$-ary decomposition operation, 
i.e.~it holds $\matG_{\beta,q} \cdot \matG_{\beta,q}^{-1}(\vecx) = \vecx$ and $\norm{\matG_{\beta,q}^{-1}(\vecx)} \leq \beta/2$ for any $\vecx \in \ring_q^n$.

\medskip
\noindent\textbf{Fully-Homomorphic Encryption.}
A fully-homomorphic encryption (FHE) for a message space $\cM$ 
consists of PPT algorithms $(\setup, \kgen, \enc, \Eval, \dec)$ with the syntax: 

\begin{description}
	\item $\pubparam\gets\setup(\secparam)$: The setup algorithm generates the public parameters $\pubparam$.
    \item $\tuple{\pk, \sk} \gets \kgen(\pubparam)$: 
    The key generation algorithm, on input the public parameters $\pubparam$, 
    generates the public and secret keys $(\pk,\sk)$.
    \item $\ctxt \gets \enc(\pk, \mu)$: 
    The encryption algorithm encrypts a message $\mu \in \cM$ w.r.t. the public key $\pk$.
    \item $\ctxt' \gets \Eval(\pk,f,(\ctxt_i)_{i \in I})$: 
    The evaluation algorithm takes the public key $\pk$, a function $f$, and a set of ciphertexts $(\ctxt_i)_{i \in I}$ indexed by $I$, and returns an evaluated ciphertext $\ctxt'$.
    \item $\mu' \gets \dec\tuple{\sk,\ctxt}$:
    The decryption algorithm, on input a secret key $\sk$ and ciphertext $\ctxt$, outputs some $\mu'$.
\end{description}

Because of the homomorphic nature of FHE, defining the correctness requirement is not straightforward and 
multiple proposals have appeared in the literature (e.g.~\cite{CiC:BCFPPR25,EC:BJSW25,CiC:ABMP25}). 
We state a version that follows the spirit of that in \cite{CiC:ABMP25},
which formalises correctness as an experiment so as to capture the adaptiveness of an adversary in choosing ciphertext evaluations,
except that we further allow the flexibility to decrypt unevaluated ciphertexts.

\begin{definition}[FHE Correctness]\label{def:FHE-correctness}
    An FHE $\FHE$ is said to have correctness error $\varepsilon$, if for any PPT $\adv$, it holds that
    \fullversion{\begin{align*}
        &\condprob{
            \dec\tuple{\sk,\ctxt} = \mu
        }{
            (\sk,\ctxt,\mu) \gets \mathsf{Correct}_{\FHE,\adv}(\secparam) 
        } 
        \geq 1 - \varepsilon(\secpar),
    \end{align*}}{\begin{align*}
        &\condprob{
            \dec\tuple{\sk,\ctxt} = \mu
        }{
            (\sk,\ctxt,\mu) \gets \mathsf{Correct}_{\FHE,\adv}(\secparam) 
        } \\
        &\geq 1 - \varepsilon(\secpar),
    \end{align*}}
    where $\mathsf{Correct}_{\FHE,\adv}$ is defined in \cref{fig:FHE-correctness}.
    Furthermore, we say that $\FHE$ is correct if $\varepsilon$ is negligible in $\secpar$.
\end{definition}


\begin{figure*}[t]
    \begin{pchstack}[center,boxed]
        \adjustbox{max width=\dimexpr\textwidth-2\fboxsep-2\fboxrule\relax}%
        {
            \procedure[]{$\mathsf{Correct}_{\FHE,\adv}(\secparam)$}
            {
                \Lquery \coloneqq \emptyset \\
                \pubparam\gets\setup(\secparam) \\
                (\pk, \sk) \gets \kgen(\pubparam) \\
                \id \gets \adv^{\enc\oracle,\Eval\oracle}(\pubparam,\pk,\sk) \\
                \pcparse (\ctxt,\mu) \gets \Lquery[\id] \\
                \pcreturn (\sk,\ctxt,\mu)
            }
            \pchspace
            \procedure[]{$\mathsf{Correct}_{\tFHE,\adv}(\secparam)$}
            {
                \Lquery \coloneqq \emptyset \\
                \pubparam\gets\setup(\secparam) \\
                (\pk, (\sk_k)_{k \in [K]}) \gets \kgen(\pubparam) \\
                (T, \id) \gets \adv^{\enc\oracle,\Eval\oracle}(\pubparam,\pk,(\sk_k)_{k \in [K]}) \\
                \pcparse (\ctxt,\mu) \gets \Lquery[\id] \\
                \pd_k \gets \ParDec(\pk, \sk_k, \ctxt) ~~\forall k \in T \\
                \pcif \mu \neq \Dec(T, (\pd_k)_{k \in T}, \ctxt) \\ 
                \pcind \pcthen \pcreturn 1 \\
                \pcelse \pcreturn 0
            }
            \pchspace
            \begin{pcvstack}            
                \procedure[]{$\enc\oracle(\id,\mu \in \cM)$}
                {
                    \ctxt \gets \enc(\pk,\mu) \\
                    \Lquery[\id] \coloneqq (\ctxt,\mu) 
                }
                \pcvspace
                \procedure[]{$\Eval\oracle(\id,f,(\id_i)_{i\in I})$}
                {
                    \pcfor i \in I : \\
                        \pcind \pcassert \Lquery[\id_i] \neq \emptyset \\
                        \pcind \pcparse (\ctxt_i,\mu_{i}) \gets \Lquery[\id_i] \\ 
                    \ctxt' \gets \Eval(\pk, f, (\ctxt_i)_{i \in I}) \\
                    \Lquery[\id] \coloneqq (\ctxt', f((\mu_{i})_{i \in I})) 
                }
            \end{pcvstack}
        }
        \end{pchstack}
    \caption{Correctness experiment for FHE and threshold FHE (\cref{def:FHE-correctness,def:tFHE-correctness}).} 
    \label{fig:FHE-correctness} 
\end{figure*}



\section{Threshold FHE Construction}\label{sec:Plivi-FHE}

\subsection{Requirements on Underlying FHE}\label{sec:FHE-requirements}

\begin{definition}[FHE with Linear Decryption]\label{def:linear-FHE-fine-grained}
  Let $\xi \in \ring$ be coprime to $q$. An FHE scheme $\FHE$ has $(\ring,\hat{n},p,q,\xi,B)$-linear-decryption with error $\varepsilon$, if the following hold:
  \begin{enumerate} 
    \item The secret key space is $\ring_{q}^{\hat{n}}$, i.e.~a secret key is a vector $\sk = \hat{\vec{s}} \in \ring_{q}^{\hat{n}}$.
    \item There exists a deterministic $\PreDecrypt$ algorithm, which inputs a ciphertext $\ctxt$ and outputs
      $(\hatvec{c}_0,\hat{c}_1) \in \ring_q^{\hat{n}} \times \ring_q \gets \PreDecrypt(\ctxt)$.
    \item For any PPT $\adv$, it holds that 
    \[
      \adjustbox{max width=\linewidth}{$\displaystyle
        \condprob{
          \begin{aligned}
            &\hat{c}_1-\hat{\vec{s}}^\tr \cdot \hat{\vecc}_0 \bmod q \\
            =~& \xi^{-2} \cdot \round{q/p} \cdot \mu + \hat{e} \\
            \land~& \norm{\hat{e}}\leq B
          \end{aligned}
        }{
          \begin{aligned}
            & (\sk,\ctxt,\mu) \gets \mathsf{Correct}_{\FHE,\adv}(\secparam) \\
            & \pcparse \sk = \hatvec{s} \\
            & (\hatvec{c}_0,\hat{c}_1) \gets \PreDecrypt(\ctxt)
          \end{aligned}
        }
      $}
    \]
    $\geq 1 - \varepsilon(\secpar)$,
    where $\mathsf{Correct}_{\FHE,\adv}$ is defined in \cref{fig:FHE-correctness}. 
  \end{enumerate}
  If $\varepsilon(\lambda)$ is negligible, we simply say that $\FHE$ has $(\ring,\hat{n},p,q,\xi,B)$-linear-decryption.
\end{definition}

We will also require the FHE scheme to satisfy a security notion that we call ``IND-CPA-with-leakage''.
This notion says that, for an FHE scheme with linear decryption, 
the scheme is IND-CPA secure even when additionally giving out a masked leakage on its error term arising from the linear part of the decryption algorithm.
Below, we recall standard IND-CPA and define the above variant. 

\begin{definition}[FHE IND-CPA(-with-Leakage) Security]\label{def:FHE-with-leakage-security}
Let $\FHE$ be an FHE scheme. We define the IND-CPA, and IND-CPA-with-$(D,\leakdim,\chi)$-Leakage experiments $\mathsf{Exp}_{\FHE,\adv}^b$ in~\cref{fig:FHE-leaky-security}. The IND-CPA and 
IND-CPA-with-$(D,\leakdim,\chi)$-Leakage advantage functions of a PPT $\adv$ against $\FHE$ are defined as
  \[
    \Adv_{\FHE,\adv\text{\dbox{,$(D,\leakdim,\chi)$}}}^{\INDCPAwLd}(\secpar)\coloneqq
    \abs{
    \ifstyle{2}{%
      \begin{aligned}
      &\prob{\mathsf{Exp}_{\FHE,\adv}^0(\secparam) = 1} \\
      -~& 
      \prob{\mathsf{Exp}_{\FHE,\adv}^1(\secparam) = 1}
      \end{aligned}
    }{
      \prob{\mathsf{Exp}_{\FHE,\adv}^0(\secparam) = 1} 
      - 
      \prob{\mathsf{Exp}_{\FHE,\adv}^1(\secparam) = 1}
    }
    }. 
  \]
    $\FHE$ is IND-CPA secure, respectively IND-CPA-with-$(D,\leakdim,\chi)$-Leakage secure, 
    if for any PPT $\adv$ it holds that
    $\mathsf{Adv}_{\FHE,\adv\text{\dbox{,$(D,\leakdim,\chi)$}}}^{\INDCPAwLd}(\secpar)\leq \negl$. 
\end{definition}


\begin{figure*}[t]
    \begin{pchstack}[center,boxed]
        \adjustbox{max width=\dimexpr\textwidth-2\fboxsep-2\fboxrule\relax}%
        {
            \begin{pcvstack}
                \procedure[]{$\mathsf{Exp}_{\FHE,\adv}^b(\secparam)$}
                {
                    \pubparam\gets\setup(\secparam)\\
                    (\pk, \sk) \gets \kgen(\pubparam) \\
                    \Lquery \coloneqq \emptyset;~~ \Lleak\coloneqq \emptyset \\
                    b' \gets \adv^{\Chal\oracle,\dbox{\scriptsize$\enc\oracle,\Eval\oracle,\Leak\oracle$}} (\pubparam,\pk) \\ 
                    \pcreturn b'
                }
                \pcvspace
                \procedure[]{$\Leak\oracle(\id, \leco\in\ring^\leakdim)$}
                {
                    \pcif \Lleak[\id] = \emptyset \pcthen \\
                    \pcind \pcassert \norm{\leco} \leq D \pccomment{norm check} \\
                    \pcind \pcassert \Lquery[\id] \neq \emptyset \\ 
                    \pcind \pclinecomment{query $\id$ honestly generated} \\
                    \pcind \pcparse (\ctxt,\mu_0,\mu_1) \gets \Lquery[\id] \\ 
                    \pcind \pcassert \mu_0 = \mu_1 \\ 
                    \pcind \pclinecomment{same plaintext in games 0 and 1} \\
                    \pcind (\hat{\vecc}_0,\hat{c}_1)\gets\PreDecrypt(\ctxt)\\
                    \pcind \mu \coloneqq \mu_0 = \mu_1\\
                    \pcind \hat{e} \coloneqq \hat{c}_1 - \sk^\tr\cdot\hat{\vecc}_0 - \xi^{-2} \round{\frac{q}{p}}\mu \bmod q \\
                    \pcind \vecf \sample \ddv_{\ring,\chi}^\leakdim \\
                    \pcind \Lleak[\id] \gets \leco \cdot \hat{e} + \vecf \\
                    \pcind \pclinecomment{error $\leco \cdot \hat{e}$ masked by fresh $\vecf$} \\ 
                    \pcreturn \Lleak[\id] 
                }
            \end{pcvstack}
            \pchspace 
            \begin{pcvstack} 
                \procedure[]{$\mathsf{Exp}_{\tFHE,\adv}^b(\secparam)$}
                {
                    \pubparam\gets\setup(\secparam)\\
                    \cdv \gets \adv(\pubparam) \pccomment{corrupt set $\cdv$} \\
                    \pcassert \cdv \subset_{< t} [K] \\
                    (\pk, (\sk_k)_{k \in [K]}) \gets \kgen(\pubparam) \\
                    \Lquery \coloneqq \emptyset;~~ \Lshare \coloneqq \emptyset;~~ \Lindex \coloneqq \emptyset \\
                    b' \gets \adv^{\enc\oracle,\Chal\oracle,\Eval\oracle,\ParDec\oracle} \tuple{\pk,\tuple{\sk_k}_{k \in \cdv}} \\ 
                    \gamechange{$\pcassert |\LShare| \leq Q$} \\
                    \gamechange{\pclinecomment{partial decryptions for $\leq Q$ ciphertexts}} \\
                    \pcreturn b'
                }
                \pcvspace
                \procedure[]{$\ParDec\oracle(\id, k \in [K])$}
                {
                    \pcassert \Lquery[\id] \neq \emptyset \\ 
                    \pclinecomment{queries $\id$ honestly generated} \\
                    \Lindex[\id] \coloneqq \Lindex[\id] \cup \set{k} \\
                \pcparse (\ctxt,\mu_0,\mu_1) \gets \Lquery[\id] \\ 
                    \pcif \mu_0 \neq \mu_1 \pcthen 
                    \pcassert |\Lindex[\id] \cup \cdv| < t \\
                    \pclinecomment{if different plaintext in games 0 and 1,} \\
                    \pclinecomment{unauthorised to decrypt} \\
                    \pcif \Lshare[\id] = \emptyset \pcthen \\
                    \pcind \pd_k \gets \ParDec\tuple{\pk, \sk_k,\ctxt} ~~\forall k \in [K] \\                     
                    \pcind \Lshare[\id] \coloneqq (\pd_k)_{k \in [K]} \\
                    \pcparse (\pd_k)_{k \in [K]} \gets \Lshare[\id] \\
                    \pcreturn \pd_k
                }        
            \end{pcvstack}
            \pchspace
            \begin{pcvstack}
                \procedure[]{$\Chal\oracle(\id,\mu_0,\mu_1)$}
                {
                    \pcif \Lquery[\id] = \emptyset \pcthen \\ 
                    \pcind \ctxt \gets \enc(\pk,\mu_b) \\
                    \pcind \Lquery[\id] \coloneqq (\ctxt,\mu_0,\mu_1) \\
                    \pcreturn \Lquery[\id]
                }
                \pcvspace
                \procedure[]{$\enc\oracle(\id,\mu)$}
                {
                    \pcif \Lquery[\id] = \emptyset \pcthen \\ 
                    \pcind \ctxt \gets \enc(\pk,\mu) \\
                    \pcind \Lquery[\id] \coloneqq (\ctxt,\mu,\mu) \\
                \pcreturn \Lquery[\id]
                }
                \pcvspace
                \procedure[]{$\Eval\oracle(\id,f,(\id_i)_{i\in I})$}
                {
                \pcif \Lquery[\id] = \emptyset \pcthen \\
                    \pcind \pcfor i \in I : \\
                        \pcind[2] \pcassert \Lquery[\id_i] \neq \emptyset \\ 
                        \pcind[2] \pclinecomment{queries $\id_i$ honestly generated} \\
                        \pcind[2] \pcparse (\ctxt_i,\mu_{0,i},\mu_{1,i}) \gets \Lquery[\id_i] \\ 
                    \pcind \ctxt' \gets \Eval(\pk, f, (\ctxt_i)_{i \in I}) \\
                    \pcind \mu_0' \coloneqq f((\mu_{0,i})_{i \in I}) \\
                    \pcind \mu_1' \coloneqq f((\mu_{1,i})_{i \in I}) \\
                    \pcind \Lquery[\id] \coloneqq (\ctxt', \mu_0', \mu_1') \\
                    \pcreturn \Lquery[\id]
                }
            \end{pcvstack}
        }
    \end{pchstack}
    \caption[{}]{IND-CPA(-with-$(D,\leakdim,\chi)$-Leakage) experiments $\mathsf{Exp}_{\FHE,\adv}^b(\secparam)$ for FHE (\cref{def:FHE-with-leakage-security}), and 
    IND-CPA experiment $\mathsf{Exp}_{\tFHE,\adv}^b(\secparam)$ for threshold FHE (\cref{def:tFHE-security}). 
    Code in \dbox{dashed box}, i.e.~the oracles $\enc\oracle$, $\Eval\oracle$ and $\Leak\oracle$, applies only to 
    IND-CPA-with-$(D,\leakdim,\chi)$-Leakage. 
    $\gamechange{Code in gray}$ for the $Q$-bounded version of the experiment.
    } 
    \label{fig:FHE-leaky-security}\label{fig:tFHE-security}  
\end{figure*}

It is easy to see that the IND-CPA-with-Leakage security is implied by standard IND-CPA security if the masking noise for the leakage is large enough, 
specifically when it is super-polynomially larger than the error term being masked. This follows from a standard noise-flooding lemma,
formalised below. \fullversion{Its proof is given in \cref{app:thm:leakage-security}.}{Its proof can be found in \thefullversion.}

\begin{restatable}[Leakage Security]{theorem}{leakagesecurity}\label{thm:leakage-security}
  Let $\mathcal{K} = \QQ(\zeta)$ be a cyclotomic field of degree $\varphi$ and expansion factor $\expansion_\infty$ w.r.t. the coefficient $\ell_\infty$-norm. Let $\ring = \ZZ[\zeta]$ be its ring of integers.
  Let $\leakdim,D\in\NN$ and $\FHE$ be an FHE scheme with $(\ring,\hat{n},p,q,\xi,B)$-linear-decryption with error $\varepsilon_\FHE$.
  Then, for any PPT adversary $\adv$ against the IND-CPA-with-$(D,\leakdim,\chi)$-Leakage security of $\FHE$ that makes at most $Q \in \poly$ queries to the oracle $\Leak\oracle$, there exists a PPT adversary $\bdv$ against the IND-CPA security of $\FHE$ such that
  \begin{align*}
    \ifstyle{2}{
      &\Adv_{\FHE,\adv,(D,\leakdim,\chi)}^{\INDCPAwL}(\secpar) \\ 
      \leq~& 
    }{
      \Adv_{\FHE,\adv,(D,\leakdim,\chi)}^{\INDCPAwL}(\secpar) \leq 
    }
    \ifstyle{2}{
      2\!\cdot \!\left( Q\!\cdot\! \varepsilon_\FHE\!+\! C \!\cdot\! \sqrt{Q\! \cdot \!\varphi\! \cdot\! \leakdim}\! \cdot\! \expansion_\infty\! \cdot\! D\!\cdot\! B/\chi\right)
      +
      \Adv_{\FHE,\bdv}^{\INDCPA}(\secpar)
    }
    {
      2\cdot \left( Q\cdot \varepsilon_\FHE + \frac{C \cdot \sqrt{Q \cdot \varphi \cdot \leakdim} \cdot \expansion_\infty \cdot D\cdot B}{\chi} \right)
      +
      \Adv_{\FHE,\bdv}^{\INDCPA}(\secpar)
    }
  \end{align*}
  where $C = \sqrt{\pi/2}$ is a constant. 
  In particular, if $\chi> \lambda^{\omega(1)}\cdot \sqrt{Q\cdot\varphi \cdot \leakdim} \cdot \expansion_\infty \cdot D\cdot B$, $\FHE$ has $(\ring,\hat{n},p,q,\xi,B)$-linear-decryption and is IND-CPA secure, then it is also IND-CPA-with-$(D,\leakdim,\chi)$-Leakage secure.
\end{restatable}

\subsection{Threshold FHE Definitions}

A $t$-out-of-$K$ threshold fully-homomorphic encryption, or $(t,K)$-threshold FHE, for a set of $K$ parties, threshold $t \leq K$, and a message space $\cM$,
consists of PPT algorithms $(\setup, \kgen, \enc, \Eval, \ParDec, \Dec)$ with the syntax: 

\begin{description}
	\item $\pubparam\gets\setup(\secparam)$: The setup algorithm generates the public parameters $\pubparam$.
    \item $(\pk, (\sk_k)_{k \in [K]}) \gets \kgen(\pubparam)$: 
    The key generation algorithm, on input the public parameters $\pubparam$, 
    generates the public key $\pk$ and a tuple of $K$ secret keys $\tuple{\sk_k}_{k \in [K]}$ for each user $k \in [K]$.
    \item $\ctxt \gets \enc(\pk, \mu)$: 
    The encryption algorithm encrypts a message $\mu \in \cM$ w.r.t. the public key $\pk$.
    \item $\ctxt' \gets \Eval(\pk,f,(\ctxt_i)_{i \in I})$: 
    The evaluation algorithm takes the public key $\pk$, a function $f$, and a set of ciphertexts $(\ctxt_i)_{i \in I}$ indexed by $I$, and returns an evaluated ciphertext $\ctxt'$.
    \item $\pd_k \gets \ParDec\tuple{\pk,\sk_k,\ctxt}$: 
    The partial decryption algorithm takes the secret key $\sk_k$ of a user $k$ and a ciphertext $\ctxt$, and outputs a partial decryption $\pd_k$.
    \item $\mu' \gets \Dec\tuple{\pk,T,\tuple{\pd_k}_{k \in T},\ctxt}$: 
    The combine algorithm, on input a tuple of partial decryptions $(\pd_k)_{k \in T}$ from a set $T$ of $t$ users and a ciphertext $\ctxt$, outputs some $\mu'$.
\end{description}

Analogous to a (non-threshold) FHE, we define correctness of a threshold FHE by an experiment, 
in order to capture the adaptive nature of an adversary in choosing ciphertext evaluations and decryption. 

\begin{definition}[Threshold FHE Correctness]\label{def:tFHE-correctness}
    A $(t,K)$-threshold FHE $\tFHE$ 
    is said to have correctness error $\varepsilon(\secpar)$, if
    for any PPT $\adv$,
    \begin{align*}
        \prob{\mathsf{Correct}_{\tFHE,\adv}(\secparam) = 1}
        \leq \varepsilon(\secpar),
      \end{align*}
    where $\mathsf{Correct}_{\tFHE,\adv}$ is defined in \cref{fig:FHE-correctness}.
    Furthermore, we say that $\tFHE$ is correct if $\varepsilon$ is negligible in $\secpar$.
\end{definition}

\begin{definition}[Threshold FHE IND-CPA Security]\label{def:tFHE-security}
  A $(t,K)$-threshold FHE $\tFHE$ is (resp.~$Q$-bounded) IND-CPA secure, 
  if for any PPT adversary $\adv$ it holds that
  \ifstyle{2}{
    $\Adv_{\tFHE,\adv}^{\INDCPA}(\secpar) \coloneqq$ 
    \[
      \abs{\prob{\mathsf{Exp}_{\tFHE,\adv}^0(\secparam) = 1} - \prob{\mathsf{Exp}_{\tFHE,\adv}^1(\secparam) = 1}} 
    \]
    $\leq \negl,$
  }{
    \[ 
      \Adv_{\tFHE,\adv}^{\INDCPA}(\secpar) \coloneqq \abs{\prob{\mathsf{Exp}_{\tFHE,\adv}^0(\secparam) = 1} - \prob{\mathsf{Exp}_{\tFHE,\adv}^1(\secparam) = 1}} \leq \negl,
    \]
  }
  where $\mathsf{Exp}_{\tFHE,\adv}^b$ (resp.~augmented with $\gamechange{code in gray}$) is defined in~\cref{fig:tFHE-security}.
\end{definition}

\begin{remark}[Alternative definition for threshold FHE security.]\label{rem:alt-security-def}
  In \cref{def:tFHE-security}, more specifically in the security experiment defined in \cref{fig:tFHE-security}, the oracles keep track of ciphertexts $\ctxt$ by their identifiers $\id$. 
  An alternative would be to track ciphertexts $\ctxt$ by (the representation of) the ciphertexts $\ctxt$ themselves.
  This alternative definition is polynomially equivalent to that stated in \cref{fig:tFHE-security} if ciphertexts are collision resistant, i.e.~for honestly generated $\pk$, it is infeasible to produce the same ciphertext $\ctxt$ by two distinct sequences of $\enc(\pk, \cdot)$ and $\eval(\pk, \cdot, \cdot)$ operations with honestly sampled randomness. 
  This alternative definition has some practical benefits, to be elaborated in \cref{rem:pseudorandom-tag}. 
\end{remark}

\subsection{Threshold FHE from FHE}\label{sec:thFHE-from-FHE}

\noindent\textbf{Construction.}
Let $1 < p \leq q$ be integers.
We require two ingredients:
\begin{itemize}
  \item $\FHE$ is an FHE scheme which has $(\ring,\hat{n},p,q,\xi,B)$-linear-decryption with error $\varepsilon_\FHE$ (\cref{def:linear-FHE-fine-grained}).
  \item $\Xi$ is a $\xi$-subtractive set of cardinality $K$, with quality measurements $(\rho,\gamma,\theta) = (\rho(t),\gamma(t),\theta(t))$ (\cref{def:subtractive-sets}).
\end{itemize}
Given the above, in \cref{fig:threshold-FHE-fine-grained} we construct a $(t,K)$-threshold FHE scheme $\tFHE_\params$
with parameters $\params = (\FHE,\ring,n,m,h,q,\beta,\bar{\chi},\chi,z,\Xi)$.


\begin{figure*}[t]
	\begin{pchstack}[center,boxed]
        \adjustbox{max width=\dimexpr\textwidth-2\fboxsep-2\fboxrule\relax}%
		{
			\begin{pcvstack}				
				\procedure[]{$\setup(\secparam)$}
				{	
					\pubparam_{\FHE} \gets \FHE.\Setup(\secparam) \\
					\pcreturn \pubparam \coloneqq \pubparam_\FHE
				}				
				\pcvspace
				\procedure[]{$\kgen(\pubparam)$}
				{	
					\pcparse \pubparam = \pubparam_\FHE \\
                    (\pk_\FHE, \sk_\FHE) \gets \FHE.\kgen(\pubparam_\FHE) \\ 
                    \pclinecomment{$\sk_\FHE = \hat{\vec{s}} \in \ring_{q}^{\hat{n}}$} \\
					\matA \sample \ring_q^{n \times m} \\
					\mat{R} \sample \ring_q^{t \times n} \\
					\vec{r}^\transpose \coloneqq \vec{v}_0^\transpose \cdot \mat{R} \bmod q \\ 
					\vec{e} \sample \DD_{\ring,\bar{\chi}}^m \\
					\vec{z}^\transpose \coloneqq (\hat{\vec{s}}^\tr,\veczero_{h}^\tr) \cdot \matG_{\beta,q} \bmod q \in \ring_q^m \\
                    \vec{b}^\transpose \coloneqq  \vec{r}^\transpose \matA + \vec{e}^\transpose + \vecz^\transpose \bmod q \\
					\pk \coloneqq (\pk_\FHE, \matA, \vec{b}) \\
					(\vec{s}_{k})_{k\in[K]} \coloneqq \xi \cdot \mat{V} \cdot \mat{R} \bmod q \\
                    (\sk_k)_{k \in [K]} \coloneqq (\vec{s}_{k})_{k\in[K]} \\
					\pcreturn (\pk,(\sk_k)_{k\in[K]})
				}
			\end{pcvstack}
			\pchspace
			\begin{pcvstack}
				\procedure[]{$\enc(\pk,\mu)$}
				{
					\pcparse \pk=(\pk_\FHE,\matA,\vecb) \\
                    \ctxt_\FHE \gets \FHE.\enc(\pk_\FHE, \mu) \\
					\rtag \sample \ring_q^h \\
					\pcreturn \ctxt \coloneqq (\ctxt_\FHE, \rtag) \\
				}
                \pcvspace
				\procedure[]{$\ParDec(\pk, \sk_k, \ctxt)$}
				{	
					\pcparse \pk = (\pk_\FHE,\matA,\vecb) \\
					\pcparse \sk_k = \vecs_k \\
					\pcparse \ctxt = (\ctxt_\FHE,\rtag) \\
					(\hatvec{c}_0,\hat{c}_1) \gets \FHE.\PreDecrypt(\ctxt_\FHE) \\
					\vecx \coloneqq \matG_{\beta,q}^{-1} \begin{pmatrix} \hat{\vec{c}}_0 \pclb \rtag \end{pmatrix} \\
                    \vecc_0 \coloneqq \matA \cdot \vecx \bmod q \\
					e_{k} \sample \DD_{\ring,\chi} \\
					\pd_{k} \coloneqq \vecs_{k}^\transpose \cdot \vec{c}_0 + e_{k} \bmod q \\
					\pd_{k} \gets \round{\pd_{k}}_{2^z} \\
					\pcreturn \pd_k
				}
			\end{pcvstack}
			\pchspace
			\begin{pcvstack}
				\procedure[]{$\Eval(\pk, f, (\ctxt_i)_{i \in I})$}
				{	
					\pcparse \pk = (\pk_\FHE,\matA,\vecb) \\
					\pcparse \ctxt_i = (\ctxt_{\FHE,i}, \rtag_i)~~\forall i \in I \\
                    \ctxt_\FHE' \gets \FHE.\Eval(\pk_\FHE, f, (\ctxt_{\FHE,i})_{i \in I}) \\
					\rtag \sample \ring_q^h \\
					\pcreturn \ctxt' \coloneqq (\ctxt_\FHE', \rtag)
				}
				\pcvspace
				\procedure[]{$\Dec(\pk, T,(\pd_k)_{k\in T},\ctxt)$}
				{	
					\pcparse \pk = (\pk_\FHE,\matA,\vecb) \\
					\pcparse \ctxt = (\ctxt_\FHE,\rtag) \\
					(\hatvec{c}_0,\hat{c}_1) \gets \FHE.\PreDecrypt(\ctxt_\FHE) \\
					\vecx \coloneqq \matG_{\beta,q}^{-1} \begin{pmatrix} \hat{\vec{c}}_0 \pclb \rtag \end{pmatrix} \\
                    c_1 \coloneqq \xi^2 \cdot \vec{b}^\tr \cdot \vecx \bmod q \\
					\pd \gets \xi \cdot \vec{v}_0^\transpose \cdot \mat{V}_T^{-1} \cdot (\pd_{k})_{k\in T} \bmod q \\
					\hat{\mu} \coloneqq \xi^2\cdot\hat{c}_1 - (c_{1} - \pd) \bmod q \in\ring_q \\
                    \mu' \gets \round{ p/q \cdot \hat{\mu} } \bmod p \in \ring_p \\
					\pcreturn \mu'
				}
			\end{pcvstack}
		}
	\end{pchstack}
	\caption[{}]{$\tFHE$ from $\FHE$. $m = (\hat{n}+h) \cdot \ceil{\log_\beta q}$. }
	\label{fig:threshold-FHE-fine-grained}
\end{figure*}

\begin{restatable}[Correctness]{theorem}{tfhecorrectness}\label{thm:correctness}
  If $\xi \in \ZZ$, 
  $\FHE$ has $(\ring,\hat{n},p,q,\xi,B)$-linear-decryption with error $\varepsilon_\FHE$ and 
  \begin{align}\label{eq:correctness}
    \frac{q}{2p}
    \geq&~ 
    \frac{\abs{q \bmod p}}{2} 
    + \xi^2\cdot B 
    \ifstyle{2}{\\ &}{} 
    + \expansion_\infty \cdot \left( \sqrt{\lambda} \cdot m\cdot\xi^2\cdot \bar{\chi}\cdot\frac{\beta}{2}+ t\cdot \rho\cdot(\sqrt{\lambda} \cdot \chi + 2^{z-1})\right).\ifstyle{2}{\nonumber}{}
  \end{align}
  Then, $\tFHE$ in \cref{fig:threshold-FHE-fine-grained} has correctness error $\varepsilon \leq \varepsilon_\FHE+\varphi\cdot(m+t)\cdot 2^{-\lambda}$. In particular, if $\min\{\log(1/\varepsilon_\FHE),\lambda -\log \varphi- \log(m+t)\}\geq\omega(\log\lambda)$, then $\tFHE$ is correct.
\end{restatable}

The proof of \cref{thm:correctness} is given in
\fullversion{\cref{app:thm:correctness}.}{\thefullversion.}
Here we provide a brief sketch of the proof.
\begin{proof}[Proof Sketch]
The proof compares the value reconstructed by threshold decryption with the
linear-decryption equation of the underlying $\FHE$ scheme. Both $\ParDec$ and
$\Dec$ compute the same gadget preimage
\[
    \vecx=\matG_{\beta,q}^{-1}
    \begin{pmatrix}
        \hat{\vecc}_0\\
        \rtag
    \end{pmatrix}.
\]
For a recovery set $T$, linearly combining the partial decryptions recovers the
shared secret $\vecr$ and gives
\[
    \pd =
    \xi^2\vecr^\transpose\matA\vecx
    +
    \xi\vecv_0^\transpose\matV_T^{-1}\vece_T
    \bmod q.
\]
On the other hand, the corresponding ciphertext component satisfies
\[
    c_1 =
    \xi^2\vecr^\transpose\matA\vecx
    +
    \xi^2\vece^\transpose\vecx
    +
    \xi^2\hat{\vec{s}}^\transpose\hat{\vecc}_0
    \bmod q.
\]
Hence $c_1-\pd$ cancels the common term
$\xi^2\vecr^\transpose\matA\vecx$ and leaves the $\FHE$ linear-decryption term, $\xi^2\hat{\vec{s}}^\transpose\hat{\vecc}_0$, plus a residual error
\[
    f \coloneqq \xi^2 \cdot\vece^\transpose\cdot\vecx - \xi \cdot\vec{v}_0^\transpose \cdot \mat{V}_T^{-1}\cdot\vece_T.
\]
In the final rounding in $\Dec$ one has to additionally account for the error $\xi^2\hat{e}$ arising from the $(\ring,\hat{n},p,q,\xi,B)$-linear-decryption property of $\FHE$.
The correctness condition in \cref{eq:correctness} is set so
that the norm of the whole error term remains below $q/(2p)$,
ensuring the scaling by $p/q$ recovers $\mu$ except with the stated failure
probability.\qed 
\end{proof}

\begin{remark}[``Average-case'' correctness analysis.]\label{rmk:average-correctness}
  In the proof of correctness (\cref{thm:correctness}) we (probabilistically) upper bound 
  $\norm{\vec{e}^\transpose \cdot \vec{x}}$ by $\expansion_\infty \cdot m \cdot \bar{\chi} \cdot \sqrt{\secpar} \cdot \beta/2$. 
  This is known as the so called ``worst-case'' correctness analysis. 
  A standard heuristic is the so-called ``average-case'' correctness analysis where one would upper bound the norm of $\vec{e}^\transpose \cdot \vec{x}$ directly by treating it as a discrete Gaussian rather than bounding $\vec{e}$ and $\vec{x}$ separately. 
  Over $\ZZ$, this would shave off a square-root-dimension factor, i.e.~$\sqrt{\varphi \cdot m}$ from the bound. 
  Taking the potential skewness of $\ring$ into account, however, we conservatively shave a factor of $\frac{\expansion_\infty}{\expansion_2} \cdot \sqrt{m} \leq \sqrt{\varphi \cdot m}$ instead, i.e.~we take $\norm{\vec{e}^\transpose \cdot \vec{x}} \leq \expansion_2 \cdot \sqrt{m} \cdot \bar{\chi} \cdot \sqrt{\secpar} \cdot \beta/2$.
  We can apply the same treatment to the term $\xi \cdot \vec{v}_0^\transpose \cdot \mat{V}_T^{-1}\cdot\vece_T$ and upper bound its norm by 
  $\expansion_2 \cdot \sqrt{t} \cdot \rho\cdot(\chi + 2^{z-1})\cdot\sqrt{\secpar}$
  instead of by
  $\expansion_\infty \cdot t\cdot \rho\cdot(\chi\cdot\sqrt{\secpar} + 2^{z-1})$.
  Here $\vece_T$ contains the rounding errors $d_k$ of $\ParDec$, and the heuristic treats each $e_k + d_k$ as noise of width at most $\chi + 2^{z-1}$.
  To summarise, with the ``average-case'' heuristic, we can replace the constraint \cref{eq:correctness} with 
  \begin{align}\label{eq:average-correctness}
      \frac{q}{2p}
      \geq&~ 
      \frac{\abs{q \bmod p}}{2} 
      +\xi^2\cdot B 
      \ifstyle{2}{\\ &}{} 
      + \expansion_2 \cdot \sqrt{\secpar} \cdot \left( \sqrt{m} \cdot \xi^2\cdot \bar{\chi}\cdot\frac{\beta}{2}+ \sqrt{t}\cdot \rho\cdot(\chi + 2^{z-1})\right). \ifstyle{2}{\nonumber}{}
  \end{align} 
\end{remark}


    \begin{restatable}[Security]{theorem}{tfhesecurity}
    \label{thm:tfhe-security}
     Let $\FHE$ be an FHE scheme with $(\ring,\hat{n},p,q,\xi,B)$-linear-decryption and
    $\Xi$ be a $\slack$-subtractive set of cardinality $K$ with quality measurements $(\rho,\gamma,\theta)$ (\cref{def:subtractive-sets}).
    Let
    $\epsilon \in (0,1/2]$,
    $Q \in \poly$,
    $q,\xi$ be coprime,
    $\chi_1^* \geq \sqrt{2} \cdot \eta_\epsilon(\ring^t)$ be such that there exists $\hat{\chi} > 0$ with
    $\gamma \leq \chi \cdot \sqrt{(2 \eta_\epsilon(\ring^t)^2 + (\chi_1^*)^2)^{-1} - \hat{\chi}^{-2}}$,
    and $\chi_2^* \geq \sqrt{2} \cdot \eta_\epsilon(\ring^m)$ be such that
    $\gamma \cdot \sqrt{\varphi m} \cdot \beta/2 \geq 1$ and
    $\sqrt{2 Q K \cdot \varphi m} \cdot \expansion_2 \cdot \tensorexpansion{t} \cdot \gamma \cdot \beta/2 \leq \chi \cdot \sqrt{(2 \eta_\epsilon(\ring^m)^2 + (\chi_2^*)^2)^{-1} - \bar{\chi}^{-2}}$, where $\tensorexpansion{t} \coloneqq 1$ if $\conductor$ is a power of two and $\tensorexpansion{t} \coloneqq \sqrt{t}$ otherwise.
    Let
    $\ddvt\coloneqq \left\{\pretag : \rtag\sample\ring_q^{h}, \pretag\coloneqq \matG_{\beta,q}^{-1}(\rtag)\right\}$.
    Define the parameter sets: 
    \begin{itemize}
        \item $\tFHE$: $\params = (\FHE,\ring,n,m,h,q,\beta,\bar{\chi},\chi,z,\Xi)$,
        \item $\LWE{}$: $\params_1 \coloneqq (\ring,h \ceil{\log_\beta q},\poly,q,\chi_1^*,\ddvt)$,
        \item $\LWE{}$: $\params_2 \coloneqq (\ring,n,m,q,\chi_2^*)$,
        \item $\FHE$: $\params_3 \coloneqq (\theta,K,\chi)$.
    \end{itemize}
    Consider the $(t,K)$-threshold FHE scheme $\tFHE_\params$ in \cref{fig:threshold-FHE-fine-grained}.  
    For all PPT $\adv$ that requests partial decryptions on at most $Q$ ciphertexts (i.e.\ in the $Q$-bounded IND-CPA experiment of \cref{def:tFHE-security}), there exist PPT $\bdv_1,\bdv_2,\bdv_3$ such that $\Adv_{\tFHE,\adv}^{\INDCPA}(\secpar) \leq \varepsilon$ where
    \begin{align*}
        \varepsilon \coloneqq&~
        2 \!\cdot\! t\! \cdot \!\left( \Adv_{\params_1,\bdv_1}^{\LWE{}}(\secpar) 
        \!+\! \Adv_{\params_2,\bdv_2}^{\LWE{}}(\secpar) \right) 
        \!+\! \frac{12\!\cdot\! Q\!\cdot\! \epsilon}{1\!-\!\epsilon}
        \ifstyle{2}{\\ &}{}
        \!+\! \Adv_{\FHE,\params_3,\bdv_3}^{\INDCPAwL}(\secpar) \\
        &\!+\! 4\! \cdot\! t\! \cdot\!
        \left( 2QK \!+\! \frac{3}{1\!-\!\epsilon} \right)\!\cdot\! \epsilon.
    \end{align*}
\end{restatable}

    The proof of \cref{thm:tfhe-security} is given in \fullversion{\cref{sec:missing-proofs-for-security}.}{\thefullversion. A more detailed description of the proof strategy can be found in \cref{sec:overview-analysis}.}
     Here we provide a brief sketch of the proof.

    \begin{proof}[Proof Sketch]
        Since the rounding $\round{\cdot}_{2^z}$ in $\ParDec$ is a deterministic public post-processing of each partial decryption, it suffices to prove security for the scheme without it.
        The proof proceeds through hybrids, $\hyb_{b,i}$ for $i=0,1,2,3$, that first remove the threshold layer and then reduce the remaining indistinguishability to the security of the underlying $\FHE$. Starting from the real experiment $\hyb_{b,0}$, the first hop to $\hyb_{b,1}$ rewrites $\vecb$ and the partial decryptions $\pd_k$ so that the threshold-layer correlations are simulated explicitly. This step is justified by standard $\LWE{}$: the term $\Adv_{\params_1,\bdv_1}^{\LWE{}}(\secpar)$ comes from the tag-dependent rerandomisation via $\matG^{-1}(\rtag)$, while $\Adv_{\params_2,\bdv_2}^{\LWE{}}(\secpar)$ comes from replacing the key-switching-key component $\vecb$. The two statistical terms in $\varepsilon$ are losses involved in the reductions from $\LWE{}$.

        The next hop, $\hyb_{b,1}\equiv\hyb_{b,2}$, replaces the partial decryptions of challenge ciphertexts by uniform values. This uses the Vandermonde structure induced by the $\slack$-subtractive set $\Xi$ and the fact that $\adv$ receives fewer than $t-|\cdv|$ honest partial decryptions for every ciphertext whose plaintext differs between games $0$ and $1$.

        Finally, $\hyb_{b,2}\equiv\hyb_{b,3}$ handles ciphertexts whose plaintext is the same in both games. Using the linear-decryption equation from \cref{def:linear-FHE-fine-grained}, the corresponding partial decryptions are rewritten so that the only remaining dependence on $\sk_\FHE=\hatvec{s}$ is through the FHE decryption noise $\hat e$. These values are simulated by the leakage oracle in the IND-CPA-with-Leakage experiment, giving the term $\Adv_{\FHE,\params_3,\bdv_3}^{\INDCPAwL}(\secpar)$. At this point, challenge shares are uniform and non-challenge shares are simulated from the leakage oracle, so summing the hybrid losses gives the stated bound $\varepsilon$.\qed
    \end{proof}

\begin{remark}[Relation to IND-CPA$^D$ security of FHE]
  By setting $K=t=1$, so that $\tFHE$ collapses to a plain FHE scheme, 
  we recover from \cref{def:tFHE-security} the IND-CPA$^D$ security notion for FHE defined in \cite{EC:LiMic21}.
  This means that our construction $\tFHE$ in \cref{fig:threshold-FHE-fine-grained} 
  is also a compiler which turns any FHE satisfying the requirements in \cref{sec:FHE-requirements} into an FHE that is IND-CPA$^D$ secure \cite{EC:LiMic21}.
\end{remark}

\begin{remark}[Pseudorandom tag and context-dependent decryption.]\label{rem:pseudorandom-tag}
  In our construction in \cref{fig:threshold-FHE-fine-grained}, the tag $\rtag$ in a ciphertext $\ctxt = (\ctxt_\FHE, \rtag)$ is randomly sampled during encryption and refreshed during homomorphic evaluation, independent of the main part $\ctxt_\FHE$ of the ciphertext.
  If the underlying $\FHE$ has collision-resistant ciphertexts in the sense described in \cref{rem:alt-security-def}, then our construction would be secure wrt.~the alternative security definition described in \cref{rem:alt-security-def} where the oracles keep track of ciphertexts via their own representations rather than identifiers. 
  This means that the security proof would also go through, in the random oracle model, if we replace the truly random tag $\rtag \sample \ring_q^h$ by a pseudorandom $\rtag \coloneqq \mathcal{H}(\ctxt_\FHE)$ where $\mathcal{H}$ is a hash function modelled as a random oracle. 
  Deriving tags in a pseudorandom manner has the obvious advantages of reducing ciphertext size and more importantly reducing use of randomness, which could be scarce in constrained devices.   
  Furthermore, we can easily upgrade our construction to support context-dependent decryption \cite{AC:BBNRS25} by hashing $\ctxt_\FHE$ together with the context string.
\end{remark}

\begin{remark}[Optimised Lagrange interpolation in $\Dec$]
\label{rem:lagrange-interpolation}
One optimisation that may not arise naturally from the protocol description concerns $\Dec$, where the message is reconstructed from the $t$ partial decryptions using the (scaled) first row $\xi\cdot\vec{v}_0^\transpose \mat{V}_T^{-1}$ of the inverse Vandermonde matrix $\mat{V}_T$ of the active nodes $\{\zeta_k\}_{k\in T}$. Instead of inverting the $t\times t$ matrix $\mat{V}_T$ over $\ring_q$, we use that this row is the vector of Lagrange basis polynomials evaluated at $0$, namely $\big(L_k(0)\big)_{k\in T}$ with $L_k(0)=\prod_{j\in T\setminus\{k\}}\zeta_j/(\zeta_j-\zeta_k)$, which costs $O(t^2)$ to form directly. Moreover, since the nodes are exactly the $K$-th roots of unity, the identity $\prod_{m\in T}(x-\zeta_m)\prod_{m\in T^c}(x-\zeta_m)=x^K-1$ lets us rewrite each weight as a product over the complement, $L_k(0)=-\tfrac{1}{K}\prod_{m\in T}(-\zeta_m)\prod_{m\in T^c}(\zeta_k-\zeta_m)$. Evaluating the row over whichever of $T$ or $T^c$ is smaller thus costs $O(t\cdot\min\{t,K-t\})$, keeping $\Dec$ cheap for large committees whenever $t$ is close to $1$ or to $K$.\footnote{Asymptotically, the row can be built in quasi-linear time $O(t\log^2 t)$ via fast multipoint evaluation~\cite{vzGG}; this appears less efficient at the committee sizes we consider, but could be relevant for larger ones.}
\end{remark}


\section{Practical Instantiation}\label{sec:implementation}
\label{sec:FHE_compatibility}

\ourwork is a generic construction to build threshold FHE from an FHE scheme. Concretely, its implementation consists of two main components: $\PreDecrypt$ and the threshold decryption layer. The first establishes the bridge between the FHE scheme and our threshold layer: given a generic FHE ciphertext, it produces a ciphertext satisfying the requirements described in \cref{sec:Plivi-FHE}. This design makes our threshold decryption layer generic wrt.~the choice of FHE scheme, at the cost of making $\PreDecrypt$ scheme dependent. Let us start by discussing potential instantiations of $\PreDecrypt$ for some of the most popular modern FHE schemes:

\medskip\noindent
\textbf{BFV/BGV/CKKS}~\cite{C:Brakerski12,EPRINT:FanVer12,ITCS:BraGenVai12,AC:CKKS17}: These schemes usually operate over rings that can be directly used by our threshold decryption layer, admitting $\PreDecrypt$ to be just the identity in some scenarios. \cref{def:linear-FHE-fine-grained} expects a BFV-style ciphertext, which can be achieved for BGV through simple transforms (see Appendix~A of \cite{C:AlpPei13}) and for (discrete~\cite{EC:BCKS24}) CKKS through bootstrapping (\emph{e.g.}, \cite{EPRINT:Kim25b}). For the original (approximate) CKKS~\cite{AC:CKKS17}, on the other hand, an additional challenge is in separating message from noise, a non-trivial security problem~\cite{EC:LiMic21} that afflicts approximate schemes whenever the existence of decryption oracles is considered. This problem is mostly orthogonal to the threshold setting, and considering potential solutions (\emph{e.g.}, ~\cite{PKC:BCDKLT25}) is left as future work.

\medskip
\noindent
\textbf{FHEW/TFHE}~\cite{EC:DucMic15,AC:CGGI16}: These schemes are typically built on top of ring variants~\cite{EC:DucMic15} of AP-style~\cite{C:AlpPei14} bootstrapping for LWE ciphertexts in the form $\mathbb{Z}_q^{\hat{n}} \times \mathbb{Z}_q$, which cannot be directly used by our threshold layer. For these cases, we employ Packing Key Switching~\cite{AC:CGGI16,ICALP:MicSor18,ACNS:CDKS21} techniques, which, given a vector of $k$ LWE ciphertexts encrypting messages $\{m_i\}_{i \in [0, k)} \in \mathbb{Z}_q^k$, produces a single ciphertext of the form $\ring_q^{\hat{n}} \times \ring_q$ (for some ring of degree at least $k$) encrypting $\sum_{i=0}^{k-1} m_i X^i \in \ring_q$. This ciphertext can either be defined over the ring required by the threshold layer, or ring switched~\cite{SCN:GHPS12} to it.  Alternatively, these schemes could also be instantiated with AP-style bootstrapping for RLWE ciphertexts~\cite{ICALP:MicSor18,CCS:GuiPer25}, which already perform similar packing by construction.

\begin{remark}[Bootstrapping]\label{rem:bootstrapping}
  In addition to structural requirements, our construction also requires the FHE scheme to provide a minimum ratio between ciphertext modulus and noise magnitude at the end of the homomorphic evaluation. For leveled schemes, one could consider increasing parameters to account for it, but this would introduce performance overheads for the entire computation. A more generic and modular solution is to perform a bootstrapping over the results of the evaluation procedure. We do not specify this bootstrapping as part of $\PreDecrypt$, as schemes may have different ways of performing it, but, in practice, including a bootstrapping at the beginning of $\PreDecrypt$ is a natural solution for most scenarios, and we adopt it in our implementation.
\end{remark}

\subsection{Implementation Details}\label{sec:imple_details}

To demonstrate compatibility with existing FHE schemes, our implementation\footnote{Available at \url{https://github.com/osdnk/tulva}.} natively consumes TFHE ciphertexts produced by the \texttt{tfhe-rs} library. A thin compatibility layer (written in Rust, with Python bindings) performs the necessary repacking and key-switching within \texttt{tfhe-rs}: since \texttt{tfhe-rs} operates on LWE ciphertexts over $\ZZ_q$, up to $\varphi$ of them are gathered into a single ciphertext over $\ring_q$, as required by our threshold layer.

As a first step of $\PreDecrypt$, the ciphertext is bootstrapped to reduce the noise magnitude, include the $\xi^{-2}$ slack, and switch to the modulus $q$ required by our construction, which is significantly larger than the $2^{64}$ modulus \texttt{tfhe-rs} typically operates on. We provide our own implementation of the \cite{AC:CGGI16} and \cite{CCS:GuiPer25} bootstrappings: the former is generally faster for small plaintext spaces ($p = 2^4$, as used in our examples), at the cost of higher memory usage and a larger public key, while the latter is more efficient for larger plaintext spaces ($p = 2^6$ or $p = 2^8$). We report runtimes only for \cite{AC:CGGI16}, since our example parameters use the small plaintext space, but also include \cite{CCS:GuiPer25} (with the \cite{ACNS:CDKS21} repacking technique) for use in other parameter regimes. Before this bootstrapping step, all homomorphic operations are performed in the smaller modulus and require no parameter adaptation compared to a non-threshold setting; we use the parameters suggested by \texttt{tfhe-rs} for $2^{-128}$ security and correctness errors.

The threshold layer is orchestrated in Python, with the most computationally intensive operations implemented in C and exposed via bindings; this keeps the critical code easy to optimise and the overall implementation simple to modify. The C layer is mainly an adaptation of the code given in \cite{C:CCCFGS25}, and Gaussian sampling, used for both key generation and noise flooding, is a generalised and optimised implementation of the one in \cite{CCS:LyuSeiSte24}.

We operate on power-of-two cyclotomic rings, the most common in practical FHE implementations, and employ standard RNS and NTT polynomial arithmetic (also known as double-CRT \cite{C:GenHalSma12}). We applied standard \texttt{AVX-512} optimisations, especially in the NTT transforms, while also optimising for cache locality and parallelism; the code remains portable to any platform, with a performance hit where \texttt{AVX-512} is unavailable.

\subsection{Parameter Selection}

\fullversion{
\medskip\noindent\textbf{Summary of constraints.}
For the convenience of referencing, we summarise \cref{thm:correctness,thm:tfhe-security,thm:leakage-security,rmk:average-correctness} in the following \cref{cor:tFHE-from-FHE}. 
We note that Item 4 is associated to the ``average-case'' correctness heuristic.
\begin{corollary}\label{cor:tFHE-from-FHE}
  Let $\ring = \ZZ[\zeta]$ be of ring degree $\varphi$ and with field expansion factor $\expansion_p$ w.r.t. the coefficient $\ell_p$-norm. 
  Let $\FHE$ be an FHE scheme. 
  Let $\Xi$ be a $\xi$-subtractive set of cardinality $K$ with quality measurements $(\rho,\gamma,\theta)$ (\cref{def:subtractive-sets}).
  Let $\delta = \delta(\lambda)$ be an independent parameter such that $2^\delta \in \lambda^{\omega(1)}$.
  Suppose $\params = (\FHE,\ring,n,m,h,q,\beta,\bar{\chi},\chi,z,\Xi)$, $\chi_1^*,\chi_2^* >0$, $\epsilon \in (0,1/2] \cap \negl$ and $Q \in \poly$ satisfy:
  \begin{enumerate}
    \item $\FHE$ has $(\ring,\hat{n},p,q,\xi,B)$-linear-decryption (\cref{def:linear-FHE-fine-grained}). 
    \item $\FHE$ is IND-CPA secure.
    \item $q,\xi$ are coprime.
    \item $\frac{\abs{q \bmod p}}{2} +\xi^2\cdot B +\expansion_2 \cdot \sqrt{\secpar} \cdot \left( \sqrt{m} \cdot \xi^2\cdot \bar{\chi}\cdot\frac{\beta}{2}+ \sqrt{t}\cdot \rho\cdot (\chi + 2^{z-1}) \right) +\frac{1}{2} \leq \frac{q}{2p}$ (\cref{eq:average-correctness}).
    \item $\chi \geq 2^\delta \cdot \sqrt{Q \cdot \varphi \cdot K} \cdot \expansion_\infty \cdot B \cdot \theta$.\footnote{We let $2^\delta$ absorb the $\sqrt{\pi/2}$ factor.} 
    \item\label{item:cor-LWE} $\LWE{\ring,h,\poly,q,\chi_1^*}$\footnote{This implies $\LWE{\ring,h \ceil{\log_\beta q},\poly,q,\chi_1^*,\ddvt}$ where $\ddvt = \mathcal{U}([-\beta/2,\beta/2)\cap\ZZ)^{h \ceil{\log_\beta q}}$ \cite[Theorem 4.3]{C:BLMR13}.}, $\chi_1^* \geq \sqrt{2} \cdot \eta_\epsilon(\ring^t)$, $\hat{\chi} >0$ and
      \[
          \frac{\gamma^2}{\chi^2} + \frac{1}{\hat{\chi}^2} \leq \frac{1}{2 \eta_\epsilon(\ring^t)^2 + (\chi_1^*)^2}.
      \]
    \item\label{item:cor-LWE-A} $\LWE{\ring,n,m,q,\chi_2^*}$, $\chi_2^* \geq \sqrt{2} \cdot \eta_\epsilon(\ring^m)$, $\gamma \cdot \sqrt{\varphi m} \cdot \beta/2 \geq 1$, and, with $\tensorexpansion{t} \coloneqq 1$ if $\conductor$ is a power of two and $\tensorexpansion{t} \coloneqq \sqrt{t}$ otherwise,
      \[
          \frac{Q \cdot K \cdot \varphi \cdot m \cdot \expansion_2^2 \cdot (\tensorexpansion{t})^2 \cdot \gamma^2 \cdot \beta^2}{2 \cdot \chi^2} + \frac{1}{\bar{\chi}^2} \leq \frac{1}{2 \eta_\epsilon(\ring^m)^2 + (\chi_2^*)^2}.
      \]
    \item\label{item:cor-eps} $8 \cdot (3Q + t \cdot QK + 3t) \cdot \epsilon \leq 2^{-\delta}$.\footnote{This makes the statistical terms of \cref{thm:tfhe-security} total at most $2^{-\delta}$, using $1/(1-\epsilon) \leq 2$.}
  \end{enumerate}
  Then, the $(t,K)$-threshold FHE scheme $\tFHE_\params$ in \cref{fig:threshold-FHE-fine-grained} is correct (with ``average-case'' correctness) and $Q$-bounded IND-CPA secure.
  In particular, if $\params$, $\chi_1^*$ and $\chi_2^*$ are such that conditions \ref{item:cor-LWE}, \ref{item:cor-LWE-A} and \ref{item:cor-eps} hold for all $Q \in \poly$,
  then $\tFHE_\params$ in \cref{fig:threshold-FHE-fine-grained} is also (unbounded) IND-CPA secure.
\end{corollary}

\begin{proof}
  Follows from putting \cref{thm:correctness,thm:tfhe-security,thm:leakage-security,rmk:average-correctness} together:
  Items 3--4 imply the correctness constraints in \cref{thm:correctness,rmk:average-correctness}; 
  Items 1, 2, 5 imply that $\FHE$ is IND-CPA-with-$(D,K,\chi)$-Leakage secure with $D = \theta$ by \cref{thm:leakage-security};
  security follows from Items 3, 6, 7, 8 by invoking \cref{thm:tfhe-security}.
  \qed
\end{proof}

}{}

We search for parameter sets that minimise the public-key $|\pk|$ and partial-decryption $|\pd|$ sizes, while keeping (partial) decryption reasonably efficient, subject to the correctness and security constraints of \cref{thm:correctness,thm:tfhe-security,thm:leakage-security,rmk:average-correctness}.
We set the plaintext modulus $p = 2^4 = 16$, a common choice for FHEW/TFHE-based applications~\cite{JC:BBBCLO23}, and target $\secpar=128$ bits of computational security, $\delta=64$ bits of statistical security, and $Q = 2^{64}$ partial decryption queries.
Recall that the underlying $\FHE$ must have $(\ring,\hat{n},p,q,\xi,B)$-linear-decryption; we insist on a noise bound $B \leq 2^{70}$ to leave enough room for $\FHE.\PreDecrypt$.
When searching for parameters using the Lattice Estimator~\cite{JMC:AlbPlaSco15}, we keep the degree $\ell$ of gadget decomposition low, for faster partial decryption and decryption, and require $\beta \in [2^{35}, 2^{50}]$ so that it corresponds to one ``limb'' in the RNS decomposition, allowing a more efficient gadget decomposition.

Following this methodology, we found parameter sets for conductors $\conductor\in\{512,4096\}$ and various threshold profiles. The small ring is selected for low latency (a single decryption is fast to compute, with reasonable per-sample efficiency), the large ring for high throughput (when many messages are to be decrypted).
Instead of specifying the threshold $t$ directly, we specify the number of parties $K$ and a window size $w$, covering all profiles with $t\in \{2,\ldots,w\} \cup \{ K - w, \ldots, K\}$. Since the partial decryption size and time scale with $\min\{t,K-t\}$, a window of size $w$ collects exactly the profiles whose cost is no worse than that at $\min\{t,K-t\}=w$.
The resulting parameters are listed in \cref{tab:tulva_params}.

When calculating the size of a public key, we assume that $\mat{A}$ in the public key is sampled using a 256-bit seed, and that the public key only stores this seed and the vector $\vec{b}$.

\providecommand{\fitbox}[1]{\sbox0{#1}\ifdim\wd0>\linewidth\resizebox{\linewidth}{!}{\usebox0}\else\usebox0\fi}

\begin{table}[t]
  \centering
  \setlength{\tabcolsep}{2.5pt}
  \caption{\ourworktfhe parameters across conductors $\conductor\in\{512,4096\}$ (ring degree $\varphi=\conductor/2$), $Q=2^{64}$, $p=2^4$. A set of width $w$ supports every threshold $t \in [2,w] \cup [K-w,K]$. Here the gadget base $\beta$ is constrained to a single RNS limb. The column $|\ctxt|$ gives the size of the input ciphertext, i.e. before $\PreDecrypt$.}\label{tab:tulva_params}
  \fitbox{%
  \begin{tabular}{lccccccccccccccc}
    \toprule
    & $K$ & $w$ & $\secpar$ & $q$ & $\ell$ & $\hat{n}$ & $n$ & $h$ & $m$ & $\chi$ & $\bar{\chi}$ & $|\ctxt|$\,[kB] & $|\pk|$\,[kB] & $|\pd|$\,[kB] & $|\pd|$/bit [b] \\ \midrule
    \texttt{S.1} & \phantom{00}8 & \phantom{0}4 & 130 & $2^{210}$ & 5 & 32 & 10 & \phantom{0}1 & 165 & $2^{193}$ & $2^{97}$ & 18 & 1083 & \phantom{0}0.7 & \phantom{0}5.5 \\
    \texttt{S.2} & \phantom{0}32 & 16 & 131 & $2^{235}$ & 5 & 36 & 16 & \phantom{0}3 & 195 & $2^{205}$ & $2^{81}$ & 18 & 1432 & \phantom{0}1.0 & \phantom{0}8.3 \\
    \texttt{S.3} & 128 & 16 & 129 & $2^{300}$ & 6 & 45 & 27 & \phantom{0}9 & 324 & $2^{240}$ & $2^{72}$ & 18 & 3038 & \phantom{0}2.0 & 16.0 \\
    \texttt{S.4} & 128 & 64 & 128 & $2^{329}$ & 7 & 49 & 38 & 17 & 462 & $2^{253}$ & $2^{45}$ & 18 & 4750 & \phantom{0}2.5 & 19.8 \\
    \texttt{S.5} & 512 & 16 & 128 & $2^{368}$ & 8 & 55 & 37 & 17 & 576 & $2^{276}$ & $2^{71}$ & 18 & 6624 & \phantom{0}3.0 & 23.8 \\
    \addlinespace
    \texttt{L.1} & \phantom{00}8 & \phantom{0}4 & 162 & $2^{215}$ & 5 & \phantom{0}5 & \phantom{0}2 & \phantom{0}1 & \phantom{0}30 & $2^{196}$ & $2^{98}$ & 32 & 1613 & \phantom{0}5.8 & \phantom{0}5.8 \\
    \texttt{L.2} & \phantom{0}32 & 16 & 144 & $2^{240}$ & 5 & \phantom{0}5 & \phantom{0}3 & \phantom{0}1 & \phantom{0}30 & $2^{208}$ & $2^{82}$ & 32 & 1800 & \phantom{0}8.5 & \phantom{0}8.5 \\
    \texttt{L.3} & 128 & 16 & 134 & $2^{308}$ & 7 & \phantom{0}6 & \phantom{0}4 & \phantom{0}2 & \phantom{0}56 & $2^{247}$ & $2^{82}$ & 32 & 4312 & 16.8 & 16.8 \\
    \texttt{L.4} & 128 & 64 & 137 & $2^{336}$ & 7 & \phantom{0}7 & \phantom{0}5 & \phantom{0}3 & \phantom{0}70 & $2^{259}$ & $2^{48}$ & 32 & 5880 & 20.5 & 20.5 \\
    \texttt{L.5} & 512 & 16 & 128 & $2^{376}$ & 8 & \phantom{0}7 & \phantom{0}5 & \phantom{0}3 & \phantom{0}80 & $2^{283}$ & $2^{75}$ & 32 & 7520 & 24.8 & 24.8 \\
    \bottomrule
  \end{tabular}}
\end{table}


\subsection{Experiments}
\label{sec:experiments}
All experiments were conducted on a single AWS \texttt{c7a.32xlarge} instance, equipped with 128 AMD EPYC 9R14 (Genoa) vCPUs and 256\,GB of RAM. For \texttt{S.[1-5]}, we considered two types of experiments: decryption of a single message and decryption of $\varphi$ messages (the full ring), the latter having a substantial impact on the $\PreDecrypt$ phase. For \texttt{L.[1-5]}, we report only the full-ring decryption results, as employing such a large ring for single-message decryption is unreasonable. The reported times are available in \cref{tab:tulva_runtimes}. We do not report the $\kgen$ time as it was not optimised (it takes at most 36 seconds) and is a one-time cost that can be amortised over many decryptions. We report $\PreDecrypt$ as a subroutine separate from $\ParDec$ and $\Dec$ to highlight the cost of the pre-decryption phase, which is the most expensive part (yet, secret independent). 

\providecommand{\fitbox}[1]{\sbox0{#1}\ifdim\wd0>\linewidth\resizebox{\linewidth}{!}{\usebox0}\else\usebox0\fi}

\begin{table}[t]
  \centering
  \setlength{\tabcolsep}{2pt}
  \caption{Runtimes of the threshold layer on the parameter sets of \cref{tab:tulva_params}, using the \cite{AC:CGGI16} bootstrapping.}\label{tab:tulva_runtimes}
  \fitbox{%
  \begin{tabular}{lccccc}
    \toprule
    Set & $t/K$ & msgs & $\PreDecrypt$\,[s] & $\ParDec$\,[ms] & $\Dec$\,[\textmu s] \\ \midrule
    \texttt{S.1} & $\phantom{0}4/8\phantom{00}$ & $1$ & \phantom{00}0.13 & 0.64 & \phantom{00}39 \\
      &  & $\varphi$ & \phantom{00}5.83 & 0.51 & \phantom{00}39 \\
    \addlinespace[2pt]
    \texttt{S.2} & $16/32\phantom{0}$ & $1$ & \phantom{00}0.16 & 0.60 & \phantom{00}91 \\
      &  & $\varphi$ & \phantom{00}7.70 & 0.43 & \phantom{00}77 \\
    \addlinespace[2pt]
    \texttt{S.3} & $16/128$ & $1$ & \phantom{00}0.24 & 1.11 & \phantom{00}94 \\
      &  & $\varphi$ & \phantom{0}12.68 & 0.77 & \phantom{00}82 \\
    \addlinespace[2pt]
    \texttt{S.4} & $64/128$ & $1$ & \phantom{00}0.36 & 2.80 & \phantom{0}223 \\
      &  & $\varphi$ & \phantom{0}22.58 & 2.98 & \phantom{0}610 \\
    \addlinespace[2pt]
    \texttt{S.5} & $16/512$ & $1$ & \phantom{00}0.47 & 3.68 & \phantom{0}534 \\
      &  & $\varphi$ & \phantom{0}31.03 & 4.26 & \phantom{0}562 \\
    \midrule
    \texttt{L.1} & $\phantom{0}4/8\phantom{00}$ & $\varphi$ & \phantom{0}74.17 & 0.63 & \phantom{00}70 \\
    \texttt{L.2} & $16/32\phantom{0}$ & $\varphi$ & \phantom{0}80.08 & 0.34 & \phantom{0}403 \\
    \texttt{L.3} & $16/128$ & $\varphi$ & 144.95 & 0.62 & \phantom{0}967 \\
    \texttt{L.4} & $64/128$ & $\varphi$ & 258.87 & 1.25 & 1495 \\
    \texttt{L.5} & $16/512$ & $\varphi$ & 303.27 & 1.10 & \phantom{0}875 \\
    \bottomrule
  \end{tabular}}
\end{table}

\cref{tab:tulva_runtimes} confirms the expected scaling of \ourworktfhe. For the small ring, a single decryption is fast across all profiles: the pre-decryption stays below one second and the online cost is at most a few milliseconds of partial decryption followed by a combination step of tens to hundreds of microseconds (e.g.\ $0.13$\,s, $0.64$\,ms and $39$\,\textmu s for \texttt{S.1}). $\PreDecrypt$ is by far the most expensive step and, as expected, scales with the number of decrypted messages, growing from sub-second for a single message to between $6$ and $31$\,s for the full $\varphi=256$ ring. At the same time, the per-message partial-decryption and combination times stay within a factor of three. Decrypting the full ring therefore amortises this fixed cost over many messages, and the large ring (\texttt{L.[1-5]}) pushes this further by packing $\varphi=2048$ messages per batch. Finally, the online cost scales with $\min\{t,K-t\}$ rather than by $K$: for instance \texttt{L.5} ($t/K=16/512$) combines in $0.88$\,ms, faster than \texttt{L.4} ($t/K=64/128$) at $1.50$\,ms, which stays consistent with \cref{rem:lagrange-interpolation}.

\medskip\noindent\textbf{Comparison with \cite{CCS:ZZLP25}.}
The concurrent work of \cite{CCS:ZZLP25} also targets LWE threshold decryption, removing the decryption noise through an interactive \emph{MPC rounding} rather than masking it. We compare their secret-dependent online latency for a single ciphertext \cite[Table~2]{CCS:ZZLP25} with our $\ParDec+\text{(network delay)}+\Dec$, matching their party count to our active threshold $t$ (\cref{tab:zzlp_comparison}).
For Tulva, the reported latency covers partial decryption, communication, and share combination after public ciphertext preprocessing has completed.
Preprocessing is excluded from this comparison and reported separately in \cref{tab:tulva_runtimes}.
The results show lower latency for Tulva and a smaller increase as network delay grows, consistent with its single communication round. Their ``preparation'' phase is interactive and generates secret correlated randomness, whereas our $\PreDecrypt$ is public, non-interactive, and ciphertext-dependent, and can be performed by an untrusted server.
Tulva's shorter secret-dependent phase comes at the cost of ciphertext-dependent preprocessing, which takes $0.13$\,s for \texttt{S.1} and $0.16$\,s for \texttt{S.2} (\cref{tab:tulva_runtimes}).
The preparation phase of \cite{CCS:ZZLP25} can instead be completed before the ciphertext becomes available.

\begin{table}[t]
  \centering
  \setlength{\tabcolsep}{5pt}
  \caption{Secret-dependent online-phase latency for a single ciphertext.}\label{tab:zzlp_comparison}
  \begin{tabular}{lcccc}
    \toprule
    Set & $t$ & ping & \cite{CCS:ZZLP25}\,[ms] & ours\,[ms] \\ \midrule
    \texttt{S.1} & \phantom{0}4 & $\phantom{0}1$\,ms & $\phantom{0}8.48$ & $\phantom{0}1.68$ \\
    \texttt{S.2} & 16 & $\phantom{0}1$\,ms & $24.23$ & $\phantom{0}1.69$ \\
    \addlinespace[2pt]
    \texttt{S.1} & \phantom{0}4 & $10$\,ms & $55.62$ & $10.68$ \\
    \texttt{S.2} & 16 & $10$\,ms & $56.98$ & $10.69$ \\
    \bottomrule
  \end{tabular}
\end{table}


\section{Examples of FHEVM Programs}\label{sec:fhevm}

Threshold FHE can be used to implement confidential smart contracts:
many users encrypt private inputs under one FHE public key, a neutral evaluator (the validator set) applies a public function $f$ to the ciphertexts while learning nothing about the inputs, and a $t$-of-$K$ committee runs the $\ParDec$ and $\Dec$ procedures of the threshold FHE scheme to reveal only a chosen part of the output.
The parties that hold the data are thus separated from the party that computes and from the parties that decrypt.
Furthermore, the committee composition can change dynamically over time.

We describe four concrete FHEVM programs and report the runtimes of their implementations using \ourworktfhe.
Detailed descriptions of the functions $f$ are deferred to \cref{app:fhevm}.
Since every value occupies a single $\ZZ_p$ element, each program is correct only under explicit domain constraints that bound its inputs, e.g.~the value ranges such as $[0,B)$ below. The exact constraints are given in \cref{app:fhevm}.
For each program we report runtimes over the five small-ring parameter sets \texttt{S.1}--\texttt{S.5} of \cref{tab:tulva_params}:
homomorphic evaluation of $f$ ($\Eval$), the functional bootstrap that carries each revealed value to the committee ($\PreDecrypt$), and the committee's partial and final decryption ($\ParDec$, $\Dec$).
The three profile-dependent phases ($\PreDecrypt$, $\ParDec$, $\Dec$) are listed per profile and $\Eval$ is listed once.
As in \cref{sec:experiments}, we omit key generation, a one-time cost that is amortised over many reveals.

\begin{table}[!t]
  \centering
  \setlength{\tabcolsep}{3pt}
  \caption{Sealed-bid auction (\cref{alg:fhevm-auction}).}
  \label{tab:fhevm-auction}
  \fullversion{}{\adjustbox{max width=\columnwidth}}%
  {%
  \begin{tabular}{@{}llcccc@{}}
    \toprule
    $\mathsf{Set}$ & $t/K$ & $\Eval$ [ms] & $\PreDecrypt$ [ms] & $\ParDec$ [ms] & $\Dec$ [\textmu s] \\
    \midrule
    \texttt{S.1} & $\phantom{0}4/8\phantom{00}$ & \multirow{5}{*}{513} & \phantom{0}275 & 0.6 & \phantom{0}42 \\
    \texttt{S.2} & $16/32\phantom{0}$ & & \phantom{0}310 & 0.6 & \phantom{0}90 \\
    \texttt{S.3} & $16/128$ & & \phantom{0}492 & 1.2 & \phantom{0}98 \\
    \texttt{S.4} & $64/128$ & & \phantom{0}724 & 2.7 & 220 \\
    \texttt{S.5} & $16/512$ & & 1120 & 4.0 & 206 \\
    \bottomrule
  \end{tabular}}
  \par\smallskip
  {\footnotesize\emph{Parameters:} $n = 4$ bidders, max. bid $B = 4$.}
\end{table}

\medskip\noindent\textbf{1. Sealed-Bid Auction (\cref{tab:fhevm-auction}).} 
A pool of $n$ users encrypt their bids over $[0, B)$ and submit them to the contract, which homomorphically computes the winning bid and bidder index.
The contract reveals the winner and the price at which the auction clears but not the losing bids.

\begin{table}[!t]
  \centering
  \setlength{\tabcolsep}{3pt}
  \caption{Confidential ERC-20 (\cref{alg:fhevm-erc20}).}
  \label{tab:fhevm-erc20}
  \fullversion{}{\adjustbox{max width=\columnwidth}}%
  {%
  \begin{tabular}{@{}llcccc@{}}
    \toprule
    $\mathsf{Set}$ & $t/K$ & $\Eval$ [s] & $\PreDecrypt$ [ms] & $\ParDec$ [ms] & $\Dec$ [\textmu s] \\
    \midrule
    \texttt{S.1} & $\phantom{0}4/8\phantom{00}$ & \multirow{5}{*}{98} & \phantom{0}310 & 0.5 & \phantom{0}38 \\
    \texttt{S.2} & $16/32\phantom{0}$ & & \phantom{0}375 & 0.4 & \phantom{0}74 \\
    \texttt{S.3} & $16/128$ & & \phantom{0}556 & 1.3 & \phantom{0}83 \\
    \texttt{S.4} & $64/128$ & & \phantom{0}902 & 3.0 & 574 \\
    \texttt{S.5} & $16/512$ & & 1156 & 4.0 & 176 \\
    \bottomrule
  \end{tabular}}
  \par\smallskip
  {\footnotesize\emph{Parameters:} $n = 256$ accounts, $R = 10$ rounds, max. balance $B = 4$, max. amount $A = 4$.}
\end{table}

\medskip\noindent\textbf{2. Confidential ERC-20 Transfer (\cref{tab:fhevm-erc20}).} 
ERC-20 is the standard smart-contract interface for fungible Ethereum tokens. 
In our specific example, account balances $[0, B)$ are encrypted, and a transfer moves an encrypted amount over $[0, A)$ between accounts. Throughout the protocol, multiple transfers can be processed sequentially (each with its own round). The contract only reveals the final balances of selected accounts. The rest of the state remains encrypted.

\begin{table}[!t]
  \centering
  \setlength{\tabcolsep}{3pt}
  \caption{Dark-pool DEX (\cref{alg:fhevm-darkpool}).}
  \label{tab:fhevm-darkpool}
  \fullversion{}{\adjustbox{max width=\columnwidth}}%
  {%
  \begin{tabular}{@{}llcccc@{}}
    \toprule
    $\mathsf{Set}$ & $t/K$ & $\Eval$ [s] & $\PreDecrypt$ [ms] & $\ParDec$ [ms] & $\Dec$ [\textmu s] \\
    \midrule
    \texttt{S.1} & $\phantom{0}4/8\phantom{00}$ & \multirow{5}{*}{9.3} & 236 & 0.6 & \phantom{0}39 \\
    \texttt{S.2} & $16/32\phantom{0}$ & & 264 & 0.4 & \phantom{0}73 \\
    \texttt{S.3} & $16/128$ & & 348 & 1.1 & \phantom{0}84 \\
    \texttt{S.4} & $64/128$ & & 507 & 3.0 & 203 \\
    \texttt{S.5} & $16/512$ & & 653 & 4.1 & 178 \\
    \bottomrule
  \end{tabular}}
  \par\smallskip
  {\footnotesize\emph{Parameters:} $S = 4$ resting orders, $R = 10$ rounds, max. price $P = 4$, max. volume $V = 4$.}
\end{table}

\medskip\noindent\textbf{3. Dark-Pool DEX (\cref{tab:fhevm-darkpool}).} The contract holds an encrypted order book. It matches incoming requests against it, publishing only the executed fills while prices over $[0, P)$, volumes over $[0, V)$, and unmatched orders remain encrypted. After removing the matched orders, the contract updates the order book with new encrypted submissions. This process repeats for $R$ rounds, and the unmatched orders are never revealed.

\begin{table}[!t]
  \centering
  \setlength{\tabcolsep}{3pt}
  \caption{Median price oracle (\cref{alg:fhevm-oracle}).}
  \label{tab:fhevm-oracle}
  \fullversion{}{\adjustbox{max width=\columnwidth}}%
  {%
  \begin{tabular}{@{}llcccc@{}}
    \toprule
    $\mathsf{Set}$ & $t/K$ & $\Eval$ [s] & $\PreDecrypt$ [ms] & $\ParDec$ [ms] & $\Dec$ [\textmu s] \\
    \midrule
    \texttt{S.1} & $\phantom{0}4/8\phantom{00}$ & \multirow{5}{*}{28} & 237 & 0.5 & \phantom{0}38 \\
    \texttt{S.2} & $16/32\phantom{0}$ & & 262 & 0.5 & \phantom{0}73 \\
    \texttt{S.3} & $16/128$ & & 348 & 1.0 & \phantom{0}80 \\
    \texttt{S.4} & $64/128$ & & 499 & 3.1 & 210 \\
    \texttt{S.5} & $16/512$ & & 680 & 4.3 & 178 \\
    \bottomrule
  \end{tabular}}
  \par\smallskip
  {\footnotesize\emph{Parameters:} $n = 128$ submissions, max. value $H = 8$.}
\end{table}

\medskip\noindent\textbf{4. Median Price Oracle (\cref{tab:fhevm-oracle}).} The contract aggregates $n$ encrypted submissions over $[0, H)$ from many nodes into their median. Computing the median homomorphically requires obliviously sorting the inputs, for which we employ Batcher's odd--even mergesort.

\medskip\noindent
Across all four programs, the threshold reveal ($\PreDecrypt + \ParDec + \Dec$) costs at most a few hundred milliseconds for the small committees and stays below $1.2$ seconds even for $K = 512$. For the computationally heavy programs this overhead is negligible relative to the evaluation. Thus, for any non-trivial computation, \ourworktfhe adds little overhead over a non-threshold deployment.

\section*{AI Disclosure}

The construction, all proof ideas and the parameter analysis are due to the authors. The first version of this work was written entirely by hand and was then proofread by Claude Opus~5 and~5.5 as well as by GPT-5.6 Sol. The models pointed out small gaps in some proofs and proposed corrections, which the authors verified before adopting, and the authors revised every passage edited by the models in subsequent editing rounds. The scripts that select the parameters of \cref{tab:tulva_params} were first written by hand, and Claude Opus~5 and~5.5 were then used to optimise their search logic and to refactor them. The authors checked that the resulting scripts are consistent with the analysis in the paper. The algorithms in the implementation were first written by hand, and Claude Opus~5 and~5.5 as well as GPT-6 Astra were used extensively to rewrite and optimise existing parts of the implementation.

\section*{Acknowledgements}

Russell W.\ F.\ Lai and Micha{\l} Osadnik are supported by the Research Council of Finland grant 358951 and the Future Makers Program project SPICED of the Technology Industries of Finland Centennial Foundation.
Ivy K.\ Y.\ Woo is supported by the Research Council of Finland grant 358950. 
Antonio Guimar{\~a}es is partially supported by grants JDC2024-055789-I and CEX2024-001471-M funded by MICIU/AEI/10.13039/501100011033 and ESF+, and by the ERC Proof of Concept grant VERIFHE (GA: 101287502). 
Valerio Cini is supported by the European Union's
Horizon Europe research and innovation programme under the
Marie Sk{\l}odowska-Curie grant agreement No.~101202597 (ACryL).
Ivy K.\ Y.\ Woo would like to thank Yongsoo Song for sharing the reduction in \cite{EPRINT:HKMNS26} during the Novel Lattice Assumptions Workshop in July 2026, which significantly improved the results of this work.

\FloatBarrier

\ifstyle{2}{
    \bibliographystyle{IEEEtran}
}{
    \fullversion{
        \bibliographystyle{alphaurl}
    }{
        \bibliographystyle{splncs04}
    }
}
\bibliography{cryptobib/abbrev3,cryptobib/crypto,local}

\ifstyle{1}{
  \appendix
  \clearpage
  \crefalias{section}{appendix} 
  \crefalias{subsection}{appendix} 
}

\ifstyle{2}{\appendices}


\section{Algorithms for FHEVM Programs}\label{app:fhevm}

For each of the four programs in \cref{sec:fhevm}, we state the plaintext algorithm and its homomorphic evaluation at the level of ciphertexts.
Throughout we write $\enc(x)$ for a ciphertext encrypting $x$, i.e.~shorthand for $\tFHE.\enc(\pk, x)$.
We use $\{0,1,\ldots,p-1\}$ to represent $\ZZ_p$.
We write $\mathbf{1}\{c\}$ for the bit $1$ when a condition $c$ holds and $0$ otherwise, and $(c\,?\,a:b)$ for the value $a$ if $c$ holds and $b$ otherwise.

\begin{table*}[t]
  \centering
  \footnotesize
  \setlength{\tabcolsep}{5pt}
  \caption{Functional-bootstrap tables instantiated by the FHEVM programs.}
  \label{tab:fhevm-tables}
  \begin{tabularx}{\linewidth}{@{}llXl@{}}
    \toprule
    Table & Definition & Description & Program(s) \\
    \midrule
    \multicolumn{4}{@{}l}{\textit{One-input table}} \\
    $T_{=c}$            & $\mathbf{1}\{x = c\}$               & equality test of $x$ against a public constant $c$ & Auction, Dark-pool \\
    $T_{\leq a,b}$      & $(x \le a)\,?\,x : b$              & keep $x$ when $x \le a$, otherwise lift it to $b$ & Dark-pool \\
    $T_{\mathsf{refresh}}$ & $x$                             & identity bootstrap (a ``refresh''): clears accumulated noise & Dark-pool \\
    \midrule
    \multicolumn{4}{@{}l}{\textit{Two-input table}} \\
    $T_>$               & $\mathbf{1}\{x > y\}$               & greater-than test of two bids & Auction \\
    $T_{\mathsf{sel}}$  & $x\,?\,y : 0$                       & multiplexer: pass value $y$ when bit $x$ is set, else $0$ & Auction, Dark-pool \\
    $T_{\mathsf{debit}}$ & $(x \ge y)\,?\,x{-}y : x$          & new sender balance: $x$ less amount $y$ when affordable & ERC-20 \\
    $T_{\mathsf{credit}}$ & $(x \ge y)\,?\,y : 0$             & amount credited to the receiver, or $0$ if balance $x < y$ & ERC-20 \\
    $T_{=}$             & $\mathbf{1}\{x = y\}$               & equality test of two encrypted values & Dark-pool \\
    $T_{\mathsf{scan}}$ & $\mathbf{1}\{x = 1 \wedge y = 0\}$  & prefix-scan step: fires when equality bit $x = 1$ and carry $y = 0$ & Dark-pool \\
    $T_\wedge$          & $x \wedge y$                        & logical AND (one-hot $x$, empty-slot bit $y$): the replace flag & Dark-pool \\
    \midrule
    \multicolumn{4}{@{}l}{\textit{Two-input sort}} \\
    $T_{\mathsf{sort}}$ & $(\min(x,y),\, \max(x,y))$          & sort the pair $(x,y)$ into $(\min, \max)$ & Dark-pool, Oracle \\
    \bottomrule
  \end{tabularx}
\end{table*}

Homomorphic addition and multiplication by a public scalar need no bootstrapping (bts) and are evaluated directly. Everything else is built from the tables in \cref{tab:fhevm-tables}, the gadgets shared across the four programs. Each table $T$ takes one or two encrypted inputs, returns one encrypted output (or two for $T_{\mathsf{sort}}$), and is evaluated by a single functional bootstrap. Using these gadgets, we describe the homomorphic realisations of the four programs below, with pseudocode listed in \cref{alg:fhevm-auction,alg:fhevm-erc20,alg:fhevm-darkpool,alg:fhevm-oracle}.
In each program, the \highlight{highlighted comments} indicate the number of bootstraps required for each step.

\begin{algorithm2e}[t]
  \caption{Sealed-bid first-price auction.}\label{alg:fhevm-auction}
  \SetKwInput{Input}{Input}
  \SetKwInput{Output}{Output}
  \DontPrintSemicolon
  \Input{distinct bids $b_0, \ldots, b_{n-1} \in \{0, \ldots, B-1\}$}
  \Output{winner index $w$ and clearing price $b_w$}
  \BlankLine
  \For{$i \gets 0$ \KwTo $n-1$}{
    $s_i \gets 0$ \tcp*{win count of $i$.}
    \For{$j \gets 0$ \KwTo $n-1$ \textbf{with} $j \neq i$}{
      \lIf{$b_i > b_j$}{$s_i \gets s_i + 1$} \tcp{\highlight{$s_i \gets s_i + T_>(b_i, b_j)$: $n(n{-}1)$ bts}}
    }
  }
  $o_i \gets \mathbf{1}\{s_i = n-1\}$ for all $i$, then $w \gets \sum_i i \cdot o_i$ \tcp*{\highlight{$w \gets \sum_i i \cdot T_{=c}(s_i)$: $n$ bts}}
  $b_w \gets \sum_i (o_i\,?\,b_i : 0)$ \tcp*{\highlight{$b_w \gets \sum_i T_{\mathsf{sel}}(o_i, b_i)$: $n$ bts}}
  \Return $(w,\ b_w)$
\end{algorithm2e}

\medskip\noindent\textbf{1. Sealed-Bid Auction (\cref{alg:fhevm-auction}).}
A first-price auction reveals only the winner and the price they pay. Assume the bids are distinct; then one bidder $i$ outbids all $n-1$ others. To find that bidder, form the win count $s_i = \sum_{j \ne i} T_>(b_i, b_j)$ for bids in $[0,B)$. Computing all $n$ of them costs $n(n-1)$ bootstraps. Bidder $i$ wins exactly when $s_i = n-1$, so $o_i = T_{=c}(s_i)$ at $c = n-1$ is a one-hot indicator of the winner. The two outputs are linear in $o_i$. The winner index is $w = \sum_i i\,o_i$; the clearing price is $b_w = \sum_i T_{\mathsf{sel}}(o_i, b_i)$, where the multiplexer $T_{\mathsf{sel}}$ keeps $b_i$ on the winning slot and zeroes the rest. Evaluation is $n^2 + n$ bootstraps (\cref{tab:fhevm-auction}), and each of $w, b_w$ is read by one $\PreDecrypt$. Correctness needs $B^2 \le p$ for the two-input comparison $T_>$ on bids in $[0,B)$, and $n \le p$ to keep the win counts $s_i$ and the winner index $w$ below $p$.

\begin{algorithm2e}[t]
  \caption{Confidential ERC-20 transfers.}\label{alg:fhevm-erc20}
  \SetKwInput{Input}{Input}
  \SetKwInput{Output}{Output}
  \DontPrintSemicolon
  \Input{sender balances $s_0, \ldots, s_{n-1}$, receiver balances $r_0, \ldots, r_{n-1}$, per-round transfer amounts $a^{(\rho)}_i$}
  \Output{updated balances $(s_i, r_i)_i$}
  \BlankLine
  \For{$\rho \gets 1$ \KwTo $R$}{
    \For{$i \gets 0$ \KwTo $n-1$}{
      \If{$s_i \ge a^{(\rho)}_i$}{
        $r_i \gets r_i + a^{(\rho)}_i$ \tcp*{\highlight{$r_i \mathrel{+}= T_{\mathsf{credit}}(s_i, a^{(\rho)}_i)$: $Rn$ bts}}
        $s_i \gets s_i - a^{(\rho)}_i$ \tcp*{\highlight{$s_i \gets T_{\mathsf{debit}}(s_i, a^{(\rho)}_i)$: $Rn$ bts}}
      }
    }
  }
  \Return $(s_i, r_i)_i$ \tcp*{queried accounts only. $1$ $\PreDecrypt$ each}
\end{algorithm2e}

\medskip\noindent\textbf{2. Confidential ERC-20 Transfer (\cref{alg:fhevm-erc20}).}
Each account on a confidential ERC-20 ledger carries an encrypted balance, and a transfer moves tokens without exposing any balance or the amount sent. Account $i$ holds a sender balance $s_i$ and a receiver balance $r_i$ as persistent state. A transfer debits the sender only when its funds suffice.
In round $\rho$, two two-input bootstraps on $(s_i, a^{(\rho)}_i)$ decide account $i$'s transfer obliviously: $T_{\mathsf{debit}}(s_i, a^{(\rho)}_i)$ returns the new sender balance ($s_i - a^{(\rho)}_i$ if $s_i \ge a^{(\rho)}_i$, else $s_i$) and $T_{\mathsf{credit}}(s_i, a^{(\rho)}_i)$ the moved amount ($a^{(\rho)}_i$ if affordable, else $0$). The receiver is credited by $r_i \mathrel{+}= T_{\mathsf{credit}}(s_i, a^{(\rho)}_i)$, and $T_{\mathsf{debit}}(s_i, a^{(\rho)}_i)$ replaces $s_i$, for $2Rn$ bootstraps in total.
A balance is read with one $\PreDecrypt$. The encoding is correct while $A B \le p$ for the two-input tables $T_{\mathsf{debit}}, T_{\mathsf{credit}}$ on balances in $[0,B)$ and amounts in $[0,A)$, and every reachable balance stays below $p$ as a receiver balance accumulates credits across rounds.

\begin{algorithm2e*}[t]
  \caption{Dark-pool order matching (stateful, $R$ rounds).}\label{alg:fhevm-darkpool}
  \SetKwInput{Input}{Input}
  \SetKwInput{Output}{Output}
  \DontPrintSemicolon
  \Input{book of $S$ resting asks $(\pi_s, \nu_s)_s$ (price, volume); requests $(\ell^{(\rho)}, q^{(\rho)})$ and pending orders $(\hat{\pi}^{(\rho)}, \hat{\nu}^{(\rho)})$, for $\rho = 1, \ldots, R$}
  \Output{per-round fills $f^{(1)}, \ldots, f^{(R)}$}
  \BlankLine
  \For{$\rho \gets 1$ \KwTo $R$}{
    \lFor{$s \gets 0$ \KwTo $S-1$}{$\kappa_s \gets (\pi_s \le \ell^{(\rho)})\,?\,\pi_s : P$}
    \tcp{\highlight{$\kappa_s \gets T_{\leq \ell^{(\rho)},P}(\pi_s)$: $RS$ bts}}
    $\beta \gets \min_s \kappa_s$\;
    \tcp{\highlight{$S{-}1$ comparisons $(\beta,\,\cdot) \gets T_{\mathsf{sort}}(\beta, \kappa_s)$, then $\beta \gets T_{\mathsf{refresh}}(\beta)$: $RS$ bts}}
    \eIf{$\beta = P$}{
      $f^{(\rho)} \gets 0$ \tcp*{nothing crosses (branchless)}
    }{
      $s^\star \gets$ first $s$ with $\pi_s = \beta$\;
      \tcp{\highlight{$\forall s$ ($S$ slots): $e_s \gets T_{=}(\pi_s, \beta)$, $m_s \gets T_{\mathsf{scan}}(e_s, g)$: $2RS$ bts}}
      $f^{(\rho)} \gets \min(q^{(\rho)}, \nu_{s^\star})$\;
      \tcp{\highlight{$v \gets T_{\mathsf{refresh}}\bigl(\textstyle\sum_s T_{\mathsf{sel}}(m_s, \nu_s)\bigr)$: $S{+}1$ bts,}}
      \tcp{\highlight{$(f^{(\rho)},\,\cdot) \gets T_{\mathsf{sort}}(q^{(\rho)}, v)$: ${+}1$ bt, total $R(S{+}2)$}}
      $\nu_{s^\star} \gets \nu_{s^\star} - f^{(\rho)}$
      \tcp{\highlight{$\forall s$: $\nu_s \mathrel{-}= T_{\mathsf{sel}}(m_s, f^{(\rho)})$ (only $s^\star$ changes): $RS$ bts}}
      \lIf{$\nu_{s^\star} = 0$}{$(\pi_{s^\star}, \nu_{s^\star}) \gets (\hat{\pi}^{(\rho)}, \hat{\nu}^{(\rho)})$}
      \tcp{\highlight{$d \gets T_{=0}(\nu_{s^\star})$, $\forall s$: $h_s \gets T_\wedge(m_s, d)$,}}
      \tcp{\highlight{$\pi_s \gets T_{\mathsf{sel}}(h_s, \hat{\pi}^{(\rho)}) + T_{\mathsf{sel}}(1{-}h_s, \pi_s)$,}}
      \tcp{\highlight{$\nu_s \gets T_{\mathsf{sel}}(h_s, \hat{\nu}^{(\rho)}) + T_{\mathsf{sel}}(1{-}h_s, \nu_s)$: $(5S{+}1)R$ bts}}
    }
  }
  \Return $(f^{(\rho)})_\rho$ \tcp*{book stays hidden. $1$ $\PreDecrypt$ per round}
\end{algorithm2e*}

\begin{algorithm2e*}[t]
  \caption{Median price oracle.}\label{alg:fhevm-oracle}
  \SetKwInput{Input}{Input}
  \SetKwInput{Output}{Output}
  \DontPrintSemicolon
  \Input{price submissions $x_0, \ldots, x_{n-1} \in \{0, \ldots, H-1\}$}
  \Output{robust aggregate price}
  \BlankLine
  $(y_0, \ldots, y_{n-1}) \gets \mathsf{sort}(x_0, \ldots, x_{n-1})$ \tcp*{nondecreasing.}
  \tcp{\highlight{Batcher network of $T_{\mathsf{sort}}$: $\bigl((\log_2 n)^2 - \log_2 n + 4\bigr)\,n/4 - 1$ bts}}
  \Return $y_{\lfloor n/2 \rfloor}$ \tcp*{median only. $1$ $\PreDecrypt$}
\end{algorithm2e*}

\medskip\noindent\textbf{3. Dark-Pool DEX (\cref{alg:fhevm-darkpool}).}
The resting book of a dark pool consists of $S$ orders $(\pi_s, \nu_s)$, prices and volumes held encrypted state and never revealed; only the trades executed against it become public. Each round matches one buy request $(\ell^{(\rho)}, q^{(\rho)})$, a limit price and a quantity, against the cheapest ask it crosses, reports the fill $f^{(\rho)}$, and updates the book, refilling an emptied slot from a pending order $(\hat{\pi}^{(\rho)}, \hat{\nu}^{(\rho)})$. To keep the matched slot secret, every step runs over all $S$ slots and carries the match in a one-hot indicator $m_s$. The round has three stages: \emph{cross} prices the eligible asks and takes the lowest, $\beta$; \emph{locate} marks the first slot at $\beta$ as $m_s$; and \emph{fill and update} computes the fill and writes the book back.

\medskip\noindent\emph{Cross.}
An ask crosses when $\pi_s \le \ell^{(\rho)}$. The crossing key $\kappa_s = T_{\le \ell^{(\rho)}, P}(\pi_s)$ keeps a crossing ask at its own price and lifts every other to a reserved value $P$ above all prices (prices in $[0,P)$), so the lowest crossing price is $\beta = \min_s \kappa_s$. This minimum is a tournament of $S-1$ pairwise sorts $T_{\mathsf{sort}}$, then one $T_{\mathsf{refresh}}$.

\medskip\noindent\emph{Locate.}
The match is the first slot at price $\beta$, a one-hot $m_s$. A per-slot equality $e_s = T_=(\pi_s, \beta)$ flags every slot at $\beta$, and a left-to-right prefix scan keeps the first. The scan carries a flag $g$, initially $0$, that flips once a match has been passed: $T_{\mathsf{scan}}$ sets $m_s = 1$ exactly when $e_s = 1$ and $g = 0$, then $g \mathrel{+}= m_s$. The scan is sequential, slot $s$ reading slot $s-1$'s flag.

\medskip\noindent\emph{Fill and update.}
The matched volume $v = \sum_s T_{\mathsf{sel}}(m_s, \nu_s)$ is cleaned by $T_{\mathsf{refresh}}$, and the fill $f^{(\rho)} = \min(q^{(\rho)}, v)$ is one further $T_{\mathsf{sort}}$. The matched ask is drawn down by $\nu_s \mathrel{-}= T_{\mathsf{sel}}(m_s, f^{(\rho)})$, and an emptied match is refilled: $d = T_{=0}(\nu_{s^\star})$ tests emptiness, $h_s = T_\wedge(m_s, d)$ flags the slot, and every slot is rewritten as $\pi_s \gets T_{\mathsf{sel}}(h_s, \hat{\pi}^{(\rho)}) + T_{\mathsf{sel}}(1 - h_s, \pi_s)$ ($\nu_s$ analogously). All $S$ slots are rewritten, so only the matched one changes.

For volumes in $[0, V)$, the round is correct while $P \le V$, $2V \le p$ (the fill sorts the volumes $q^{(\rho)}$ and the matched volume, both $< p/2$), and $P < p/2$.

\medskip\noindent\textbf{4. Median Price Oracle (\cref{alg:fhevm-oracle}).}
A median price oracle aggregates $n$ encrypted submissions into their median. The circuit sorts them obliviously and returns the middle element. The sorting algorithm is Batcher's odd-even merge-sort. Each comparator is one two-input sort $T_{\mathsf{sort}}$, hence one bootstrap. For $n$ a power of two the network has $\bigl((\log_2 n)^2 - \log_2 n + 4\bigr)\,n/4 - 1$ comparators ($1471$ at $n = 128$). Reading the median $y_{\lfloor n/2 \rfloor}$ is one $\PreDecrypt$. Correctness needs every submission in $[0,H)$ with $H = p/2$, the domain of the two-input sort.

\fullversion{

\section{Additional Preliminaries}\label{sec:add-prelim}
We give some additional preliminaries beyond those in \cref{sec:prelim} that are used in the proofs of our main theorems.

\medskip
We write $\otimes$ for the tensor (i.e.~Kronecker) product of matrices.
The matrix tensor product satisfies the mixed product property: For all matrices $\mat{A}$, $\mat{B}$, $\mat{C}$ and $\mat{D}$ of suitable dimensions, we have
\(
(\mat{A}\otimes\mat{B}) (\mat{C}\otimes\mat{D})
= \mat{A}\mat{C} \otimes \mat{B}\mat{D}
\).

\medskip
\noindent\textbf{Supplementary lemmas.}
Below are lemmas on the ring expansion factors, subtractive sets, Gaussian norm bounds, and noise flooding.

\begin{lemma}\label{lem:norm-conversion}
    If $\conductor$ is a power of two, then $\expansion_2 \leq \sqrt{\varphi}$ and $\expansion_\infty \leq \varphi$.
    If $\conductor$ is an odd prime, then $\expansion_2 \leq \conductor$ and $\expansion_\infty \leq 2\varphi$. 
\end{lemma}

\begin{proof}
    The claims about $\expansion_\infty$ follow by realising that the proof of \cite[Proposition 2]{C:AlbLai21} not only applies to the ring of integers $\ring$ of $\field$ but also to $\field$ in general. 
    For the claims about $\expansion_2$, we use Items (i) and (ii) of \cite[Lemma 1]{EPRINT:KLNO24}, which state the following: 
    For any $x \in \field$, 
    (i) $\norm{\psi(x)}_2 \leq \sqrt{\frac{\rad(\conductor)}{\conductor}} \cdot \norm{\sigma(x)}_2$ and 
    (ii) $\norm{\sigma(x)}_2 \leq \sqrt{\hat{\conductor}} \cdot \norm{\psi(x)}_2$
    where $\rad(\conductor)$ is the radical (i.e.~product of all distinct prime factors) of $\conductor$ and $\hat{\conductor} = \conductor$ if $\conductor$ is odd and $\hat{\conductor} = \conductor/2$ if $\conductor$ is even.
    Consider $a, b \in \field$. Substituting $x = a \cdot b$, we have 
    $\norm{\psi(a \cdot b)}_2 
    \leq \sqrt{\frac{\rad(\conductor)}{\conductor}} \cdot \norm{\sigma(a \cdot b)}_2
    \leq \sqrt{\frac{\rad(\conductor)}{\conductor}} \cdot \norm{\sigma(a)}_2 \cdot \norm{\sigma(b)}_2
    \leq \sqrt{\frac{\rad(\conductor)}{\conductor}} \cdot \hat{\conductor} \cdot \norm{\psi(a)}_2 \cdot \norm{\psi(b)}_2$. 
    Observe that $\sqrt{\frac{\rad(\conductor)}{\conductor}} \cdot \hat{\conductor} = \conductor$ if $\conductor$ is an odd prime and $\sqrt{\frac{\rad(\conductor)}{\conductor}} \cdot \hat{\conductor} = \sqrt{\hat{\conductor}} = \sqrt{\varphi}$ if $\conductor$ is a power of 2.
    \qed
\end{proof}

\begin{lemma}\label{lem:induced-norm-embeddings}
    Let $\conductor$ be a power of two.
    Apply each component $\sigma_i$, $i \in \ZZ_\conductor^\times$, of the canonical embedding $\sigma$ of $\field$ entrywise to vectors and matrices over $\field$.
    Then for any matrix $\matM \in \field^{n \times d}$,
    \[
        \norm{\matM}_{2,2} = \max_{i \in \ZZ_\conductor^\times} \norm{\sigma_i(\matM)}_{2,2}.
    \]
\end{lemma}

\begin{proof}
    Throughout, let $i$ and $j$ range over $\ZZ_\conductor^\times$, let $\sigma_i$ be the embedding with $\sigma_i(\zeta) = \zeta^i$, and $\matM \in \field^{n \times d}$.
    We assume $\conductor \geq 2$, since $\conductor = 1$ and $\conductor = 2$ give the same field $\field = \QQ$.
    Then $\ZZ_\conductor^\times$ is the set of odd residues modulo $\conductor$, and $\varphi = \conductor/2$.
    We will show that (1) $\Psi(\matM)$ is unitarily equivalent to the block-diagonal matrix 
    \[
        \matD \coloneqq \begin{bmatrix} \sigma_1(\matM) & & \\ & \ddots & \\ & & \sigma_{\conductor-1}(\matM) \end{bmatrix} \in \CC^{\varphi n \times \varphi d},
    \]
    (2) the former implies $\norm{\matM}_{2,2} = \norm{\matD}_{2,2}$, and (3) $\norm{\matD}_{2,2} = \max_i \norm{\sigma_i(\matM)}_{2,2}$.

    \medskip\noindent(1): First, let $\matU \in \CC^{\varphi \times \varphi}$ be the matrix such that
    its $(i,k)$-th entry is $\zeta^{ik}/\sqrt{\varphi}$, for all rows $i$ and columns $k \in \set{0, \ldots, \varphi-1}$, so that  
    \[
        \sigma(x) = \sqrt{\varphi} \cdot \matU \cdot \psi(x) \quad\forall x \in \field.
    \]
    That is, $\matU$ is the conversion matrix between the coefficient and canonical embeddings of $\vecx \in \field$ up to scaling.
    The inner product of rows $i$ and $j$ of $\matU$ is $\frac{1}{\varphi} \sum_{k=0}^{\varphi-1} \omega^k$ where $\omega \coloneqq \zeta^{i-j}$.
    If $i = j$ this equals $1$.
    If $i \neq j$, then $i - j$ is even and nonzero modulo $\conductor$, so $\omega \neq 1$ and $\omega^\varphi = \zeta^{\conductor (i-j)/2} = 1$, hence the sum is $(\omega^\varphi - 1)/(\omega - 1) = 0$.
    Therefore $\matU$ is unitary.

    Next, for a vector $\vecx \in \field^d$, let 
    \[
        \hat{\sigma}(\vecx) \coloneqq (\sigma_i(\vecx))_i \in \CC^{\varphi d} 
    \]
    list the embeddings of $\vecx$ grouped by embedding rather than by entry. Also, let $\matP_d$ be the permutation matrix reordering coordinates from grouping by entry to grouping by embedding, so that we have  
    \[
        \hat{\sigma}(\vecx) = \sqrt{\varphi} \cdot \matW_d \cdot \psi(\vecx) 
        \quad\text{where}\quad 
        \matW_d \coloneqq \matP_d \cdot (\matI_d \otimes \matU).
    \]
    Observe that the matrix $\matW_d$ is unitary, as it is a product of unitary matrices.
    Since each $\sigma_i$ is a ring homomorphism, $\sigma_i(\matM \vecx) = \sigma_i(\matM) \cdot \sigma_i(\vecx)$ for every $i$, implying $\hat{\sigma}(\matM \vecx) = \matD \cdot \hat{\sigma}(\vecx)$. 
    Together with $\psi(\matM \vecx) = \Psi(\matM) \cdot \psi(\vecx)$ (\cref{sec:prelim}), this gives
    \[
        \matW_n \cdot \Psi(\matM) \cdot \psi(\vecx) = \matD \cdot \matW_d \cdot \psi(\vecx)
        \quad\forall \vecx \in \field^d.
    \]
    As $\psi(\field^d) = \QQ^{\varphi d}$ contains a basis of $\CC^{\varphi d}$, it follows that the block diagonalisation of $\Psi(\matM)$ is 
    \[
        \Psi(\matM) = \matW_n^* \cdot \matD \cdot \matW_d,
    \]
    that is, $\Psi(\matM)$ is unitarily equivalent to $\matD$.
    
    \medskip\noindent(2): 
    Since $\Psi(\matM)$ is real, its conjugate transpose is $\Psi(\matM)^\tr$, and since $\matW_n$ is unitary, (1) gives
    \[
        \Psi(\matM)^\tr \cdot \Psi(\matM)
        = \matW_d^* \cdot \matD^* \cdot \matW_n \cdot \matW_n^* \cdot \matD \cdot \matW_d
        = \matW_d^* \cdot (\matD^* \matD) \cdot \matW_d.
    \]
    Hence $\Psi(\matM)^\tr \Psi(\matM)$ and $\matD^* \matD$ are unitarily similar and have the same eigenvalues, so $\smax(\Psi(\matM)) = \smax(\matD)$.
    Since $\norm{\matM}_{2,2} = \smax(\Psi(\matM))$ by \cite[Lemma~2]{CiC:LaiSwaWoo25} (\cref{sec:prelim}), we get $\norm{\matM}_{2,2} = \norm{\matD}_{2,2}$.

    \medskip\noindent(3): The matrix $\matD^* \matD$ is block diagonal with diagonal blocks $\sigma_i(\matM)^* \sigma_i(\matM) \in \CC^{d \times d}$.
    The eigenvalues of a block-diagonal matrix are those of its diagonal blocks, so the largest eigenvalue of $\matD^* \matD$ is the maximum over $i$ of the largest eigenvalue of $\sigma_i(\matM)^* \sigma_i(\matM)$.
    Taking square roots gives 
    \[
        \norm{\matD}_{2,2} = \smax(\matD) = \max_i \smax(\sigma_i(\matM)) = \max_i \norm{\sigma_i(\matM)}_{2,2}.
    \]

    Finally, combining (2) and (3), we get $\norm{\matM}_{2,2} = \norm{\matD}_{2,2} = \max_i \norm{\sigma_i(\matM)}_{2,2}$.
    \qed
\end{proof}

\begin{lemma}[{Adapted from \cite[Theorems 1 and 2]{C:AlbLai21}}]\label{lem:xi-subtractive}
    If $\conductor$ is a power of two, then $\set{1, \zeta, \ldots, \zeta^{\conductor-1}}$ is $\slack$-subtractive for $\slack(t) = 2^{\ceil{\log t}}$. 
    If $\conductor = p$ is a prime, then the set $\set{\sum_{j=0}^{i-1} \zeta^j: i \in [p-1]}$ is $\slack$-subtractive for $\slack(t) = 1$.
\end{lemma}

The following lemma, adapted from \cite{C:KLLPS26}, provides bounds on the quality of the subtractive set we adopt in our implementation. Its proof is given in \fullversion{\cref{app:lem:xi-subtactive-trig-spaced}}{\thefullversion}.

\begin{lemma}[Quality of Subtractive Sets. {Adapted from \cite[Corollary 5.1]{C:KLLPS26}}]\label{lem:xi-subtactive-trig-spaced}
    Let $\conductor$ and $K \leq \conductor$ be powers of two, $\eta \coloneqq \zeta^{\conductor/K}$ a primitive $K$-th root of unity,
    and let $\Xi' \coloneqq \set{\eta^i : 0 \leq i < K}$. For any $2 \leq t \leq K$ and $\xi = \slack(t) = 2^{\ceil{\log t}}$, the set $\Xi'$ and any subset thereof have the quality 
    \begin{align*}
        \rho(t) \leq~& \xi \cdot 2^{-(t-1)} \prod_{k=1}^{\floor{t/2}} \frac{1}{\sin^2(\pi k/K)}, \\
        \theta(t) \leq~& \xi \cdot A, \\
        \gamma(t) \leq~& \xi \cdot \sqrt{ (K - t + 1) (A^2 + (t-1) B^2) + 1 },
    \end{align*}
    where
    \begin{align*}
        &A \coloneqq 2^{t-1} \prod_{k=1}^{\floor{(t-2)/2}} \cos^2\!\big(\tfrac{\pi k}{K}\big),
        \fullversion{\qquad}{\\ &}
        B \coloneqq \prod_{k=1}^{\floor{(t-3)/2}} \cos^2\!\big(\tfrac{\pi k}{K}\big) \prod_{k=1}^{\floor{(t-1)/2}} \frac{1}{\sin^2(\pi k/K)}.
    \end{align*}
\end{lemma}

\begin{lemma}[Adapted from~{\cite[Lemma~1.5]{MA:Banaszczyk93}}]\label{lem:norm_bound}
  Let $\ring = \ZZ[\zeta]$ with ring degree $\varphi$. For any $m \in \NN$, $\chi >0$, $\delta >0$,
  \begin{align}
    \condprob{\norm{\vece}_2 > \chi \sqrt{\varphi m}}{\vece \sample \DD^m_{\ring,\chi}} \leq~& 2^{-\varphi m}, \label{eq:bound-2-canonical} \\
    \condprob{\norm{\vece}_\infty > \chi\sqrt{\delta}}{\vece \sample \DD^m_{\ring,\chi}} \leq~& \varphi m\cdot 2^{-\delta}. \label{eq:bound-infinity-coefficient}
   \end{align}
\end{lemma}

The following lemma is standard and can be shown easily by bounding the KL divergence and then via Pinsker's inequality.

\begin{lemma}[Noise Flooding]\label{lem:noise_flooding}
    For any $n\in\NN$, $\chi > 0$, $\vec{c}\in\ring^n$, it
    holds that
    $\mathrm{KL}(\DD_{\ring,\chi}^n, \DD_{\ring,\chi}^n+\vec{c}) \leq \pi \cdot \frac{\norm{\vec{c}}_2^2}{\chi^2}$ and 
    $\mathsf{SD}(\DD_{\ring,\chi}^n, \DD_{\ring,\chi}^n+\vec{c})\leq \sqrt{\frac{\pi}{2}} \cdot \frac{\norm{\vec{c}}_2}{\chi}$,
    where $\mathrm{KL}$ and $\mathsf{SD}$ denote the KL divergence and the statistical distance, respectively.
    In particular, if $\chi\geq\lambda^{\omega(1)}\cdot \norm{\vec{c}}_2$, one has $\DD_{\ring,\chi}^n\approx_s\DD_{\ring,\chi}^n+\vec{c}$.
\end{lemma}

For our security proofs, we will need to apply noise flooding over $Q$ many adaptive oracle queries. Naively making $Q$ hybrid hops will lead to a factor $Q$ blow-up in the advantage bound. The following lemma provides a tighter bound with only $\sqrt{Q}$ blow-up, it follows from bounding the KL divergence over all $Q$ queries carefully and then applying Pinsker's inequality. 

\begin{lemma}\label{lem:adaptive_noise_flooding}
    Let $n\in\NN$ and $\chi,B>0$. For $b\in\set{0,1}$ and any (computationally unbounded) adversary $\adv$, let $\mathsf{Flood}^b_\adv$ be the game defined below:
    \begin{center}
        \begin{pchstack}[center]
            \procedure{$\mathsf{Flood}^b_\adv$}{
                \pcreturn \adv^{\oracle_b}()
            }
            \pchspace
            \procedure{$\oracle_b(\vecc \in\ring^n)$}{
                \pcassert \norm{\vecc}_2 \le B \\
                \vece \sample \DD_{\ring,\chi}^n \\
                \pcreturn \veca \coloneqq b \cdot \vecc + \vece 
            }
        \end{pchstack}
    \end{center}
    For any (computationally unbounded) $\adv$ that makes at most $Q \in \NN$ queries to $\oracle_b$, it holds that 
    \[
        \abs{\prob{\mathsf{Flood}^0_\adv=1}-\prob{\mathsf{Flood}^1_\adv=1}}\le \sqrt{\tfrac{\pi}{2}}\cdot\frac{\sqrt{Q} \cdot B}{\chi}.
    \]
\end{lemma}

\begin{proof}
    Fix any $\adv$. As $\adv$ can be unbounded, we may assume it is deterministic -- fixing its coins to the ones such that the LHS of the concerned inequality is maximised. The view of $\adv$ is the tuple of all oracle answers $(\veca_1,\ldots,\veca_Q)$. 
    Below, denote by $P$ and $\hat{P}$ the joint distributions of these answers in $\mathsf{Flood}^1_\adv$ and $\mathsf{Flood}^0_\adv$, respectively. 
    The LHS quantity is thus at most $\mathsf{SD}(P,\hat{P})$, which we bound in the rest. 
    We will use the chain-rule for KL divergence, for which we provide a proof below for completeness.

    \begin{proposition}
    For any two distributions $P,\hat P$ over tuples $a=(a_1,\ldots,a_Q)$, the KL divergence between P and $\hat P$ satisfies
    \begin{align}\label{eq:KL-chain-rule}
        \mathrm{KL}(P \,\|\, \hat P)
        = \sum_{i \in [Q]} \sum_{a_{<i}} P(a_{<i}) \cdot \mathrm{KL}\big(P(a_i\mid a_{<i}) \,\big\|\, \hat P(a_i\mid a_{<i})\big).
    \end{align}
    \end{proposition}
    
    \begin{proof}    
        The probability chain rule gives $P(a)=\prod_{i\in[Q]}P(a_i\mid a_{<i})$ and $\hat P(a)=\prod_{i\in[Q]}\hat P(a_i\mid a_{<i})$, where $a_{<i}=(a_1,\ldots,a_{i-1})$. Hence
        \begin{align*}
            \mathrm{KL}(P \,\|\, \hat P)
            &= \sum_{a} P(a) \log \frac{P(a)}{\hat P(a)} \\
            &= \sum_{a} P(a) \sum_{i \in [Q]} \log \frac{P(a_i\mid a_{<i})}{\hat P(a_i\mid a_{<i})}  \\
            &= \sum_{i \in [Q]} \sum_{a} P(a) \log \frac{P(a_i\mid a_{<i})}{\hat P(a_i\mid a_{<i})}  \\
            &= \sum_{i \in [Q]} \sum_{a_{\le i}} \Big(\sum_{a_{>i}} P(a_{\le i}, a_{>i})\Big) \log \frac{P(a_i\mid a_{<i})}{\hat P(a_i\mid a_{<i})}  \\
            &= \sum_{i \in [Q]} \sum_{a_{\le i}} P(a_{\le i}) \log \frac{P(a_i\mid a_{<i})}{\hat P(a_i\mid a_{<i})}  \\
            &= \sum_{i \in [Q]} \sum_{a_{<i}} P(a_{<i}) \sum_{a_i} P(a_i\mid a_{<i}) \log \frac{P(a_i\mid a_{<i})}{\hat P(a_i\mid a_{<i})}  \\
            &= \sum_{i \in [Q]} \sum_{a_{<i}} P(a_{<i}) \cdot \mathrm{KL}\big(P(a_i\mid a_{<i}) \,\big\|\, \hat P(a_i\mid a_{<i})\big).
        \end{align*}
        In the above, the fourth line uses 
        $a=(a_{\le i}, a_{>i})$ where we denote $a_{\le i}=(a_1,\ldots,a_i)$ and $a_{>i}=(a_{i+1},\ldots,a_Q)$.
        The fifth line holds by definition of marginal distribution, i.e.~$\sum_{a_{>i}}P(a_{\le i},a_{>i})=P(a_{\le i})$. 
        The sixth line uses $P(a_{\le i})=P(a_{<i})\,P(a_i\mid a_{<i})$, and the last line is by realising that the inner sum over $a_i$ is the conditional KL divergence $\mathrm{KL}(P(a_i\mid a_{<i})\,\|\,\hat P(a_i\mid a_{<i}))=\sum_{a_i}P(a_i\mid a_{<i})\log\frac{P(a_i\mid a_{<i})}{\hat P(a_i\mid a_{<i})}$. 
        \qed
    \end{proof}

    With $\adv$ deterministic, once the prefix $\veca_{<i}$ is fixed, the query $\vecc_i$ is a fixed vector and the only fresh randomness in $\veca_i$ is $\vece_i\sample\DD_{\ring,\chi}^n$, drawn independently of the prefix. 
    The conditional distributions are therefore $P(\veca_i\mid\veca_{<i})=\vecc_i+\DD_{\ring,\chi}^n$ and $\hat{P}(\veca_i\mid\veca_{<i})=\DD_{\ring,\chi}^n$. 
    
    By the KL divergence bound in \cref{lem:noise_flooding} and the norm-check in the oracle $\oracle_b$, for every prefix $\veca_{<i}$, 
    \[
        \mathrm{KL}\big(P(\veca_i\mid\veca_{<i})\,\big\|\,\hat{P}(\veca_i\mid\veca_{<i})\big)
        \leq \frac{\pi \norm{\vecc_i}_2^2}{\chi^2} 
        \leq \frac{\pi B^2}{\chi^2}. 
    \]
    Note that the above upper bound is a fixed value independent of the distribution of $\norm{\vecc_i}$ (the latter may adaptively depend on the prefix $\veca_{<i}$). 
    Plugging into \cref{eq:KL-chain-rule}, we get 
    \begin{align*}
        \mathrm{KL}(P\,\|\,\hat{P})
        &\leq \sum_{i\in[Q]} \sum_{\veca_{<i}} P(\veca_{<i})\cdot \frac{\pi B^2}{\chi^2}
        = \frac{\pi Q B^2}{\chi^2},
    \end{align*}
    where the last equality uses $\sum_{\veca_{<i}} P(\veca_{<i})=1$. Finally, using Pinsker's inequality gives
    \[
        \mathsf{SD}(P,\hat{P})
        \leq \sqrt{\frac{1}{2} \mathrm{KL}(P\,\|\,\hat{P})} 
        \leq \sqrt{\tfrac{\pi}{2}}\cdot\frac{\sqrt{Q} \cdot B}{\chi}.
        \tag*{\qed}
    \]
\end{proof}


\section{Proof of \texorpdfstring{\cref{lem:xi-subtactive-trig-spaced}}{Lemma \ref{lem:xi-subtactive-trig-spaced}} (Quality of Subtractive Sets)}\label{app:lem:xi-subtactive-trig-spaced}

To prove \cref{lem:xi-subtactive-trig-spaced}, we will use the following \cref{lem:xi-subtactive-trig} which provides more general bounds for a subtractive set that consists of roots of unity. It is adapted from \cite{C:KLLPS26} to suit the norm convention in this work. 

\begin{lemma}[{Adapted from \cite[Theorem 4.3]{C:KLLPS26}}]\label{lem:xi-subtactive-trig}
    Let $\conductor$ be a power of two, so that $\varphi = \conductor/2$, let $\Xi = \set{1, \zeta, \ldots, \zeta^{\conductor-1}}$, and let $\xi = \slack(t) = 2^{\ceil{\log t}}$.
    For any $2 \leq t \leq \conductor$, the set $\Xi$ and any subset thereof have the quality
    \begin{align*}
        \rho(t) \leq~& \xi \cdot 2^{-(t-1)} \prod_{k=1}^{\floor{t/2}} \frac{1}{\sin^2(\pi k/\conductor)}, \\
        \theta(t) \leq~& \xi \cdot A, \\
        \gamma(t) \leq~& \xi \cdot \sqrt{ (\conductor - t + 1) (A^2 + (t-1) B^2) + 1 },
    \end{align*}
    where 
    \begin{align*}
        A \coloneqq 2^{t-1} \prod_{k=1}^{\floor{(t-2)/2}} \cos^2\!\big(\tfrac{\pi k}{\conductor}\big),
        \qquad 
        B \coloneqq \prod_{k=1}^{\floor{(t-3)/2}} \cos^2\!\big(\tfrac{\pi k}{\conductor}\big) \prod_{k=1}^{\floor{(t-1)/2}} \frac{1}{\sin^2(\pi k/\conductor)}.
    \end{align*}
\end{lemma}

\begin{proof}
    Throughout this proof, for a set $S \subseteq \set{0} \cup \Xi$ and an element $\beta \notin S$, we write
    $\lambda_\mu^{S \to \beta} \coloneqq \prod_{\nu \in S \setminus \set\mu} \frac{\beta - \nu}{\mu - \nu}$
    for the Lagrange coefficient of the element $\mu \in S$.%
    \footnote{This is consistent with the notations in \cite{C:KLLPS26}.}
    With this notation, the row vector $\xi \cdot \vec{v}_0^\transpose \cdot \mat{V}_T^{-1}$ has its $\mu$-th entry being $\xi \cdot \lambda_\mu^{T \to 0}$, and the matrix $\xi \cdot \mat{V}_{\Xi \setminus T^*} \cdot \mat{V}_{\set{0} \cup T^*}^{-1}$ has its $(\alpha, \mu)$-th entry being $\xi \cdot \lambda_\mu^{\set{0} \cup T^* \to \alpha}$,
    where the row index $\alpha$ ranges over $\Xi \setminus T^*$ and the column index $\mu$ over $\set{0} \cup T^*$.
    By \cref{lem:xi-subtractive}, all these elements lie in $\ring$.

    Since $\conductor$ is a power of two,
    \begin{align}\label{eq:bound-coeff-by-canon}
        \norm{x} \leq \norm{\sigma(x)}_\infty = \max_\ell \abs{\sigma_\ell(x)}
    \end{align}
    for any $x \in \field$, 
    so it suffices to bound the canonical embedding.
    Fix one embedding $\sigma_\ell$, sending $\zeta$ to a primitive $\conductor$-th root of unity $\omega = \exp(\pi i r/\varphi)$ with $r$ coprime to $\conductor$.
    For exponents $a, b$ we have $\abs{\sigma_\ell(\zeta^a)} = 1$, $\abs{\sigma_\ell(\zeta^a - 0)} = 1$, and
    \begin{align*}
        \abs{\sigma_\ell(\zeta^a - \zeta^b)} = \abs{\omega^{a-b} - 1} = \abs{\exp(2\pi i \, r(a-b)/\conductor) - 1} = 2\,\abs{\sin\!\big(\pi r (a-b)/\conductor\big)}.
    \end{align*}
    As $r$ is coprime to $\conductor$, multiplication by $r$ permutes the integers modulo $\conductor$, so we may assume $r = 1$ without loss of generality.

    \begin{proposition}[Sine-product]    
        For any $m$ distinct non-zero integers $d_1, \ldots, d_m$ modulo $\conductor$,
        \begin{align}\label{eq:sine-product-claim}
            \prod_{j=1}^{m} \frac{1}{\abs{\sin(\pi d_j/\conductor)}} \leq \prod_{k=1}^{\ceil{m/2}} \frac{1}{\sin^2(\pi k/\conductor)}.
        \end{align}
    \end{proposition}
    
    \begin{proof}
        Indeed, $\abs{\sin(\pi d/\conductor)}$ depends only on $\bar d \coloneqq \min(d, \conductor - d) \in \set{1, \ldots, \varphi}$, since $\sin(\pi(\conductor - d)/\conductor) = \sin(\pi d/\conductor)$, and is strictly increasing in $\bar d$.
        Each value of $\bar d$ is attained by at most two integers modulo $\conductor$, namely $d$ and $\conductor - d$.
        Hence the LHS in the proposition is the largest when the multiset $\set{\bar d_j}$ takes the smallest possible values each repeated at most twice, that is, $\set{1, 1, 2, 2, \ldots}$ of size $m$, whose largest element is $\ceil{m/2}$.
        When $m$ is odd we enlarge the product by one factor $\sin^{-2} \geq 1$, which gives \cref{eq:sine-product-claim}.
        \qed
    \end{proof}

    \begin{proposition}[Cosine-product]    
        For any $m$ distinct non-zero integers $d_1, \ldots, d_m$ modulo $\conductor$,
        \begin{align}\label{eq:cosine-product-claim}
            \prod_{j=1}^{m} \abs{\sin(\pi d_j/\conductor)} \leq \prod_{k=1}^{\floor{(m-1)/2}} \cos^2(\pi k/\conductor).
        \end{align}
    \end{proposition}

    \begin{proof}    
        Writing $e \coloneqq \varphi - \bar d \in \set{0, \ldots, \varphi - 1}$, we have $\abs{\sin(\pi d/\conductor)} = \sin(\pi(\varphi - e)/\conductor) = \cos(\pi e/\conductor)$, which is the largest when $e$ is the smallest.
        Each value $e = k \geq 1$ is attained by at most two integers modulo $\conductor$ (those with $\bar d = \varphi - k$), and $e = 0$ by one (the integer $d = \varphi$), so the LHS in the proposition is the largest for the multiset $\set{0, 1, 1, 2, 2, \ldots}$ of size $m$.
        For even $m$ the largest element $m/2$ is unpaired, and bounding its single factor by $\cos(\pi(m/2)/\conductor) \leq 1$ gives \cref{eq:cosine-product-claim}.
        \qed
    \end{proof}

    \medskip\noindent\underline{Bound on $\rho$.}
    By definition $\rho(t) = \max_{T \subseteq_t \Xi} \norm{\slack \cdot \vec{v}_0^\transpose \cdot \mat{V}_T^{-1}} = \max_{T \subseteq_t \Xi} \norm{(\xi \cdot \lambda_\mu^{T \to 0})_{\mu \in T}}$,
    that is, the maximum of $\norm{\xi \cdot \lambda_\mu^{T \to 0}}$ over all $\mu \in T$ and all sets $T \subseteq_t \Xi$.
    By \cref{eq:bound-coeff-by-canon}, to bound $\rho(t)$ it suffices to bound $\abs{\sigma_\ell(\xi \cdot \lambda_\mu^{T \to 0})}$ for any $\mu$,$T$ and $\ell$.

    Fix a size-$t$ set $T \subseteq_t \Xi$ and $\mu \in T$, and write each $\nu \in T$ as $\zeta^{c_\nu}$.
    Then $\abs{\sigma_\ell(\nu)} = 1$ and $\abs{\sigma_\ell(\mu - \nu)} = 2\abs{\sin(\pi(c_\mu - c_\nu)/\conductor)}$, so
    \begin{align*}
        \abs{\sigma_\ell(\xi \cdot \lambda_\mu^{T \to 0})}
        = \xi \prod_{\nu \in T \setminus \set\mu} \frac{\abs{\sigma_\ell(\nu)}}{\abs{\sigma_\ell(\mu - \nu)}}
        = \xi \cdot 2^{-(t-1)} \prod_{\nu \in T \setminus \set\mu} \frac{1}{\abs{\sin(\pi(c_\mu - c_\nu)/\conductor)}}.
    \end{align*}
    The $t-1$ integers $c_\mu - c_\nu$ are distinct and non-zero modulo $\conductor$, so \cref{eq:sine-product-claim} with $m = t-1$ and $\ceil{(t-1)/2} = \floor{t/2}$ upper-bounds this by
    $\xi \cdot 2^{-(t-1)} \prod_{k=1}^{\floor{t/2}} \sin^{-2}(\pi k/\conductor)$,
    which is the stated bound.

    \medskip\noindent\underline{Bound on $\gamma$.}
    By definition $\gamma(t) = \max_{T^* \subseteq_{t-1} \Xi} \norm{\slack \cdot \mat{V}_{\Xi} \cdot \mat{V}_{\set{0} \cup T^*}^{-1}}_{2,2}$.
    Fix any $T^* \subseteq_{t-1} \Xi$, define the shorthand $S \coloneqq \set{0} \cup T^*$, and denote $\matM \coloneqq \xi \cdot \mat{V}_{\Xi \setminus T^*} \cdot \mat{V}_S^{-1}$.
    We bound separately the rows of $\xi \cdot \mat{V}_\Xi \cdot \mat{V}_S^{-1}$ indexed by $\Xi \setminus T^*$, which form $\matM$, and those indexed by $T^*$, and then combine the two bounds.

    For any row $\alpha \in \Xi \setminus T^*$ and column $\mu \in S$, the $(\alpha, \mu)$-th entry of $\matM$ is
    \begin{align}\label{eq:express_M_i-mu}
        \mat{M}_{\alpha, \mu} = \xi \cdot \lambda_\mu^{S \to \alpha} = \xi \prod_{\nu \in S \setminus \set\mu} \frac{\alpha - \nu}{\mu - \nu}.
    \end{align}
    Write the evaluation point as $\alpha = \zeta^a$ and each node $\nu \in T^*$ as $\nu = \zeta^{c_\nu}$.
    Relative to the bound on $\rho(t)$, the main difference here is that the numerator is $\alpha - \nu$ (instead of $-\nu$) which is no longer a unit of $\ring$, and will yield a sine factor $\abs{\sigma_\ell(\alpha - \nu)} = 2\abs{\sin(\pi(a - c_\nu)/\conductor)}$ as shown below.

    For a canonical embedding $\sigma_\ell$, extend $\sigma_\ell$ entrywise to matrices and vectors of ring elements, so that $\sigma_\ell(\mat{M}) = \big( \sigma_\ell(\mat{M}_{\alpha, \mu}) \big)_{\alpha, \mu} \in \CC^{(\conductor - t + 1) \times t}$.
    Since $\conductor$ is a power of two, \cref{lem:induced-norm-embeddings} applies to $\matM$, and 
    \begin{align}\label{eq:bound-22-norm-by-Frobenius}
        \norm{\mat{M}}_{2,2}
        = \max_\ell \norm{\sigma_\ell(\mat{M})}_{2,2}
        \leq \max_\ell \norm{\sigma_\ell(\mat{M})}_F
        = \max_\ell \sqrt{\sum_{\alpha, \mu} \abs{\sigma_\ell(\mat{M}_{\alpha,\mu})}^2},
    \end{align}
    where the inequality upper-bounds $\norm{\cdot}_{2,2}$ by the Frobenius norm $\norm{\cdot}_F$.
    In the following we bound $\abs{\sigma_\ell(\mat{M}_{\alpha, \mu})}$ for each index $\alpha,\mu$. We recall that $\matM_{\alpha,\mu}$ can be expressed as in \cref{eq:express_M_i-mu}.
    
    Fix any row $\alpha \in \Xi\setminus T^*$. We consider two cases: the column index of $\mat{M}_{\alpha, \mu}$ is $\mu = 0$ or $\mu \in T^*$.
    
    For $\mu = 0$, the denominator $\prod_{\nu \in T^*}(0 - \nu)$ in \cref{eq:express_M_i-mu} is a product of units, so
    \begin{align*}
        \abs{\sigma_\ell(\xi \cdot \lambda_0^{S \to \alpha})} = \xi \prod_{\nu \in T^*} \abs{\sigma_\ell(\alpha - \nu)} = \xi \cdot 2^{t-1} \prod_{\nu \in T^*} \abs{\sin(\pi(a - c_\nu)/\conductor)}.
    \end{align*}
    The $t-1$ integers $a - c_\nu$ are distinct and non-zero modulo $\conductor$, so \cref{eq:cosine-product-claim} with $m = t-1$ upper-bounds this by
    \begin{align*}
        \abs{\sigma_\ell(\xi \cdot \lambda_0^{S \to \alpha})} \leq \xi \cdot 2^{t-1} \prod_{k=1}^{\floor{(t-2)/2}} \cos^2(\pi k/\conductor)
        = \xi \cdot A.
    \end{align*}
    For $\mu \in T^*$, writing $\mu = \zeta^c$ for some $c$ and factoring out the term with $\nu = 0$, 
    \begin{align*}
        \abs{\sigma_\ell(\xi \cdot \lambda_\mu^{S \to \alpha})}
        = \xi \cdot \frac{\abs{\sigma_\ell(\alpha-0)}}{\abs{\sigma_\ell(\mu-0)}} \prod_{\nu \in T^* \setminus \set\mu} \frac{\abs{\sigma_\ell(\alpha - \nu)}}{\abs{\sigma_\ell(\mu - \nu)}}
        = \xi \prod_{\nu \in T^* \setminus \set\mu} \frac{\abs{\sin(\pi(a - c_\nu)/\conductor)}}{\abs{\sin(\pi(c - c_\nu)/\conductor)}},
    \end{align*}
    where the last line is because both $\alpha$ and $\mu$ are units so $\abs{\sigma_\ell(\alpha)} = \abs{\sigma_\ell(\mu)} = 1$.
    Within each of the numerator and the denominator in \cref{eq:express_M_i-mu}, the $t-2$ integers $a - c_\nu$ (respectively $c - c_\nu$) are distinct and non-zero modulo $\conductor$. Hence, we can bound the numerator by \cref{eq:cosine-product-claim} with $m = t-2$, and the inverse of the denominator by \cref{eq:sine-product-claim} with $m = t-2$ (using $\ceil{(t-2)/2} = \floor{(t-1)/2}$), obtaining
    \begin{align*}
        \abs{\sigma_\ell(\xi \cdot \lambda_\mu^{S \to \alpha})} \leq \xi \prod_{k=1}^{\floor{(t-3)/2}} \cos^2(\pi k/\conductor) \prod_{k=1}^{\floor{(t-1)/2}} \frac{1}{\sin^2(\pi k/\conductor)}
        = \xi \cdot B.
    \end{align*}

    For any $\alpha \in \Xi \setminus T^*$, each row of $\sigma_\ell(\mat{M})$ indexed by $\alpha$ has 1 entry with $\mu = 0$ and $t-1$ entries with $\mu \in T^*$. Using the bounds from the above two cases,
    \[
        \sum_{\mu} \abs{\sigma_\ell(\mat{M}_{\alpha,\mu})}^2
        \leq (\xi \cdot A)^2 + (t-1)(\xi \cdot B)^2 = \xi^2 \cdot (A^2 + (t-1)B^2).
    \] 
    Summing over the $\conductor - t + 1$ rows, we get that 
    \begin{align*}
        \sum_{\alpha, \mu} \abs{\sigma_\ell(\mat{M}_{\alpha,\mu})}^2 \leq \xi^2 \cdot (\conductor - t + 1) (A^2 + (t-1)B^2).
    \end{align*}
    This bound is independent of $\ell$, so taking square roots yields an upper-bound for $\norm{\mat{M}}_{2,2}$ by \cref{eq:bound-22-norm-by-Frobenius}.
    
    It remains to account for the rows of $\xi \cdot \mat{V}_\Xi \cdot \mat{V}_S^{-1}$ indexed by $T^*$.
    Each $\mu \in T^*$ indexes a row of $\mat{V}_S$, so row $\mu$ of $\mat{V}_{T^*} \cdot \mat{V}_S^{-1}$ is the unit row vector selecting the column $\mu$.
    Hence, up to a permutation of rows, $\xi \cdot \mat{V}_\Xi \cdot \mat{V}_S^{-1}$ is $\matM$ stacked on $\xi \cdot \matP$, where $\matP \in \set{0,1}^{(t-1) \times t}$ selects the entries of a vector indexed by $T^*$.
    For any $\vec{x} \in \field^t$, since $\xi$ is an integer and $\norm{\matP \vec{x}}_2 \leq \norm{\vec{x}}_2$,
    \begin{align*}
        \norm{\xi \cdot \mat{V}_\Xi \cdot \mat{V}_S^{-1} \cdot \vec{x}}_2^2
        =~& \norm{\matM \vec{x}}_2^2 + \xi^2 \cdot \norm{\matP \vec{x}}_2^2 \\
        \leq~& \left( \norm{\matM}_{2,2}^2 + \xi^2 \right) \cdot \norm{\vec{x}}_2^2 \\
        \leq~& \xi^2 \cdot \big( (\conductor - t + 1) (A^2 + (t-1)B^2) + 1 \big) \cdot \norm{\vec{x}}_2^2.
    \end{align*}
    Maximising over $T^*$ yields the stated bound on $\gamma(t)$.

    \medskip\noindent\underline{Bound on $\theta$.}
    This quantity concerns the same matrix $\matM \coloneqq \xi \cdot \mat{V}_{\Xi \setminus T^*} \cdot \mat{V}_{\{0\} \cup T^*}^{-1}$ as above.
    In particular, $\theta(t) = \max_{T^*} \norm{\mat{M} \cdot \vec{v}_0}$ is the $\ell_\infty$-norm of the column with $\mu = 0$. We have
    \[
        \theta(t) \leq \max_\alpha \norm{\mat{M}_{\alpha, 0}}
        \leq \max_\alpha \norm{\sigma(\mat{M}_{\alpha, 0})}_\infty
        \leq
        \xi \cdot A,
    \]
    where the second inequality is by \cref{eq:bound-coeff-by-canon}.
    Maximising over $T^*$ gives the bound on $\theta(t)$.
    \qed
\end{proof}

Selecting subtractive set elements that are evenly spread among the $\conductor$ monomials reduces the effective ring degree and can lead to better upper bounds. This is observed in \cite{AC:CinLaiWoo25} and formalised in \cite[Section 5.1]{C:KLLPS26}, and is also the choice in our implementation, captured by \cref{lem:xi-subtactive-trig-spaced}.

\begin{proof}[of \cref{lem:xi-subtactive-trig-spaced}]
    Since $\conductor$ is a power of two and $K \mid \conductor$, the element $\eta = \zeta^{\conductor/K}$ is a primitive $K$-th root of unity with $\eta^{K/2} = \zeta^{\conductor/2} = -1$, so the subring $\ring' \coloneqq \ZZ[\eta] \subseteq \ring$ is isomorphic to the $K$-th cyclotomic ring $\ZZ[Y]/(Y^{K/2}+1)$, of degree $K/2$.
    The set $\Xi' = \set{\eta^i : 0 \leq i < K}$ is precisely the set of powers of $\eta$, which is a $\slack$-subtractive set over $\ring'$ for the same slack $\xi = 2^{\ceil{\log t}}$ by \cref{lem:xi-subtractive} (as $K$ is also a power of two).
    Applying \cref{lem:xi-subtactive-trig} over $\ring'$ gives the four bounds with conductor $K$ and degree $K/2$.

    It remains to note that all the norms involved are unchanged when an element of $\ring'$ is viewed in $\ring$.
    Any $a = \sum_{j=0}^{K/2-1} c_j \cdot \eta^j \in \ring'$ equals $\sum_j c_j \cdot \zeta^{(\conductor/K) j}$, whose exponents $(\conductor/K) j < \varphi$ are distinct, so the coefficient vector of $a$ in the power basis of $\ring$ is that of $a$ in $\ring'$ padded with zeros.
    Hence $\norm{a}_p$ agrees in $\ring$ and $\ring'$ for every $p$, and the canonical-embedding magnitudes $\abs{\sigma_\ell(a)}$ used in the proof of \cref{lem:xi-subtactive-trig} are those of $\ring'$, which give the trigonometric factors with argument $\pi k/K$.

    For $\gamma$, fix $T^* \subseteq_{t-1} \Xi'$, set $S \coloneqq \set{0} \cup T^*$, and let $\matM \coloneqq \xi \cdot \mat{V}_{\Xi' \setminus T^*} \cdot \mat{V}_S^{-1}$ be the matrix with entries in $\ring'$, having $K - t + 1$ rows indexed by $\alpha \in \Xi' \setminus T^*$ and $t$ columns indexed by $\mu \in S$.
    Viewing $\matM$ as a matrix over $\ring$, \cref{eq:bound-22-norm-by-Frobenius} bounds its $\ell_2$-induced operator norm by $\norm{\matM}_{2,2} \leq \max_\ell \norm{\sigma_\ell(\matM)}_F$, where $\ell$ ranges over the $\varphi$ embeddings of $\ring$.
    By the previous paragraph, the magnitude $\abs{\sigma_\ell(\matM_{\alpha,\mu})}$ viewed in $\ring$ equals the magnitude in $\ring'$, whose trigonometric factors carry the argument $\pi k/K$.
    Hence the bounds $\abs{\sigma_\ell(\matM_{\alpha,0})} \leq \xi A$ and $\abs{\sigma_\ell(\matM_{\alpha,\mu})} \leq \xi B$ for $\mu \in T^*$, established in the proof of \cref{lem:xi-subtactive-trig}, hold with $A$ and $B$ evaluated at conductor $K$.
    As each of the $K - t + 1$ rows carries one entry of the first kind and $t-1$ of the second,
    \begin{align*}
        \norm{\sigma_\ell(\matM)}_F^2 = \sum_{\alpha, \mu} \abs{\sigma_\ell(\matM_{\alpha,\mu})}^2 \leq \xi^2 \cdot (K - t + 1)\big(A^2 + (t-1)B^2\big)
    \end{align*}
    for every $\ell$.
    Accounting for the $t-1$ rows of $\xi \cdot \mat{V}_{\Xi'} \cdot \mat{V}_S^{-1}$ indexed by $T^*$ exactly as in the proof of \cref{lem:xi-subtactive-trig} gives
    \[
        \norm{\xi \cdot \mat{V}_{\Xi'} \cdot \mat{V}_S^{-1}}_{2,2}^2 \leq \norm{\matM}_{2,2}^2 + \xi^2 \leq \xi^2 \cdot \big( (K - t + 1)(A^2 + (t-1)B^2) + 1 \big).
    \]
    Taking square roots and maximising over $T^*$ yields the stated bound for $\gamma$.
    \qed
\end{proof}

\section{Proof of \texorpdfstring{\cref{thm:leakage-security}}{Theorem \ref{thm:leakage-security}} (FHE Leakage Security)}\label{app:thm:leakage-security}

\leakagesecurity*

\begin{proof}
We consider the following sequence of hybrids. 
    
\begin{description}
    \item $\hyb_{b,0}$: This is the IND-CPA-with-$(D,\leakdim,\chi)$-Leakage  security experiment for challenge bit $b$. To recall, the view of the adversary consists of 
    \begin{align*}    
        &\text{public parameters } \pubparam,\qquad \text{public key } \pk, \\
        &\text{ciphertexts } \ctxt \quad \text{from $\enc\oracle$, $\Chal\oracle$ or $\Eval\oracle$, and their underlying plaintexts } \mu_b, \\
        &\text{leakages } \vecu\cdot\hat{e}+\vecf \quad\text{from $\Leak\oracle$ on ciphertexts $\ctxt$}. 
    \end{align*}
    \item $\hyb_{b,1}$: We change the experiment as follows.
        For each query $(\id,\vecu)$ to the oracle $\Leak\oracle$, compute $\hat{e}\coloneqq  \hat{c}_1 - \sk^\tr\cdot\hat{\vecc}_0 - \xi^{-2} \round{\frac{q}{p}}\mu \bmod q $. If $\norm{\hat{e}} > B$, then set $L_\Leak[\id]\coloneqq \vecf \sample \DD_{\ring,\chi}^\leakdim$ and use this as the oracle answer. Otherwise, there is no change, i.e.~set $L_\Leak[\id]\coloneqq \vecu \cdot \hat{e} + \vecf$.
    \item $\hyb_{b,2}$: We change the experiment as follows. 
        For each $(\id,\vecu)$ to the oracle $\Leak\oracle$, sample $\vecf\sample\DD_{\ring,\chi}^{\leakdim}$. The leakage is computed as $L_\Leak[\id]\coloneqq \vecf$.
\end{description}

For $i=0,1,2$, let $\Adv_i(\adv)$ denote the advantage of $\adv$ in experiments $(\hyb_{0,i},\hyb_{1,i})$, that is
\[
\Adv_i(\adv) \coloneqq \abs{\prob{\hyb_{0,i,\adv}(\secparam) = 1} - \prob{\hyb_{1,i,\adv}(\secparam) = 1}}
\]
By definition of $\hyb_{b,0}$, one has
\[
\Adv_{\FHE,\adv,(D,\leakdim,\chi)}^{\INDCPAwL}(\secpar)=\Adv_0(\adv)
\]
\begin{itemize}
    \item By \cref{lem:abort-correctness}, if $\FHE$ has $(\ring,\hat{n},p,q,\xi,B)$-linear-decryption with error $\varepsilon_\FHE$, we have
    \[
    \Adv_0(\adv)\leq \Adv_1(\adv)+2\cdot Q\cdot \varepsilon_\FHE.
    \]
    \item By \cref{lem:noise-flooding-step}, we have
    \[
    \Adv_1(\adv)\leq \Adv_2(\adv)+2\cdot  \frac{C \cdot \sqrt{Q\cdot\varphi \cdot \leakdim} \cdot \expansion_\infty \cdot D \cdot B}{\chi}.
    \]
    \item By \cref{lem:reduction-to-IND-CPA}, for every PPT $\adv$, there exists a PPT $\bdv$ such that      
    \[
    \Adv_2(\adv)\leq \Adv_{\FHE,\bdv}^{\INDCPA}(\lambda).
    \]
\end{itemize}

Putting everything together, one obtains
\[
\Adv_{\FHE,\adv,(D,\leakdim,\chi)}^{\INDCPAwL}(\secpar)\leq 2\cdot \left(Q\cdot \varepsilon_\FHE+\frac{C \cdot \sqrt{Q\cdot\varphi \cdot \leakdim} \cdot \expansion_\infty \cdot D \cdot B}{\chi}\right)+\Adv_{\FHE,\bdv}^{\INDCPA}(\lambda),
\]
as claimed.

\begin{lemma}\label{lem:abort-correctness}
    Adopt the notations in \cref{thm:leakage-security}. 
    We have $\Adv_0(\adv)\leq \Adv_1(\adv)+2\cdot Q\cdot \varepsilon_\FHE$.
\end{lemma}

\begin{proof}
    By the definition of $(\ring,\hat{n},p,q,\xi,B)$-linear-decryption with error $\varepsilon_\FHE$, the norm of each $\hat{e}$ computed by the $\Leak\oracle$ is bounded by $B$ except with probability $\varepsilon_\FHE$. By a hybrid argument over the $Q$ queries made by $\adv$ to $\Leak\oracle$, it follows that, for each $b\in\{0,1\}$, the statistical distance between $\hyb_{b,0}$ and  $\hyb_{b,1}$ is at most $Q\cdot \varepsilon_\FHE$. Thus, by the triangle inequality, 
    \begin{align*}
        \Adv_0(\adv)&=\abs{\prob{\hyb_{0,0,\adv}(\secparam) = 1} - \prob{\hyb_{1,0,\adv}(\secparam) = 1}}\\
        &= \left|\,
            \begin{aligned}
                &\prob{\hyb_{0,0,\adv}(\secparam) = 1} - \prob{\hyb_{0,1,\adv}(\secparam) = 1} \\
                {}+{} &\prob{\hyb_{0,1,\adv}(\secparam) = 1} - \prob{\hyb_{1,1,\adv}(\secparam) = 1} \\
                {}+{} &\prob{\hyb_{1,1,\adv}(\secparam) = 1} - \prob{\hyb_{1,0,\adv}(\secparam) = 1}
            \end{aligned}
        \,\right|\\  
        &\leq
            \begin{aligned}[t]
                &\abs{\prob{\hyb_{0,0,\adv}(\secparam) = 1} - \prob{\hyb_{0,1,\adv}(\secparam) = 1}} \\
                {}+{} &\abs{\prob{\hyb_{0,1,\adv}(\secparam) = 1} - \prob{\hyb_{1,1,\adv}(\secparam) = 1}} \\
                {}+{} &\abs{\prob{\hyb_{1,1,\adv}(\secparam) = 1} - \prob{\hyb_{1,0,\adv}(\secparam) = 1}}
            \end{aligned}\\ 
        &\leq Q\cdot \varepsilon_\FHE + \Adv_1(\adv) + Q\cdot \varepsilon_\FHE\\
        &= \Adv_1(\adv)+2\cdot Q\cdot \varepsilon_\FHE.
        \tag*{\qed}
    \end{align*}
\end{proof}

\begin{lemma}\label{lem:noise-flooding-step}
    Adopt the notations in \cref{thm:leakage-security}. 
    We have $\Adv_1(\adv)\leq \Adv_2(\adv)+2\cdot \frac{C\cdot \sqrt{Q\cdot\varphi \cdot \leakdim} \cdot \expansion_\infty \cdot D \cdot B}{\chi}$.
\end{lemma}

\begin{proof}
    The only difference between $\Hyb_{b,1}$ and $\Hyb_{b,2}$ is that whenever $\norm{\hat{e}} \leq B$, the answer from $\Leak\oracle$ is changed from $\vecu \cdot \hat{e} + \vecf$ to $\vecf$. 
    For each such query $\vecu$, it holds that 
    \[
    \norm{\vecu \cdot \hat{e}}_2 
    \leq \sqrt{\varphi \cdot \leakdim} \cdot \norm{\vecu \cdot \hat{e}}_\infty 
    \leq \sqrt{\varphi \cdot \leakdim} \cdot \expansion_\infty \cdot \norm{\vecu}_\infty \cdot \norm{\hat{e}}_\infty  
    \leq  \sqrt{\varphi \cdot \leakdim} \cdot \expansion_\infty \cdot D \cdot B,
    \]
    where we use that $\norm{\vecu}_\infty \leq D$ by the norm-check in $\leak\oracle$. 
    
    For $b \in \{0,1\}$, we invoke \cref{lem:adaptive_noise_flooding} with $n = \leakdim$ and norm bound $B' \coloneqq \sqrt{\varphi \cdot \leakdim} \cdot \expansion_\infty \cdot D \cdot B$.
    The (unbounded) reduction is trivial, so we only give a sketch. It simulates $\hyb_{b,1}$ and $\hyb_{b,2}$ as follows: on each $\leak\oracle$ query it computes $\hat{e}$ using $\sk$, queries its oracle $\oracle_b$ on $\vecc \coloneqq \vecu \cdot \hat{e}$ if $\norm{\hat{e}} \leq B$ and $\vecc \coloneqq \veczero$ otherwise (so $\norm{\vecc}_2 \leq B'$ in both cases), and answers the query with the oracle output.
    The remaining parts are identical in both hybrids, which are computed by the reduction itself.
    It is direct to see that it perfectly simulates either hybrid depending on the challenge bit in the experiment in \cref{lem:adaptive_noise_flooding}.
    In particular, when $\norm{\hat{e}} > B$, the reduction answers $\vecf$ as defined in $\Hyb_{b,1}$.
    Since there are at most $Q$ queries, \cref{lem:adaptive_noise_flooding} gives
    \[
        \abs{\prob{\hyb_{b,1,\adv}(\secparam) = 1} - \prob{\hyb_{b,2,\adv}(\secparam) = 1}} \leq (C\cdot \sqrt{Q \cdot \varphi \cdot \leakdim} \cdot \expansion_\infty \cdot D\cdot B)/\chi
    \]
    for both $b \in \{0,1\}$.
    Then by the triangle inequality, 
    \begin{align*}
        \Adv_1(\adv)&=\abs{\prob{\hyb_{0,1,\adv}(\secparam) = 1} - \prob{\hyb_{1,1,\adv}(\secparam) = 1}}\\
        &= \left|\,
            \begin{aligned}
                &\prob{\hyb_{0,1,\adv}(\secparam) = 1} - \prob{\hyb_{0,2,\adv}(\secparam) = 1} \\
                {}+{} &\prob{\hyb_{0,2,\adv}(\secparam) = 1} - \prob{\hyb_{1,2,\adv}(\secparam) = 1} \\
                {}+{} &\prob{\hyb_{1,2,\adv}(\secparam) = 1} - \prob{\hyb_{1,1,\adv}(\secparam) = 1}
            \end{aligned}
        \,\right|\\  
        &\leq
        \begin{aligned}[t]
            &\abs{\prob{\hyb_{0,1,\adv}(\secparam) = 1} - \prob{\hyb_{0,2,\adv}(\secparam) = 1}} \\
            {}+{} &\abs{\prob{\hyb_{0,2,\adv}(\secparam) = 1} - \prob{\hyb_{1,2,\adv}(\secparam) = 1}} \\
            {}+{} &\abs{\prob{\hyb_{1,2,\adv}(\secparam) = 1} - \prob{\hyb_{1,1,\adv}(\secparam) = 1}}
        \end{aligned}\\ 
        &\leq (C\cdot \sqrt{Q \cdot \varphi \cdot \leakdim} \cdot \expansion_\infty \cdot D\cdot B)/\chi + \Adv_2(\adv) + (C\cdot \sqrt{Q \cdot \varphi \cdot \leakdim} \cdot \expansion_\infty \cdot D\cdot B)/\chi\\
        &= \Adv_2(\adv)+2\cdot (C\cdot \sqrt{Q \cdot \varphi \cdot \leakdim} \cdot \expansion_\infty \cdot D\cdot B)/\chi.
        \tag*{\qed}
    \end{align*}
\end{proof}

\begin{lemma}\label{lem:reduction-to-IND-CPA}
    Adopt the notations in \cref{thm:leakage-security}. 
    We have $\Adv_2(\adv)\leq \Adv_{\FHE,\bdv}^{\INDCPA}(\lambda)$.
\end{lemma}

\begin{proof}
     Consider a PPT algorithm $\adv$ distinguishing $\hyb_{0,2}$ and $\hyb_{1,2}$ with advantage $\Adv_2(\adv)$. We construct a reduction $\bdv$ against the IND-CPA security of $\FHE$, which proceeds as follows:
\begin{itemize}
    \item It receives $(\pubparam,\pk)$ from the IND-CPA challenger and sends them to $\adv$.
    \item It internally maintains lists $\Lquery,\Lleak$, initially empty. It answers $\adv$'s queries to $\enc\oracle,\Eval\oracle,\Leak\oracle$ by honestly performing the oracle computations (which are identical in $\hyb_{0,2}$ and $\hyb_{1,2}$) itself, and forwards the $\Chal\oracle$ queries to its own $\Chal\oracle$ oracle available from the IND-CPA challenger. In more detail, it answers the oracle queries as follows:
    \begin{itemize}
    \item $\enc\oracle(\id,\mu)$: If $\Lquery[\id]=\emptyset$, it computes $\ctxt\gets\enc(\pk,\mu)$ and sets $\Lquery[\id]\coloneqq(\ctxt,\mu,\mu)$. It returns $\Lquery[\id]$ to $\adv$.
    \item $\Chal\oracle(\id,\mu_0,\mu_1)$: If $\Lquery[\id]=\emptyset$, it queries its own $\Chal\oracle$ on $(\id,\mu_0,\mu_1)$ to receive $\ctxt$, and sets $\Lquery[\id]\coloneqq(\ctxt,\mu_0,\mu_1)$. It returns $\Lquery[\id]$ to $\adv$.
    \item $\Eval\oracle(\id,f,(\id_i)_{i\in I})$: If $\Lquery[\id]=\emptyset$, it asserts $\Lquery[\id_i]\ne\emptyset$ and parses $(\ctxt_i,\mu_{0,i},\mu_{1,i})\gets\Lquery[\id_i]$ for each $i\in I$, then computes $\ctxt'\gets\Eval(\pk,f,(\ctxt_i)_{i\in I})$, and sets $\Lquery[\id]\coloneqq(\ctxt', \allowbreak f((\mu_{0,i})_{i\in I}), \allowbreak f((\mu_{1,i})_{i\in I}))$. It returns $\Lquery[\id]$ to $\adv$.
    \item $\Leak\oracle(\id,\vecu\in\ring^\leakdim)$: If $\Lleak[\id]=\emptyset$, it asserts $\norm{\vecu}\leq D$ and $\Lquery[\id]\neq\emptyset$, parses $(\ctxt,\mu_0,\mu_1)\gets\Lquery[\id]$ and asserts $\mu_0=\mu_1$, samples $\vecf\sample\DD_{\ring,\chi}^{\leakdim}$, and sets $\Lleak[\id]\coloneqq\vecf$. It returns $\Lleak[\id]$ to $\adv$.
    \end{itemize}
    \item Eventually, it receives $b'$ from $\adv$ and returns $b'$.
\end{itemize}       

If $\bdv$ is interacting with the IND-CPA  experiment $0$, then queries to $\Chal\oracle(\id,\mu_0,\mu_1)$ return $\ctxt\gets\enc(\pk,\mu_0)$. Therefore, the view of $\adv$ is identical to $\hyb_{0,2}$.

If $\bdv$ is interacting with the IND-CPA  experiment $1$, then queries to $\Chal\oracle(\id,\mu_0,\mu_1)$ return $\ctxt\gets\enc(\pk,\mu_1)$. Therefore, the view of $\adv$ is identical to $\hyb_{1,2}$.

It follows that $\bdv$ perfectly simulates the view of either $\hyb_{0,2}$ or $\hyb_{1,2}$ to $\adv$, 
if it is interacting either with the IND-CPA experiment $0$ or $1$. Therefore 
\(
\Adv_2(\adv)\leq \Adv_{\FHE,\bdv}^{\INDCPA}(\lambda)
\)
as claimed.
\qed        
\end{proof}
\qed
\end{proof}

\section{Proof of \texorpdfstring{\cref{thm:correctness}}{Theorem \ref{thm:correctness}} (\ourwork Correctness)}\label{app:thm:correctness}

\tfhecorrectness*

\begin{proof}
    First observe that, for any $\tFHE$ ciphertext $(\ctxt_\FHE,\rtag)$, 
    the third and fourth lines of $\ParDec$ and $\Dec$ in \cref{fig:threshold-FHE-fine-grained} are identical,
    i.e.~they compute the same $\vecx = \matG_{\beta,q}^{-1} \begin{pmatrix} \hat{\vecc}_0 \\ \rtag \end{pmatrix}$, where 
    $(\hat{\vecc}_0,\hat{c}_1)\gets\PreDecrypt(\ctxt_\FHE)$. 

    The algorithm $\ParDec$, on input $(\pk,\sk_k,\ctxt)$, computes the following quantity
    \begin{align}\label{eq:partial-decryption-share}
        \pd_k =\round{\vecs_k^\transpose\cdot\vecc_0+e_k}_{2^z}
        =\vecs_k^\transpose\cdot\matA\cdot\vecx+e_k + d_k \bmod q,
    \end{align}
    where $d_k \coloneqq \round{y}_{2^z} - y \in \ring$ is the rounding error, for $y$ the representative of $\vecs_k^\transpose\cdot\vecc_0+e_k \bmod q$ with coefficients in $[-q/2,q/2)$.
    Writing $y_i$ for the coefficients of $y$, each coefficient of $d_k$ is $2^z \cdot (\round{y_i/2^z} - y_i/2^z)$, so $\norm{d_k} \leq 2^{z-1}$.

    Below, for any recovery set $T$, denote by $\vec{e}_T \coloneqq (e_k + d_k)_{k \in T}$ the vertical concatenation of the errors $e_k + d_k$.
    The algorithm $\Dec$, on input $(\pk,T,(\pd_k)_{k\in T}, \ctxt)$, computes $\pd$ and $c_1$ which are  
    \begin{align}\label{eq:partial-decryption}
            \pd
            &= \xi \cdot\vec{v}_0^\transpose \cdot \mat{V}_T^{-1} \cdot (\pd_{k})_{k\in T} \bmod q \nonumber\\
            &=\xi \cdot\vec{v}_0^\transpose \cdot \mat{V}_T^{-1} \cdot \left(\vecs_k^\transpose\cdot\matA\cdot\vecx+e_k + d_k\right)_{k\in T} \bmod q
                \pccomment{by \cref{eq:partial-decryption-share}} \nonumber\\
            &=\xi^2 \cdot \vec{v}_0^\transpose \cdot \mat{V}_T^{-1} \cdot \matV_T\cdot\matR\cdot\matA\cdot\vecx+\xi \cdot\vec{v}_0^\transpose \cdot \mat{V}_T^{-1}\cdot\vece_T\bmod q 
                \pccomment{$(\vec{s}_{k})_{k\in T} =\xi\cdot \mat{V}_T \cdot \mat{R} \bmod q$, $(e_k + d_k)_{k\in T} = \vece_T$} \nonumber\\
            &=\xi^2 \cdot \vec{v}_0^\transpose \cdot\matR\cdot\matA\cdot\vecx+\xi \cdot\vec{v}_0^\transpose \cdot \mat{V}_T^{-1}\cdot\vece_T\bmod q \nonumber\\
            &=\xi^2 \cdot\vecr^\transpose\cdot\matA\cdot\vecx+\xi \cdot\vec{v}_0^\transpose \cdot \mat{V}_T^{-1}\cdot\vece_T\bmod q
                \pccomment{$\vecr^\transpose= \vec{v}_0^\transpose \cdot \matR\bmod q$} 
    \end{align}
    and  
    \begin{align}\label{eq:c1-component}
        c_1
        &=\xi^2 \cdot\vecb^\transpose\cdot\vecx\bmod q \nonumber\\
        &=\xi^2 \cdot(\vecr^\transpose\cdot\matA+\vece^\transpose+\vecz^\transpose)\cdot\vecx\bmod q\pccomment{$\vecb^{\transpose}=\vecr^\transpose\matA+\vece^\transpose+\vecz^\transpose\bmod q$} \nonumber\\
        &=\xi^2 \cdot(\vecr^\transpose\cdot\matA+\vece^\transpose+(\hat{\vec{s}}^\tr,\veczero_{h}^\tr)\cdot \matG_{\beta,q} )\cdot\vecx\bmod q
            \pccomment{$\vec{z}^\transpose = ((\hat{\vec{s}}^\tr,\veczero_{h}^\tr)\cdot \matG_{\beta,q} \bmod q)$} \nonumber\\
        &=\xi^2 \cdot\vecr^\transpose\cdot\matA\cdot\vecx+\xi^2 \cdot\vece^\transpose\cdot\vecx+\xi^2\cdot((\hatvec{s}^\tr,\veczero_{h}^\tr)\cdot\matG_{\beta,q}\cdot\vecx)\bmod q \nonumber\\
        &=\xi^2 \cdot\vecr^\transpose\cdot\matA\cdot\vecx+\xi^2 \cdot\vece^\transpose\cdot\vecx+\xi^2\cdot\left((\hat{\vec{s}}^\tr,\veczero_{h}^\tr)\cdot\matG_{\beta,q}\cdot\matG_{\beta,q}^{-1} \begin{pmatrix} \hat{\vecc}_0\\ \rtag \end{pmatrix}\right)\bmod q
            \pccomment{$\vecx=\matG_{\beta,q}^{-1} \begin{pmatrix} \hat{\vecc}_0\\ \rtag \end{pmatrix}$} \nonumber\\
        &=\xi^2 \cdot\vecr^\transpose\cdot\matA\cdot\vecx
        +\xi^2 \cdot\vece^\transpose\cdot\vecx
        +\xi^2 \cdot \hat{\vec{s}}^\tr\cdot \hat{\vecc}_0 \bmod q.
    \end{align}

    Putting \cref{eq:c1-component,eq:partial-decryption} together, we have 
    \begin{align}\label{eq:c_1-pd}
        c_1 - \pd
        =~&
        \left( \xi^2 \cdot\vecr^\transpose\cdot\matA\cdot\vecx
            +\xi^2 \cdot\vece^\transpose\cdot\vecx
            +\xi^2\cdot \hat{\vec{s}}^\tr\cdot \hat{\vecc}_0\right)
        - \left(\xi^2 \cdot\vecr^\transpose\cdot\matA\cdot\vecx+\xi \cdot\vec{v}_0^\transpose \cdot \mat{V}_T^{-1} \cdot \vece_T \right) \bmod q\nonumber\\
        =~& 
        \xi^2\cdot \hat{\vec{s}}^\tr \cdot \hat{\vecc}_0 
        + \xi^2 \cdot\vece^\transpose\cdot\vecx - \xi \cdot\vec{v}_0^\transpose \cdot \mat{V}_T^{-1}\cdot\vece_T 
        \bmod q \nonumber\\
        =~&
        \xi^2 \cdot \hat{\vec{s}}^\tr \cdot \hat{\vecc}_0 + f \bmod q, 
    \end{align} 
    where in the last line we use the shorthand 
    \begin{align*}
        f \coloneqq
         \xi^2 \cdot\vece^\transpose\cdot\vecx - \xi \cdot\vec{v}_0^\transpose \cdot \mat{V}_T^{-1}\cdot\vece_T.
    \end{align*}
    Using \cref{eq:c_1-pd}, we obtain 
    \begin{align}\label{eq:vprime}
            \hat{\mu}
            =~& \xi^2\cdot \hat{c}_1- (c_{1} - \pd)  \bmod q\nonumber\\
            =~&
           \xi^2 \cdot \hat{c}_1 - \xi^2 \cdot \hat{\vec{s}}^\tr \cdot \hat{\vecc}_0 - f  \bmod q\nonumber\\
            =~&
            \xi^2 \cdot \left(\xi^{-2}\cdot\round{\frac{q}{p}}\cdot \mu+\hat{e}\right) - f \bmod q \nonumber\\            
            =~& 
            \round{\frac{q}{p}}\cdot \mu + \xi^2\cdot\hat{e} - f  \bmod q\nonumber\\
            =~& 
            \round{\frac{q}{p}}\cdot \mu + \xi^2\cdot\hat{e} - f  + q\cdot h \pccomment{for some $h\in\ring$}
    \end{align}
    where the third equality follows from the $(\ring,\hat{n},p,q,\xi,B)$-linear-decryption property of $\FHE$ (\cref{def:linear-FHE-fine-grained}),
    which guarantees that 
    \[
        \hat{c}_1-\hat{\vec{s}}^\tr \cdot \hat{\vecc}_0 \bmod q 
            = \xi^{-2} \cdot \round{q/p} \cdot \mu + \hat{e} 
    \]
    except with probability at most $\varepsilon_\FHE$.
    Finally, the algorithm $\Dec$ returns 
    \begin{align*}
    \mu'&=\round{\frac{p}{q} \cdot \hat{\mu}} \bmod p \\
    &=\round{ \frac{p}{q}\cdot\left( \round{\frac{q}{p}}\cdot\mu+\xi^2\cdot\hat{e}-f+q\cdot h\right)} \bmod p \pccomment{By \cref{eq:vprime}}\\
        &=\round{ \frac{p}{q}\cdot \left(\frac{q}{p}-\fracpart{\frac{q}{p}}\right)\cdot\mu+\frac{p}{q} \cdot\xi^2\cdot\hat{e}-\frac{p}{q}\cdot f+p\cdot h} \bmod p \pccomment{$\frac{q}{p}=\round{\frac{q}{p}}+\fracpart{\frac{q}{p}}$}\\
        &=\round{ \mu-\frac{p}{q}\cdot \fracpart{\frac{q}{p}}\cdot\mu+\frac{p}{q} \cdot\xi^2\cdot\hat{e}-\frac{p}{q}\cdot f+p\cdot h} \bmod p \\
        &=\mu+\round{ -\frac{p}{q}\cdot \fracpart{\frac{q}{p}}\cdot\mu+\frac{p}{q} \cdot\xi^2\cdot\hat{e}-\frac{p}{q}\cdot f}+p\cdot h\bmod p \\
        &=\mu+\round{ -\frac{p}{q}\cdot \fracpart{\frac{q}{p}}\cdot\mu+\frac{p}{q} \cdot\xi^2\cdot\hat{e}-\frac{p}{q}\cdot f} \bmod p.
    \end{align*}
    Thus, the plaintext $\mu$ is correctly recovered as long as 
   \begin{align*}
        \round{ -\frac{p}{q}\cdot \fracpart{\frac{q}{p}}\cdot\mu+\frac{p}{q} \cdot\xi^2\cdot\hat{e}-\frac{p}{q}\cdot f}=0&\iff \norm{-\frac{p}{q}\cdot \fracpart{\frac{q}{p}}\cdot\mu+\frac{p}{q} \cdot\xi^2\cdot\hat{e}-\frac{p}{q}\cdot f}\leq \frac{1}{2}\\
        &\iff \norm{- \fracpart{\frac{q}{p}}\cdot\mu+\xi^2\cdot\hat{e}-f}\leq \frac{q}{2p}, 
   \end{align*}
   which is implied by 
   \begin{align}\label{eq:correctness-sufficient-condition}
        \norm{\fracpart{\frac{q}{p}}\cdot\mu} +\norm{\xi^2\cdot\hat{e}} +\norm{f}\leq \frac{q}{2p}.
   \end{align}
   In the rest, we show that \cref{eq:correctness} in the theorem statement implies \cref{eq:correctness-sufficient-condition} in the above except with probability at most $\varepsilon_\FHE+\varphi\cdot(m+t)\cdot 2^{-\lambda}$, which completes the proof. We have  
   \begin{align*}
        \norm{\fracpart{\frac{q}{p}}\cdot\mu}&\leq \abs{\fracpart{\frac{q}{p}} \cdot \frac{p}{2}} = \frac{\abs{q \bmod p}}{2},\\ \\
        \norm{\xi^2\cdot\hat{e}}&=\norm{\xi^2}\cdot\norm{\hat{e}}\pccomment{$\xi\in\ZZ$}\\ 
        &\leq \xi^2\cdot B, \pccomment{By $(\ring,\hat{n},p,q,\xi,B)$-linear-decryption with error $\varepsilon_\FHE$}\\ \\
        \norm{f}&=\norm{\xi^2 \cdot\vece^\transpose\cdot\vecx-\xi \cdot\vec{v}_0^\transpose \cdot \mat{V}_T^{-1}\cdot\vece_T}\\
        &\leq \norm{\xi^2 \cdot\vece^\transpose\cdot\vecx}+\norm{\xi \cdot\vec{v}_0^\transpose \cdot \mat{V}_T^{-1}\cdot\vece_T}\\
        &\leq \Big(\expansion_\infty \cdot m\cdot\norm{\xi^2}\cdot\norm{\vece}\cdot\norm{\vecx}\qquad\qquad
            \pccomment{$\xi\in\ZZ$, $\norm{\vece^\tr\cdot\vecx}\leq \expansion_\infty \cdot m\cdot\norm{\vece}\cdot\norm{\vecx}$}\\
        &\qquad\qquad + \expansion_\infty \cdot t\cdot  \norm{\xi\cdot\vecv_0^\tr\cdot\matV_T^{-1}} \cdot\norm{\vece_T}\Big)
            \pccomment{$\norm{\xi \cdot\vec{v}_0^\transpose \cdot \mat{V}_T^{-1}\cdot\vece_T}\leq \expansion_\infty \cdot t\cdot \norm{\xi \cdot\vec{v}_0^\transpose \cdot \mat{V}_T^{-1}}\cdot\norm{\vece_T}$}\\
         &\leq \bigg(\expansion_\infty \cdot m\cdot\xi^2\cdot \bar{\chi}\cdot\sqrt{\lambda}\cdot\frac{\beta}{2}
            \pccomment{$\norm{\vece}\leq \bar{\chi}\cdot\sqrt{\lambda}$ by \cref{eq:bound-infinity-coefficient}, $\norm{\vecx}\leq \frac{\beta}{2}$, $\xi\in\ZZ$}\\
        &\qquad\qquad + \expansion_\infty \cdot t\cdot \rho\cdot(\chi\cdot\sqrt{\lambda} + 2^{z-1})\bigg)
            \pccomment{$\norm{\vece_T}\leq \chi\cdot\sqrt{\lambda} + 2^{z-1}$ by \cref{eq:bound-infinity-coefficient}, $\norm{d_k} \leq 2^{z-1}$}\\
        &= \expansion_\infty \cdot \bigg( \sqrt{\lambda} \cdot m\cdot\xi^2\cdot \bar{\chi}\cdot\frac{\beta}{2}+ t\cdot \rho\cdot(\sqrt{\lambda} \cdot \chi + 2^{z-1})\bigg),
   \end{align*} 
   where we recall that $\rho$ in the above is a quality measurement of $\Xi$ (\cshref{def:subtractive-sets}).
   Putting the above bounds together with \cref{eq:correctness} yields \cref{eq:correctness-sufficient-condition} as claimed. 
   Indeed, if $\min\{\log(1/\varepsilon_\FHE),\lambda -\log \varphi-\log(m+t)\}\geq\omega(\log\lambda)$, then the above bounds hold except with probability negligible in $\secpar$.
   \qed

\end{proof}


\section{Proofs for \texorpdfstring{\Cref{thm:tfhe-security}}{Theorem \ref{thm:tfhe-security}} (\ourwork Security)}\label{sec:missing-proofs-for-security}

This section provides the proof of \cref{thm:tfhe-security} and is structured as follows.
In \cref{sec:thm-statement-and-proof}, we recall the theorem statement and state its proof which relies on several bridging lemmas.
Then, in \cref{sec:security-bridging-lemmas}, we state and prove these bridging lemmas \cref{lem:thLWE-hybrid-step,lem:H1-H2,lem:H2-H3,lem:H3-H4}.
One of the bridging lemmas (\cref{lem:thLWE-hybrid-step}) will further rely on a separate utility lemma (\cref{lem:th-LWE-variant}), the latter being a generalisation of a lemma in \cite{AC:CinLaiWoo25} and is to be provided in \cref{app:utility-and-adpLLWE}.

\subsection{Theorem Statement and Proof}\label{sec:thm-statement-and-proof}

\tfhesecurity*

\begin{proof}
    Let $\tFHE'$ denote the scheme in \cref{fig:threshold-FHE-fine-grained} with the step $\pd_k \gets \round{\pd_k}_{2^z}$ of $\ParDec$ removed, i.e.~the scheme without bit dropping.
    Since $\round{\cdot}_{2^z}$ is a deterministic public function applied only to the outputs of $\ParDec$, any adversary $\adv$ against $\tFHE$ yields an adversary $\adv'$ against $\tFHE'$ that runs $\adv$, forwards all its oracle queries, and applies $\round{\cdot}_{2^z}$ to every partial decryption before passing it to $\adv$.
    The adversary $\adv'$ is PPT, requests partial decryptions on the same ciphertexts as $\adv$, and perfectly simulates $\mathsf{Exp}^b_{\tFHE,\adv}$ for both $b \in \set{0,1}$, so $\Adv^{\INDCPA}_{\tFHE,\adv}(\secpar) = \Adv^{\INDCPA}_{\tFHE',\adv'}(\secpar)$.
    It therefore suffices to show $\Adv^{\INDCPA}_{\tFHE',\adv}(\secpar) \leq \varepsilon$ for every PPT $\adv$ against $\tFHE'$ that requests partial decryptions on at most $Q$ ciphertexts, with $\varepsilon$ as in \cref{thm:tfhe-security}.
    In the rest of this proof, $\adv$ denotes such an adversary and partial decryptions are those in $\mathsf{Exp}^b_{\tFHE',\adv}$.

    Throughout this proof (and in \cref{lem:th-LWE-variant,lem:thLWE-hybrid-step,lem:H1-H2,lem:H2-H3,lem:H3-H4} which support this proof), we adopt the following notations: 
    \begin{itemize}
        \item For any subset $\cdv \subseteq [K]$, denote by $\mat{V}_{\cdv} \in \ring^{|\cdv| \times t}$ the submatrix of $\mat{V}$ whose rows are indexed by $\cdv$. Similarly, $\mat{S}_{\cdv} \in \ring_q^{|\cdv| \times n}$ is the submatrix of $\mat{S}=\matV\cdot\matR $ whose rows are indexed by $\cdv$,  i.e., $\mat{S}_{\cdv} = \mat{V}_{\cdv} \cdot \mat{R}$. 
        \item For any set $\cdv \subset [K]$ of corrupt parties with $|\cdv| <t$, we define:
        \begin{itemize}
            \item $\hdv \coloneqq [K] \setminus \cdv$ the set of honest (i.e.~non-corrupt) parties, and 
            \item $\mat{T}_\cdv \in \ring^{t \times t}$ a matrix such that 
            it holds $\xi^{-1} \cdot \mat{T}_\cdv \cdot \mat{V}_{\cdv^*} = \mat{I}_t$
            for some size-$t$ set $\cdv^*$ satisfying $(\set{0} \cup \cdv) \subseteq \cdv^* \subseteq [0,K]$.
            ($\mat{T}_\cdv$ is not unique given $\cdv$ and we may fix an arbitrary one.) 
        \end{itemize}
    \end{itemize}

    We consider the following sequence of hybrids. 
    
    \begin{description}
    \item $\hyb_{b,0}$: This is the  $\tFHE'$ security experiment for challenge bit $b$. To recall, the view of the adversary consists of $\pk = (\pk_\FHE, \matA, \vec{b})$ where
    
            \begin{align*}    
        \vec{b}^\transpose \coloneqq \vec{r}^\transpose \matA + \vec{e}^\transpose + \vecz^\tr \bmod q,
        && 
        \vecz^\tr = (\hat{\vecs}^\tr,\veczero_h^\tr) \cdot \matG_{\beta,q} \bmod q,
    \end{align*}
    corrupt user keys $(\vecs_k)_{k \in \cdv}$, ciphertexts $\ctxt = (\ctxt_\FHE, \rtag)$ from $\enc\oracle, \Chal\oracle$ or $\Eval\oracle$, partial decryptions $\pd_k = \vecs_k^\transpose \matA \vecx + e_k \bmod q$ where $\vecx \coloneqq \matG^{-1}\begin{pmatrix} \hat{\vecc}_0 \\ \rtag \end{pmatrix}$ and the vector $\hatvec{c}_0$ is obtained from running 
    \(
        (\hat{\vecc}_0,\hat{c}_1) \gets \FHE.\PreDecrypt(\ctxt_\FHE),
    \)
    which is possible because of the $(\ring,\hat{n},p,q,\xi,B)$-linear-decryption property of $\FHE$. 
    (We recall that $\PreDecrypt$ is a public operation, thus the adversary can also compute the above and observe $(\hat{\vecc}_0,\hat{c}_1)$.)
    
    \item $\hyb_{b,1}$: This experiment is identical to $\hyb_{b,0}$ except for the following change.
        For each $\ctxt$, sample a random vector $\vecw_\vecx \sample \ring_q^{t-|\cdv|-1}$ and an error $\vece_\vecx \sample \DD_{\ring,\chi}^{K}$, its partial decryptions are computed as 
\begin{align} 
(\pd_k)_{k \in [K]} 
&=
\begin{bNiceArray}{c|c|c}  
\Block{1-3}{\mat{V}_\hdv \cdot \mat{T}_\cdv}  \\ \hline 
\quad &\quad & \xi\mat{I}_\cdv
\end{bNiceArray}
{\begin{bmatrix}
(\vecb - \vece - \vecz)^\transpose \vecx \\
\vecw_\vecx \\
\mat{S}_\cdv \mat{A} \vecx
\end{bmatrix}}
+ \vec{e}_\vecx \bmod q \nonumber\\
\eqqcolon~&
\begin{bNiceArray}{w{c}{2.5em}|w{c}{2.5em}|w{c}{3em}}
\Block{1-3}{\mat{V}_\hdv \cdot (\vecu_0,\matU_1,\matU_2)}  \\ \hline 
\quad &\quad & \xi\mat{I}_\cdv
\end{bNiceArray}
{\begin{bmatrix}
(\vecb - \vece - \vecz)^\transpose \vecx \\
\vecw_\vecx \\
\mat{S}_\cdv \mat{A} \vecx
\end{bmatrix}}
+ \vec{e}_\vecx \bmod q
\label{eq:simulation-with-random-b}
\end{align}
        where we recall $\vecx \coloneqq \matG^{-1}\begin{pmatrix} \hat{\vecc}_0\\ \rtag \end{pmatrix}$ and $\vecz^\tr = (\hat{\vecs}^\tr,\veczero_h^\tr) \cdot \matG_{\beta,q} \bmod q$.

        In the last line, we define the shorthand 
        \begin{equation}\label{eq:definition-of-u0-U1-U2}
        \matT_\cdv = (\vecu_0, \matU_1, \matU_2) \in \ring^{t \times 1} \times \ring^{t \times t-|\cdv|-1} \times \ring^{t \times |\cdv|},
        \end{equation} 
        for the block columns of $\mat{T}_\cdv$,
        which will be used in the subsequent hybrids.

    \item $\hyb_{b,2}$: This experiment is identical to $\hyb_{b,1}$ except for the following change.
    For each $\ParDec\oracle$ query $(\id, k)$, if $\Lquery[\id] = (\ctxt, \mu_0, \mu_1)$ is such that
    $\mu_0 \neq \mu_1$ and $k \notin \cdv$, we change the partial decryption answer for $k$ to a
    uniformly random $\ring_q$ element. For corrupt parties $k \in \cdv$, the answer is kept as
    in $\hyb_{b,1}$.

    \item $\hyb_{b,3}$: This experiment is identical to $\hyb_{b,2}$ except for the following change. For each $\ParDec\oracle$ query $\id$, if $\Lquery[\id] = (\ctxt, \mu_0, \mu_1)$ is such that $\mu_0 = \mu_1$, we express its partial decryptions as 
   \begin{align}
        (\pd_k)_{k \in [K]} \nonumber  
        &=
        \begin{bmatrix} \mat{V}_\hdv \vecu_0 \\ \matzero \end{bmatrix}
        \begin{bmatrix}(\vecb - \vece)^\transpose \vecx - \left(\hat{c}_1- \left(\round{\frac{q}{p}}\cdot \xi^{-2} \cdot \mu + \hat{e}\right)\right) \end{bmatrix} \\ 
        &\quad+ 
        {\begin{bNiceArray}{w{c}{2.5em}|w{c}{3em}}  
            \Block{1-2}{\mat{V}_\hdv (\matU_1,\matU_2)}  \\ \hline 
            \quad & \xi\mat{I}_\cdv \end{bNiceArray}}
        {\begin{bmatrix}
            \vecw_\vecx \\ \mat{S}_\cdv \mat{A} \vecx 
        \end{bmatrix}} 
        + \vece_\vecx 
        \bmod q, \label{eq:simulation-from-mu}
    \end{align}
    where $\mu \coloneqq \mu_0 = \mu_1$, and $\hat{e} \in \ring$ is defined as
    \begin{equation}\label{eq:hate}
        \hat{e} \coloneqq \hat{c}_1 - \hat{\vecs}^\tr\cdot\hat{\vecc}_0 - \xi^{-2} \round{\frac{q}{p}}\mu \bmod q
    \end{equation}
    We recall $\vecx \coloneqq \matG^{-1}\begin{pmatrix} \hat{\vecc}_0 \\ \rtag \end{pmatrix}$, and $\vecu_0, \matU_1, \matU_2$ are defined in $\hyb_{b,1}$.
    \end{description}
    For $i\in[0,3]$, let $\Adv_i(\adv)$ denote the advantage of $\adv$ in experiments $(\hyb_{0,i},\hyb_{1,i})$, that is
    \(
        \Adv_i(\adv) \coloneqq \abs{\prob{\hyb_{0,i,\adv}(\secparam) = 1} - \prob{\hyb_{1,i,\adv}(\secparam) = 1}}.
    \)
    By definition of $\hyb_{b,0}$, one has
    \(
        \Adv_{\tFHE',\adv}^{\INDCPA}(\secpar)=\Adv_0(\adv).
    \)
  By \cref{lem:thLWE-hybrid-step}, under the condition stated in \cref{thm:tfhe-security}, for any PPT $\adv$, there exist PPT $\bdv_1,\bdv_2$ such that
    \begin{align*}
        \Adv_0(\adv)
        \leq~&
        \Adv_1(\adv)+2\cdot t\cdot\Bigg(\Adv_{\params_1,\bdv_1}^{\LWE{}}(\secpar)
        +\Adv_{\params_2,\bdv_2}^{\LWE{}}(\secpar)
        + \frac{6 \cdot Q \cdot \epsilon}{t(1-\epsilon)}
        + 4 \cdot QK \cdot \epsilon
        + \frac{6 \cdot \epsilon}{1-\epsilon} \Bigg).
    \end{align*}
    By \cref{lem:H1-H2}, 
        \(
         \Adv_1(\adv)=  \Adv_2(\adv).
        \)
    By \cref{lem:H2-H3},
        \(
         \Adv_2(\adv)= \Adv_3(\adv).
        \)
    By \cref{lem:H3-H4}, $\forall$ PPT $\adv$, $\exists$ PPT $\bdv_3$ such that  
        \(
        \Adv_3(\adv)\leq \Adv_{\FHE,\bdv_3,\params_3}^{\INDCPAwL}(\secpar).
        \)
    Putting everything together, one obtains the claimed result.
    \Cref{lem:thLWE-hybrid-step,lem:H1-H2,lem:H2-H3,lem:H3-H4} and their proofs are provided in \cref{sec:security-bridging-lemmas}.
    \qed
\end{proof}

\subsection{Proofs of \texorpdfstring{\Cref{lem:thLWE-hybrid-step,lem:H1-H2,lem:H2-H3,lem:H3-H4}}{Lemmas \ref{lem:thLWE-hybrid-step}, \ref{lem:H1-H2}, \ref{lem:H2-H3} and \ref{lem:H3-H4}}}\label{sec:security-bridging-lemmas}

\begin{lemma}\label{lem:thLWE-hybrid-step}
    Adopt the notations in \cref{thm:tfhe-security}. 
        Suppose $q$ and $\xi$ are coprime, $0 < \epsilon \leq 1/2$, and the noise parameters $\chi_1^*,\chi_2^*$ satisfy both of the following:
        \begin{itemize}
            \item $\chi_1^* \geq \sqrt{2} \cdot \eta_\epsilon(\ring^t)$ is such that there exists $\hat{\chi} >0$ satisfying
                \(
                    \gamma \leq \chi \cdot \sqrt{(2 \eta_\epsilon(\ring^t)^2 + (\chi_1^*)^2)^{-1} - \hat{\chi}^{-2}};
                \)
            \item $\chi_2^* \geq \sqrt{2} \cdot \eta_\epsilon(\ring^m)$,
                $\gamma \cdot \sqrt{\varphi m} \cdot \beta/2 \geq 1$, and, with $\tensorexpansion{t} \coloneqq 1$ if $\conductor$ is a power of two and $\tensorexpansion{t} \coloneqq \sqrt{t}$ otherwise,
                \(
                    \sqrt{2 Q K \cdot \varphi m} \cdot \expansion_2 \cdot \tensorexpansion{t} \cdot \gamma \cdot \beta/2 \leq \chi \cdot \sqrt{(2 \eta_\epsilon(\ring^m)^2 + (\chi_2^*)^2)^{-1} - \bar{\chi}^{-2}}.
                \)
        \end{itemize}
            Then, for any PPT $\adv$, there exist PPT $\bdv_1,\bdv_2$ such that
            \begin{align*}
                \Adv_0(\adv)
                \leq~&
                \Adv_1(\adv)+2\cdot t\cdot\Bigg(\Adv_{(\ring,h \ceil{\log_\beta q},\poly,q,\chi_1^*,\ddvt),\bdv_1}^{\LWE{}}(\secpar)
                +\Adv_{(\ring,n,m,q,\chi_2^*),\bdv_2}^{\LWE{}}(\secpar) \\
                &\qquad\qquad+ \frac{6 \cdot Q \cdot \epsilon}{t(1-\epsilon)}
                + 4 \cdot QK \cdot \epsilon
                + \frac{6 \cdot \epsilon}{1-\epsilon} \Bigg)
            \end{align*}
\end{lemma}
\begin{proof}
    Let $\params = ((\ring,n,m,q,Q,\bar{\chi},\chi,\ddvA,\ddvt,\beta/2),(t,K,\Xi))$, where
    \[
    \ddvA \coloneqq \mathcal{U}(\ring_q^{n \times m}),\quad\text{and}\quad \ddvt\coloneqq\left\{\pretag : \rtag\sample\ring_q^h, \pretag\coloneqq \matG_{\beta,q}^{-1}(\rtag)\right\}.
    \]
    We will show that for any PPT $\adv$, there exists a PPT $\bdv$ such that
    \begin{align*}
         \Adv_0(\adv)&\leq \Adv_1(\adv)+2\cdot\Adv_{\params,\bdv}^{\mathsf{RealIdeal}}(\secpar),
    \end{align*}
    where $\Adv_{\params,\bdv}^{\mathsf{RealIdeal}}(\secpar)$ is defined in \cref{lem:th-LWE-variant}.
    The claimed result will follow by plugging the upper-bound for $\Adv_{\params,\bdv}^{\mathsf{RealIdeal}}(\secpar)$ from \cref{lem:th-LWE-variant}, and by recalling that the parameter $\ddvA$ is omitted from the parameters of the $\LWE{}$ instance when $\ddvA$ is the uniform distribution (see \cref{def:LWE}).

    Suppose $\adv$ is a PPT distinguisher for ($\hyb_{b,0},\hyb_{b,1})$. We construct a reduction $\bdv$ distinguishing $\mathsf{Real}_{\params,\bdv}$ and $\mathsf{Ideal}_{\params,\bdv}$.
    The reduction $\bdv$ works as follows:
\begin{itemize}
    \item It receives $\matA\in\ring_q^{n\times m}$ from the challenger.
        It samples $\pubparam_\FHE\gets\FHE.\Setup(1^\lambda)$, sets $\pubparam\coloneqq\pubparam_\FHE$,
        and passes $\pubparam$ to $\adv$, receiving from which $\mathcal{C}\subset_{<t}[K]$.
    \item It samples $(\pk_\FHE,\sk_\FHE)\gets\FHE.\kgen(\pp_\FHE)$, parses
        $\sk_\FHE=\hat{\vecs}\in\ring_q^{\hat{n}}$, and sets
        $\vecz^\tr\coloneqq(\hat{\vecs}^\tr,\veczero_h^\tr)\cdot\matG_{\beta,q}\bmod q \in \ring_q^m$.
        It forwards $(\mathcal{C},\vecz)$ to the challenger, receiving from which
        $(\vecb,\matS_\mathcal{C}=(\sk_k)_{k\in\mathcal{C}})$.
    \item It sets $\pk\coloneqq(\pk_\FHE,\matA,\vecb)$ and
        $(\sk_k)_{k\in\mathcal{C}}\coloneqq(\vecs_k)_{k\in\mathcal{C}}$, and sends
        $(\pk,(\sk_k)_{k\in\mathcal{C}})$ to $\adv$.
    \item It internally maintains lists $\Lquery,\Lshare,\Lindex$, initially empty. 
        Define two subroutines $\mathsf{Pre}$ and $\mathsf{Tag}$ below, which will be used by the reduction to answer oracle queries, and where $\mathsf{Tag}$ invokes the experiment's oracle $\Gen\oracle$:
        \begin{pchstack}[center]
            \procedure[]{$\mathsf{Pre}(\ctxt_\FHE)$}
            {
                (\hat{\vecc}_0,\hat{c}_1) \gets \PreDecrypt(\ctxt_\FHE) \\
                \pcreturn \matG_{\beta,q}^{-1}(\hat{\vecc}_0)
            }
            \pchspace
            \procedure[]{$\mathsf{Tag}(\id,\ctxt_\FHE)$}
            {
                \prectxt \coloneqq \mathsf{Pre}(\ctxt_\FHE) \\
                \vecx \gets \Gen\oracle(\id,\prectxt) \\
                \pcparse \begin{pmatrix}\prectxt \pclb \pretag \end{pmatrix} \gets \vecx \\
                \rtag \coloneqq \matG_{\beta,q} \cdot \pretag \bmod q \\
                \pcreturn \ctxt \coloneqq (\ctxt_\FHE,\rtag)
            }
        \end{pchstack}
        The reduction answers the oracle queries of $\adv$ as follows:
    \begin{itemize}
        \item $\enc\oracle(\id,\mu)$: If $\Lquery[\id]=\emptyset$,
            it computes
            \begin{align*}
                \ctxt_\FHE \gets \FHE.\enc(\pk_\FHE,\mu)
                &&\text{and}&&
                \ctxt \gets \mathsf{Tag}(\id,\ctxt_\FHE),
                &&\text{and sets}~ \Lquery[\id]\coloneqq(\ctxt,\mu,\mu).
            \end{align*}
            It returns $\Lquery[\id]$ to $\adv$.
        \item $\Chal\oracle(\id,\mu_0,\mu_1)$: If $\Lquery[\id]=\emptyset$,
            it computes
            \begin{align*}
                \ctxt_\FHE \gets \FHE.\enc(\pk_\FHE,\mu_b)
                &&\text{and}&&
                \ctxt \gets \mathsf{Tag}(\id,\ctxt_\FHE),
                &&\text{and sets}~ \Lquery[\id]\coloneqq(\ctxt,\mu_0,\mu_1).
            \end{align*}
            It returns $\Lquery[\id]$ to $\adv$.
        \item $\Eval\oracle(\id,f,(\id_i)_{i\in I})$: If $\Lquery[\id]=\emptyset$, then:
            \begin{itemize}
            \item For each $i\in I$, it returns $\bot$ if $\Lquery[\id_i]=\emptyset$, and otherwise parses $(\ctxt_i,\mu_{0,i},\mu_{1,i})\gets\Lquery[\id_i]$ and $(\ctxt_{\FHE,i},\rtag_i)\gets\ctxt_i$.
            \item It computes
            \begin{align*}
                &\ctxt_\FHE'\gets\Eval(\pk,f,(\ctxt_{\FHE,i})_{i\in I})
                \quad\text{and}\quad
                \ctxt \gets \mathsf{Tag}(\id,\ctxt_\FHE'), \\
                &\text{and sets}~ \Lquery[\id]\coloneqq\bigl(\ctxt,~ f((\mu_{0,i})_{i\in I}),~ f((\mu_{1,i})_{i\in I})\bigr).
            \end{align*}
            \end{itemize}
            It returns $\Lquery[\id]$ to $\adv$.
        \item $\ParDec\oracle(\id,k)$: If $\Lquery[\id]=\emptyset$, it returns $\bot$. Otherwise:
        \begin{itemize}
            \item It adds $\{k\}$ to $\Lindex[\id]$, and parses $(\ctxt,\mu_0,\mu_1)\gets\Lquery[\id]$ and $(\ctxt_\FHE,\rtag)\gets\ctxt$.
            \item If $\mu_0 \neq \mu_1$ and $|\Lindex[\id]\cup\mathcal{C}|\geq t$, it returns $\bot$.
            \item Otherwise, if $\LShare[\id,k]=\emptyset$, it queries
            \begin{align*}
                d_k\gets\Share\oracle(\id,k)
                \quad\text{and sets}~ \LShare[\id,k]\coloneqq d_k.
            \end{align*}
            It returns $\LShare[\id,k]$ to $\adv$.
        \end{itemize}
    \end{itemize}
    \item Eventually receive $b'$ from $\adv$, return $b'$.
\end{itemize}
  Below, we show that $\bdv$ perfectly simulates the view of either $\hyb_{b,0}$ or $\hyb_{b,1}$ to $\adv$,
        if it is interacting either with the $\mathsf{Real}_{\params,\bdv}$ or the $\mathsf{Ideal}_{\params,\bdv}$ experiment.

If $\bdv$ is interacting with the $\mathsf{Real}_{\params,\bdv}$ experiment, then 
\begin{align*}
    \vecb^\tr &= \vecr^\tr\matA+\vece^\tr+\vecz^\tr \bmod q\\
    \vecx&=\begin{pmatrix} \prectxt \\ \pretag \end{pmatrix}\quad\text{for }\prectxt=\matG_{\beta,q}^{-1} \begin{pmatrix} \hat{\vecc}_0 \end{pmatrix} \text{ and }\pretag \sample \ddvt\\
    \vecd^\tr&=(\pd_k)_{k \in [K]}=\xi\cdot\matV\matR\matA\vecx+\vece_\vecx\bmod q
\end{align*}
for $\vec{e}_\vecx \sample \DD_{\ring,\chisub}^K$. Using that $\ddvt=\left\{\pretag : \rtag\sample\ring_q^h, \pretag\coloneqq \matG_{\beta,q}^{-1}(\rtag)\right\}$ and that $(\vecs_k)_{k\in[K]}=\xi\cdot\matV\matR\bmod q$, one obtains that
\[
(\pd_k)_{k \in [K]}=\underbrace{\xi\cdot\matV\cdot\matR}_{(\vecs_k)_{k\in[K]}}\cdot\underbrace{\matA\cdot\vecx}_{\vecc_0}+\underbrace{\vece_\vecx}_{(e_k)_{k\in[K]}}=\left(\vecs_k^\tr\cdot\vecc_0+e_k\right)_{k\in[K]}
\]
Therefore, the view of $\adv$ is identical to $\hyb_{b,0}$.

If $\bdv$ is interacting with the $\mathsf{Ideal}_{\params,\bdv}$ experiment, then 
\begin{align*}
    \vecb&\sample \ring_q^m\\
    \vecx&=\begin{pmatrix} \prectxt \\ \pretag\end{pmatrix}\quad\text{for }\prectxt=\matG_{\beta,q}^{-1} \begin{pmatrix} \hat{\vecc}_0  \end{pmatrix} \text{ and }\pretag \sample \ddvt\\
    \vecd^\tr&=(\pd_k)_{k \in [K]}=
    \begin{bNiceArray}{c|c|c}  
        \Block{1-3}{\mat{V}_\hdv \cdot \mat{T}_\cdv}  \\ \hline 
        \quad &\quad & \xi\mat{I}_\cdv \end{bNiceArray}
    {\begin{bmatrix}
        (\vecb - \vece - \vecz)^\transpose \vecx \\ \vecw_\vecx \\ \mat{S}_\cdv \mat{A} \vecx
    \end{bmatrix}} + \vec{e}_\vecx \bmod q 
\end{align*}
Using as above that $\ddvt=\left\{\pretag : \rtag\sample\ring_q^h, \pretag\coloneqq \matG_{\beta,q}^{-1}(\rtag)\right\}$, we conclude that the view of $\adv$ is identical to $\hyb_{b,1}$.
By the triangle inequality, it follows that 
\begin{align*}
    \Adv_0(\adv)
    &= \abs{\prob{\hyb_{0,0,\adv}(\secparam) = 1} - \prob{\hyb_{1,0,\adv}(\secparam) = 1}} \\
    &= \left|\,
        \begin{aligned}
            &\prob{\hyb_{0,0,\adv}(\secparam) = 1} - \prob{\hyb_{0,1,\adv}(\secparam) = 1} \\
            {}+{} &\prob{\hyb_{0,1,\adv}(\secparam) = 1} - \prob{\hyb_{1,1,\adv}(\secparam) = 1} \\
            {}+{} &\prob{\hyb_{1,1,\adv}(\secparam) = 1} - \prob{\hyb_{1,0,\adv}(\secparam) = 1}
        \end{aligned}
    \,\right| \\
    &\leq
        \begin{aligned}[t]
            &\abs{\prob{\hyb_{0,0,\adv}(\secparam) = 1} - \prob{\hyb_{0,1,\adv}(\secparam) = 1}} \\
            {}+{} &\abs{\prob{\hyb_{0,1,\adv}(\secparam) = 1} - \prob{\hyb_{1,1,\adv}(\secparam) = 1}} \\
            {}+{} &\abs{\prob{\hyb_{1,1,\adv}(\secparam) = 1} - \prob{\hyb_{1,0,\adv}(\secparam) = 1}}
        \end{aligned} \\
    &\leq \Adv_{\params,\bdv}^{\mathsf{RealIdeal}}(\secpar) + \Adv_1(\adv) + \Adv_{\params,\bdv}^{\mathsf{RealIdeal}}(\secpar) \\
    &= \Adv_1(\adv)+2\cdot\Adv_{\params,\bdv}^{\mathsf{RealIdeal}}(\secpar),
\end{align*}
as wanted. As mentioned at the beginning of the proof, plugging the upper-bound for $\Adv_{\params,\adv}^{\mathsf{RealIdeal}}(\secpar)$ from \cref{lem:th-LWE-variant} yields the claimed results.
\qed
\end{proof}

\begin{lemma}\label{lem:H1-H2}
    Adopt the notations in \cref{thm:tfhe-security}. We have $\Adv_1(\adv)=\Adv_2(\adv)$.
\end{lemma}

\begin{proof}
    Fix any $\id$ such that $\Lquery[\id] = (\ctxt, \mu_0, \mu_1)$ and $\mu_0 \neq \mu_1$.
    Denote by $\bar{\hdv} \subseteq \hdv = [K] \setminus \cdv$ the set of honest parties for which an adversary has queried its partial decryptions. 
    By design of the security experiment, it must hold that $|\bar{\hdv}| < t - |\cdv|$. 

    Now recall that, in $\hyb_{b,1}$, the set of partial decryptions for the parties $\bar{\hdv}$ is defined by \cref{eq:simulation-with-random-b}, which is 
    \begin{align*}
        (\pd_k)_{k \in \bar{\hdv}}
        =~&
        \begin{bmatrix}  
            {\mat{V}_{\bar{\hdv}} \cdot (\vecu_0,\matU_1,\matU_2)} \end{bmatrix}
        {\begin{bmatrix}
            (\vecb - \vece - \vecz)^\transpose \vecx \\ \vecw_\vecx \\ \mat{S}_\cdv \mat{A} \vecx 
        \end{bmatrix}} + \vec{e}_\vecx \bmod q \\
        =~&
        \left(
        \matV_{\bar{\hdv}} \cdot \vecu_0 \cdot (\vecb - \vece - \vecz)^\tr \vecx
        + 
        \matV_{\bar{\hdv}} \cdot \matU_1 \cdot \vecw_\vecx 
        +
        \matV_{\bar{\hdv}} \cdot \matU_2 \cdot \matS_\cdv \matA \vecx \right)
            + \vece_\vecx \bmod q, 
    \end{align*}
    where we recall $\vecw_\vecx$ is a random vector sampled by the experiment  
    and $\vecx = \matG^{-1}\begin{pmatrix} \hat{\vecc}_0 \\ \rtag \end{pmatrix}$.
    By \cite[Lemma 9, Eprint]{AC:CinLaiWoo25}, the columns of $\matV_{\bar{\hdv}} \cdot \matU_1$ generate $\ring_q^{|\bar{\hdv}|}$. 
    Therefore, since $\vecw_\vecx$ is uniformly random, so is $\matV_{\bar{\hdv}} \cdot \matU_1 \cdot \vecw_\vecx \bmod q$.
    That is, 
    \begin{align*}
        \matV_{\bar{\hdv}} \cdot \matU_1 \cdot \vecw_\vecx \bmod q \sim \mathcal{U}(\ring_q^{|\bar{\hdv}| }),
    \end{align*}
    and therefore we get
    \begin{align*}
        (\pd_k)_{k \in \bar{\hdv}}
        \equiv~&
        \left(
        \matV_{\bar{\hdv}} \cdot \vecu_0 \cdot (\vecb - \vece - \vecz)^\tr \vecx
        +
        \mathcal{U}(\ring_q^{|\bar{\hdv}|})
        +
        \matV_{\bar{\hdv}} \cdot \matU_2 \cdot \matS_\cdv \matA \vecx \right)
        + \vece_\vecx \bmod q \\
        \equiv~&
        \mathcal{U}(\ring_q^{|\bar{\hdv}|}).
    \end{align*}
    In other words, the set of partial decryptions $(\pd_k)_{k \in \bar{\hdv}}$ for $\id$ is uniformly random,
    identical to that in $\hyb_{b,2}$. Therefore, $\Adv_1(\adv)= \Adv_2(\adv)$.
    \qed 
\end{proof}

\begin{lemma}\label{lem:H2-H3}
    Adopt the notations in \cref{thm:tfhe-security}. 
    We have $\Adv_2(\adv)=\Adv_3(\adv)$.
\end{lemma}

\begin{proof}
    Fix any $\id$ such that $\Lquery[\id] = (\ctxt, \mu_0, \mu_1)$, $(\hat{\vecc}_0,\hat{c}_1)\gets\PreDecrypt(\ctxt_\FHE)$, and $\mu_0 = \mu_1$. 
    Denote $\mu \coloneqq \mu_0 = \mu_1$ and let 
    \begin{equation}\label{eq:def-hat-e}
    \hat{e} \coloneqq \hat{c}_1 - \hat{\vecs}^\tr\cdot\hat{\vecc}_0 - \xi^{-2} \round{\frac{q}{p}}\mu \bmod q
    \end{equation}
    Notice that
    \begin{align*}
        \vecz^\transpose \vecx \bmod q 
        =~& 
        ((\hat{\vecs}^\tr,\veczero_h^\tr) \cdot \matG_{\beta,q} ) \cdot \matG_{\beta,q}^{-1}\begin{pmatrix} \hat{\vecc}_0 \\ \rtag \end{pmatrix} \bmod q \\
        =~& 
        \left( (\hat{\vecs}^\tr,\veczero_h^\tr) \cdot \begin{pmatrix} \hat{\vecc}_0 \\ \rtag \end{pmatrix} \right) \bmod q \\
        =~& 
        (\hat{\vecs}^\tr \cdot \hat{\vecc}_0) \bmod q \\
        =~&
        \left(\hat{c}_1-\left( \round{\frac{q}{p}}\cdot \xi^{-2} \cdot \mu + \hat{e}\right)+q\cdot h''\right) \bmod q \pccomment{By \cref{eq:def-hat-e}}\\
        =~&
        \left(\hat{c}_1-\left( \round{\frac{q}{p}}\cdot \xi^{-2} \cdot \mu + \hat{e}\right)\right)\bmod q.
    \end{align*} 
    Hence, 
    \begin{align*}    
        &{\begin{bNiceArray}{w{c}{2.5em}|w{c}{2.5em}|w{c}{3em}}  
            \Block{1-3}{\mat{V}_\hdv \cdot (\vecu_0,\matU_1,\matU_2)}  \\ \hline 
            \quad &\quad & \xi\mat{I}_\cdv \end{bNiceArray}}
        {\begin{bmatrix}
            (\vecb - \vece)^\transpose \vecx - \vecz^\transpose \vecx \\ \vecw_\vecx \\ \mat{S}_\cdv \mat{A} \vecx 
        \end{bmatrix}} + \vec{e}_\vecx \bmod q \\
        =~&
        {\begin{bNiceArray}{w{c}{2.5em}|w{c}{2.5em}|w{c}{3em}}  
            \Block{1-3}{\mat{V}_\hdv \cdot (\vecu_0,\matU_1,\matU_2)}  \\ \hline 
            \quad &\quad & \xi\mat{I}_\cdv \end{bNiceArray}}
        {\begin{bmatrix}
            (\vecb - \vece)^\transpose \vecx - \left(\hat{c}_1-\left(\round{\frac{q}{p}}\cdot \xi^{-2}\cdot  \mu + \hat{e}\right)\right) \\ 
            \vecw_\vecx \\ \mat{S}_\cdv \mat{A} \vecx 
        \end{bmatrix}} + \vece_\vecx \bmod q \\
        =~&
        \begin{bmatrix} \mat{V}_\hdv \vecu_0 \\ \matzero \end{bmatrix}
        \begin{bmatrix} (\vecb - \vece)^\transpose \vecx - \left(\hat{c}_1- \left(\round{\frac{q}{p}}\cdot \xi^{-2}\cdot \mu + \hat{e}\right)\right) \end{bmatrix}
        + 
        {\begin{bNiceArray}{w{c}{2.5em}|w{c}{3em}}  
            \Block{1-2}{\mat{V}_\hdv (\matU_1,\matU_2)}  \\ \hline 
            \quad & \xi\mat{I}_\cdv \end{bNiceArray}}
        {\begin{bmatrix}
            \vecw_\vecx \\ \mat{S}_\cdv \mat{A} \vecx 
        \end{bmatrix}} + \vece_\vecx \bmod q, 
    \end{align*} 
    that is, the partial decryptions for $\id$ in $\hyb_{b,2}$ and $\hyb_{b,3}$ are identical. Therefore, $\Adv_2(\adv)=\Adv_3(\adv)$. 
    \qed
\end{proof}

\begin{lemma}\label{lem:H3-H4}
    Adopt the notations in \cref{thm:tfhe-security}. 
    Then, for any PPT $\adv$ there exists PPT $\bdv_3$ such that $\Adv_3(\adv)\leq \Adv_{\FHE,\bdv_3,(\theta,K,\chi)}^{\INDCPAwL}(\secpar)$.
\end{lemma}

\begin{proof} 
        Suppose $\adv$ is a PPT  distinguisher for ($\hyb_{0,3},\hyb_{1,3})$. We construct a reduction $\bdv_3$ against the IND-CPA under leakage security of $\FHE$, which proceeds as follows:
\begin{itemize}
    \item It receives $(\pubparam_\FHE,\pk_\FHE)$ from the IND-CPA-under-leakage challenger,
        parses $\pubparam=\pubparam_\FHE$, and passes $\pubparam$ to $\adv$, receiving from which
        $\mathcal{C}\subset_{<t}[K]$.
    \item It samples $\matA\sample\ring_q^{n\times m}$, $\matR\sample\ring_q^{t\times n}$,
        $\vecb\sample\ring_q^m$, and $\vece\sample\DD_{\ring,\bar{\chi}}^m$, sets $(\sk_k)_{k\in[K]}\coloneqq\xi\cdot\matV\cdot\matR\bmod q$ and
        $\pk\coloneqq(\pk_\FHE,\matA,\vecb)$, and sends $(\pk,(\sk_k)_{k\in\mathcal{C}})$ to $\adv$.
    \item It internally maintains lists $\Lquery,\Lshare,\Lindex$, initially empty. 
    It answers $\adv$'s queries to $\enc\oracle,\Chal\oracle,\Eval\oracle$ by forwarding to its own $\enc\oracle,\Chal\oracle,\Eval\oracle$ from the IND-CPA-under-leakage challenger, and appending a freshly sampled $\rtag$ itself. In more detail, it answers the oracle queries as follows:
    \begin{itemize}
        \item $\enc\oracle(\id,\mu)$: If $\Lquery[\id]=\emptyset$, then 
            \begin{align*}
                \text{query}~ (\ctxt_\FHE,\mu,\mu) \gets \enc\oracle(\id,\mu),
                \qquad 
                \text{sample}~ \rtag\sample\ring_q^h,
                \qquad 
                \text{set}~ \ctxt \coloneqq (\ctxt_\FHE,\rtag) 
            \end{align*}
            and $\Lquery[\id]\coloneqq(\ctxt,\mu,\mu)$. It returns $\Lquery[\id]$ to $\adv$.
        \item $\Chal\oracle(\id,\mu_0,\mu_1)$: If $\Lquery[\id]=\emptyset$, then
            \begin{align*}
                \text{query}~ (\ctxt_\FHE,\mu_0,\mu_1) \gets \Chal\oracle(\id,\mu_0,\mu_1),
                \qquad
                \text{sample}~ \rtag\sample\ring_q^h,
                \qquad
                \text{set}~ \ctxt \coloneqq (\ctxt_\FHE,\rtag)
            \end{align*}
            and $\Lquery[\id]\coloneqq(\ctxt,\mu_0,\mu_1)$. It returns $\Lquery[\id]$ to $\adv$.
        \item $\Eval\oracle(\id,f,(\id_i)_{i\in I})$: If $\Lquery[\id]=\emptyset$, it parses
            $(\ctxt_i,\mu_{0,i},\mu_{1,i})\gets\Lquery[\id_i]$ for each $i\in I$, then
            \begin{align*}
                \text{query}~ (\ctxt_\FHE',\mu_0',\mu_1') \gets \Eval\oracle(\id,f,(\id_i)_{i\in I}),
                \qquad
                \text{sample}~ \rtag\sample\ring_q^h,
                \qquad
                \text{set}~ \ctxt' \coloneqq (\ctxt_\FHE',\rtag)
            \end{align*}
            and $\Lquery[\id]\coloneqq(\ctxt',\mu_0',\mu_1')$.
            It returns $\Lquery[\id]$ to $\adv$.
    \end{itemize}
    \item To answer $\adv$'s oracle queries to $\ParDec\oracle(\id,k)$: If $\Lquery[\id]=\emptyset$, the reduction returns $\bot$. 
        Otherwise it adds $\{k\}$ to $\Lindex[\id]$ and parses $(\ctxt,\mu_0,\mu_1)\gets\Lquery[\id]$. If $\mu_0\ne\mu_1$ and $|\Lindex[\id]\cup\mathcal{C}|\geq t$, it returns $\bot$. Otherwise, if $\Lshare[\id]=\emptyset$, it computes
        $(\pd_k)_{k\in[K]}$ by distinguishing two cases:
    \begin{itemize}
        \item If $\mu_0=\mu_1$: It simulates $(\pd_k)_{k\in[K]}$ using the $\leak\oracle$ answer from the challenger. Specifically, 
        \begin{align*}
            \text{set}~\leco\coloneqq \begin{bmatrix}\matV_\hdv \cdot \vecu_0\\ \matzero\end{bmatrix},
            \qquad 
            \text{query}~ \leak \gets \Leak\oracle(\id,\leco),
        \end{align*}
        where $\vecu_0$ was defined in \cref{eq:definition-of-u0-U1-U2}. 
        Then, it parses $(\ctxt_\FHE,\rtag)\gets\ctxt$, computes
        $(\hat{\vecc}_0,\hat{c}_1)\gets\PreDecrypt(\ctxt_\FHE)$, samples
        $\vecw_\vecx\sample\ring_q^{t-|\cdv|-1}$, sets
        $\mu\coloneqq \mu_0=\mu_1$ and
        $\vecx\coloneqq\matG^{-1}\begin{pmatrix} \hat{\vecc}_0\\ \rtag \end{pmatrix}$, 
        and it sets
        \begin{align*}
            (\pd_k)_{k\in[K]} 
            \coloneqq& \begin{bmatrix} \mat{V}_\hdv \vecu_0 \\ \matzero \end{bmatrix}
            \begin{bmatrix}(\vecb - \vece)^\transpose \vecx - \left(\hat{c}_1- \round{\frac{q}{p}}\cdot \xi^{-2} \cdot \mu\right) \end{bmatrix}  
            +\leak +{\begin{bNiceArray}{w{c}{2.5em}|w{c}{3em}}  
                \Block{1-2}{\mat{V}_\hdv (\matU_1,\matU_2)}  \\ \hline 
                \quad & \xi\mat{I}_\cdv \end{bNiceArray}}
            {\begin{bmatrix}
                \vecw_\vecx \\ \mat{S}_\cdv \mat{A} \vecx 
            \end{bmatrix}} 
            \bmod q.
        \end{align*}
        \item If $\mu_0\ne\mu_1$: It parses $(\ctxt_\FHE,\rtag)\gets\ctxt$, computes
        $(\hat{\vecc}_0,\hat{c}_1)\gets\PreDecrypt(\ctxt_\FHE)$ and
        $\vecx\coloneqq\matG^{-1}\begin{pmatrix} \hat{\vecc}_0\\ \rtag \end{pmatrix}$, and samples
        $\vece_\vecx\sample\DD_{\ring,\chi}^{|\cdv|}$. It then samples the honest coordinates uniformly
        at random and computes the corrupt coordinates as real partial decryptions, that is
        \begin{align*}
            (\pd_k)_{k\in\hdv}\sample\ring_q^{|\hdv|}
            \qquad\text{and}\qquad
            (\pd_k)_{k\in\cdv}\coloneqq\xi\cdot\mat{S}_\cdv\mat{A}\vecx+\vece_\vecx\bmod q.
        \end{align*}
    \end{itemize}
    In either case, it sets $\Lshare[\id]\coloneqq(\pd_k)_{k\in[K]}$. It returns $\pd_k\gets\Lshare[\id]$ to $\adv$.
    \item Eventually, it receives $b'$ from $\adv$ and returns $b'$.
\end{itemize}

Below we analyse the reduction's simulation. Recall that $\cdv^*$ satisfies $(\set{0} \cup \cdv) \subseteq \cdv^* \subseteq [0,K]$. Notice that 
\begin{align*}
    \norm{\leco}
    = \norm{\begin{bmatrix}\matV_\hdv \cdot \vecu_0\\
    \matzero\end{bmatrix}} 
    = \norm{\matV_\hdv \cdot \xi \cdot \matV_{\cdv^*}^{-1} \cdot \vecv_0}
    \leq \theta(t),
    \end{align*}
    where the second equality is by definition of $\vecu_0 = \matT_\cdv \cdot \vecv_0 = \xi \cdot \matV_{\cdv^*}^{-1} \cdot \vecv_0$,  $\theta(t)$ is one of the quality measurements of the subtractive set $\Xi$ (\cref{def:subtractive-sets}, \cref{eq:subtractive-set-qualities}), and the inequality holds by taking $T^* = \cdv^* \setminus \{0\}$ in \cref{eq:subtractive-set-qualities}. Here $\matV_\hdv$ carries more rows than $\matV_{\Xi \setminus T^*}$ when $|\cdv| < t-1$, but the additional rows, indexed by the filler honest parties $(\cdv^* \setminus \{0\}) \setminus \cdv$, vanish on $\vecu_0$ since $\matV_{\cdv^*} \cdot \vecu_0 = \xi \cdot \vecv_0$, so the two give the same norm. Therefore, the queries made by $\bdv_3$ to the oracle $\Leak\oracle$ are valid.

    Below, we show that $\bdv_3$ perfectly simulates the view of either $\hyb_{0,3}$ or $\hyb_{1,3}$ to $\adv$, 
    if it is interacting either with the IND-CPA-with-$(\theta(t),K,\chi)$-leakage experiment $0$ or $1$.

    Consider the $\mu_0=\mu_1$ case.
    By definition of the IND-CPA-with-$(\theta(t),K,\chi)$-leakage security experiment (\cref{fig:FHE-leaky-security}), one has
    \begin{align*}
        \leak = \leco\cdot\hat{e}+\vecf 
        = \begin{bmatrix}\matV_\hdv \cdot \vecu_0\\
        \matzero\end{bmatrix} \cdot \hat{e}+\vecf.
    \end{align*}
    where
    \[
    \hat{e} \coloneqq \hat{c}_1 - \hat{\vecs}^\tr\cdot\hat{\vecc}_0 - \xi^{-2} \round{\frac{q}{p}}\mu \bmod q
    \]
    as also defined in \cref{eq:hate}.
    It follows that
    \begin{align*}
        (\pd_k)_{k\in[K]}
        =~& \begin{bmatrix} \mat{V}_\hdv \vecu_0 \\ \matzero \end{bmatrix}
        \begin{bmatrix}(\vecb - \vece)^\transpose \vecx - \left(\hat{c}_1- \round{\frac{q}{p}}\cdot \xi^{-2} \cdot \mu\right) \end{bmatrix}+\leak +{\begin{bNiceArray}{w{c}{2.5em}|w{c}{3em}}  
            \Block{1-2}{\mat{V}_\hdv (\matU_1,\matU_2)}  \\ \hline 
            \quad & \xi\mat{I}_\cdv \end{bNiceArray}}
        {\begin{bmatrix}
            \vecw_\vecx \\ \mat{S}_\cdv \mat{A} \vecx 
        \end{bmatrix}} 
        \bmod q\\
        =~& \begin{bmatrix} \mat{V}_\hdv \vecu_0 \\ \matzero \end{bmatrix}
        \begin{bmatrix}(\vecb - \vece)^\transpose \vecx - \left(\hat{c}_1- \round{\frac{q}{p}}\cdot \xi^{-2} \cdot \mu\right) \end{bmatrix} \\ 
        &+ \left(\begin{bmatrix}\matV_\hdv \vecu_0\\
        \matzero\end{bmatrix} \cdot \hat{e}+\vecf\right) +{\begin{bNiceArray}{w{c}{2.5em}|w{c}{3em}}  
            \Block{1-2}{\mat{V}_\hdv (\matU_1,\matU_2)}  \\ \hline 
            \quad & \xi\mat{I}_\cdv \end{bNiceArray}}
        {\begin{bmatrix}
            \vecw_\vecx \\ \mat{S}_\cdv \mat{A} \vecx 
        \end{bmatrix}} 
        \bmod q\\
        =~& \begin{bmatrix} \mat{V}_\hdv \vecu_0 \\ \matzero \end{bmatrix}
        \begin{bmatrix}(\vecb - \vece)^\transpose \vecx - \left(\hat{c}_1- \round{\frac{q}{p}}\cdot \xi^{-2} \cdot \mu-\hat{e}\right) \end{bmatrix} +{\begin{bNiceArray}{w{c}{2.5em}|w{c}{3em}}  
            \Block{1-2}{\mat{V}_\hdv (\matU_1,\matU_2)}  \\ \hline 
            \quad & \xi\mat{I}_\cdv \end{bNiceArray}}
        {\begin{bmatrix}
            \vecw_\vecx \\ \mat{S}_\cdv \mat{A} \vecx 
        \end{bmatrix}} +\vecf
        \bmod q\\
        =~& \begin{bmatrix} \mat{V}_\hdv \vecu_0 \\ \matzero \end{bmatrix}
        \begin{bmatrix}(\vecb - \vece)^\transpose \vecx - \left(\hat{c}_1- \left(\round{\frac{q}{p}}\cdot \xi^{-2} \cdot \mu+\hat{e}\right)\right) \end{bmatrix} +{\begin{bNiceArray}{w{c}{2.5em}|w{c}{3em}}  
            \Block{1-2}{\mat{V}_\hdv (\matU_1,\matU_2)}  \\ \hline 
            \quad & \xi\mat{I}_\cdv \end{bNiceArray}}
        {\begin{bmatrix}
            \vecw_\vecx \\ \mat{S}_\cdv \mat{A} \vecx 
        \end{bmatrix}} +\vece_\vecx
        \bmod q
    \end{align*}
    where in the last line we have used that $\vecf$ and $\vece_\vecx$ are identically distributed.
    This matches the distribution of $(\pd_k)_{k \in [K]}$ defined in \cref{eq:simulation-from-mu}.
    
    For the $\mu_0\ne\mu_1$ case, recall that the distribution of $(\pd_k)_{k \in [K]}$ is identical in $\hyb_{b,3}$ and $\hyb_{b,2}$, where in the latter $(\pd_k)_{k \in \hdv}$ are uniformly random and $(\pd_k)_{k \in \cdv}$ are real partial decryptions.
    This is exactly the distribution of $(\pd_k)_{k \in [K]}$ sampled by the reduction.

    It follows that the $\ParDec\oracle$ is correctly simulated in both cases.

    If $\bdv_3$ is interacting with the IND-CPA-with-$(\theta(t),K,\chi)$-leakage experiment $0$, then queries to $\Chal\oracle(\id,\mu_0,\mu_1)$ return $\ctxt_\FHE\gets\FHE.\enc(\pk_\FHE,\mu_0)$ and therefore $\ctxt\gets\enc(\pk,\mu_0)$. Therefore, the view of $\adv$ is identical to $\hyb_{0,3}$.

    If $\bdv_3$ is interacting with the IND-CPA-with-$(\theta(t),K,\chi)$-leakage experiment $1$, then queries to $\Chal\oracle(\id,\mu_0,\mu_1)$ return $\ctxt_\FHE\gets\FHE.\enc(\pk_\FHE,\mu_1)$ and therefore $\ctxt\gets\enc(\pk,\mu_1)$. Therefore, the view of $\adv$ is identical to $\hyb_{1,3}$.

    Therefore,
        $\Adv_3(\adv)\leq \Adv_{\FHE,\bdv_3,(\theta,K,\chi)}^{\INDCPAwL}(\secpar)$, as claimed.
    \qed        
\end{proof}


\section{A Utility Lemma and Intermediate Assumption}\label{app:utility-and-adpLLWE}

In \cref{app:leaky-lwe-def}, we state the Leaky LWE assumption and provide lemmas on reducing it from standard LWE. The Leaky LWE assumption serves as an intermediate interface which simplifies our proof steps. 
Using this, in \cref{sec:th-LWE-variant}, we state and prove a utility lemma (\cref{lem:th-LWE-variant}), which concerns the computational indistinguishability of a pair of real and ideal experiments. This utility lemma is invoked in the proof of \cref{thm:tfhe-security}, specifically in \cref{lem:thLWE-hybrid-step}, in \cref{sec:missing-proofs-for-security}. Additional results assisting this utility lemma are in \cref{app:core2,app:core1-alt}.

\subsection{Intermediate lattice assumption}\label{app:leaky-lwe-def}

The Leaky LWE assumption \cite{CiC:LaiSwaWoo25} supports leakages on the LWE errors, i.e.~the adversary declares the set of leakage functions and sees the noisy evaluated LWE errors, conditioned on that the leakage functions are norm-bounded. 
We state both the semi- and fully-adaptive versions below. 

\begin{definition}[$\elLWE{}$ {\cite{CiC:LaiSwaWoo25}}, $\Adp\elLWE{}$]\label{def:leakyLWE}
    Let 
    \(\params = (\ring,n,m,Q,q,\bar{\chi},\chi,\cL,\ddvA)\) be parametrised by $\secpar$,
    where $\ddvA \sim \ring_q^{n \times m}$ is a distribution.
    The semi-adaptive Leaky LWE assumption $\elLWE{\params}$ (resp.~adaptive Leaky LWE assumption $\Adp\LLWE{\params}$) states that 
    for any PPT adversary $\adv$, 
    for the experiments $\mathsf{Exp}\mhyphen\mathsf{LLWE}^b$ given in~\cref{fig:leaky-lwe}, it holds that 
    \begin{align*}
        \Adv_{\params,\adv}^{\scriptsize\dbox{$\LLWE{}$},\pcbox{\Adp\LLWE{}}} \coloneqq
        \abs{
            \prob{\mathsf{Exp}\mhyphen\mathsf{LLWE}^0_{\params, \adv}(\secparam) = 1} - \prob{\mathsf{Exp}\mhyphen\mathsf{LLWE}^1_{\params, \adv}(\secparam) = 1}
        } \leq \negl.
    \end{align*}
    The parameter $\ddvA$ is suppressed if it is the uniform distribution over $\ring_q^{n \times m}$.
\end{definition}

\cite{CiC:LaiSwaWoo25}, resp. \cite{EPRINT:HKMNS26}%
\footnote{More specifically, \cite{EPRINT:HKMNS26} showed that the Adaptive Hint-MLWE assumption is implied by Hint-MLWE \cite{C:KLSS23b}, the latter further implied by LWE. Leaky LWE is a variant of Hint-MLWE, and the reduction of \cite{EPRINT:HKMNS26} equally applies to Leaky LWE.}%
, showed that the semi-adaptive, resp. full-adaptive, Leaky LWE is implied by standard LWE under a polynomial-size modulus $q$. We recall both results below, with that of \cite{EPRINT:HKMNS26} adapted to the formulation and norm convention in this work. For completeness, we also provide the adapted proof.

\begin{lemma}[Special case of~{\cite[Theorem 3]{CiC:LaiSwaWoo25}}]\label{lem:lwe-to-ellwe}
    Let $\ring, n, m, Q, q, \epsilon, \alpha, \bar{\chi}, \chi, \allowbreak \chi^*, \cL, \ddv$ be parametrised by $\secpar$, where
    $\ring = \ZZ[\zeta]$,
    $n,m,Q,q \in \NN$,
    $\ddv \sim \ring_q^{n \times m}$,
    $0 < \epsilon \leq 1/2$,
    $\chi^* \geq \sqrt{2} \cdot \eta_\epsilon(\ring^m)$,
    $\bar{\chi} >0$, $\chi >0$, 
    and $\cL = \set{ \mat{Z} \in \ring^{m \times Q} : \norm{\mat{Z}^\tr}_{2,2} \leq \alpha}$,
    where $\alpha = \chi \cdot \sqrt{(2 \eta_\epsilon(\ring^m)^2 + (\chi^*)^2)^{-1} - \bar{\chi}^{-2}}$.%

    Denote 
    \[
        \params = (\ring,n,m,q,\chi^*,\ddv), \qquad 
        \params_L = (\ring,n,m,Q,q,\bar{\chi},\chi,\cL,\ddv).
    \]
    $\forall$ PPT $\adv$, $\exists$ PPT $\bdv$ such that
    \begin{align*}
        \Adv_{\params_L,\adv}^{\LLWE{}}(\secpar)
        \leq \Adv_{\params,\bdv}^{\LWE{}}(\secpar) + 6 \cdot \epsilon/ (1-\epsilon).
    \end{align*}
    In particular, if $\epsilon \in \negl$ and the $\LWE{\params}$ assumption holds, then the $\elLWE{\params_L}$ assumption holds.  
\end{lemma}

\begin{figure}[t!]
    \begin{pchstack}[center,boxed]
        \adjustbox{max width=\dimexpr\columnwidth-2\fboxsep-2\fboxrule\relax}%
        {
        \procedure{$\mathsf{Exp}\mhyphen\mathsf{LLWE}^b_{\params, \adv}(\secparam)$}
        {
            \mat{A} \gets \ddvA \\
            \vecr \sample \ring_q^{n} ;~ \barvec{e} \sample \DD_{\ring,\bar{\chi}}^{m} \\ 
            \pcif b=0 \pcthen \vecb^\tr \coloneqq \vecr^\tr \mat{A} + \barvec{e}^\tr \bmod q \\
            \pcif b=1 \pcthen \vecb \sample \ring_q^{m} \\
            \begin{subprocedure}
                \dbox{
                \pseudocode{
                    \mat{Z} \gets \adv(\mat{A}) \\ 
                    \pcassert \mat{Z} \in \cL \subseteq \ring^{m \times Q} \\
                    \vece \sample \DD_{\ring,\chi}^{Q} ;~
                    \vecl^\tr \coloneqq \barvec{e}^\tr \cdot \mat{Z} + \vec{e}^\tr \\
                    \pcreturn \adv(\vecb, \vecl) 
                }
                } 
            \end{subprocedure} \\
            \begin{subprocedure}
                \pcbox{
                \pseudocode{
                    i \coloneqq 0 \\
                    \pcreturn \adv^{\Leak\oracle}(\mat{A},\vec{b})
                }
                }
            \end{subprocedure}
        }
        \pchspace
        \procedure{$\Leak\oracle(\vecz)$}
        {
            \pcassert i < Q \\
            i \coloneqq i+1 \\
            \pcassert \vecz \in \cL \subseteq \ring^{m} \\
            e \sample \DD_{\ring,\chi} ;~
            \ell \coloneqq \barvec{e}^\tr \cdot \vecz + e \\
            \pcreturn \ell
        }
        }
    \end{pchstack}
    \caption[{}]{
        Experiments for Leaky LWE (\cref{def:leakyLWE}), \dbox{dashed box} for semi-adaptive, \pcbox{\text{solid box}} for adaptive.
    }\label{fig:leaky-lwe}
\end{figure}

The following is a special case of the Gaussian convolution lemma in \cite{PKC:GMPW20}, needed for \cref{lem:LLWE-implies-AdpLLWE}.

\begin{lemma}[Special case of {\cite[Theorem 3.1 (eprint)]{PKC:GMPW20}}]\label{lem:gaussian-convolution}
    Let $\ring = \ZZ[\zeta]$ have ring degree $\varphi$ and coefficient embedding $\psi$, let $m \in \NN$, $0 < \epsilon \leq 1/2$, and $s, \kappa > 0$ with $s > \eta_\epsilon(\ring^m)$.
    Let $\vecz \in \ring^m$ be arbitrary satisfying
    \begin{align}\label{eq:convolution-bound}
        \norm{\Psi(\vecz^\tr)}_{2,2} \leq \kappa \cdot \sqrt{1 - \eta_\epsilon(\ring^m)^2/s^2}
    \end{align}
    and let $\Sigma_\vecz \coloneqq s^2 (\kappa^2 \matI_\varphi - \Psi(\vecz^\tr)\Psi(\vecz^\tr)^\tr)$.
    Then the statistical distance between
    \begin{align*}
        \set{\vecz^\tr \vecy + f ~\middle|~ \vecy \sample \DD_{\ring,s}^m,~f \sample \DD_{\ring,\sqrt{\Sigma_\vecz}}}
        &&\text{and}&&
        \DD_{\ring,\kappa s}
    \end{align*}
    is at most $\epsilon/(1-\epsilon) \leq 2\epsilon$.
\end{lemma}

\begin{proof}
    In this proof we use the shorthands 
    \[
        \matM \coloneqq \Psi(\vecz^\tr), \quad
        \mu \coloneqq \norm{\matM}_{2,2}, \quad
        \Sigma \coloneqq \Sigma_\vecz, \quad 
        \eta \coloneqq \eta_\epsilon(\ring^m) = \eta_\epsilon(\ZZ^{\varphi m}). 
    \]
    The bijection $\psi$ preserves statistical distances and maps the distributions
    \[
        \DD_{\ring,s}^m,~ \DD_{\ring,\sqrt{\Sigma}},~ \DD_{\ring,\kappa s},~ \vecz^\tr \vecy
        \quad\text{to}\quad 
        \DD_{\ZZ^{\varphi m},s},~ \DD_{\ZZ^{\varphi},\sqrt{\Sigma}},~ \DD_{\ZZ^{\varphi},\kappa s},~ \matM \cdot \psi(\vecy)
    \]
    respectively, so it suffices to bound the distance between $\ddv \coloneqq \{\matM \vecy + \vecf \mid \vecy \sample \DD_{\ZZ^{\varphi m},s},~\vecf \sample \DD_{\ZZ^{\varphi},\sqrt{\Sigma}}\}$ and $\DD_{\ZZ^{\varphi},\kappa s}$.

    We apply \cite[Theorem 3.1 (eprint)]{PKC:GMPW20}, for which we express the distribution $\ddv$ as a linear image of one discrete Gaussian vector.
    Let $\Lambda \coloneqq \ZZ^{\varphi m+\varphi}$, $\matS \coloneqq \diag{s \cdot \matI_{\varphi m}, \sqrt{\Sigma}}$ and $\matT \coloneqq (\matM \mid \matI_{\varphi}) \in \ZZ^{\varphi \times (\varphi m+\varphi)}$.
    Since $\matS\matS^\tr = \diag{s^2\matI_{\varphi m}, \Sigma}$ is block diagonal, a sample $\vecx \sample \DD_{\Lambda,\matS}$ is $\vecx = \begin{bmatrix} \vecy \\ \vecf \end{bmatrix}$ with independent $\vecy \sample \DD_{\ZZ^{\varphi m},s}$ and $\vecf \sample \DD_{\ZZ^{\varphi},\sqrt{\Sigma}}$, and $\matT\vecx = \matM\vecy + \vecf$.
    Hence $\matT \cdot \DD_{\Lambda,\matS}$ is exactly the distribution of $\matM\vecy + \vecf$ in $\ddv$, and it remains to check the requirements of \cite[Theorem 3.1 (eprint)]{PKC:GMPW20}:
    \begin{itemize}
        \item $\matS$ has full column rank and $\Lambda \subseteq \Span(\matS) = \RR^{\varphi m + \varphi}$.
        \item $\ker(\matT) = \set{(\vecu,-\matM \vecu) : \vecu \in \RR^{\varphi m}}$ is spanned by $\Lambda \cap \ker(\matT) = \set{(\vecu,-\matM\vecu) : \vecu \in \ZZ^{\varphi m}}$, using $\matM \in \ZZ^{\varphi \times \varphi m}$, i.e.~$\ker(\matT)$ is a $\Lambda$-subspace.
        \item $\eta_\epsilon(\Lambda \cap \ker(\matT)) \leq \matS$.
        Recall that for matrices $\matS' \leq \matT'$ means $\matS'\matS'^\tr \preceq \matT'\matT'^\tr$ \cite{PKC:GMPW20}.
        Since $\Lambda \cap \ker(\matT) = \matC\ZZ^{\varphi m}$ for $\matC \coloneqq \begin{bmatrix} \matI_{\varphi m} \\ -\matM \end{bmatrix}$, by \cite[Definition 2.4 (eprint)]{PKC:GMPW20} the condition is same as 
        \[
            \eta_\epsilon(\matB\ZZ^{\varphi m}) \leq 1
            \quad\text{where}\quad
            \matB \coloneqq \matS^{-1}\matC = \begin{bmatrix} s^{-1}\matI_{\varphi m} \\ -\sqrt{\Sigma}^{-1}\matM \end{bmatrix}.
        \]
        As $\eta_\epsilon(\ZZ^{\varphi m}) = \eta$, \cite[Lemma 2.5 (eprint)]{PKC:GMPW20} gives $\eta_\epsilon(\matB\ZZ^{\varphi m}) \leq \eta\matB$, so it suffices to show $\eta\matB \leq \matI_{\varphi m+\varphi}$, i.e.\ $\eta^2\matB^\tr\matB \preceq \matI_{\varphi m}$.
        Write $\matX \coloneqq \matM^\tr\matM$, whose eigenvalues lie in $[0,\mu^2]$, with $\mu < \kappa$ by \eqref{eq:convolution-bound}.
        Using $\Sigma^{-1} = s^{-2}(\kappa^2\matI_\varphi - \matM\matM^\tr)^{-1}$ and $\matM^\tr(\kappa^2\matI_\varphi - \matM\matM^\tr)^{-1} = (\kappa^2\matI_{\varphi m} - \matX)^{-1}\matM^\tr$ (the ``push-through identity''),
        \begin{align*}
            \matB^\tr\matB
            &= s^{-2}\matI_{\varphi m} + \matM^\tr\Sigma^{-1}\matM
             = s^{-2}\bigl(\matI_{\varphi m} + (\kappa^2\matI_{\varphi m} - \matX)^{-1}\matX\bigr) \\
            &= \frac{\kappa^2}{s^2}(\kappa^2\matI_{\varphi m} - \matX)^{-1}
             \preceq \frac{\kappa^2}{s^2(\kappa^2-\mu^2)}\matI_{\varphi m}
             \preceq \eta^{-2}\matI_{\varphi m},
        \end{align*}
        where the last step is \eqref{eq:convolution-bound} squared and rearranged.
    \end{itemize}
    Therefore $\matT \cdot \DD_{\Lambda,\matS}$ and $\DD_{\matT\Lambda,\matT\matS}$ are within relative error $2\epsilon/(1-\epsilon)$, hence within statistical distance $\epsilon/(1-\epsilon) \leq 2\epsilon$ using $\epsilon \leq 1/2$.
    Finally $\matT \Lambda = \ZZ^{\varphi}$ and $(\matT\matS)(\matT\matS)^\tr = s^2 \matM\matM^\tr + \Sigma = \kappa^2 s^2 \matI_{\varphi}$, so $\DD_{\matT\Lambda,\matT\matS} = \DD_{\ZZ^{\varphi},\kappa s}$.
    \qed
\end{proof}

The next lemma, adapted from \cite{EPRINT:HKMNS26}, shows that fully-adaptive Leaky LWE is implied by the semi-adaptive counterpart. 

\begin{lemma}[Adapted from proof of {\cite[Theorem 1]{EPRINT:HKMNS26}}]\label{lem:LLWE-implies-AdpLLWE}
    Let $\ring, n, m, Q, q, \epsilon, \alpha, \kappa, \bar{\chi}, \chi, \ddv$ be parametrised by $\secpar$, where
    $\ring = \ZZ[\zeta]$ has ring degree $\varphi$ and field expansion factor $\expansion_2$ w.r.t. the coefficient $\ell_2$-norm,
    and $n,m,Q,q \in \NN$,
    $\ddv \sim \ring_q^{n \times m}$,
    $0 < \epsilon \leq 1/2$,
    and $\alpha \geq 1$, $\kappa, \bar{\chi}, \chi >0$ satisfy
    \begin{align}\label{eq:rerand-condition}
        \expansion_2 \cdot \alpha \leq \kappa \cdot \sqrt{1 - \kappa^2 \cdot \eta_\epsilon(\ring^m)^2/\chi^2}.
    \end{align}
    Define the parameters
    \begin{align*}
        \params_L = (\ring,n,m,mQ,q,\bar{\chi},\chi/\kappa,\cL_L,\ddv) 
        &&\text{and}&&
        \params_A = (\ring,n,m,Q,q,\bar{\chi},\chi,\cL_A,\ddv)
    \end{align*}
    where
    \(
    \cL_L = \{ \matZ \in \ring^{m \times mQ} : \norm{\matZ^\tr}_{2,2} \leq \sqrt{Q} \}\) and \(
    \cL_A = \{ \vecz \in \ring^{m} : \norm{\vecz}_2 \leq \alpha \}
    \).
    $\forall$ PPT $\adv$, $\exists$ PPT $\bdv$ such that
    \begin{align*}
        \Adv_{\params_A,\adv}^{\Adp\LLWE{}}(\secpar)
        \leq \Adv_{\params_L,\bdv}^{\LLWE{}}(\secpar) + 4 \cdot Q \cdot \epsilon.
    \end{align*}
    In particular, if $Q \in \poly$, $\epsilon \in \negl$,
    and the $\LLWE{\ring,n,m,mQ,q,\bar{\chi},\chi/\kappa,\cL_L,\ddv}$ assumption holds for some $\kappa$ satisfying \cref{eq:rerand-condition},
    then the $\Adp\LLWE{\ring,n,m,Q,q,\bar{\chi},\chi,\cL_A,\ddv}$ assumption holds.
\end{lemma}

\begin{proof}
    Consider the following hybrid experiments, where $\Sigma_\vecz \coloneqq (\chi/\kappa)^2 (\kappa^2 \matI_\varphi - \Psi(\vecz^\tr) \Psi(\vecz^\tr)^\tr)$:
    \begin{description}
        \item $\hyb_{0}$: The fully-adaptive Leaky LWE experiment (\cref{fig:leaky-lwe}, solid box) with challenge bit $b=0$.
        \item $\hyb_{1}$: Almost identical to $\hyb_{0}$, except that the challenger additionally samples $\vece_1,\ldots,\vece_Q \sample \DD_{\ring,\chi/\kappa}^m$, and it answers the $i$-th oracle query $\Leak\oracle(\vecz)$ with $(\barvec{e} + \vece_i)^\tr \cdot \vecz + f$, where $f \sample \DD_{\ring,\sqrt{\Sigma_\vecz}}$ is sampled freshly and dependent on the query $\vecz$. 
        \item $\hyb_{2}$: Almost identical to $\hyb_{1}$, except that the challenger samples the vector $\vecb \sample \ring_q^m$ uniformly randomly.
        \item $\hyb_{3}$: Almost identical to $\hyb_{2}$, except that challenger undoes the change in $\hyb_1$, i.e.~
        each oracle query $\Leak\oracle(\vecz)$ is answered with $\barvec{e}^\tr \cdot \vecz + e$ where $e \sample \DD_{\ring,\chi}$ is a fresh error, as in $\hyb_0$.
    \end{description}

    We have the following:
    \begin{itemize}
        \item By \cref{lem:LLWE-impliesAdpLLWE-convolution}, the advantage in distinguishing $\Hyb_0$ and $\Hyb_1$ is at most $2 Q \epsilon$.
        \item By \cref{lem:LLWE-impliesAdpLLWE-reduction}, the advantage in distinguishing $\hyb_{1}$ and $\hyb_{2}$
        equals $\Adv_{\params_L,\bdv}^{\LLWE{}}(\secpar)$ for some PPT $\bdv$.
        \item By the same argument as for $\Hyb_0$ and $\Hyb_1$, the advantage in distinguishing $\Hyb_2$ and $\Hyb_3$ is again at most $2 Q \epsilon$.
    \end{itemize}
    The proof is completed by noting that $\hyb_3$ is identical to the fully-adaptive Leaky LWE experiment with challenge bit $b=1$, and the claimed bound follows from the triangle inequality.

    \begin{lemma}\label{lem:LLWE-impliesAdpLLWE-convolution}
        For any (computationally unbounded) $\adv$, the advantage in distinguishing $\Hyb_0$ and $\Hyb_1$ is at most $2 Q \epsilon$.
    \end{lemma}

    \begin{proof}
    For $i \in \set{0,1,\ldots,Q}$, let $\hyb_{0,i}$ be identical to $\hyb_0$ except that the first $i$ oracle queries are answered as in $\hyb_1$, so that $\hyb_{0,0} = \hyb_0$ and $\hyb_{0,Q} = \hyb_1$.
    Fix any $i \in [Q]$. 
    The answer to the $i$-th $\Leak\oracle$ query $\vecz$ in $\hyb_{0,i}$ is
    \[
        \hatvec{e}_i^\tr \cdot \vecz + f = \barvec{e}^\tr \cdot \vecz + \left(\vece_i^\tr \cdot \vecz + f\right).
    \]
    By \cref{lem:gaussian-convolution}, applied with $\vecy = \vece_i$ and using $\kappa s = \chi$, the summand $\vece_i^\tr \cdot \vecz + f$ is within statistical distance $2\epsilon$ to $\DD_{\ring,\chi}$, which is the distribution of the corresponding error in $\hyb_{0,i-1}$. As $\hyb_{0,i-1}$ and $\hyb_{0,i}$ are otherwise identical, their statistical distance is at most $2\epsilon$.
    Summing over $i \in [Q]$ yields the claim.

    In the rest we show that the requirements in \cref{lem:gaussian-convolution} are satisfied, so that it can indeed be applied. Let $\vecz \in \cL_A$ be any admissible query and write $\vecz = (z_1,\ldots,z_m)$. For any $\vecx = (x_1,\ldots,x_m) \in \field^m$ we have
    \[
        \norm{\vecz^\tr \vecx}_2 \leq \expansion_2 \cdot \sum_{j \in [m]} \norm{z_j}_2 \cdot \norm{x_j}_2 \leq \expansion_2 \cdot \norm{\vecz}_2 \cdot \norm{\vecx}_2
    \]
    by the triangle inequality and Cauchy--Schwarz, hence
    \[
        \norm{\vecz^\tr}_{2,2} = \norm{\Psi(\vecz^\tr)}_{2,2} \leq \expansion_2 \cdot \norm{\vecz}_2 \leq \expansion_2 \cdot \alpha.
    \]
    The above together with \cref{eq:rerand-condition} imply \cref{eq:convolution-bound} for $s = \chi/\kappa$, i.e. 
    \[
        \norm{\Psi(\vecz^\tr)}_{2,2} \leq \kappa \cdot \sqrt{1 - \kappa^2 \cdot \eta_\epsilon(\ring^m)^2 / \chi^2}
    \]
    holds.
    Clearly \cref{eq:rerand-condition} also implies the other requirement of $\kappa \cdot \eta_\epsilon(\ring^m)/\chi < 1$.    
    \qed
    \end{proof}

    \begin{lemma}\label{lem:LLWE-impliesAdpLLWE-reduction}
        For any PPT $\adv$ there exists a PPT $\bdv$ such that the advantage of $\adv$ in distinguishing $\Hyb_1$ and $\Hyb_2$ is equal to $\Adv_{\params_L,\bdv}^{\LLWE{}}(\secpar)$.
    \end{lemma}

    \begin{proof}
        Given a PPT $\adv$ distinguishing $\Hyb_1$ and $\Hyb_2$, we construct a PPT reduction $\bdv$ against $\LLWE{\params_L}$, as follows:
        \begin{itemize}
        \item Receive $\matA$ from the $\LLWE{}$ challenger, reply with the leakage matrix 
        $\matZ \coloneqq (\matI_m | \cdots | \matI_m) \in \ring^{m \times mQ}$. 
        \item Receive $(\vecb,\vecl)$ from the $\LLWE{}$ challenger, pass $(\matA,\vecb)$ to $\adv$. Also, write $\vecl^\tr = (\vecl_1^\tr,\ldots,\vecl_Q^\tr)$ with each $\vecl_i \in \ring^m$.
        \item Maintain a counter $i \coloneqq 0$. Answer $\adv$'s oracle queries $\Leak\oracle(\vecz)$ as follows:
        \begin{itemize}
            \item Assert $i < Q$ and $\vecz \in \cL_A$, and set $i \coloneqq i+1$.
            \item Sample $f \sample \DD_{\ring,\sqrt{\Sigma_\vecz}}$ and answer with $\vecl_i^\tr \cdot \vecz + f$.
        \end{itemize}
        \item Eventually receive a bit $b'$ from $\adv$, return $b'$.
    \end{itemize}

    The reduction $\bdv$ is admissible as an $\LLWE{\params_L}$ adversary, since
    for any $\vecy \in \field^{m}$ we have $\matZ^\tr \cdot \vecy = (\vecy^\tr,\ldots,\vecy^\tr)^\tr \in \field^{mQ}$, which implies $\norm{\matZ^\tr \cdot \vecy}_2 = \sqrt{Q} \cdot \norm{\vecy}_2$, so $\norm{\matZ^\tr}_{2,2} = \sqrt{Q}$ and $\matZ \in \cL_L$.

    The simulation is perfect: in the $\LLWE{\params_L}$ experiment with challenge bit $b$, the matrix $\matA$ and the error $\barvec{e} \sample \DD_{\ring,\bar{\chi}}^m$ are distributed as in $\hyb_1$ and $\hyb_2$, the vector $\vecb$ is as in $\hyb_1$ if $b=0$ and as in $\hyb_2$ if $b=1$. In either game, by the choice of $\matZ$, 
    \begin{align*}    
        \vecl^\tr 
        =~& \barvec{e}^\tr \cdot \matZ + \vece^\tr \pccomment{$\vece \sample \DD_{\ring,\chi/\kappa}^{mQ}$} \\
        =~& \barvec{e}^\tr \cdot (\matI_m | \cdots | \matI_m) + (\vece_1^\tr | \ldots | \vece_Q^\tr) \pccomment{$(\vece_1^\tr \,|\, \ldots \,|\, \vece_Q^\tr) \coloneqq \vece^\tr$} \\
        =~& (\barvec{e}^\tr + \vece_1^\tr \,|\, \ldots \,|\, \barvec{e}^\tr + \vece_Q^\tr) \\
        \iff \quad 
        \vecl_i =~& \barvec{e} + \vece_i ~~\forall i \in [Q], \pccomment{$\vece_i \sample \DD_{\ring,\chi/\kappa}^{m}$}
    \end{align*}
    hence each $\Leak\oracle(\vecz)$ query answer is $(\barvec{e} + \vece_i)^\tr \cdot \vecz + f$, as wanted.
    Overall, $\bdv$'s answers are exactly those of $\hyb_1$, resp.~$\hyb_2$.
    \qed
    \end{proof}
    \qed
\end{proof}

By \cref{lem:lwe-to-ellwe,lem:LLWE-implies-AdpLLWE}, standard LWE thus implies fully-adaptive Leaky LWE, summarised below.

\begin{corollary}[Adapted from {\cite[Theorem 1]{EPRINT:HKMNS26}}]\label{cor:lwe-to-adpllwe}
    Let $\ring, n, m, Q, q, \epsilon, \alpha, \kappa, \bar{\chi}, \chi, \chi^*, \ddv$ be parametrised by $\secpar$, where
    $\ring = \ZZ[\zeta]$ has ring degree $\varphi$ and field expansion factor $\expansion_2$ w.r.t. the coefficient $\ell_2$-norm,
    $n,m,Q,q \in \NN$,
    $\ddv \sim \ring_q^{n \times m}$,
    $0 < \epsilon \leq 1/2$,
    $\chi^* \geq \sqrt{2} \cdot \eta_\epsilon(\ring^m)$,
    and $\alpha \geq 1$, $\kappa, \bar{\chi}, \chi > 0$ satisfy both
    \begin{align}
        \expansion_2 \cdot \alpha \leq~& \kappa \cdot \sqrt{1 - \kappa^2 \cdot \eta_\epsilon(\ring^m)^2/\chi^2},\label{eq:adpllwe-constraint1} \\
        \sqrt{Q} \leq~& \frac{\chi}{\kappa} \cdot \sqrt{(2 \eta_\epsilon(\ring^m)^2 + (\chi^*)^2)^{-1} - \bar{\chi}^{-2}}.\label{eq:adpllwe-constraint2}
    \end{align}
    Denote
    \[
        \params = (\ring,n,m,q,\chi^*,\ddv), \qquad
        \params_A = (\ring,n,m,Q,q,\bar{\chi},\chi,\cL_A,\ddv),
    \]
    where $\cL_A = \{ \vecz \in \ring^{m} : \norm{\vecz}_2 \leq \alpha \}$.
    $\forall$ PPT $\adv$, $\exists$ PPT $\bdv$ such that
    \begin{align*}
        \Adv_{\params_A,\adv}^{\Adp\LLWE{}}(\secpar)
        \leq \Adv_{\params,\bdv}^{\LWE{}}(\secpar) + 4 \cdot Q \cdot \epsilon + 6 \cdot \epsilon/(1-\epsilon).
    \end{align*}
    In particular, if $Q \in \poly$, $\epsilon \in \negl$ and the $\LWE{\params}$ assumption holds, then the $\Adp\LLWE{\params_A}$ assumption holds.
\end{corollary}

\begin{proof}
    Follows from putting \cref{lem:lwe-to-ellwe,lem:LLWE-implies-AdpLLWE} together. In particular, we invoke \cref{lem:lwe-to-ellwe} 
    with $\LLWE{}$ parameters $\params' = (\ring,n,m,mQ,q,\bar{\chi},\chi/\kappa,\cL',\ddv)$, 
    where $\cL' = \{ \matZ \in \ring^{m \times mQ} : \norm{\matZ^\tr}_{2,2} \leq \alpha' \}$ with $\alpha' = \frac{\chi}{\kappa} \cdot \sqrt{(2 \eta_\epsilon(\ring^m)^2 + (\chi^*)^2)^{-1} - \bar{\chi}^{-2}}$.
    
    Since $\cL' \supseteq \cL_L \coloneqq \{ \matZ \in \ring^{m \times mQ} : \norm{\matZ^\tr}_{2,2} \leq \sqrt{Q} \}$ by \cref{eq:adpllwe-constraint2}, $\LLWE{}$ with $\params'$ implies
    $\LLWE{}$ with $\params_L = (\ring,n,m,mQ,q,\bar{\chi},\chi/\kappa,\cL_L,\ddv)$, the latter being the same as that used in \cref{lem:LLWE-implies-AdpLLWE}. 

    \cref{eq:adpllwe-constraint1} is identical to \cref{eq:rerand-condition} required by \cref{lem:LLWE-implies-AdpLLWE}.
    \qed
\end{proof}

\begin{remark}[On choosing $\kappa$ for \cref{lem:LLWE-implies-AdpLLWE,cor:lwe-to-adpllwe}]\label{rem:llwe-kappa}
Keeping everything else unchanged, a smaller $\kappa$ in \cref{lem:LLWE-implies-AdpLLWE,cor:lwe-to-adpllwe} is desired, 
since it leads to a larger error parameter $\chi/\kappa$ for $\LLWE{}$ in \cref{lem:LLWE-implies-AdpLLWE}, i.e. a harder $\LLWE{}$ problem, 
and it makes \cref{eq:adpllwe-constraint2} in \cref{cor:lwe-to-adpllwe} more easily satisfied.

Writing $\eta \coloneqq \eta_\epsilon(\ring^m)$, 
a $\kappa$ satisfying \cref{eq:adpllwe-constraint1} exists if and only if 
\begin{align}\label{eq:adpllwe-mod1}
    \chi \geq 2 \cdot \expansion_2 \cdot \alpha \cdot \eta, 
\end{align}
in which case any $\kappa$ such that $\kappa^2 \in [\kappa_-^2, \kappa_+^2]$ is admissible, where
\(
    \kappa_\pm^2 \coloneqq \frac{\chi^2}{2\eta^2}\left(1 \pm \sqrt{1 - (2 \cdot \expansion_2 \cdot \alpha \cdot \eta/\chi)^2}\right)%
\)%
.\footnote{Squaring \cref{eq:adpllwe-constraint1} gives $\kappa^4 \cdot \eta^2/\chi^2 - \kappa^2 + \expansion_2^2 \cdot \alpha^2 \leq 0$, a quadratic in $\kappa^2$ whose roots are the stated $\kappa_\pm^2$. The latter are real if and only if $2 \cdot \expansion_2 \cdot \alpha \cdot \eta/\chi \leq 1$.}
For the minimum $\kappa_-$, it holds that 
\[
    \expansion_2 \cdot \alpha \leq \kappa_- \leq \sqrt{2}\cdot \expansion_2 \cdot \alpha,
\]
which follows from 
$\kappa_-^2 \cdot (1 - \kappa_-^2 \cdot \eta^2/\chi^2) = \expansion_2^2 \cdot \alpha^2$ and from $1 - \sqrt{1-t} \leq t$ for $t \in [0,1]$, the latter giving $\kappa_-^2 \leq \frac{\chi^2}{2\eta^2}\cdot \frac{4 \cdot \expansion_2^2 \cdot \alpha^2 \cdot \eta^2}{\chi^2}$.

A sufficient condition for \cref{eq:adpllwe-constraint2} at $\kappa = \kappa_-$ is thus 
\begin{align}\label{eq:adpllwe-mod2}
    \sqrt{2Q} \cdot \expansion_2 \cdot \alpha \leq \chi \cdot \sqrt{(2 \eta^2 + (\chi^*)^2)^{-1} - \bar{\chi}^{-2}}.
\end{align}
In the latter sections where we invoke \cref{cor:lwe-to-adpllwe}, 
we will always directly require \cref{eq:adpllwe-mod2} instead of \cref{eq:adpllwe-constraint1,eq:adpllwe-constraint2}. This is sufficient, since \cref{eq:adpllwe-mod2} and the condition $\chi^* \geq \sqrt{2} \cdot \eta_\epsilon(\ring^m)$ jointly imply \cref{eq:adpllwe-mod1}%
\footnote{Since $\chi^* \geq \sqrt{2} \cdot \eta$, the radicand in \cref{eq:adpllwe-mod2} is at most $(2\eta^2 + 2\eta^2)^{-1} = (4\eta^2)^{-1}$, so the right-hand side of \cref{eq:adpllwe-mod2} is at most $\chi/(2\eta)$. Hence \cref{eq:adpllwe-mod2} gives $\chi \geq 2 \cdot \sqrt{2Q} \cdot \expansion_2 \cdot \alpha \cdot \eta$, which implies \cref{eq:adpllwe-mod1} since $Q \geq 1$.}%
, and \cref{eq:adpllwe-mod1,eq:adpllwe-mod2} in turn imply that \cref{eq:adpllwe-constraint1,eq:adpllwe-constraint2} hold for the choice $\kappa = \kappa_-$ by the above.
\end{remark}

\begin{remark}
    The semi-adaptive Leaky LWE in \cite{CiC:LaiSwaWoo25} more generally supports bounded-norm leakage also on the LWE secret $\vecr$, and admits a reduction from LWE when $\vecr$ is distributed as Gaussian. 
    The adaptive version can be generalised in an analogous way, where the oracle $\Leak\oracle$ returns leakages of the form $(\vecr^\tr,\barvec{e}^\tr)\cdot \vecz + e$. 
    When $\vecr$ is Gaussian vector, \cref{lem:LLWE-implies-AdpLLWE,cor:lwe-to-adpllwe} can be generalised analogously. 
\end{remark}


\subsection{Utility Lemma}\label{sec:th-LWE-variant}

This subsection is devoted to proving that the two experiments
$\mathsf{Real}$ and $\mathsf{Ideal}$
in \cref{fig:exp-th-LWE} are computationally indistinguishable under two $\LWE{}$ assumptions,
one with respect to the tag distribution $\ddvt$ and one with respect to the public matrix distribution $\ddvA$.
This is formally given in \cref{lem:th-LWE-variant}.
The majority of its proof is a generalisation of \cite[Lemma 10, Eprint]{AC:CinLaiWoo25}.


\begin{figure}[t!]
    \begin{pchstack}[center,boxed]
        {
        \begin{pcvstack}
            \procedure{$\mathsf{Real}_{\params, \adv}(\secparam)$}
            {
                \LShare\coloneqq \emptyset;~~ \Lquery \coloneqq \emptyset\\
                \mat{A} \sample \ddvA \\
                (\cdv, \vecz) \gets \adv\tuple{\mat{A}} ;~ \pcassert |\cdv| < t \\
                \mat{R} \sample \ring_q^{t \times n} \\
                \mat{S}_\cdv \coloneqq \mat{V}_\cdv \cdot \mat{R} \bmod q  \in \ring_q^{\cdv \times n} \\
                \vecr \coloneqq \vec{v}_0^\transpose \cdot \mat{R} \bmod q \\
                \vec{e} \sample \DD_{\ring,\bar{\chi}}^m \\
                \vec{b}^\transpose \coloneqq \vecr^\transpose \mat{A} + \vec{e}^\transpose + \vecz^\transpose \bmod q \\
                \\
                b' \gets \adv^{\Gen\oracle,\Share\oracle_0}\tuple{\vec{b}, \xi \cdot \mat{S}_\cdv} \\
                \pcassert |\LShare| \leq Q \pccomment{at most $Q$ shared $\mat{A}\vec{x}$} \\
                \pcreturn b'
            }
            \pcvspace
            \procedure{$\Gen\oracle(\id,\tildevec{c}) \pccomment{$\tildevec{c} \in \ring^{(m-m_t)}$}$}
            {
                \pcif \Lquery[\id] = \emptyset \pcthen \\
                \pcind \pcassert \norm{\tildevec{c}} \leq B_c \\
                \pcind \tildevec{t} \sample \ddvt ;~ 
                \vecx \coloneqq \begin{bmatrix} \tildevec{c} \pclb \tildevec{t} \end{bmatrix} \\
                \Lquery[\id] \coloneqq \vecx \\
                \pcreturn \vecx \gets \Lquery[\id]
            }
            \pcvspace
            \procedure{$\Share\oracle_b(\id, k \in [K])$}
            {
                \pcif \LShare[\id] = \emptyset \pcthen \\
                \pcind \pcassert \Lquery[\id] \neq \emptyset ;~ 
                \pcparse \vecx \gets \Lquery[\id] \\
                \pcind \pcif b = 0 \pcthen \\
                \pcind[2] \vec{e}_\vecx \sample \DD_{\ring,\chisub}^K ;~ 
                    \vecd \coloneqq \xi \cdot \mat{V} \mat{R} \mat{A}\vecx + \vec{e}_\vecx \bmod q \\
                \pcind \pcif b = 1 \pcthen \\
                \pcind[2] (\vecd, \state) \gets \sdv_{\mathsf{Share}}(\vecx, \state) \\
                \pcind \LShare[\id] \coloneqq \vecd \\
                \pcparse \vecd \gets \LShare[\id] \in \ring_q^{K} \\
                \pcreturn d_k
            }
        \end{pcvstack}
        \pchspace
        \begin{pcvstack}
            \procedure{$\mathsf{Ideal}_{\params, \adv}(\secparam)$}
            {
                \LShare\coloneqq \emptyset;~~ \Lquery \coloneqq \emptyset\\
                \mat{A} \sample \ddvA \\
                (\cdv,\vecz) \gets \adv\tuple{\mat{A}} ;~ \pcassert |\cdv| < t \\
                \\
                \mat{S}_\cdv \sample \ring_q^{\cdv \times n} \\
                \\
                \\
                \vec{b} \sample \ring_q^m \\
                \state \gets \sdvinit(\mat{A},\vec{b},\mat{S}_\cdv,\vecz) \pccomment{State of simulator $\sdv$} \\
                b' \gets \adv^{\Gen\oracle,\Share\oracle_1}\tuple{\vec{b}, \xi \cdot \mat{S}_\cdv} \\
                \pcassert |\LShare| \leq Q \pccomment{at most $Q$ shared $\mat{A}\vec{x}$} \\
                \pcreturn b'
            }
            \pcvspace
            \procedure{$\sdvinit(\mat{A},\vec{b},\mat{S}_\sdv,\vecz)$}
            {   
                \vec{e} \sample \DD_{\ring,\bar{\chi}}^{m} ;~
                \hatvec{b} \coloneqq \vec{b}-\vec{e} - \vecz \\
                \pcreturn \state \coloneqq (\mat{A},\vec{b},\mat{S}_\sdv, \hatvec{b})
            }
            \pcvspace
            \procedure{$\sdv_{\mathsf{Share}}(\vecx,\state)$}
            {   
                \vec{e}_\vecx \sample \DD_{\ring,\chi}^{K} \\
                \vecw_\vecx \sample \ring_q^{t-|\cdv|-1} \\
                \vecd \coloneqq  
                {\begin{bNiceArray}{c|c|c}  
                    \Block{1-3}{\mat{V}_\hdv \cdot \mat{T}_\cdv}  \pclb\hline 
                    \quad &\quad & \xi\mat{I}_\cdv \end{bNiceArray}}
                {\begin{bmatrix}
                    \hatvec{b}^\transpose \vecx \pclb \vecw_\vecx \pclb \mat{S}_\cdv \mat{A} \vecx 
                \end{bmatrix}} + \vec{e}_\vecx \bmod q \\
                \pcreturn (\vecd, \state)
            }        
        \end{pcvstack}    
        }
    \end{pchstack}
\caption{$\mathsf{Real}_{\params,\adv}$ and $\mathsf{Ideal}_{\params,\adv}$ 
with $\params = ((\ring,n,m,q,Q,\bar{\chi},\chi,\ddvA,\ddvt,B_c),(t,K,\Xi))$. 
}\label{fig:exp-th-LWE}
\end{figure}

\begin{lemma}\label{lem:th-LWE-variant}
    Let $\params = ((\ring,n,m,q,Q,\bar{\chi},\chi,\ddvA,\ddvt,B_c),(t,K,\Xi))$ be parametrised by $\secpar$,
    where $\ring = \ZZ[\zeta]$ has ring degree $\varphi$ and field expansion factor $\expansion_p$ w.r.t. the coefficient $\ell_p$-norm, 
    $\Xi$ is a size-$K$ $\slack$-subtractive set with quality measurements $(\rho,\gamma,\theta) = (\rho(t),\gamma(t),\theta(t))$ (\cref{def:subtractive-sets}), 
    $\ddvA$ and $\ddvt$ are distributions over $\ring_q^{n \times m}$ and $\ring^{m_t}$ respectively with $m_t < m$, and $B_c >0$.
    Let $B_t >0$ be such that 
    $\Pr(\norm{\tildevec{t}} \leq B_t | \tildevec{t} \sample \ddvt) = 1$.
    
    Consider the experiments $\mathsf{Real}_{\params,\adv}$ and $\mathsf{Ideal}_{\params,\adv}$ in \cref{fig:exp-th-LWE},
    and define the advantage bound
    \begin{align*}
        \Adv_{\params,\adv}^{\mathsf{RealIdeal}}(\secpar) \coloneqq \abs{ \prob{\mathsf{Real}_{\params, \adv}(\secparam) = 1} - \prob{\mathsf{Ideal}_{\params, \adv}(\secparam) = 1}}.
    \end{align*} 
    Suppose $q$ and $\xi$ are coprime, $0 < \epsilon \leq 1/2$, and the noise parameters $\chi_1^*, \chi_2^*$ satisfy both of the following:
    \begin{itemize}
        \item $\chi_1^* \geq \sqrt{2} \cdot \eta_\epsilon(\ring^t)$ is such that there exists $\hat{\chi} >0$ satisfying
            \[
                \gamma \leq \chi \cdot \sqrt{(2 \eta_\epsilon(\ring^t)^2 + (\chi_1^*)^2)^{-1} - \hat{\chi}^{-2}};
            \]
        \item $\chi_2^* \geq \sqrt{2} \cdot \eta_\epsilon(\ring^m)$,
            $\gamma \cdot \sqrt{\varphi m} \cdot \max(B_c,B_t) \geq 1$, and, with $\tensorexpansion{t} \coloneqq 1$ if $\conductor$ is a power of two and $\tensorexpansion{t} \coloneqq \sqrt{t}$ otherwise,
            \begin{align}\label{eq:th-LWE-variant-constraint}
                \sqrt{2 Q K \cdot \varphi m} \cdot \expansion_2 \cdot \tensorexpansion{t} \cdot \gamma \cdot \max(B_c,B_t) \leq \chi \cdot \sqrt{(2 \eta_\epsilon(\ring^m)^2 + (\chi_2^*)^2)^{-1} - \bar{\chi}^{-2}}.
            \end{align}
    \end{itemize}
    Then, for any PPT $\adv$, there exist PPT $\bdv_1,\bdv_2$ such that
    \begin{align*}
        \frac{1}{t} \cdot\Adv_{\params,\adv}^{\mathsf{RealIdeal}}(\secpar)
        \leq~&
        \Adv_{(\ring,m_t,\poly,q,\chi_1^*,\ddvt),\bdv_1}^{\LWE{}}(\secpar)
        +\Adv_{(\ring,n,m,q,\chi_2^*,\ddvA),\bdv_2}^{\LWE{}}(\secpar) \\
        &+ \frac{6 \cdot Q \cdot \epsilon}{t(1-\epsilon)}
        + 4 \cdot Q K \cdot \epsilon
        + \frac{6 \cdot \epsilon}{1-\epsilon}.
    \end{align*}
    In particular, if $QK \in \poly$, $\epsilon \in \negl$, and the $\LWE{\ring,m_t,\poly,q,\chi_1^*,\ddvt}$ and $\LWE{\ring,n,m,q,\chi_2^*,\ddvA}$ assumptions hold, then
    the experiments $\mathsf{Real}_{\params,\adv}$ and $\mathsf{Ideal}_{\params,\adv}$ are computationally indistinguishable.
\end{lemma}

\begin{proof}
    Let the experiment $\mathsf{Real}_{\params, \adv}$ be $\Hyb_0$, and we define experiments $\hyb_1$, $\hyb_2$, $\Hyb_3$ and $\Hyb_4$ in \cref{fig:core_dist}.
    Below we argue that $\Hyb_i$ is indistinguishable from $\Hyb_{i+1}$ for each $i \in \set{0,1,2,3}$.
    The proof is then completed by noting that $\Hyb_4$ is identical to the experiment $\mathsf{Ideal}_{\params, \adv}$.


\begin{figure}[t]
    \begin{pchstack}[center,boxed]        
    \adjustbox{max width=\dimexpr\textwidth-2\fboxsep-2\fboxrule\relax}%
    {
        \begin{pcvstack}
            \procedure{$\Hyb_1(\secparam)$}
            {
                \LShare[\cdot] \coloneqq \emptyset \\
                \mat{A} \sample \ddvA \\
                (\cdv, \vecz) \gets \adv\tuple{\mat{A}} ;~ \pcassert |\cdv| < t \\
                \diff{\begin{bmatrix} \vec{r}^\transpose \pclb \mat{W} \pclb \mat{S}_\cdv \end{bmatrix} \sample \ring_q^{t \times n}} \\[5pt]
                \vec{e} \sample \DD_{\ring,\bar{\chi}}^m \\
                \vec{b}^\transpose \coloneqq \diff{\vecr^\transpose} \mat{A} + \vec{e}^\transpose + \vecz^\transpose \bmod q \\
                b' \gets \adv^{\Share\oracle_{\Hyb1}}\tuple{\vec{b}, \xi \cdot \mat{S}_\cdv} \\
                {\pcassert |\LShare| \leq Q} \\
                \pcreturn b'
            }
            \pcvspace
            \procedure{$\Share\oracle_{\Hyb1}(\id, k \in [K])$}
            {
                \pcif \LShare[\id] = \emptyset \pcthen \\
                \pcind \pcassert \Lquery[\id] \neq \emptyset \\ 
                \pcind \pcparse \vecx \gets \Lquery[\id] \\
                \pcind \vec{e}_\vecx \sample \DD_{\ring,\chisub}^K \\ 
                \\[4pt]
                \pcind \matM \coloneqq 
                \diff{\begin{bNiceArray}{c|c|c}  
                    \Block{1-3}{\mat{V}_\hdv \cdot \mat{T}_\cdv}  \pclb\hline 
                    \quad &\quad & \xi \mat{I}_\cdv \end{bNiceArray}} \\
                \pcind \vecd \coloneqq 
                    \matM \cdot 
                    \diff{\begin{bmatrix} \vec{r}^\transpose \pclb \mat{W} \pclb \mat{S}_\cdv \end{bmatrix}}
                    \mat{A}\vec{x} + \vec{e}_\vecx \bmod q \\
                \pcind \LShare[\id] \coloneqq \vecd \in \ring_q^{K} \\
                \pcparse \vecd \gets \LShare[\id]  \\
                \pcreturn d_k 
            }
        \end{pcvstack} 
        \pchspace
        \begin{pcvstack}
            \procedure{$\Hyb_2(\secparam)$}
            {
                \LShare[\cdot] \coloneqq \emptyset \\
                \mat{A} \sample \ddvA \\
                (\cdv, \vecz) \gets \adv\tuple{\mat{A}} ;~ \pcassert |\cdv| < t \\
                \mat{S}_\cdv \sample \ring_q^{\cdv \times n} \\
                \diff{\begin{bmatrix} \vec{r}_m^\transpose \pclb \mat{W}_m \end{bmatrix} \sample \ring_q^{(t - |\cdv|) \times m}} \\
                \vec{e} \sample \DD_{\ring,\bar{\chi}}^m \\
                \vec{b}^\transpose \coloneqq \diff{\vec{r}_m^\transpose} + \vec{e}^\transpose + \vecz^\transpose \bmod q \\
                b' \gets \adv^{\Share\oracle_{\Hyb2}}\tuple{\vec{b}, \xi \cdot \mat{S}_\cdv} \\
                {\pcassert |\LShare| \leq Q} \\
                \pcreturn b'
            }
            \pcvspace
            \procedure{$\Share\oracle_{\Hyb2}(\id, k \in [K])$}
            {
                \pcif \LShare[\id] = \emptyset \pcthen \\
                \pcind \pcassert \Lquery[\id] \neq \emptyset \\ 
                \pcind \pcparse \vecx \gets \Lquery[\id] \\
                \pcind \vec{e}_\vecx \sample \DD_{\ring,\chisub}^K \\
                \\[4pt]
                \pcind \matM \coloneqq 
                {\begin{bNiceArray}{c|c|c}  
                    \Block{1-3}{\mat{V}_\hdv \cdot \mat{T}_\cdv}  \pclb\hline 
                    \quad &\quad & \xi \mat{I}_\cdv \end{bNiceArray}} \\
                \pcind \vecd \coloneqq 
                    \matM \cdot 
                    {\begin{bmatrix} \diff{\vec{r}_m^\transpose} \pclb \diff{\mat{W}_m} \pclb \mat{S}_\cdv \mat{A} \end{bmatrix}}
                    \vec{x} + \vec{e}_\vecx \bmod q \\
                \pcind \LShare[\id] \coloneqq \vecd \in \ring_q^{K} \\
                \pcparse \vecd \gets \LShare[\id]  \\
                \pcreturn d_k
            }
        \end{pcvstack}
        \pchspace
        \begin{pcvstack}
            \procedure{$\Hyb_3(\secparam)$}
            {
                \LShare[\cdot] \coloneqq \emptyset \\
                \mat{A} \sample \ddvA \\
                (\cdv, \vecz) \gets \adv\tuple{\mat{A}} ;~ \pcassert |\cdv| < t \\
                \mat{S}_\cdv \sample \ring_q^{\cdv \times n} \\
                \vec{r}_m^\transpose \sample \ring_q^{1 \times m} \\[12pt]
                \vec{e} \sample \DD_{\ring,\bar{\chi}}^m \\
                \vec{b}^\transpose \coloneqq \vec{r}_m^\transpose + \vec{e}^\transpose + \vecz^\transpose \bmod q \\
                b' \gets \adv^{\Share\oracle_{\Hyb3}}\tuple{\vec{b}, \xi \cdot \mat{S}_\cdv} \\
                {\pcassert |\LShare| \leq Q} \\
                \pcreturn b'
            }            
            \pcvspace
            \procedure{$\Share\oracle_{\Hyb3}(\id, k \in [K])$}
            {
                \pcif \LShare[\id] = \emptyset \pcthen \\
                \pcind \pcassert \Lquery[\id] \neq \emptyset \\ 
                \pcind \pcparse \vecx \gets \Lquery[\id] \\
                \pcind \vec{e}_\vecx \sample \DD_{\ring,\chisub}^K \\ 
                \pcind \diff{\vecw_\vecx \sample \ring_q^{t-|\cdv|-1}} \\
                \pcind \matM \coloneqq 
                {\begin{bNiceArray}{c|c|c}  
                    \Block{1-3}{\mat{V}_\hdv \cdot \mat{T}_\cdv}  \pclb\hline 
                    \quad &\quad & \xi \mat{I}_\cdv \end{bNiceArray}} \\
                \pcind \vecd \coloneqq 
                    \matM \cdot 
                    {\begin{bmatrix} \vec{r}_m^\transpose \vecx \pclb \diff{\vecw_\vecx} \pclb \mat{S}_\cdv \mat{A} \vecx \end{bmatrix}}
                    + \vec{e}_\vecx \bmod q \\
                \pcind \LShare[\id] \coloneqq \vecd \in \ring_q^{K} \\
                \pcparse \vecd \gets \LShare[\id]  \\
                \pcreturn d_k
            }
        \end{pcvstack}
        \pchspace
        \begin{pcvstack}
            \procedure{$\Hyb_4(\secparam)$}
            {
                \LShare[\cdot] \coloneqq \emptyset \\
                \mat{A} \sample \ddvA \\
                (\cdv, \vecz) \gets \adv\tuple{\mat{A}} ;~ \pcassert |\cdv| < t \\
                \mat{S}_\cdv \sample \ring_q^{\cdv \times n} \\
                \\
                \\
                \vec{e} \sample \DD_{\ring,\bar{\chi}}^m \\
                {\diff{\vec{b}^\transpose \sample \ring_q^{m}}} \\
                b' \gets \adv^{\Share\oracle_{\Hyb4}}\tuple{\vec{b}, \xi \cdot \mat{S}_\cdv} \\
                {\pcassert |\LShare| \leq Q} \\
                \pcreturn b'
            }        
            \pcvspace
            \procedure{$\Share\oracle_{\Hyb4}(\id, k \in [K])$}
            {
                \pcif \LShare[\id] = \emptyset \pcthen \\
                \pcind \pcassert \Lquery[\id] \neq \emptyset \\ 
                \pcind \pcparse \vecx \gets \Lquery[\id] \\
                \pcind \vec{e}_\vecx \sample \DD_{\ring,\chisub}^K \\ 
                \pcind \vecw_\vecx \sample \ring_q^{t-|\cdv|-1} \\
                \pcind \matM \coloneqq 
                {\begin{bNiceArray}{c|c|c}  
                    \Block{1-3}{\mat{V}_\hdv \cdot \mat{T}_\cdv}  \pclb\hline 
                    \quad &\quad & \xi \mat{I}_\cdv \end{bNiceArray}} \\
                \pcind \vecd \coloneqq 
                    \matM \cdot 
                    {\begin{bmatrix} \diff{(\vec{b} - \vec{e} - \vecz)^\transpose} \vecx ~ \pclb \vecw_\vecx \pclb \mat{S}_\cdv \mat{A} \vecx \end{bmatrix}} 
                    + \vec{e}_\vecx \bmod q \\
                \pcind \LShare[\id] \coloneqq \vecd \in \ring_q^{K} \\
                \pcparse \vecd \gets \LShare[\id]  \\
                \pcreturn d_k
            }
        \end{pcvstack}
    }
    \end{pchstack}
    \caption{Hybrid distributions $\Hyb_1$ to $\Hyb_4$ for proof of \cref{lem:th-LWE-variant}.
    For each, the changes relative to the previous hybrid are marked in red. 
    The oracle $\Gen\oracle$ is unchanged and omitted.
    }\label{fig:core_dist}
\end{figure}

    \paragraph{$\hyb_1$.}
    This is functionally equivalent to $\hyb_0 = \mathsf{Real}_{\params, \adv}$. 
    Notice that $\mat{V} = 
    \begin{bmatrix}  
        \mat{V}_{\hdv} \\
        \mat{V}_\cdv 
    \end{bmatrix}$
    by definition, thus 
    the only differences in $\hyb_1$ are that we sample uniformly random $\begin{bmatrix} \vec{r}^\transpose \\ \mat{W} \\ \mat{S}_\cdv \end{bmatrix}$,
    and we swap 
    \begin{align*}
        \begin{bmatrix} \vec{v}_0^\tr \\ \xi \cdot \mat{V} \end{bmatrix}
        \mat{R}
        = 
        \begin{bmatrix}  
            \vec{v}_0^\transpose \\
            \xi \cdot \mat{V}_{\hdv} \\
            \xi \cdot \mat{V}_\cdv 
        \end{bmatrix}
        \mat{R} 
        \qquad\text{to}\qquad
        \matM \cdot  
        \begin{bmatrix} \vec{r}^\transpose \\ \mat{W} \\ \mat{S}_\cdv \end{bmatrix}
        =
        \begin{bNiceArray}{c|c|c}  
            1 & \quad & \quad \\\hline 
            \Block{1-3}{\mat{V}_\hdv \cdot \mat{T}_\cdv} \\\hline 
            \quad & \quad & \xi \mat{I}_\cdv  
        \end{bNiceArray}
        \begin{bmatrix} \vec{r}^\transpose \\ \mat{W} \\ \mat{S}_\cdv \end{bmatrix}.
    \end{align*} 
    We recall $\mat{T}_\cdv \in \ring^{t \times t}$ is such that
    it holds $\xi^{-1} \cdot \mat{T}_\cdv \cdot \mat{V}_{\cdv^*} = \mat{I}_t$
    for some size-$t$ set $\cdv^*$ satisfying $(\set{0} \cup \cdv) \subseteq \cdv^* \subseteq [0,K]$.
    To see that the two terms are functionally equivalent,
    we denote $\bar{\cdv} \coloneqq \cdv^* \setminus (\set{0} \cup \cdv)$ and notice 
    \begin{align*}
        \begin{bmatrix} \vec{v}_0^\tr \\ \xi \cdot \mat{V} \end{bmatrix} \mat{R} 
        =
        \begin{bmatrix}  
            \vec{v}_0^\transpose \\
            \xi \cdot \mat{V}_{\hdv} \\
            \xi \cdot \mat{V}_\cdv 
        \end{bmatrix}
        \mat{R}
        =
        \begin{bmatrix}  
            1 \cdot \vec{v}_0^\transpose \\
            \xi \cdot \xi^{-1} \cdot \mat{V}_\hdv \cdot \mat{T}_\cdv \cdot \matV_{\cdv^*} \\
            \xi \cdot \matI_\cdv \cdot \mat{V}_\cdv 
        \end{bmatrix}
        \mat{R}
        =
        \begin{bNiceArray}{c|c|c}  
            1 & \quad & \quad \\\hline 
            \Block{1-3}{\mat{V}_\hdv \cdot \mat{T}_\cdv} \\\hline 
            \quad & \quad & \xi \mat{I}_\cdv  
        \end{bNiceArray}
        \underbrace{
        \begin{bmatrix}  
            \vec{v}_0^\transpose \\
            \mat{V}_{\bar{\cdv}} \\
            \mat{V}_\cdv 
        \end{bmatrix}
        }_{= \mat{V}_{\cdv^*}} \mat{R} 
        \equiv 
        \begin{bNiceArray}{c|c|c}  
            1 & \quad & \quad \\\hline 
            \Block{1-3}{\mat{V}_\hdv \cdot \mat{T}_\cdv} \\\hline 
            \quad & \quad & \xi \mat{I}_\cdv  
        \end{bNiceArray}
        \underbrace{
        \begin{bmatrix} \vec{r}^\transpose \\ \mat{W} \\ \mat{S}_\cdv \end{bmatrix}
        }_{\text{uniform}} \bmod~q,
    \end{align*}
    where the first $=$ is by definition of $\mat{V}$, the second $=$ holds since $\xi^{-1} \cdot \mat{T}_\cdv \cdot \mat{V}_{\cdv^*} = \mat{I}_t$, 
    and the last $\equiv$ holds because $\xi$ and $q$ are coprime and thus $\mat{V}_{\cdv^*}^{-1}$ exists over $\ring_q$.%
    \footnote{By definition of $\Xi$, $\xi \cdot \matV_{\cdv^*}^{-1} \in \ring^{t \times t}$. 
    Since $\xi$ and $q$ are coprime, there exists $a \in \ring$ s.t. $a \cdot \xi = 1 + q\cdot h$ for some $h \in \ZZ$. 
    We have $a \cdot \xi \cdot \matV_{\cdv^*}^{-1} \cdot \matV_{\cdv^*} = (1 + q \cdot h) \cdot \matI$, and $a \cdot \xi \cdot \matV_{\cdv^*}^{-1} \in \ring^{t \times t}$.
    Therefore, $a \cdot \xi \cdot \matV_{\cdv^*}^{-1} \bmod q$ is the inverse of $\matV_{\cdv^*}$ over $\ring_q$.}

    \paragraph{$\hyb_2$.}
    We swap the vector $\mat{b}^\transpose$ and each vector $\vec{d}$
    from 
    \begin{align*}
        \vec{r}^\transpose \mat{A} + \mat{e}^\transpose + \vec{z}^\tr \bmod q
        &&\text{and}&& 
        \matM \cdot 
        \begin{bmatrix} \vec{r}^\transpose \\ \mat{W} \\ \mat{S}_\cdv \end{bmatrix}
        \mat{A} \vecx 
        + \vece_\vecx \bmod q \\
    \text{to}\qquad 
        \diff{\vec{r}_m^\transpose} + \vec{e}^\transpose + \vec{z}^\tr \bmod q 
        &&\text{and}&& 
        \matM \cdot 
        \begin{bmatrix} \diff{\vec{r}_m^\transpose} \\ \diff{\mat{W}_m} \\ \mat{S}_\cdv \mat{A} \end{bmatrix}
        \vecx
        + \vece_\vecx \bmod q
    \end{align*}
    respectively, where $\vec{r}_m, \mat{W}_m$ are uniformly random. 
    For this, we will invoke \cref{thm:core1-alt}, instantiating its parameters $(t,t_2)$ with $(t-|\cdv|, |\cdv|)$ and its matrix with
    \[
        \matM = (\matM_1,\matM_2) \coloneqq
        \begin{bNiceArray}{c|c}
            \Block{1-2}{\mat{V}_\hdv \cdot \mat{T}_\cdv} \\\hline
            \quad & \xi \mat{I}_\cdv
        \end{bNiceArray}
        \in \ring^{K \times ((t-|\cdv|) + |\cdv|)},
    \]
    so that its secret matrices are $\matS_1 = \begin{bmatrix} \vec{r}^\tr \\ \mat{W} \end{bmatrix}$ and $\matS_2 = \mat{S}_\cdv$,
    and its unit vector $\vecu$ selects the row $\vec{r}^\tr$ of $\matS_1$, whose noisy image is $\vecb$.
    We bound the $\ell_2$-induced norm of $\matM_1$ and the norm of $\vecx \coloneqq \begin{bmatrix} \tildevec{c} \\ \tildevec{t} \end{bmatrix}$:
    \begin{align*}
        \norm{\matM_1}_{2,2} \leq~& \norm{\matV_\hdv \cdot \matT_\cdv}_{2,2} = \norm{\xi \cdot \matV_\hdv \cdot \matV_{\cdv^*}^{-1}}_{2,2} \leq \gamma, \\
        \norm{\vecx}_\infty \leq~& \max(B_c,B_t),
    \end{align*}
    where the first inequality holds since $\matM_1$ consists of the first $t-|\cdv|$ columns of $\matM$, the block $\xi \matI_\cdv$ contributing to $\matM_2$ only, and where we recall $\gamma$ is a quality measurement of $\Xi$ (ref. \cref{eq:subtractive-set-qualities}): the last inequality holds by taking $T^* = \cdv^* \setminus \set{0}$ in \cref{eq:subtractive-set-qualities}, since the rows of $\matV_\hdv$ are a subset of the rows of $\matV_\Xi$ and removing rows does not increase the $\ell_2$-induced norm.
    The remaining requirements of \cref{thm:core1-alt} are satisfied by assumptions in this lemma:
    We set the noise parameter $\chi^* = \chi_2^*$, the norm bound $\beta = \max(B_c,B_t)$.
    The two constraints in \cref{thm:core1-alt}, at its instantiated threshold $t - |\cdv|$, follow from $\gamma \cdot \sqrt{\varphi m} \cdot \max(B_c,B_t) \geq 1$ and \cref{eq:th-LWE-variant-constraint}, since $1 \leq \tensorexpansion{t - |\cdv|} \leq \tensorexpansion{t}$.
    Then, by \cref{thm:core1-alt}, whose advantage bound carries the multiplicative factor $t - |\cdv| \leq t$, the advantage of any PPT algorithm $\bdv_2$ in distinguishing $\hyb_1$ and $\hyb_2$ is upper-bounded by
    \[
        t \cdot \left( \Adv_{(\ring,n,m,q,\chi_2^*,\ddvA),\bdv_2}^{\LWE{}}(\secpar)
        + 4 \cdot QK \cdot \epsilon
        + \frac{6 \cdot \epsilon}{1-\epsilon} \right).
    \]
    (To use \cref{thm:core1-alt}, the trivial reduction forwards $\matA$ to $\adv$, adds $\vecz^\tr$ to the vector $\vecb^\tr$ received from the experiment $\mathsf{Exp2}^b$, passes $\xi \matS_2$ in place of $\matS_2$, and answers the first $\Share\oracle$ query on each $\id$ by querying $\oracle_b$ once on the vector $\vecx$ recorded in $\Lquery[\id]$ and storing the answer under $\id$, and every later query on the same $\id$ from that stored answer.)%

    \paragraph{$\hyb_3$.}
    We swap each vector $\vecd$ from 
    \begin{align*}
        \matM \cdot 
        \begin{bmatrix} \vec{r}_m^\transpose \\ \mat{W}_m \\ \mat{S}_\cdv \mat{A} \end{bmatrix} \vecx
        + \vece_\vecx \bmod q
        &&\text{to}&& 
        \matM \cdot 
        \begin{bmatrix} \vec{r}_m^\transpose \vecx \\ \diff{\vecw_\vecx} \\ \mat{S}_\cdv \mat{A} \vecx \end{bmatrix}
        + \vece_\vecx \bmod q
    \end{align*}
    where $\vecw_\vecx$ is uniformly random.
    For this, we will invoke \cref{thm:core2-fix},
    for which we again write $\matM = (\matM_0,\matM_1,\matM_2)$; here $\matM_1$ (width $t_1 = t-|\cdv|-1$) is the block swapped to uniform, while $\matM_0,\matM_2$ (total width $1+|\cdv|$) are left unchanged and together play the role of \cref{thm:core2-fix}'s second block, so we instantiate \cref{thm:core2-fix} with $t_1 = t-|\cdv|-1$ and $t_2 = 1+|\cdv|$.
    We bound the norm of $\matM_1$:
    \[
        \norm{\matM_1}_{2,2} \leq \norm{\matV_\hdv \cdot \matT_\cdv}_{2,2} = \norm{\xi \cdot \matV_\hdv \cdot \matV_{\cdv^*}^{-1}}_{2,2} \leq \gamma,
    \]
    where we recall $\gamma$ is a quality measurement of $\Xi$ (ref. \cref{eq:subtractive-set-qualities}), and the last inequality holds for the same reason as in the $\hyb_2$ step.
    Then, by \cref{thm:core2-fix}, whose advantage bound carries the multiplicative factor $t_1 = t-|\cdv|-1 \leq t$, the advantage of any PPT algorithm $\bdv_1$ in distinguishing $\hyb_2$ and $\hyb_3$ is upper-bounded by
    \[
        t \Adv_{(\ring,m_t,\poly,q,\chi_1^*,\ddvt),\bdv_1}^{\LWE{}}(\secpar)
        + \frac{6 \cdot Q \cdot \epsilon}{1-\epsilon},
    \]
    where $\chi_1^* \geq \sqrt{2} \cdot \eta_\epsilon(\ring^t)$ is such that
    there exists $\hat{\chi} >0$ satisfying
    \[
        \gamma \leq \chi \cdot \sqrt{(2 \eta_\epsilon(\ring^t)^2 + (\chi_1^*)^2)^{-1} - \hat{\chi}^{-2}},
    \]
    as assumed in the statement of this lemma.
    (To use \cref{thm:core2-fix}, the trivial reduction queries $\mathsf{Tag}\oracle$ from the experiments in \cref{thm:core2-fix} to obtain $\tildevec{t}$ to answer $\Gen\oracle$ queries, simply forwards the $\Share\oracle$ queries to the oracle $\oracle$, and generates the rest itself.)

    \paragraph{$\hyb_4$.}
    This is functionally equivalent to $\hyb_3$. 
    We rewrite the vectors $\vec{b} = \vec{r}_m + \vec{e} + \vec{z} \bmod q$ and $\vec{r}_m \sample \ring_q^{m}$ in $\hyb_3$ 
    as $\vec{b} \sample \ring_q^{m}$ and $(\vec{b} - \vec{e} - \vec{z})$ in $\hyb_4$. 
    This is a simple change of variables of uniformly random vectors, and clearly the two hybrids are identically distributed.
    \qed
\end{proof}


We provide \cref{thm:core1-alt,thm:core2-fix} involved in \cref{lem:th-LWE-variant} in the remaining subsections.


\subsection{\texorpdfstring{\cref{thm:core2-fix}}{Theorem \ref{thm:core2-fix}} and proofs}\label{app:core2}

\begin{theorem}\label{thm:core2-fix}
    Let $\ring,m,Q,q,m_c,K,t_1,t_2,\chi^*,\chi,\gamma,\ddvt$ and $\epsilon \in (0,1)$ be parametrised by $\secpar$,
    where $\ring = \ZZ[\zeta]$ has ring degree $\varphi$,
    $m,Q,m_c,K,t_1,t_2 \in \poly$ and $\ddvt$ is a distribution over $\ring_q^{m}$.
    Suppose $\chi^*,\chi$ are such that $\chi^* \geq \sqrt{2} \cdot \eta_\epsilon(\ring^{t_1})$, and there exists $\hat{\chi} >0$ such that $\gamma \leq \chi \cdot \sqrt{(2 \eta_\epsilon(\ring^{t_1})^2 + (\chi^*)^2)^{-1} - \hat{\chi}^{-2}}$. Denote
    \begin{align*}
        \params = (\ring,m,Q,q,m_c,K,t_1,t_2,\chi,\gamma,\ddvt).
    \end{align*}
    For any PPT $\adv$, define the following experiment $\mathsf{Exp1}^b_{\params,\adv}$:
    \begin{pchstack}[center]
        \begin{pcvstack}
            \procedure{$\mathsf{Exp1}^b_{\params, \adv}(\secparam)$}
            {
                \Lquery[\cdot] \coloneqq \emptyset ;~ \Ltag[\cdot] \coloneqq \emptyset \\
                \begin{pmatrix}
                    (\matM=(\matM_1,\matM_2),\matS_2) \pclb
                    \in \ring^{K \times (t_1+t_2)} \times \ring_q^{t_2 \times (m_c+m)}
                \end{pmatrix}
                \gets \adv(\secparam) \\
                \pcassert \norm{\matM_1}_{2,2} \leq \gamma \\
                \matS_1 \sample \ring_q^{t_1 \times (m_c+m)} \\
                b' \gets \adv^{\mathsf{Tag}\oracle,\oracle_b}(\secparam) \\
                \pcreturn b'
            }
            \pcvspace
            \procedure{$\mathsf{Tag}\oracle()$}
            {
                i \coloneqq |\Ltag| + 1 ;~ \vec{t}_i \sample \ddvt \\
                \Ltag[i] \coloneqq \vec{t}_i \\
                \pcreturn (i, \vec{t}_i)
            }
        \end{pcvstack}
        \pchspace
        \procedure{$\oracle_b(\vecc \in \ring^{m_c}, i)$}
        {
            \pcassert |\Lquery| < Q \\
            \pcif \Lquery[i] = \emptyset \pcthen \\
            \pcind \pcassert \Ltag[i] \neq \emptyset ;~ \pcparse \vec{t}_i \gets \Ltag[i] \\
            \pcind \vecx \coloneqq \begin{bmatrix}\vecc \pclb \vec{t}_i \end{bmatrix} ;~ \vece \sample \ddv_{\ring,\chi}^{K} \\
            \pcind\pcif b=0 \pcthen
                \vecd \coloneqq \matM \cdot \begin{bmatrix} \mat{S}_1 \pclb \mat{S}_2 \end{bmatrix} \vecx + \vece \\
            \pcind \pcif b=1 \pcthen
                \vecw \sample \ring_q^{t_1} ;~
                \vecd \coloneqq \matM \cdot \begin{bmatrix} \vecw \pclb \mat{S}_2 \vecx \end{bmatrix} + \vece \\
            \pcind \Lquery[i] \coloneqq \vecd \\
            \pcreturn \vecd \gets \Lquery[i]
        }
    \end{pchstack}
    and denote the advantage of $\adv$ in distinguishing the above experiments by
    \begin{align}\label{eq:core2-fix-advantage}
        \Adv_{\params,\adv}^{\mathsf{Exp1}} \coloneqq
        \abs{ \prob{\mathsf{Exp1}^0_{\params, \adv}(\secparam) = 1} - \prob{\mathsf{Exp1}^1_{\params, \adv}(\secparam) = 1}}.
    \end{align}
    For any PPT $\adv$, there exists a PPT $\bdv$ such that
    \[
        \Adv_{\params,\adv}^{\mathsf{Exp1}}(\secpar) \leq t_1 \cdot \Adv_{(\ring,m,\poly,q,\chi^*,\ddvt),\bdv}^{\LWE{}}(\secpar) + 2Q \cdot \frac{3\epsilon}{1-\epsilon}.
    \]
    In particular, if the $\LWE{\ring,m,\poly,q,\chi^*,\ddvt}$ assumption holds,
    then the experiments $\mathsf{Exp1}^0_{\params,\adv}$ and $\mathsf{Exp1}^1_{\params,\adv}$ are computationally indistinguishable.
\end{theorem}

\begin{proof}
    Follows directly by putting \cref{lem:LLWE-reprove,lem:core2-fix} together. 
    (We rename $\bar{\chi}$ in \cref{lem:LLWE-reprove} to $\hat{\chi}$.) \qed
\end{proof}

We will use the following statistical lemma from \cite{CiC:LaiSwaWoo25} to prove \cref{lem:LLWE-reprove}.
The proof of \cref{lem:LLWE-reprove} is essentially an adaption of \cite[Thm.3]{CiC:LaiSwaWoo25}. 

\begin{lemma}[Special case of {\cite[Thm. 2]{CiC:LaiSwaWoo25}}]\label{lem:elLWE-thm2}
    Let $m_1, m_2 \in \NN$ and $s_1,s_2,\chi^* >0$.
    Fix any $\mat{Z} \in \ZZ^{m_1 \times m_2}$.
    Define the following matrices (dependent on $s_1,s_2,\mat{Z}$): 
    \begin{align}\label{eq:elLWE-define-matrices}
        \mat{\Gamma} \coloneqq~& s_2^{-2} \cdot \mat{Z} \in \RR^{m_1 \times m_2}, \nonumber \\
        \mat{\Theta}^{-1} \coloneqq~& s_1^{-2} \cdot \mat{I}_{m_1} + s_2^{-2} \cdot \mat{Z} \mat{Z}^\tr \in \RR^{m_1 \times m_1}, \\
        \mat{\Phi}^{-1} \coloneqq~& s_2^{-2} \cdot \mat{I}_{m_2} - \mat{\Gamma}^\tr \cdot \mat{\Theta} \cdot \mat{\Gamma} \in \RR^{m_2 \times m_2}. \nonumber 
    \end{align}
    Let $\epsilon \in (0,1)$ and $\upsilon^{-1} \coloneqq s_1^{-2} + \smax(\mat{Z} \mat{Z}^\transpose) \cdot s_2^{-2}$.
    Suppose $(\chi^*)^2 < \upsilon$ and $\left( (\chi^*)^{-2} + (\upsilon - (\chi^*)^2)^{-1} \right)^{-1/2} \geq \eta_\epsilon(\ZZ^{m_1})$, which implies $\sqrt{\upsilon} \geq \eta_\epsilon(\ZZ^{m_1})$.
    Then, the statistical distance between the distributions 
    \begin{align*}  
        \set{ \left( \vec{e}_1^\transpose, \vec{e}_1^\transpose \mat{Z} + \vec{e}_2^\transpose \right)
        ~\middle|~
        \begin{aligned}
        &\vec{e}_1 \sample \ddv_{\ZZ^{m_1}, s_1} \\
        &\vec{e}_2 \sample \ddv_{\ZZ^{m_2}, s_2} 
        \end{aligned}
        } 
        &&\text{and}&&
        \set{ \left( \vec{e}^\transpose + \vec{f}_1^\transpose, \vec{f}_2^\transpose \right)
        ~\middle|~
        \begin{aligned}
        &\vec{e} \sample \ddv_{\ZZ^{m_1}, \chi^*} ;~
        \vec{f}_2 \sample \ddv_{\ZZ^{m_2}, \sqrt{\mat{\Phi}}} \\
        &\vec{c} \coloneqq \mat{\Theta} \cdot \mat{\Gamma} \cdot \vec{f}_2 \\
        &\vec{f}_1 \sample \ddv_{\ZZ^{m_1}, \sqrt{\mat{\Theta} - (\chi^*)^2 \mat{I}_{m_1}}, \vec{c}}
        \end{aligned}
        }
    \end{align*} 
    is at most $3\epsilon/(1-\epsilon)$.
\end{lemma}

\begin{lemma}\label{lem:LLWE-reprove}
    Let $\ring,m,Q,q,t_1,K,\bar{\chi},\chi^*,\chi,\gamma,\ddvt$ and $\epsilon \in (0,1)$ be parametrised by $\secpar$,
    where $\ring = \ZZ[\zeta]$ has ring degree $\varphi$,
    $m,Q,t_1 \in \poly$ and $\ddvt$ is a distribution over $\ring_q^{m}$.
    Let $\chi^* \geq \sqrt{2} \cdot \eta_\epsilon(\ring^{t_1})$, and 
    \(
        \gamma \leq \chi \cdot \sqrt{ \left( 2 \eta_\epsilon(\ring^{t_1})^2 + (\chi^*)^2 \right)^{-1} - \bar{\chi}^{-2} }.
    \)
    Denote
    \begin{align*}
        \params = (\ring,m,Q,q,t_1,K,\bar{\chi},\chi,\gamma,\ddvt).
    \end{align*}
    For any PPT $\adv$, define the following experiment $\mathsf{Exp}^b_{\params,\adv}$: 
    \begin{pchstack}[center]
        \adjustbox{max width=\columnwidth}%
        {
        \begin{pcvstack}
            \procedure{$\mathsf{Exp}^b_{\params, \adv}(\secparam)$}
            {
                \Lquery[\cdot] \coloneqq \emptyset ;~ \Ltag[\cdot] \coloneqq \emptyset \pccomment{initial lists}\\
                \mat{Z} \in \ring^{t_1 \times K} \gets \adv(\secparam) ;~ \pcassert \norm{\mat{Z}^\tr}_{2,2} \leq \gamma \\
                \mat{S} \sample \ring_q^{t_1 \times m} \\
                b' \gets \adv^{\mathsf{Tag}\oracle,\leak\oracle_b}(\cdot) \\
                \pcreturn b'
            }
            \pcvspace
            \procedure{$\mathsf{Tag}\oracle()$}
            {
                i \coloneqq |\Ltag| + 1 ;~ \vec{t}_i \sample \ddvt \pccomment{$\vec{t}_i \in \ring_q^{m}$} \\
                \Ltag[i] \coloneqq \vec{t}_i \\
                \pcreturn (i, \vec{t}_i)
            }
        \end{pcvstack}
        \pchspace
            \procedure{$\leak\oracle_b(i)$}
            {
                \pcassert |\Lquery| < Q \\
                \pcif \Lquery[i] = \emptyset \pcthen \\
                \pcind \pcassert \Ltag[i] \neq \emptyset ;~ \pcparse \vec{t}_i \gets \Ltag[i] \\
                \pcind \barvec{e}_i \sample \ddv_{\ring,\bar{\chi}}^{t_1} ;~
                    \vec{e}_i \sample \ddv_{\ring,\chi}^{K} \\
                \pcind \pcif b = 0 \pcthen \vec{b}_i \coloneqq \mat{S} \cdot \vec{t}_i + \barvec{e}_i \bmod q \\
                \pcind \pcif b = 1 \pcthen \vec{b}_i \sample \ring_q^{t_1} \\
                \pcind \vec{l}_i^\tr \coloneqq \barvec{e}_i^\tr \cdot \mat{Z} + \vec{e}_i^\tr \bmod q \\
                \pcind \Lquery[i] \coloneqq (\vec{b}_i, \vec{l}_i) \\
                \pcreturn (\vec{b}_i, \vec{l}_i) \gets \Lquery[i]
            }
        }
    \end{pchstack}
    and denote the advantage of $\adv$ in distinguishing the above experiments by
    \begin{align*}
        \Adv_{\params,\adv}^{\mathsf{Exp}} \coloneqq
        \abs{ \prob{\mathsf{Exp}^0_{\params, \adv}(\secparam) = 1} - \prob{\mathsf{Exp}^1_{\params, \adv}(\secparam) = 1}}.
    \end{align*}
    For any PPT $\adv$, there exists a PPT $\bdv$ such that
    \[
        \Adv_{\params,\adv}^{\mathsf{Exp}}(\secpar) \leq t_1 \cdot \Adv_{(\ring,m,\poly,q,\chi^*,\ddvt),\bdv}^{\LWE{}}(\secpar) + 2Q \cdot \frac{3\epsilon}{1-\epsilon}.
    \]
    In particular, if the $\LWE{\ring,m,\poly,q,\chi^*,\ddvt}$ assumption holds,
    then the experiments $\mathsf{Exp}^0_{\params,\adv}$ and $\mathsf{Exp}^1_{\params,\adv}$ are computationally indistinguishable.
\end{lemma}

\begin{proof}
    We will rely on that LWE with a uniformly random secret matrix $\matS \sample \ring_q^{t_1 \times m}$ is no easy than LWE with a uniformly random secret row vector. This is stated in the following proposition. 
    
    \begin{proposition}
       Let $\ring,m,q,t_1,\chi^*,\ddvA,\ddvt$ be parametrised by $\secpar$,
       where $t_1 \in \poly$ and $\ddvt$ is a distribution over $\ring_q^{m}$. 
       For any PPT $\adv$, there exists a PPT $\bdv$ such that  
       \[
            \Adv_{(\ring,m,\poly,t_1,q,\chi^*,\ddvt),\adv}^{\LWE{}}(\secpar)
            \leq 
            t_1 \cdot \Adv_{(\ring,m,\poly,q,\chi^*,\ddvt),\bdv}^{\LWE{}}(\secpar),
       \]
    \end{proposition}

    The above follows directly from standard hybrid argument, its proof omitted.

    In the rest we show that, given a PPT $\adv$ distinguishing $\mathsf{Exp}^0$ and $\mathsf{Exp}^1$, we construct a PPT $\bdv$ against the $\LWE{\ring,m,\poly,t_1,q,\chi^*,\ddvt}$ assumption.
    We will show that if $\adv$ distinguishes the two experiments with advantage $\Adv_{\params,\adv}^{\mathsf{Exp}}(\secpar)$,
    then $\bdv$ breaks the $\LWE{\ring,m,\poly,t_1,q,\chi^*,\ddvt}$ assumption with advantage at least
    \begin{align}\label{eq:LLWE-reprove-advantage}
        \Adv_{\params,\adv}^{\mathsf{Exp}}(\secpar) - 2Q \cdot \frac{3\epsilon}{1-\epsilon}.
    \end{align}
    The lemma then follows from combining with the proposition above.

    The PPT $\bdv$ works as follows:
    \begin{itemize}
        \item Receive for $\adv$ the matrix $\mat{Z}$, and assert $\norm{\mat{Z}^\tr}_{2,2} \leq \gamma$.
        \item Internally maintain the lists $\Lquery$ and $\Ltag$. Also, let
        \begin{align*} 
            \mat{\Gamma} = \chi^{-2} \cdot \Psi(\mat{Z}^\tr)^\tr, \qquad
            \mat{\Theta} = \big( \bar{\chi}^{-2} \cdot \mat{I}_{\varphi t_1} + \chi^{-2} \cdot \Psi(\mat{Z}^\tr)^\tr \Psi(\mat{Z}^\tr) \big)^{-1}, \qquad
            \mat{\Phi} = \big( \chi^{-2} \cdot \mat{I}_{\varphi K} - \mat{\Gamma}^\tr \, \mat{\Theta} \, \mat{\Gamma} \big)^{-1},
        \end{align*}
        where we recall that $\psi \colon \field \to \ZZ^\varphi$ is the coefficient embedding w.r.t. the power basis (\cref{sec:prelim}), and $\Psi(\matM)$ is the matrix of $\vecx \mapsto \matM \vecx$ in the coefficient embedding, so that $\psi(\matM \vecx) = \Psi(\matM) \cdot \psi(\vecx)$.
        In particular $\psi((\barvec{e}^\tr \matZ)^\tr) = \psi(\matZ^\tr \barvec{e}) = \Psi(\matZ^\tr) \cdot \psi(\barvec{e})$, so the leakage $\barvec{e}^\tr \matZ + \vece^\tr$ is, in coefficients, $\psi(\barvec{e})^\tr \cdot \Psi(\matZ^\tr)^\tr + \psi(\vece)^\tr$, and \cref{lem:elLWE-thm2} is applied with the integer matrix $\Psi(\matZ^\tr)^\tr \in \ZZ^{\varphi t_1 \times \varphi K}$.%

        \item For any query to $\mathsf{Tag}\oracle$, do the following: 
        \begin{itemize}
            \item Query the LWE challenger to obtain an LWE sample $(\vect_i,\hatvec{b}_i)$.
            \item Store $\Ltag[i] \coloneqq (\vect_i,\hatvec{b}_i)$, i.e. including also the LWE sample $\hatvec{b}_i$, and answer with $\vect_i$. 
        \end{itemize}

        \item For any query $i$ to $\mathsf{Leak}\oracle$, do the following:
        \begin{itemize}
            \item Assert $|\Lquery| < Q$, i.e. there has been fewer than $Q$ distinct queries.
            \item If $\Lquery[i] = \emptyset$, assert $\Ltag[i] \neq \emptyset$ and parse $(\vect_i,\hatvec{b}_i) \gets \Ltag[i]$.
            \item Sample $\vecl_i \sample \ddv_{\ring^{K}, \sqrt{\mat{\Phi}}}$, set 
            $\vecc_i \coloneqq \psi^{-1}(\mat{\Theta} \cdot \mat{\Gamma} \cdot \psi(\vecl_i))$\footnote{We abuse notation slightly, because $\mat{\Theta} \cdot \mat{\Gamma} \cdot \psi(\vecl_i)$ may not lie in the image of $\psi$. This is not an issue, as we only operate on $\psi(\vecc)$.} and $\matF \coloneqq \mat{\Theta} - (\chi^*)^2 \mat{I}_{\varphi t_1}$, and sample
            $\vecf_i \sample \ddv_{\ring^{t_1}, \sqrt{\matF}, \vecc_i}$. Set
            \(
                \vecb_i \coloneqq \hatvec{b}_i + \vecf_i \bmod q
            \)
            and $\Lquery[i] \coloneqq (\vecb_i, \vecl_i)$.
            \item Answer with $(\vecb_i, \vecl_i) \gets \Lquery[i]$.
        \end{itemize}
        \item Eventually receive a bit $b'$ from $\adv$, return $b'$. 
    \end{itemize}
    We analyse the output of $\bdv$. 

    \medskip\noindent
    \underline{If $\bdv$ interacts with the LWE experiment $0$}:
    In this case, each sample $(\vect_i,\hatvec{b}_i)$ from the LWE challenger is such that 
    $\hatvec{b}_i = \matS \cdot \vect_i + \hatvec{e}_i \bmod q$, where $\hatvec{e}_i \sample \ddv_{\ring,\chi^*}^{t_1}$.
    The simulated $\leak\oracle$ answer is thus 
    \begin{align*}
        (\vecb_i, \vecl_i) 
        =
        (\matS \cdot \vect_i + \hatvec{e}_i + \vecf_i,~ \vecl_i) 
        \approx_s
        (\matS \cdot \vect_i + \barvec{e}_i,~ \barvec{e}_i^\tr \cdot \matZ + \vece_i^\tr),
    \end{align*}
    where the last line follows from \cref{lem:elLWE-thm2}, with statistical distance at most $\frac{3\epsilon}{1-\epsilon}$.
    Since there are at most $Q$ many queries, 
    the overall view of $\adv$ is close to that from $\mathsf{Exp}^0$ with statistical distance at most 
    $Q \cdot \frac{3\epsilon}{1-\epsilon}$.

    \medskip\noindent
    \underline{If $\bdv$ interacts with the LWE experiment $1$}:
    In this case, each sample $(\vect_i,\hatvec{b}_i)$ from the LWE challenger is such that 
    $\hatvec{b}_i \sample \ring_q^{t_1}$ is uniformly random.
    The simulated $\leak\oracle$ answer is thus 
    \begin{align*}
        (\vecb_i, \vecl_i) 
        = 
        (\hatvec{b} + \vecf_i,~ \vecl_i) 
        = 
        (\mathcal{U}(\ring_q^{t_1}),~ \vecl_i) 
        \approx_s 
        (\mathcal{U}(\ring_q^{t_1}),~ \barvec{e}_i^\tr \cdot \matZ + \vece_i^\tr), 
    \end{align*}
    where the last line again follows from \cref{lem:elLWE-thm2}, with statistical distance at most $\frac{3\epsilon}{1-\epsilon}$.
    Since there are at most $Q$ many queries, 
    the overall view of $\adv$ is close to that from $\mathsf{Exp}^1$ with statistical distance at most 
    $Q \cdot \frac{3\epsilon}{1-\epsilon}$.

    Combining the two cases yields \cref{eq:LLWE-reprove-advantage}.

    It remains to see that \cref{lem:elLWE-thm2} can indeed be invoked. 
    Write $\upsilon^{-1} \coloneqq \bar{\chi}^{-2} + \smax(\Psi(\mat{Z}^\tr)^\tr \cdot \Psi(\mat{Z}^\tr)) \cdot \chi^{-2}$.
    By the lemma condition,
    \[
        \smax(\Psi(\mat{Z}^\tr)^\tr \cdot \Psi(\mat{Z}^\tr)) = \norm{\mat{Z}^\tr}_{2,2}^2 \leq \gamma^2 \leq \chi^2 \big( (2\eta_\epsilon(\ring^{t_1})^2 + (\chi^*)^2)^{-1} - \bar{\chi}^{-2} \big), 
    \]
    which implies $\upsilon^{-1} \leq (2\eta_\epsilon(\ring^{t_1})^2 + (\chi^*)^2)^{-1}$, and therefore 
    $\upsilon - (\chi^*)^2 \geq 2\eta_\epsilon(\ring^{t_1})^2$. 
    Together with the other condition of $(\chi^*)^2 \geq 2\eta_\epsilon(\ring^{t_1})^2$, 
    we have 
    $(\chi^*)^2 < \upsilon$ and $\big( (\chi^*)^{-2} + (\upsilon - (\chi^*)^2)^{-1} \big)^{-1/2} \geq \eta_\epsilon(\ring^{t_1})$, hence the requirements of \cref{lem:elLWE-thm2} are satisfied.
    \qed
\end{proof}

\begin{lemma}\label{lem:core2-fix}
    Let $\ring,m,Q,q,m_c,K,t_1,t_2,\hat{\chi},\chi,\gamma,\ddvt$ be parametrised by $\secpar$, where $\ddvt$ is a distribution over $\ring_q^{m}$. Denote
    \begin{align*}
        \params = (\ring,m,Q,q,m_c,K,t_1,t_2,\chi,\gamma,\ddvt), \qquad
        \params_{\mathsf{E}} = (\ring,m,Q,q,t_1,K,\hat{\chi},\chi,\gamma,\ddvt).
    \end{align*}
    Consider the experiment $\mathsf{Exp1}_{\params,\adv}^b$ defined in \cref{thm:core2-fix} and the advantage $\Adv_{\params,\adv}^{\mathsf{Exp1}}$ defined in \cref{eq:core2-fix-advantage}, and let $\mathsf{Exp}$ denote the experiment of \cref{lem:LLWE-reprove} with parameters $\params_{\mathsf{E}}$.
    For any PPT $\adv$, there exists a PPT $\bdv$ such that
    \[
        \Adv_{\params,\adv}^{\mathsf{Exp1}}(\secpar) = \Adv_{\params_{\mathsf{E}},\bdv}^{\mathsf{Exp}}(\secpar).
    \]
\end{lemma}

\begin{proof}
    Given a PPT $\adv$ distinguishing $\mathsf{Exp1}^0$ and $\mathsf{Exp1}^1$, we construct a PPT $\bdv$ distinguishing $\mathsf{Exp}^0$ and $\mathsf{Exp}^1$ in \cref{lem:LLWE-reprove}, as follows:
    \begin{itemize}
        \item Receive from $\adv(\secparam)$ the matrices $(\matM = (\matM_1,\matM_2),\matS_2)$, and assert $\norm{\matM_1}_{2,2} \leq \gamma$.
        \item Send $\matZ \coloneqq \matM_1^\tr \in \ring^{t_1 \times K}$ to the $\mathsf{Exp}$ challenger. 
        \item Sample $\matS_{1,c} \sample \ring_q^{t_1 \times m_c}$. 
        \item Maintain a list $\Lquery$. Answer $\adv$'s oracle queries as follows:
        \begin{itemize}
            \item $\mathsf{Tag}\oracle()$: forward the query to the tag oracle of $\mathsf{Exp}$, obtain $(i,\vec{t}_i)$, and answer with $(i,\vec{t}_i)$.
            \item $\oracle_b(\vecc,i)$: if $\Lquery[i] = \emptyset$, query the leak oracle $\leak\oracle_b(i)$ of $\mathsf{Exp}$ to obtain $(\vec{b}_i,\vec{l}_i)$, set $\vecx \coloneqq \begin{bmatrix}\vecc \\ \vec{t}_i \end{bmatrix}$, compute
            \[
                \vecd \coloneqq \matM_1 \cdot (\matS_{1,c} \cdot \vecc) + \matM_1 \cdot \vec{b}_i - \vec{l}_i + \matM_2 \cdot (\matS_2 \cdot \vecx) \bmod q,
            \]
            and set $\Lquery[i] \coloneqq \vecd$. Answer with $\vecd \gets \Lquery[i]$.
        \end{itemize}
        \item Eventually receive a bit $b'$ from $\adv$, return $b'$.
    \end{itemize}
    Below we show that $\bdv$ perfectly simulates the view of either $\mathsf{Exp1}^0$ or $\mathsf{Exp1}^1$ to $\adv$,
    if it interacts respectively with $\mathsf{Exp}^0$ or $\mathsf{Exp}^1$, thus the lemma follows.

    \medskip\noindent
    \underline{If $\bdv$ interacts with $\mathsf{Exp}^0$}:
    By \cref{lem:LLWE-reprove}, for each fresh index $i$ the leak oracle returns
    \[
        \vec{b}_i = \matS \cdot \vec{t}_i + \barvec{e}_i \bmod q, \qquad
        \vec{l}_i = \matM_1 \cdot \barvec{e}_i + \vec{e}_i,
    \]
    where $\matS \sample \ring_q^{t_1 \times m}$ is the fixed secret of $\mathsf{Exp}$, $\barvec{e}_i \sample \ddv_{\ring,\hat{\chi}}^{t_1}$, $\vec{e}_i \sample \ddv_{\ring,\chi}^{K}$, and we used $\matZ = \matM_1^\tr$. Substituting,
    \begin{align*}
        \vecd
        =~& \matM_1 \cdot \matS_{1,c} \cdot \vecc + \matM_1 \cdot (\matS \cdot \vec{t}_i + \barvec{e}_i) - (\matM_1 \cdot \barvec{e}_i + \vec{e}_i) + \matM_2 \cdot \matS_2 \cdot \vecx \\
        =~& \matM_1 \cdot (\matS_{1,c} \cdot \vecc + \matS \cdot \vec{t}_i) - \vec{e}_i + \matM_2 \cdot \matS_2 \cdot \vecx \\
        =~& \matM_1 \cdot \matS_1 \cdot \vecx - \vec{e}_i + \matM_2 \cdot \matS_2 \cdot \vecx
        \equiv \matM \cdot \begin{bmatrix} \matS_1 \\ \matS_2 \end{bmatrix} \cdot \vecx + \vec{e}_i \bmod q,
    \end{align*}
    where $\matS_1 \coloneqq (\matS_{1,c},\matS) \in \ring_q^{t_1 \times (m_c+m)}$, so that $\matS_1 \cdot \vecx = \matS_{1,c} \cdot \vecc + \matS \cdot \vec{t}_i \bmod q$,
    and the last $\equiv$ holds since the discrete Gaussian over $\ring$ is symmetric around zero.
    Since $\matS_{1,c}$ and $\matS$ are both uniform, $\matS_1$ is uniform over $\ring_q^{t_1 \times (m_c+m)}$, exactly as in $\mathsf{Exp1}^0$. Thus $\vecd$ is distributed identically to the oracle answer in $\mathsf{Exp1}^0$, and $\bdv$ perfectly simulates $\mathsf{Exp1}^0$.

    \medskip\noindent
    \underline{If $\bdv$ interacts with $\mathsf{Exp}^1$}:
    In this case, $\vec{b}_i \sample \ring_q^{t_1}$ is uniform and independent, while $\vec{l}_i = \matM_1 \barvec{e}_i + \vec{e}_i$ as before. Substituting,
    \begin{align*}
        \vecd
        =~& \matM_1 \cdot \matS_{1,c} \cdot \vecc + \matM_1 \cdot \vec{b}_i - (\matM_1 \cdot \barvec{e}_i + \vec{e}_i) + \matM_2 \cdot \matS_2 \cdot \vecx \\
        =~& \matM_1 \cdot (\matS_{1,c} \cdot \vecc + \vec{b}_i - \barvec{e}_i) - \vec{e}_i + \matM_2 \cdot \matS_2 \cdot \vecx \\
        =~& \matM_1 \cdot \vecw - \vec{e}_i + \matM_2 \cdot \matS_2 \cdot \vecx
        \equiv \matM \cdot \begin{bmatrix} \vecw \\ \matS_2 \cdot \vecx \end{bmatrix} + \vec{e}_i \bmod q,
    \end{align*}
    where $\vecw \coloneqq \matS_{1,c} \vecc + \vec{b}_i - \barvec{e}_i \bmod q$. Since $\vec{b}_i$ is uniformly random and independent across queries, so is $\vecw$, identically distributed as the freshly sampled $\vecw$ in $\mathsf{Exp1}^1$. Thus $\vecd$ is distributed identically to the oracle answer in $\mathsf{Exp1}^1$, and $\bdv$ perfectly simulates $\mathsf{Exp1}^1$.

    \medskip\noindent
    \underline{Norm bound and query count}:
    The matrix $\matZ = \matM_1^\tr$ sent by $\bdv$ to the $\mathsf{Exp}$ challenger is admissible,
    because
    \[
        \norm{\matZ^\tr}_{2,2} = \norm{\matM_1}_{2,2} \leq \gamma.
    \]%
    $\bdv$ simply forwards each $\mathsf{Tag}\oracle$ query, and it queries $\leak\oracle$ exactly once per fresh index $i$, hence at most $Q$ times by the assertion $|\Lquery| < Q$ in $\oracle_b$ of $\mathsf{Exp}^b$.
    \qed
\end{proof}

%

\subsection{\texorpdfstring{\cref{thm:core1-alt}}{Theorem \ref{thm:core1-alt}} and proofs}\label{app:core1-alt}

\begin{theorem}\label{thm:core1-alt}
    Let $\ring,n,m,Q,q,K,t,t_2,\epsilon,\bar{\chi},\chi,\chi^*,\gamma,\beta,\ddvA$ be parametrised by $\secpar$,
    where $\ring = \ZZ[\zeta]$ has ring degree $\varphi$ and field expansion factor $\expansion_2$ w.r.t. the coefficient $\ell_2$-norm,
    $n,m,Q,K,t,t_2 \in \poly$, $\ddvA$ is a distribution over $\ring_q^{n \times m}$,
    $0 < \epsilon \leq 1/2$ and $\chi^* \geq \sqrt{2} \cdot \eta_\epsilon(\ring^m)$.
    Denote by $\vecu^\tr \coloneqq (1,0,\ldots,0) \in \ring^{1 \times t}$ the first unit vector, and
    \begin{align*}
        \params = (\ring,n,m,Q,q,K,t,t_2,\bar{\chi},\chi,\gamma,\beta,\ddvA), \qquad
        \params_{\mathsf{LWE}} = (\ring,n,m,q,\chi^*,\ddvA).
    \end{align*}
    Let $\tensorexpansion{t} \coloneqq 1$ if $\conductor$ is a power of two and $\tensorexpansion{t} \coloneqq \sqrt{t}$ otherwise.
    Suppose that $\gamma \cdot \sqrt{\varphi m} \cdot \beta \geq 1$ and
    \begin{align}\label{eq:core1-alt-constraint}
        \sqrt{2 Q K \cdot \varphi m} \cdot \expansion_2 \cdot \tensorexpansion{t} \cdot \gamma \cdot \beta \leq \chi \cdot \sqrt{(2 \eta_\epsilon(\ring^m)^2 + (\chi^*)^2)^{-1} - \bar{\chi}^{-2}}.
    \end{align}%
    For any PPT $\adv$, define the following experiment $\mathsf{Exp2}_{\params,\adv}^b$:
    \begin{pchstack}[center]
        \procedure{$\mathsf{Exp2}^b_{\params, \adv}(\secparam)$}
        {
            i \coloneqq 0 \\
            \matA \sample \ddvA \\
            \matM = (\matM_1,\matM_2) \in \ring^{K \times (t+t_2)} \gets \adv(\matA) \\
            \pcassert \norm{\matM_1}_{2,2} \leq \gamma \\
            \matS_2 \sample \ring_q^{t_2 \times n} ;~
            \barvec{e} \sample \ddv_{\ring,\bar{\chi}}^m \\
            \pcif b=0 \pcthen \\
                \pcind \matS_1 \sample \ring_q^{t \times n} \\
                \pcind \vecb^\tr \coloneqq \vecu^\tr \matS_1 \matA + \barvec{e}^\tr \bmod q \\
            \pcif b=1 \pcthen \\
                \pcind \barmat{S}_1 \sample \ring_q^{t \times m} \\
                \pcind \vecb^\tr \coloneqq \vecu^\tr \barmat{S}_1 + \barvec{e}^\tr \bmod q \\
            b' \gets \adv^{\oracle_b}(\matA,\vecb,\matS_2) \\
            \pcreturn b'
        }
        \pchspace
        \procedure{$\oracle_b(\vecx \in \ring^m)$}
        {
            \pcassert i < Q \\ 
            i \coloneqq i+1 \\
            \pcassert \norm{\vecx}_\infty \leq \beta \\
            \vece \sample \ddv_{\ring,\chi}^{K} \\
            \pcif b=0 \pcthen
                \vecd \coloneqq \matM \cdot \begin{bmatrix} \matS_1 \matA \pclb \matS_2 \matA \end{bmatrix} \vecx + \vece \bmod q \\
            \pcif b=1 \pcthen
                \vecd \coloneqq \matM \cdot \begin{bmatrix} \barmat{S}_1 \pclb \matS_2 \matA \end{bmatrix} \vecx + \vece \bmod q \\
            \pcreturn \vecd
        }
    \end{pchstack}%
    and denote the advantage of $\adv$ in distinguishing the above experiments by
    \begin{align}\label{eq:core1-alt-advantage}
        \Adv_{\params,\adv}^{\mathsf{Exp2}} \coloneqq
        \abs{ \prob{\mathsf{Exp2}^0_{\params, \adv}(\secparam) = 1} - \prob{\mathsf{Exp2}^1_{\params, \adv}(\secparam) = 1}}.
    \end{align}
    $\forall$ PPT $\adv$, $\exists$ PPT $\bdv$ such that
    \begin{align*}
        \Adv_{\params,\adv}^{\mathsf{Exp2}}(\secpar)
        \leq t \cdot \left( \Adv_{\params_{\mathsf{LWE}},\bdv}^{\LWE{}}(\secpar) + 4 \cdot QK \cdot \epsilon + 6 \cdot \epsilon/(1-\epsilon) \right).
    \end{align*}
    In particular, if $QK \in \poly$, $\epsilon \in \negl$ and the $\LWE{\params_{\mathsf{LWE}}}$ assumption holds,
    then the experiments $\mathsf{Exp2}^0_{\params,\adv}$ and $\mathsf{Exp2}^1_{\params,\adv}$ are computationally indistinguishable.
\end{theorem}

\begin{proof}
    Follows from putting \cref{cor:lwe-to-adpllwe,lem:AdpLLWE-IotimesA,lem:mLLWE-from-AdpLLWE-alt,lem:core1-alt} together. In particular:
    \begin{itemize}
        \item By \cref{cor:lwe-to-adpllwe}, invoked with the query bound $QK$ and with its leakage parameter $\alpha$ instantiated as $\tensorexpansion{t} \cdot \gamma \cdot \sqrt{\varphi m} \cdot \beta$, so that its leakage set is
        $\cL = \{ \vecz \in \ring^{m}: \norm{\vecz}_2 \leq \tensorexpansion{t} \cdot \gamma \cdot \sqrt{\varphi m} \cdot \beta \}$,
        the $\LWE{\params_{\mathsf{LWE}}}$ assumption implies $\Adp\LLWE{}$ with
        $\params^* = (\ring,n,m,QK,q,\bar{\chi},\chi,\cL,\ddvA)$,
        at an additive loss of $4 \cdot QK \cdot \epsilon + 6 \cdot \epsilon/(1-\epsilon)$.
        Its requirements hold: $\alpha \geq 1$ follows from $\gamma \cdot \sqrt{\varphi m} \cdot \beta \geq 1$ as $\tensorexpansion{t} \geq 1$, $\chi^* \geq \sqrt{2} \cdot \eta_\epsilon(\ring^m)$ is assumed,
        and \cref{eq:core1-alt-constraint} is \cref{eq:adpllwe-mod2} with $\alpha$ instantiated and with $Q$ replaced by $QK$, which suffices by \cref{rem:llwe-kappa}.
        \item By \cref{lem:AdpLLWE-IotimesA}, $\Adp\LLWE{}$ with $\params^*$ implies
        $\Adp\LLWE{}$ with $\params^\dagger = (\ring,nt,mt,QK,q,\bar{\chi},\chi,\cL^\dagger, \matI_t \otimes \ddvA)$,
        where $\cL^\dagger = \{\vecz \in \ring^{mt}: \norm{\vecz}_2 \leq \tensorexpansion{t} \cdot \gamma \cdot \sqrt{\varphi m} \cdot \beta \}$ is the leakage set of \cref{lem:mLLWE-from-AdpLLWE-alt}, at a loss of a factor $t$.
        \item By \cref{lem:mLLWE-from-AdpLLWE-alt}, $\Adp\LLWE{}$ with $\params^\dagger$ implies the indistinguishability of the experiments $\mathsf{Exp3}^0$ and $\mathsf{Exp3}^1$ defined therein, and by \cref{lem:core1-alt} the latter implies the claimed indistinguishability. Both steps are tight.
    \end{itemize}
    Putting the advantage bounds in the above lemmas together yields the stated bound. \qed
\end{proof}

\cref{lem:AdpLLWE-IotimesA} below is a standard hybrid argument, reducing
adaptive leaky LWE for matrix $\matA$ to that for matrix $\matI \otimes \matA$. 
It is a direct adaption of \cite[Proposition 1, Eprint]{AC:CinLaiWoo25}.

\begin{lemma}\label{lem:AdpLLWE-IotimesA}
    Let $\ring,n,m,Q,q,\bar{\chi},\chi,\ddvA$ be parametrised by $\secpar$, where $\ddvA$ is a distribution over $\ring_q^{n \times m}$. 
    Let $t \in \poly$,
    $\cL = \set{\vecz \in \ring^{m}: \norm{\vecz}_2 \leq \alpha}$ and $\cL^\dagger = \set{\vecz \in \ring^{mt}: \norm{\vecz}_2 \leq \alpha}$ for some $\alpha > 0$.
    Let 
    \begin{align*}
        \params^* = (\ring,n,m,Q,q,\bar{\chi},\chi,\cL,\ddvA), \qquad \params^\dagger = (\ring,nt,mt,Q,q,\bar{\chi},\chi,\cL^\dagger, \matI_t \otimes \ddvA).
    \end{align*}
    For any PPT $\adv$, there exists a PPT $\bdv$ such that 
    \[
        \Adv_{\params^\dagger,\adv}^{\Adp\LLWE{}}(\secpar) \leq t \cdot \Adv_{\params^*,\bdv}^{\Adp\LLWE{}}(\secpar).
    \]
    In particular, if the $\Adp\elLWE{\params^*}$ assumption holds, then the $\Adp\elLWE{\params^\dagger}$ assumption holds. 
\end{lemma}

\begin{proof}
    For each $i \in \set{0,1,\ldots,t}$, define the following $i$-th hybrid experiment:
    \begin{pchstack}[center]
        \procedure{$\mathsf{Exp}_i(\secparam)$}
        {
            \mat{A} \gets \ddvA ;~ j \coloneqq 0 \\
            \pcfor k \in [t] \pcdo~ \barvec{e}_k \sample \DD_{\ring,\bar{\chi}}^m \\
            \pcfor k \in [1,\ldots,i] \pcdo~
                \vec{b}_k \sample \ring_q^m \\
            \pcfor k \in [i+1,\ldots,t] \pcdo~
                \vec{r}_k \sample \ring_q^n ;~
                \vec{b}_k^\transpose \coloneqq \vec{r}_k^\transpose \mat{A} + \barvec{e}_k^\transpose \bmod q \\
            \vec{b}^\transpose \coloneqq (\vec{b}_1^\transpose, \ldots, \vec{b}_t^\transpose) ;~
            \barvec{e}^\transpose \coloneqq (\barvec{e}_1^\transpose,\ldots,\barvec{e}_t^\transpose) \\
            \pcreturn \adv^{\oracle}(\matA, \vec{b})
        }
        \pchspace
        \procedure{$\oracle(\vecz)$}
        {
            \pcassert j < Q \\ 
            j \coloneqq j+1 \\
            \pcassert \vecz \in \cL^\dagger \\
            \pcparse \begin{pmatrix} \vecz_1 \pclb \vdots \pclb \vecz_t \end{pmatrix} \gets \vecz \pccomment{$\vecz_k \in \ring^{m} ~\forall k \in [t]$} \\
            e \sample \DD_{\ring,\chi};~ \ell \coloneqq \sum_{k \in [t]} \barvec{e}_k^\transpose \cdot \vecz_k + e \\
            \pcreturn \ell
        }
    \end{pchstack}
    It is direct to verify that the $0$-th hybrid is identical to the $\Adp\elLWE{\params^\dagger}$ experiment $0$, as 
    \begin{align*}
        \vec{b}^\transpose =~& (\vec{r}_1^\transpose \mat{A} + \barvec{e}_1^\transpose,~ \ldots,~ \vec{r}_t^\transpose \mat{A} + \barvec{e}_t^\transpose)
        = (\vec{r}_1^\transpose,\ldots,\vec{r}_t^\transpose) (\mat{I}_t \otimes \mat{A}) + \barvec{e}^\transpose \bmod q,
    \end{align*}
    and each oracle answer is 
    \begin{align*}
        \ell =~& \sum_{k \in [t]} \barvec{e}_k^\transpose \cdot \vecz_k + e
        = \barvec{e}^\transpose \cdot \vecz + e. 
    \end{align*}
    Also, the $t$-th hybrid is identical to the $\Adp\elLWE{\params^\dagger}$ experiment $1$.
    In the rest we show the $(i-1)$-th and $i$-th hybrids are computationally indistinguishable under the $\Adp\elLWE{\params^*}$ assumption, for each $i \in [t]$.

    Fix any $i \in [t]$. Given any PPT $\adv$ against distinguishing the $(i-1)$-th and $i$-th hybrids, we construct a PPT $\bdv$ against $\Adp\elLWE{\params^*}$.
    Let $\bdv$ proceed as follows:
    \begin{itemize}
        \item Receive $(\mat{A}, \vecb_i)$ from the $\Adp\elLWE{\params^*}$ challenger.
        \item For each $k \in [1,\ldots,t]\setminus \set{i}$, 
            sample $\barvec{e}_k \sample \DD_{\ring,\bar{\chi}}^m$. 
        \item For each $k \in [1,\ldots,i-1]$, 
            sample random $\vec{b}_k \sample \ring_q^m$. 
        \item For each $k \in [i+1,\ldots,t]$, 
            sample random $\vec{r}_k \sample \ring_q^n$, set $\vec{b}_k^\transpose \coloneqq \vec{r}_k^\transpose \mat{A} + \barvec{e}_k^\transpose \bmod q$.
        \item Set $\vec{b}^\transpose \coloneqq (\vec{b}_1^\transpose, \ldots, \vec{b}_t^\transpose)$.
            Send $(\matA,\vec{b})$ to $\adv$.
        \item Internally maintain a counter $j \coloneqq 0$.
        For each oracle query $\oracle(\vecz)$ made by $\adv$, do the following:
        \begin{itemize}
            \item Assert $j < Q$ and set $j \coloneqq j+1$.
            \item Assert $\vecz \in \cL^\dagger$. Parse $\begin{pmatrix} \vecz_1 \\ \vdots \\ \vecz_t \end{pmatrix} \gets \vecz$.
            \item For each $k \in [1,\ldots,t]\setminus \set{i}$, let $\ell_k \coloneqq \barvec{e}_k^\transpose \cdot \vecz_k$.
            \item Send $\vecz_{i}$ to the oracle $\Leak\oracle$ available from the $\Adp\elLWE{\params^*}$ challenger, receive from which $\ell_i$.
            \item Answer with $\ell \coloneqq \sum_{k \in [t]} \ell_k$.
        \end{itemize}
        \item Eventually receive a bit $b'$ from $\adv$. Return $b'$.
    \end{itemize}
    It is direct to see that $\bdv$ perfectly simulates the view of the $(i-1)$-th and the $i$-th hybrid to $\adv$, 
    if it interacts with the $\Adp\elLWE{\params^*}$ experiment $0$ and $1$, respectively.
    In particular, in either case, 
    $\ell_i = \barvec{e}_i^\transpose \cdot \vecz_i + e$ for error $e \sample \DD_{\ring,\chi}$, hence     
    the simulated oracle answer is 
    \[ 
        \ell = \sum_{k \in [t]} \ell_k = \sum_{k \in [t]} \barvec{e}_k^\transpose \cdot \vecz_k + e.
    \]
    The vector $\vecz_i$ sent from $\bdv$ to the $\Adp\elLWE{\params^*}$ challenger is admissible, 
    i.e.~$\vecz_i \in \cL$,
    since $\vecz \in \cL^\dagger$ implies $\norm{\vecz}_2 \leq \alpha$, and we have $\norm{\vecz_i}_2 \leq \norm{\vecz}_2 \leq \alpha$.
    \qed
\end{proof}

\cref{lem:mLLWE-from-AdpLLWE-alt} below restates $\Adp\LLWE{}$ w.r.t. $\matI \otimes \matA$ in the matrix form used in the sequel, in which the secret is a matrix and the leakage on the error matrix is a noisy bilinear form.

\begin{lemma}\label{lem:mLLWE-from-AdpLLWE-alt}
    Let $\ring,n,m,Q,q,K,t,\bar{\chi},\chi,\gamma,\beta,\ddvA$ be parametrised by $\secpar$,
    where $\ring = \ZZ[\zeta]$ has conductor $\conductor$ and ring degree $\varphi$,
    $n,m,Q,K,t \in \poly$ and $\ddvA$ is a distribution over $\ring_q^{n \times m}$.
    Denote
    \begin{align*}
        \alpha \coloneqq~& \tensorexpansion{t} \cdot \gamma \cdot \sqrt{\varphi m} \cdot \beta, \\
        \params_{\mathsf{E}} =~& (\ring,n,m,Q,q,K,t,\bar{\chi},\chi,\gamma,\beta,\ddvA), \qquad
        \params^\dagger = (\ring,nt,mt,QK,q,\bar{\chi},\chi,\cL^\dagger, \matI_t \otimes \ddvA),
    \end{align*}
    where $\tensorexpansion{t} \coloneqq 1$ if $\conductor$ is a power of two and $\tensorexpansion{t} \coloneqq \sqrt{t}$ otherwise, and $\cL^\dagger = \{\vecz \in \ring^{mt}: \norm{\vecz}_2 \leq \alpha\}$.
    For any PPT $\adv$, define the following experiment $\mathsf{Exp3}^b_{\params_{\mathsf{E}},\adv}$:
    \begin{pchstack}[center]
        \procedure{$\mathsf{Exp3}^b_{\params_{\mathsf{E}}, \adv}(\secparam)$}
        {
            i \coloneqq 0 \\
            \matA \sample \ddvA ;~
            \barmat{E} \sample \ddv_{\ring,\bar{\chi}}^{t \times m} \\
            \pcif b=0 \pcthen \\
                \pcind \matS \sample \ring_q^{t \times n} ;~
                \matB \coloneqq \matS \matA + \barmat{E} \bmod q \\
            \pcif b=1 \pcthen \\
                \pcind \matB \sample \ring_q^{t \times m} \\
            \pcreturn \adv^{\Leak\oracle}(\matA,\matB)
        }
        \pchspace
        \procedure{$\Leak\oracle(\matY \in \ring^{K \times t}, \vecx \in \ring^m)$}
        {
            \pcassert i < Q \\
            i \coloneqq i+1 \\
            \pcassert \norm{\matY}_{2,2} \leq \gamma \land \norm{\vecx}_\infty \leq \beta \\
            \vece \sample \ddv_{\ring,\chi}^{K} \\
            \vecl \coloneqq \matY \barmat{E} \vecx + \vece \bmod q \\
            \pcreturn \vecl
        }
    \end{pchstack}
    and denote the advantage of $\adv$ in distinguishing the above experiments by
    \begin{align*}
        \Adv_{\params_{\mathsf{E}},\adv}^{\mathsf{Exp3}} \coloneqq
        \abs{ \prob{\mathsf{Exp3}^0_{\params_{\mathsf{E}}, \adv}(\secparam) = 1} - \prob{\mathsf{Exp3}^1_{\params_{\mathsf{E}}, \adv}(\secparam) = 1}}.
    \end{align*}
    For any PPT $\adv$, there exists a PPT $\bdv$ such that
    \[
        \Adv_{\params_{\mathsf{E}},\adv}^{\mathsf{Exp3}}(\secpar) = \Adv_{\params^\dagger,\bdv}^{\Adp\LLWE{}}(\secpar).
    \]
    In particular, if the $\Adp\LLWE{\params^\dagger}$ assumption holds,
    then the experiments $\mathsf{Exp3}^0_{\params_{\mathsf{E}},\adv}$ and $\mathsf{Exp3}^1_{\params_{\mathsf{E}},\adv}$ are computationally indistinguishable.
\end{lemma}

\begin{proof}
    Given a PPT $\adv$ distinguishing $\mathsf{Exp3}^0$ and $\mathsf{Exp3}^1$, we construct a PPT $\bdv$ against $\Adp\elLWE{\params^\dagger}$, as follows:
    \begin{itemize}
        \item Receive $(\matI_t \otimes \matA, \barvec{b}) \in \ring_q^{nt \times mt} \times \ring_q^{mt}$ from the $\Adp\elLWE{}$ challenger.
        Parse $\barvec{b}^\tr = (\barvec{b}_1^\tr, \ldots, \barvec{b}_t^\tr)$ with $\barvec{b}_j \in \ring_q^m$ for $j \in [t]$, set $\matB \in \ring_q^{t \times m}$ to be the matrix whose $j$-th row is $\barvec{b}_j^\tr$, and pass $(\matA,\matB)$ to $\adv$.
        \item For each query $(\matY,\vecx)$ made by $\adv$: 
        \begin{itemize}
            \item Assert $i < Q$ and set $i \coloneqq i+1$. Assert $\norm{\matY}_{2,2} \leq \gamma$ and $\norm{\vecx}_\infty \leq \beta$.
            \item Let $\vecy_k^\tr$ denote the $k$-th row of $\matY$ and $\vecz_k \coloneqq \vecy_k \otimes \vecx \in \ring^{mt}$, i.e.~$\vecz_k$ consists of $t$ blocks of length $m$, the $j$-th being $y_{k,j} \cdot \vecx$. 
            \item For each $k \in [K]$, query the oracle $\Leak\oracle(\vecz_k)$ available from the $\Adp\elLWE{}$ challenger, and receive an answer $\ell_k$. Answer $\adv$ with $\vecl \coloneqq (\ell_1,\ldots,\ell_K)^\tr$.%
        \end{itemize}
        \item Eventually receive $b'$ from $\adv$, return $b'$.
    \end{itemize}
    Write $\barvec{e}^\tr = (\barvec{e}_1^\tr,\ldots,\barvec{e}_t^\tr)$, where each $\barvec{e}_j \in \ring^m$, for the error of LWE samples in the $\Adp\elLWE{}$ experiment, and let $\barmat{E} \in \ring^{t \times m}$ be the matrix whose $j$-th row is $\barvec{e}_j^\tr$.

    In the $\Adp\elLWE{}$ experiment $0$, we have $\barvec{b}^\tr = \vecr^\tr (\matI_t \otimes \matA) + \barvec{e}^\tr \bmod q$ for a uniform $\vecr \sample \ring_q^{nt}$. Writing $\vecr^\tr = (\vecr_1^\tr,\ldots,\vecr_t^\tr)$ where each $\vecr_j \in \ring_q^n$, then we have 
    \begin{align*}
        \barvec{b}_j^\tr = \vecr_j^\tr \matA + \barvec{e}_j^\tr \bmod q ~~\forall j \in [t], \qquad \text{equivalently}\quad \matB = \matS \matA + \barmat{E} \bmod q 
    \end{align*}
    where $\matS \in \ring_q^{t \times n}$ is uniform, whose $j$-th row is $\vecr_j^\tr$. 
    In the $\Adp\elLWE{}$ experiment $1$, the matrix $\matB$ is uniform. 
    In either case, $\barmat{E} \sample \ddv_{\ring,\bar{\chi}}^{t \times m}$, as in $\mathsf{Exp3}$.
    
    For an oracle query $(\matY,\vecx)$, the answer is $\vecl = (\ell_1,\ldots,\ell_K)^\tr$ where, for each $k \in [K]$,  
    \[
        \ell_k = \barvec{e}^\tr \cdot \vecz_k + e_k
        = \sum_{j \in [t]} \barvec{e}_j^\tr \cdot (y_{k,j} \cdot \vecx) + e_k
        = \sum_{j \in [t]} y_{k,j} \cdot (\barvec{e}_j^\tr \vecx) + e_k
        = \vecy_k^\tr \barmat{E} \vecx + e_k,
    \]
    with $e_k \sample \ddv_{\ring,\chi}$. Vertically concatenating, we have $\vecl = \matY \barmat{E} \vecx + \vece$ as in $\mathsf{Exp3}$.
    
    Hence, $\bdv$ perfectly simulates $\mathsf{Exp3}^0$, resp.~$\mathsf{Exp3}^1$, when interacting with the $\Adp\elLWE{}$ experiment $0$, resp. experiment $1$.
    Note that the $K$ answers $\ell_1,\ldots,\ell_K$ have independent errors even if $\matY$ has two equal rows, since $\Leak\oracle$ in \cref{fig:leaky-lwe} samples a fresh error on every query.

    \medskip\noindent
    \underline{Query count}:
    $\bdv$ makes $K$ queries to $\Leak\oracle$ per query of $\adv$, hence at most $QK$ in total by the assertion $i < Q$ in $\Leak\oracle$ of $\mathsf{Exp3}^b$.

    \medskip\noindent
    \underline{Norm bound}:
    It remains to check that the queries of $\bdv$ satisfies the norm bound constraint.
    For each $k \in [K]$, the entries of $\vecz_k$ are $y_{k,j} \cdot x_i$ for $(j,i) \in [t] \times [m]$.
    Regrouping them by the columns of $\matY$, and writing $\vecu_j \in \ring^t$ for the $j$-th unit vector,
    \[
        \sum_{k \in [K]} \norm{\vecz_k}_2^2
        = \sum_{i \in [m]} \sum_{j \in [t]} \norm{\matY \cdot (x_i \cdot \vecu_j)}_2^2
        \leq \sum_{i \in [m]} \sum_{j \in [t]} \norm{\matY}_{2,2}^2 \cdot \norm{x_i}_2^2
        = t \cdot \norm{\matY}_{2,2}^2 \cdot \norm{\vecx}_2^2,
    \]
    hence, for every $k \in [K]$,
    \begin{align*}
        \norm{\vecz_k}_2 \leq~& \sqrt{t} \cdot \norm{\matY}_{2,2} \cdot \norm{\vecx}_2 \\
        \leq~& \sqrt{t} \cdot \gamma \cdot \sqrt{\varphi m} \cdot \norm{\vecx}_\infty \\
        \leq~& \sqrt{t} \cdot \gamma \cdot \sqrt{\varphi m} \cdot \beta.
    \end{align*}
    
    If $\conductor$ is a power of two, the factor $\sqrt{t}$ in the above bound can be dropped, as follows.
    The idea is to measure $\vecz_k$ in the canonical embedding, where the product of two field elements has multiplicative norm, and then to relate the resulting bound back to $\norm{\matY}_{2,2}$.
    We use the following three facts.
    \begin{enumerate}[label=(\roman*)]
        \item Extend each canonical embedding $\sigma_\iota$ of $\field$ (\cref{sec:prelim}) entrywise to vectors and matrices over $\field$ (so that e.g.~$\sigma_\iota(\matY) \in \CC^{K \times t}$).
        Since $\sigma_\iota$ is a ring homomorphism, $\sigma_\iota(\matM \vecv) = \sigma_\iota(\matM) \cdot \sigma_\iota(\vecv)$ and $\sigma_\iota(\veca \otimes \vecw) = \sigma_\iota(\veca) \otimes \sigma_\iota(\vecw)$ for any matrix $\matM$ and vectors $\vecv, \veca, \vecw$ over $\field$ of compatible dimensions.
        \item As $\conductor$ is a power of two, $\sigma$ is a scaled $\ell_2$-isometry, $\norm{\sigma(w)}_2 = \sqrt{\varphi} \cdot \norm{w}_2$ for any $w \in \field$ (\cref{sec:prelim}), hence $\norm{\vecv}_2^2 = \frac{1}{\varphi} \sum_\iota \norm{\sigma_\iota(\vecv)}_2^2$ for any vector $\vecv$ over $\field$.
        \item By \cref{lem:induced-norm-embeddings}, $\max_\iota \norm{\sigma_\iota(\matM)}_{2,2} = \norm{\matM}_{2,2}$ for any matrix $\matM$ over $\field$.
    \end{enumerate}
    Given these facts, for each $k \in [K]$,
    \begin{align*}
        \norm{\vecz_k}_2^2
        =~& \frac{1}{\varphi} \sum_\iota \norm{\sigma_\iota(\vecz_k)}_2^2 \quad\pccomment{by (ii)}\\
        =~& \frac{1}{\varphi} \sum_\iota \norm{\sigma_\iota(\vecy_k)}_2^2 \cdot \norm{\sigma_\iota(\vecx)}_2^2 \quad\pccomment{by (i) and $\norm{\veca \otimes \vecw}_2 = \norm{\veca}_2 \cdot \norm{\vecw}_2$}\\
        \leq~& \max_\iota \norm{\sigma_\iota(\matY)}_{2,2}^2 \cdot \frac{1}{\varphi} \sum_\iota \norm{\sigma_\iota(\vecx)}_2^2 \\
        =~& \max_\iota \norm{\sigma_\iota(\matY)}_{2,2}^2 \cdot \norm{\vecx}_2^2 \quad\pccomment{by (ii)} \\
        \leq~& \norm{\matY}_{2,2}^2 \cdot \norm{\vecx}_2^2 \quad\pccomment{by (iii)}
    \end{align*}
    where the first inequality holds because $\sigma_\iota(\vecy_k)^\tr$ is the $k$-th row of $\sigma_\iota(\matY)$ and the $\ell_2$-norm of a row of a complex matrix is at most its $\ell_2$-induced norm.
    Taking square roots and using $\norm{\matY}_{2,2} \leq \gamma$ and $\norm{\vecx}_2 \leq \sqrt{\varphi m} \cdot \beta$ as before yields 
    \[
        \norm{\vecz_k}_2 \leq \norm{\matY}_{2,2} \cdot \norm{\vecx}_2 \leq \gamma \cdot \sqrt{\varphi m} \cdot \beta, 
    \]
    i.e.~the bound of the general case but without the factor $\sqrt{t}$.
    \qed
\end{proof}

\cref{lem:core1-alt} below gives a reduction from $\mathsf{Exp3}$ to the distinguishing problem defined in \cref{thm:core1-alt}.

\begin{lemma}\label{lem:core1-alt}
    Let $\ring,n,m,Q,q,K,t,t_2,\bar{\chi},\chi,\gamma,\beta,\ddvA$ be parametrised by $\secpar$, where $\ddvA$ is a distribution over $\ring_q^{n \times m}$. Denote
    \begin{align*}
        \params = (\ring,n,m,Q,q,K,t,t_2,\bar{\chi},\chi,\gamma,\beta,\ddvA), \qquad
        \params_{\mathsf{E}} = (\ring,n,m,Q,q,K,t,\bar{\chi},\chi,\gamma,\beta,\ddvA).
    \end{align*}
    Consider the experiment $\mathsf{Exp2}_{\params,\adv}^b$ defined in \cref{thm:core1-alt} and the advantage $\Adv_{\params,\adv}^{\mathsf{Exp2}}$ defined in \cref{eq:core1-alt-advantage}, and let $\mathsf{Exp3}$ denote the experiment of \cref{lem:mLLWE-from-AdpLLWE-alt} with parameters $\params_{\mathsf{E}}$.
    For any PPT $\adv$, there exists a PPT $\bdv$ such that
    \[
        \Adv_{\params,\adv}^{\mathsf{Exp2}}(\secpar) = \Adv_{\params_{\mathsf{E}},\bdv}^{\mathsf{Exp3}}(\secpar).
    \]
\end{lemma}

\begin{proof}
    Given a PPT $\adv$ distinguishing $\mathsf{Exp2}^0$ and $\mathsf{Exp2}^1$, we construct a PPT $\bdv$ distinguishing $\mathsf{Exp3}^0$ and $\mathsf{Exp3}^1$, as follows:
    \begin{itemize}
        \item Receive $(\matA,\matB) \in \ring_q^{n \times m} \times \ring_q^{t \times m}$ from the $\mathsf{Exp3}$ challenger.
        \item Pass $\matA$ to $\adv$, and receive from which $\matM = (\matM_1,\matM_2) \in \ring^{K \times (t+t_2)}$. Assert $\norm{\matM_1}_{2,2} \leq \gamma$.
        \item Sample $\matS_2 \sample \ring_q^{t_2 \times n}$, set $\vecb^\tr \coloneqq \vecu^\tr \matB \bmod q$ where $\vecu^\tr \coloneqq (1,0,\ldots,0) \in \ring^{1 \times t}$ is the first unit vector, and pass $(\matA,\vecb,\matS_2)$ to $\adv$.
        \item Maintain a counter $i \coloneqq 0$. Answer each oracle query $\oracle_b(\vecx)$ of $\adv$ as follows: assert $i < Q$, set $i \coloneqq i+1$, and assert $\norm{\vecx}_\infty \leq \beta$. Then query $\Leak\oracle(\matM_1,\vecx)$ of $\mathsf{Exp3}$ to obtain $\vecl$, compute
        \[
            \vecd \coloneqq \matM_1 \cdot \matB \vecx - \vecl + \matM_2 \cdot \matS_2 \matA \vecx \bmod q,
        \]
        and answer with $\vecd$.
        \item Eventually receive a bit $b'$ from $\adv$, return $b'$.
    \end{itemize}
    Below we show that $\bdv$ perfectly simulates the view of either $\mathsf{Exp2}^0$ or $\mathsf{Exp2}^1$ to $\adv$,
    if it interacts respectively with $\mathsf{Exp3}^0$ or $\mathsf{Exp3}^1$, thus the lemma follows.

    In both cases, write $\matW \coloneqq \matB - \barmat{E} \bmod q$ for the error matrix $\barmat{E} \sample \ddv_{\ring,\bar{\chi}}^{t \times m}$ of $\mathsf{Exp3}$, and $\barvec{e}^\tr \coloneqq \vecu^\tr \barmat{E}$ for its first row, which follows $\ddv_{\ring,\bar{\chi}}^m$. By definition of the $\Leak\oracle$ answer in $\mathsf{Exp3}$, we have $\vecl = \matM_1 \cdot \barmat{E} \vecx + \vece$ for a fresh $\vece \sample \ddv_{\ring,\chi}^K$, hence for each query $\vecx$ the simulated oracle answer is
    \begin{align*}
        \vecd
        =~& \matM_1 \cdot \matB \vecx - (\matM_1 \cdot \barmat{E} \vecx + \vece) + \matM_2 \cdot \matS_2 \matA \vecx \\
        =~& \matM_1 \cdot \matW \vecx - \vece + \matM_2 \cdot \matS_2 \matA \vecx
        \equiv \matM \cdot \begin{bmatrix} \matW \\ \matS_2 \matA \end{bmatrix} \vecx + \vece \bmod q,
    \end{align*}
    where the last $\equiv$ holds since the discrete Gaussian over $\ring$ is symmetric around zero. Moreover, the vector simulated to $\adv$ is $\vecb^\tr = \vecu^\tr \matB = \vecu^\tr \matW + \barvec{e}^\tr \bmod q$. It remains to identify $\matW$ in the two cases.

    \medskip\noindent
    \underline{If $\bdv$ interacts with $\mathsf{Exp3}^0$}:
    Here $\matB = \matS \matA + \barmat{E}$ for a uniform $\matS \sample \ring_q^{t \times n}$, so $\matW = \matS \matA$. Thus $\vecb^\tr = \vecu^\tr \matS \matA + \barvec{e}^\tr$ and $\vecd \equiv \matM \cdot \begin{bmatrix} \matS \matA \\ \matS_2 \matA \end{bmatrix} \vecx + \vece$, which is the view of $\mathsf{Exp2}^0$ with $\matS_1 \coloneqq \matS$.

    \medskip\noindent
    \underline{If $\bdv$ interacts with $\mathsf{Exp3}^1$}:
    Here $\matB \sample \ring_q^{t \times m}$ is uniform and independent of $\barmat{E}$, so $\matW = \matB - \barmat{E}$ is uniform over $\ring_q^{t \times m}$ and independent of $\barmat{E}$. Thus $\vecb^\tr = \vecu^\tr \matW + \barvec{e}^\tr$ and $\vecd \equiv \matM \cdot \begin{bmatrix} \matW \\ \matS_2 \matA \end{bmatrix} \vecx + \vece$, which is the view of $\mathsf{Exp2}^1$ with $\barmat{S}_1 \coloneqq \matW$.

    \medskip\noindent
    \underline{Norm bound and query count}:
    The tuple $(\matM_1,\vecx)$ sent by $\bdv$ to the $\mathsf{Exp3}$ challenger is admissible, since $\norm{\matM_1}_{2,2} \leq \gamma$ and $\norm{\vecx}_\infty \leq \beta$ are exactly the assertions of $\mathsf{Exp2}^b$.
    Finally, $\bdv$ queries $\Leak\oracle$ exactly once per query of $\adv$, hence at most $Q$ times by the assertion $i < Q$ in $\oracle_b$ of $\mathsf{Exp2}^b$.
    \qed
\end{proof}

}{
  \section{Ethics considerations}
  None.

  \newpage 

  \section{Meta-Review}

  The following meta-review was prepared by the program committee for the 2027
  IEEE Symposium on Security and Privacy (S\&P) as part of the review process as
  detailed in the call for papers.

  \subsection{Summary}
  The paper presents a method for turning an FHE scheme with a linear decryption function into an asynchronous threshold-FHE scheme. While this has already been done, the main innovation lies in the threshold-decryption method. Conceptually, the target threshold-FHE ciphertext is first (key+scheme)-switched to a threshold-PKE scheme, then the threshold decryption of this tPKE scheme is used. The paper presents a security argument (sketched, from a full version) for the construction's claimed properties, and provides an open-source implementation and evaluation. 

  \subsection{Scientific Contributions}
  \begin{itemize}
    \item Provides a Valuable Step Forward in an Established Field
  \end{itemize}

  \subsection{Reasons for Acceptance}
  \begin{enumerate}
  \item The paper provides a valuable step forward in the field of threshold-FHE, by proposing an asynchronous scheme that requires less communication and is more scalable than the state of the art. The construction is presented formally, with a clear model, security properties, and a security proof sketch (extracted from the paper's full version).
  \item The paper is accompanied by an open-source implementation, and presents concrete FHEVM use cases.
  \end{enumerate}

  \subsection{Noteworthy Concerns} 
  \begin{enumerate} 
    \item The core construction assumes a trusted dealer for the secret keys. The feasibility of a distributed key-generation protocol is left as future work.
    \item The comparison to Zyskind et al. \cite{CCS:ZZLP25} lacks a multi-dimensional characterization of the involved trade-offs. Notably, the chosen evaluation profile (key-dependent phase latency) discards the latency incurred by `PreProc`. Yet this pre-processing is the main latency bottleneck of the proposed construction, and its characteristics differ from the pre-processing of \cite{CCS:ZZLP25} (public vs. secret-dependent, ciphertext-dependent vs. fully offline). Hence, it remains an open question in which deployment scenarios \cite{CCS:ZZLP25} outperforms Tulva.
    \item The security proof is provided as a high-level sketch extracted from the full version. Only this proof sketch was peer-reviewed.
  \end{enumerate}
}

\end{document}
